\documentclass[a4paper,11pt]{article}

\usepackage{jheppub}

\usepackage[numbers,sort&compress]{natbib} 
\usepackage{url}
\usepackage{amsmath}
\usepackage{rotating}

\usepackage{multirow}

\usepackage[T1]{fontenc} 

\usepackage{float, array, xspace, amscd, amsthm}
\usepackage{fancyhdr}
\usepackage{longtable}

\newtheorem{theorem}{Theorem}
\newtheorem{lemma}{Lemma}
\newtheorem{proposition}{Proposition}
\newtheorem{definition}{Definition}

\allowdisplaybreaks

\usepackage[usenames,dvipsnames,svgnames,table,x11names]{xcolor}
\definecolor{MagentaXD}{RGB}{204, 48, 152}
\definecolor{MagentaXDdetail}{RGB}{150, 79, 126}
\definecolor{GreenMAF}{RGB}{28, 112, 46}
\definecolor{GreenMAFdetail}{RGB}{80, 117, 88}
\definecolor{detail}{RGB}{110,110,110}
\definecolor{quantumviolet}{HTML}{53257F} 
\definecolor{quantumgray}{HTML}{555555} 
\definecolor{quantumgreen}{HTML}{007474} 
\definecolor{quantumblue}{HTML}{002366} 
\definecolor{quantumpurple}{HTML}{66023C} 
\definecolor{quantumdarkviolet}{HTML}{5D3954} 

\newcommand{\Cbb}{\mathbb{C}}

\newcommand{\Irr}{\operatorname{Irr}}
\newcommand{\Tr}{\operatorname{Tr}}
\newcommand{\End}{\operatorname{End}}

\definecolor{nblue}{rgb}{0.2,0.2,0.7}
\definecolor{ngreen}{rgb}{0.2,0.6,0.2}
\definecolor{nred}{rgb}{0.7,0.2,0.2}
\definecolor{nblack}{rgb}{0,0,0}

\newcommand{\be}{\begin{equation}}
	\newcommand{\ee}{\end{equation}}
\def\bea#1\eea{\begin{align}#1\end{align}}

\usepackage{tikz}
\usetikzlibrary{positioning,intersections}
\usetikzlibrary{calc}
\usetikzlibrary{shapes.geometric}
\usetikzlibrary{tqft} 
\usepackage{braids}
\usepackage{tikz-cd}

\newif\ifcomments
\commentstrue

\newif\ifdetails
\detailstrue

\usepackage{upgreek}
\usepackage[normalem]{ulem}
\usepackage{mathtools}
\usepackage{datetime}
\usepackage{cancel}

\usepackage[all,cmtip,knot]{xy}
\usepackage{xypic}

\newcommand{\orcid}[1]{\href{https://orcid.org/#1}{\includegraphics[width=8pt]{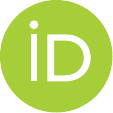}}}

\usepackage{amsmath,amsfonts,amsthm}

\theoremstyle{definition}
\newtheorem{example}{Example}

\theoremstyle{remark}
\newtheorem{remark}{Remark}[section]

\theoremstyle{remark}

\usepackage{bm,latexsym,galois,euscript,dsfont}
\newcommand\EA{\EuScript{A}}
\newcommand\EB{\EuScript{B}}
\newcommand\EC{\EuScript{C}}
\newcommand\ED{\EuScript{D}}
\newcommand\EE{\EuScript{E}}
\newcommand\EM{\EuScript{M}}
\newcommand\EN{\EuScript{N}}
\newcommand\EP{\EuScript{P}}

\newcommand{\Acal}{\mathcal{A}}
\newcommand{\Bcal}{\mathcal{B}}

\newcommand{\Hilb}{\mathsf{Hilb}}
\newcommand\Fun{\mathsf{Fun}}
\newcommand\Vect{\mathsf{Vect}}

\newcommand\Rep {\mathsf{Rep}}

\newcommand\id {\mathrm{id}}
\newcommand\Hom {\mathrm{Hom}}
\newcommand{\one}{\mathbf{1}}

\newcommand\Mod{\mathsf{Mod}}

\usepackage{mathtools,amssymb,scalerel}

\newcommand\biopencrossl{%
	\mathrel{\scalerel*{>\kern-.4\LMpt\joinrel\blacktriangleleft}{x}}}
\newcommand\biopencrossr{%
	\mathrel{\scalerel*{\blacktriangleright\joinrel\kern-.4\LMpt<}{x}}}

\settimeformat{ampmtime}

\begin{document}

\title{Weak Hopf symmetry and tube algebra of the generalized multifusion string-net model}

\author[a,b]{Zhian Jia\orcid{0000-0001-8588-173X},}

\author[c,d]{Sheng Tan\orcid{0009-0008-3318-9942},}

\author[a,b]{Dagomir Kaszlikowski}

\affiliation[a]{Centre for Quantum Technologies, National University of Singapore, Singapore 117543, Singapore}
\affiliation[b]{Department of Physics, National University of Singapore, Singapore 117543, Singapore}
\affiliation[c]{Beijing Institute of Mathematical Sciences and Applications, Beijing, 101408, China}
\affiliation[d]{Yau Mathematical Sciences Center, Tsinghua University, Beijing, 100084, China}

\emailAdd{giannjia@foxmail.com}
\emailAdd{tan296@bimsa.cn}
\emailAdd{phykd@nus.edu.sg}

\abstract{
We investigate the multifusion generalization of string-net ground states and lattice Hamiltonians, delving into its associated weak Hopf symmetry.  For the multifusion string-net, the gauge symmetry manifests as a general weak Hopf algebra, leading to a reducible vacuum string label; the charge symmetry, serving as a quantum double of gauge symmetry, constitutes a connected weak Hopf algebra. This implies that the associated topological phase retains its characterization by a unitary modular tensor category (UMTC). The bulk charge symmetry can also be captured by a weak Hopf tube algebra. We offer an explicit construction of the weak Hopf tube algebra structure and thoroughly discuss its properties. The gapped boundary and domain wall models are extensively discussed, with these $1d$ phases characterized by unitary multifusion categories (UMFCs). We delve into the gauge and charge symmetries of these $1d$ phases, as well as the construction of the boundary and domain wall tube algebras. Additionally, we illustrate that the domain wall tube algebra can be regarded as a cross product of two boundary tube algebras. 
As an application of our model, we elucidate how to interpret the defective string-net as a restricted multifusion string-net.
}

\maketitle


\section{Introduction}

Exotic quantum phases of matter have attracted substantial attention in recent decades \cite{sachdev1999quantum, Wen2004, zeng2015quantum,simon2023topological}. Apart from their fundamental significance, these phases hold potential applications in quantum information and computation tasks, such as the development of robust topological quantum error-correction codes (QECC) \cite{Dennis2002topological, Terhal2015quantum} and topological quantum computation (TQC) \cite{Kitaev2003, freedman2002modular, Nayak2008,wang2010topological,pachos2012introduction}.
There are roughly three types of quantum phases at zero temperature \cite{chatterjee2023emergent,Wen2023emergent}: (i) Gapped phase, characterized by the absence of low-energy excitations, with all excitations exhibiting an energy gap; (ii) Gapless phase with finite low-energy modes, exemplified by Dirac or Weyl semimetals, superfluids, etc.; (iii) Gapless phase with infinite low-energy modes, as observed in Fermi metals, Bose metals, etc.
In this work, we will focus on the gapped phase.

In a gapped quantum phase, there exist two types of excitations: local excitations, which can be created or annihilated by local operators, and topological excitations, which cannot. A gapped phase featuring topological excitations is termed a topologically ordered phase.
The universal low-energy and long-distance properties of gapped phases of quantum matter are described by topological quantum field theory (TQFT) \cite{witten1988topological,dijkgraaf1990topological,Wen2004,wang2010topological,pachos2012introduction,simon2023topological,zeng2015quantum}.
Topological phases display unique attributes like long-range entanglement \cite{zeng2015quantum}, fractional statistics \cite{wilczek1990fractional}, and topologically protected ground state degeneracy \cite{Wen1989,Wen1990,Kitaev2003,Levin2005}. These characteristics distinguish them from conventional models, notably deviating from Landau's spontaneous symmetry breaking framework \cite{landau1937theorie}.


Topologically ordered phases are termed anomaly-free if they can be realized by a local lattice Hamiltonian in the same dimension \cite{kong2014braided,Wen2015classifying}. Otherwise, they are referred to as anomalous.
In the context of $2d$ fusion topologically ordered phases\,\footnote{Here by fusion we mean the vacuum sector of the topological phase is simple.}, it is well known that they are classified, up to $E_8$ quantum Hall states, by the unitary modular tensor categories (UMTCs) representing the topological excitations in the bulk \cite{KITAEV2006}.
The topological excitations in the bulk are point-like quasi-particles that cannot be created or annihilated by local operators from the vacuum sector or its direct sums. They exhibit fusion, splitting into a direct sum of other topological excitations, and braiding interactions among each other.
The vacuum sector $\one$ is the tensor unit of the UMTC $\EB$. 
The  $2d$ fusion topologically ordered phases can also be classified into two families depending on its chiral central charge $c=c_R-c_L$ \cite{Kane1997quantized}: (i) chiral $c\neq 0$, where time-reversal symmetry is broken, and (ii) non-chiral $c=0$, where time-reversal symmetry is preserved.
In the UMTC framework, a non-chiral UMTC denoted as $\EB$ is defined such that there exists a unitary multifusion category (UMFC) $\ED$ where $\EB$ is braided monoidal equivalent to the Drinfeld center $\mathcal{Z}(\ED)$.

Turaev and Viro introduced a state-sum invariant for 3-manifolds rooted in a ribbon fusion category \cite{turaev1992state}. This framework was later broadened by Barrett and Westbury to encompass the spherical fusion category \cite{barrett1996invariants}.
The Levin-Wen string-net model\,\footnote{Throughout this work, when referring to the ``Levin-Wen string-net model'', we specifically mean the string-net model with input data being a unitary fusion category (UFC).} provides a lattice interpretation of the Turaev-Viro-Barrett-Westbury TQFT \cite{Levin2005,kirillov2011stringnet,lan2014topological,Lin2014generalizations,Lin2021generalized,Hahn2020generalized} and a construction of the explicit Hamiltonian. 
The Levin-Wen string-net model is a lattice realization of non-chiral topologically ordered phases.
Kitaev quantum double model is also a crucial lattice model for non-chiral topological phases, which can be regarded as a lattice gauge theory \cite{Kitaev2003}. It has been established that for a $C^*$ Hopf algebra $H$, the Kitaev quantum double model is equivalent to the Levin-Wen string-net model with input UFC $\EC=\Rep(H)$ \cite{Buerschaper2009mapping,Buerschaper2013a,buerschaper2013electric,jia2023boundary}.
However, it is established that not all UFCs can be represented as the representation category of a Hopf algebra. Thus, we need to consider more general quantum symmetries, and it turns out the symmetries are characterized by weak Hopf algebras \cite{BOHM1998weak}.
The quantum symmetry behind a general Levin-Wen string-net model is a weak Hopf algebra in the sense that every UFC can be expressed as a representation category for some weak Hopf algebra \cite{szlachanyi2000finite,ostrik2003module}.
The lattice gauge theory with gauge symmetry characterized by a weak Hopf algebra is established in Refs.~\cite{Jia2023weak,chang2014kitaev}. This prompts the question: what is the general string-net realization of a topological phase with general weak Hopf symmetry? The answer lies in the multifusion string-net model. This is because the representation category of a weak Hopf algebra is generally not a fusion category but a multifusion category.

In this work, our primary focus is on the multifusion generalization of the Levin-Wen string-net model, exploring both the  macroscopic (algebraic) theory and the microscopic lattice Hamiltonian realization.
We will use the correspondence between the multifusion string-net and the weak Hopf lattice gauge theory to study the bulk topological phases, gapped boundary theory and gapped domain wall theory.
The connection and difference between the multifusion string-net and the Levin-Wen string-net will be stressed.

To characterize the topological excitations, we introduce a weak Hopf tube algebra. It is worth noting that the investigation of tube algebras has a long history since Ocneanu's original works \cite{ocneanu1994chirality,ocneanu2001operator}.
A key property of the tube algebra is that its representation category is equivalent to the category of topological excitations of the model \cite{ocneanu1994chirality,ocneanu2001operator,simon2023topological,izumi2000structure,izumi2001structure,muger2003subfactorsI,Christian2023,kawagoe2024levinwen}. Previous formulations of tube algebras lack a coalgebra structure and therefore do not qualify as weak Hopf algebras. In line with this principle, we use the term ``tube algebra'' to denote the local algebra for which the representation category coincides with the topological phase.
For the boundary of the Levin-Wen model, a weak Hopf local algebra is proposed in Ref.~\cite{Kitaev2012boundary}, where the authors utilize topological local moves to derive the simplest configuration of a boundary tube region and construct the weak Hopf structure for this region.
However, for the bulk, performing topological local moves to obtain the simplest configuration results in a Q-algebra \cite{lan2014topological}, which lacks a coalgebra structure.
Our observation suggests that we should treat the bulk excitation as a domain wall defect. By doing so, when performing the topological local moves, we can obtain a tube region with a symmetric structure, enabling us to equip it with a weak Hopf algebra structure.

The gapped boundary and domain wall theory for Levin-Wen string-net model (and Kitaev quantum double model) have been extensively investigated in Refs.~\cite{bravyi1998quantum,Bombin2008family,freedman2001projective,Beigi2011the,Kitaev2012boundary,Levin2013protected,Levin2013protected,Wang2015boundary,Cong2017,wang2020electric,jia2023boundary,Jia2023weak}.
For the multifusion string-net model, we also explore the gapped boundary and gapped domain wall theories from the perspective of weak Hopf symmetry. The weak Hopf tube algebras for the boundary and domain wall are also constructed.
A crucial observation is that the $1d$ phases for the multifusion string-net model are generally not UFCs, but rather UMFCs.
We develop the anyon condensation theory for these scenarios, where the condensation is generally characterized by a pair of algebras in the bulk phase: one condensable algebra and another algebra over this condensable algebra.

As an application of our model, we also demonstrate how to interpret a defective Levin-Wen string-net as a special case of the multifusion string-net following the discussion outlined in Ref.~\cite{Kitaev2012boundary}.

This paper is organized as follows:

In Sec.~\ref{sec:mac}, we delve into the macroscopic (algebraic) theory of the multifusion string-net. We emphasize the weak Hopf symmetry underlying the model and provide a review of the necessary preliminaries of concepts about multifusion categories and weak Hopf algebras in this section.
We introduce the weak Hopf gauge symmetry for the multifusion string-net, based on the fact that a multifusion string-net is equivalent to a weak Hopf lattice gauge theory.
We also introduce the weak Hopf charge symmetry, for which the topological charges are irreducible representations of this weak Hopf algebra. As will be discussed in Sec.~\ref{sec:tube}, we construct the explicit weak Hopf tube algebra, which can be regarded as the charge symmetry of the multifusion string-net.
The boundary tube algebra constructed in Sec.~\ref{sec:bdtheory} can be regarded as a gauge symmetry of the bulk via the bulk-boundary duality.

In Sec.~\ref{Sec:SNmulti}, we give a detailed construction of the multifusion string-net ground state and the lattice Hamiltonian. We show that this lattice Hamiltonian is a local commutative projector (LCP) Hamiltonian and some properties of the Hamiltonian are discussed in detail.

In Sec.~\ref{sec:tube}, we develop the theory of the bulk tube algebra. By treating the bulk excitations as defects on domain walls and performing the topological local moves, we derive a tube algebra. We then construct a $C^*$ weak Hopf algebra structure for this tube algebra. The bulk topological excitations are represented by the tube algebra, and we also explore the Morita theory of the tube algebra.

Sec.~\ref{sec:bdtheory} and Sec.~\ref{sec:walltheory} focus on the gapped boundary and domain wall theories of the multifusion string-net. We demonstrate that these $1d$ phases are characterized by UMFCs. The gauge symmetry and charge symmetry for these phases are thoroughly discussed. Furthermore, we construct the boundary and domain wall tube algebras, which can be interpreted as the charge symmetries of the corresponding phases. Additionally, we illustrate how to derive the $1d$ UMFC topological phase from the bulk UMFC phases by establishing the corresponding anyon condensation theory.

In the last part, we present an application of the multifusion string-net model, demonstrating how to interpret the defective Levin-Wen string-net as a restricted multifusion string-net. In Sec.~\ref{sec:conclusion}, we offer  concluding remarks and discuss potential future directions and questions.

\section{Macroscopic theory of multifusion string-net and weak Hopf symmetry}
\label{sec:mac}

In Landau's symmetry-breaking theory \cite{landau1937theorie}, symmetries are described by groups. Consider a quantum system $(H,\mathcal{V}_{GS})$, where $H$ is the Hamiltonian and $\mathcal{V}_{GS}$ is the ground state subspace. A Hamiltonian symmetry $G_H$ is a group for which there exists a unitary or antiunitary representation $g\mapsto U_g$ such that $[U_g, H]=0$ for all $g\in G_H$.
The ground state symmetry $G_{GS}$ is a group such that $U_g$ stabilizes $\mathcal{V}_{GS}$ for all $g\in G_{GS}$.
It is clear that $G_H\subset G_{GS}$.
The topologically ordered phase is beyond the Landau's symmetry breaking paradigm, as their symmetries are described by quantum groups \cite{bais2003hopf,Kitaev2012boundary,Buerschaper2013a,meusburger2017kitaev,jia2023boundary,Jia2023weak,szlachanyi2023oriented} or fusion categories \cite{Levin2005,lan2014topological,Lin2014generalizations,Lin2021generalized}.

The original Levin-Wen string-net model is constructed from a UFC $\EC$, where the topological excitations are described by the UMTC $\mathcal{Z}(\EC)$, which is the Drinfeld center of $\EC$. It is natural to inquire about the quantum symmetry of the string-net model. The answer to this question is that its symmetry is characterized by a weak Hopf algebra $W$.
For any UFC $\EC$ we can find a \emph{connected weak Hopf algebra} such that $\EC$ is equivalent to the representation category $\Rep(W)$ of $W$ as fusion category \cite{ostrik2003module}.
This enables us to establish an equivalence between the string-net model and lattice gauge theory with gauge symmetry described by a weak Hopf algebra, which has been recently constructed as a generalization of the Kitaev quantum double model \cite{chang2014kitaev, Jia2023weak}.
It is worthy to mentioning that Hopf algebra lattice gauge theory \cite{Buerschaper2013a,meusburger2017kitaev,jia2023boundary} is not enough to describe this equivalence, since there exist some examples of monoidal categories that cannot be regarded as the representation category of Hopf algebras \cite{ostrik2003module}.
This implies that we need to consider weak Hopf lattice gauge theory. Consequently, a natural question arises: what is the string-net model corresponding to a general weak Hopf lattice gauge theory? The answer is the multifusion string-net model, which will be elaborated later in this work.
To summarize, we have a correspondence as shown in Table~\ref{tab:MFSN}.

\begin{table}[t]
\centering \small 
\begin{tabular} {|l|c|c|} 
\hline
   &Weak Hopf lattice gauge theory & Multifusion string-net model   \\ \hline
 Bulk gauge symmetry  & Weak Hopf algebra $W$ & UMFC $\EC=\mathsf{Rep}(W)$  \\ \hline
 Bulk phase  & 
 $\EP=\mathsf{Rep}(D(W))$ & $\mathsf{Fun}_{\EC |\EC}(\EC,\EC) $  \\ \hline
 Boundary gauge symmetry& $W$-comodule algebra $\mathfrak{A}$ & $\EC$-module category ${_{\mathfrak{A}}}\EM={_{\mathfrak{A}}}\mathsf{Mod}$ \\\hline
 Boundary phase & $\EB\simeq {_{\mathfrak{A}}}\mathsf{Mod}_{\mathfrak{A}}^W$ & $\mathsf{Fun}_{\EC}({_{\mathfrak{A}}}\EM,{_{\mathfrak{A}}}\EM)$ \\\hline
Boundary defect & ${_{\mathfrak{B}}}\mathsf{Mod}_{\mathfrak{A}}^W$ & $\mathsf{Fun}_{\EC}({_{\mathfrak{A}}}\EM,{_{\mathfrak{B}}}\EM)$ \\\hline
\end{tabular}
\caption{The correspondence between descriptions of non-chiral topologically ordered phases in the framework of the lattice gauge theory with weak Hopf gauge symmetry and the multifusion string-net model can be elucidated through a dictionary.\label{tab:MFSN}}
\end{table}

\subsection{Unitary multifusion category}

Let us first recall some basic definitions and results about the multifusion category.
See Refs.~\cite{etingof2016tensor,turaev2016quantum,turaev2017monoidal} for more details.
The reader familiar with the notion of the multifusion category might skip this part and get back if necessary.


\begin{itemize}
    \item A $\Cbb$-linear category $\EC$ is a category where all Hom spaces are finite-dimensional vector spaces and the compositions of morphisms  are linear maps with respect to each component. An object $X\in \EC$ is called simple if $\operatorname{End}(X)\cong \Cbb \id_X$.
    $\EC$ is called finite if there is only a finite number of inequivalent simple objects in $\EC$; we usually denote the set of equivalence classes of simple objects of $\EC$ as $\operatorname{Irr}(\EC)$. $\EC$ is called semisimple if every object in $\EC$ is a direct sum of simple objects.
    
    \item A monoidal category $\EC$ is a category equipped with a tensor product $\otimes:\EC\times \EC \to \EC$ and a tensor unit $\one$. The tensor product is associative with the associativity isomorphisms $a_{X,Y,Z}:(X\otimes Y)\otimes Z \to X\otimes (Y\otimes Z)$, which are required to obey the pentagon relation. The left and right unit isomorphisms $l_X: \one \otimes X \to X$ and $r_X: X\otimes \one \to X$ are required to obey the triangle relation.
    A monoidal category $\EC$ is called $\Cbb$-linear if $\EC$ is $\Cbb$-linear and the tensor product of morphisms is $\Cbb$-linear with respect to each component. In our context, a monoidal category that is $\Cbb$-linear will be referred to as a tensor category.

    \item A monoidal category $\EC$ is called rigid if each object $X$ has a left dual ${^{\vee}}\!X$ and a right dual $X^{\vee}$ together with the duality maps\,\footnote{In this work, the notation $d_X$ primarily represents the quantum dimension of an object. However, in this particular context, we employ it to symbolize the ``death'' map. This dual usage should not lead to confusion given the distinct scenarios where each is applied.}:
    \begin{align}
          &  \text{left:}\quad   b_X:\one \to X\otimes X^{\vee},\quad d_X: X^{\vee}\otimes X\to \one, \\
           &  \text{right:}\quad   b'_X:\one \to {^{\vee}}\!X\otimes X,\quad d'_X: X\otimes {^{\vee}}\!X\to \one,
    \end{align}
    which satisfy the following conditions:
    \begin{align}
       & \id_{X^{\vee}}= (d_X\otimes \id_{X}) \comp (\id_{X^{\vee}}\otimes b_X), \quad 
        \id_X=(\id_X\otimes d_X)\comp (b_X\otimes \id_X),\\
     & \id_{{^{\vee}}\!X}=(\id_{{^{\vee}}\!X}\otimes d'_X)\comp (b'_X\otimes \id_{{^{\vee}}\!X}),\quad 
     \id_X=(d'_X\otimes \id_X)\comp (\id_X\otimes b'_X).
    \end{align}
    $\EC$ is called sovereign if ${^{\vee}}\!X \cong X^{\vee}$ for all $X\in \EC$.
    
    \item A $\dagger$-category is a $\Cbb$-linear category $\EC$ equipped with an involutive antilinear contravariant functor $\dagger:\EC \to \EC$ which acts trivially on objects $X^{\dagger}=X$ for all $X\in \EC$. When acting on Hom spaces, it satisfies 
    (i) $f^{\dagger}\in \Hom(Y,X)$ if $f\in \Hom(X,Y)$ and  $f^{\dagger\dagger}=f$ for all $f\in \Hom(X,Y)$; 
    (ii) $(c_1 f_1+c_2 f_2)^{\dagger}=\bar{c}_1 f_1^{\dagger}+\bar{c}_2 f_2^{\dagger}$ for all $c_1,c_2\in \Cbb$ and $f_1,f_2\in \Hom (X,Y)$;
    (iii) ${\id}_X^{\dagger}=\id_X$ for all $X\in \EC$\,\footnote{This condition can be inferred from the others: $\id=\id^{\dagger \dagger} = (\id \comp \id^{\dagger})^{\dagger}=\id\comp \id^{\dagger}=\id^{\dagger}$. 
    But for clarity, we have explicitly stated it.}; (iv) $(f\comp g)^{\dagger}=g^{\dagger}\comp f^{\dagger}$ for all morphisms $f\in \Hom(Y,Z),g\in \Hom(X,Y)$.
    A $\dagger$-category $\EC$ is called unitary if $f^{\dagger}\comp f=0$ implies $f=0$.

    \item A monoidal $\dagger$-category $\EC$ is a category that possesses both a monoidal and a $\dagger$-category structures. It ensures:
   (i) For all $f\in \Hom(X,Y)$ and $g\in \Hom(Z,W)$, $(f\otimes g)^{\dagger}=f^{\dagger}\otimes g^{\dagger}$; 
   (ii) The associativity isomorphisms uphold $a_{X,Y,Z}^{\dagger}=a^{-1}_{X,Y,Z}$, while the left and right isomorphisms abide by $l_X^{\dagger}=l_X^{-1}$ and $r_X^{\dagger}=r_X^{-1}$ respectively.
    
    \item A pivotal structure on a rigid monoidal category $\EC$ is a natural isomorphism $p: \id_{\EC} \to (\cdot)^{\vee\vee}$ which satisfies $p_{X\otimes Y}\cong p_X\otimes p_Y$ for all $X, Y\in \EC$ (here ``$\cong$'' means that they are equal up to some canonical isomorphisms).
    A pivotal category is a rigid monoidal category with a given pivotal structure.
    Via pivotal structure $\{p_X\}$, we can define left and right traces of a morphism $f\in \End(X)$ as (we omit the isomorphism $\one^{\vee}\cong \one$)
    \begin{gather}
       \operatorname{Tr}^L(f)=  b_{X^{\vee}}^{\vee} \comp \gamma_{X^{\vee\vee}\otimes X^{\vee}} \comp 
        (p_{X^{\vee}} \otimes p_X) \comp
        (\id_{X^{\vee}}\otimes f) \comp (\id_{X^{\vee}}\otimes p_{X}^{-1} ) \comp b_{X^{\vee}}, \\
        \Tr^R(f)=   b_{X^{\vee}}^{\vee}  \comp  \gamma_{X\otimes X^{\vee}} \comp (p_{X}\otimes \id_{X^{\vee}})\comp (f\otimes \id_{X^{\vee}}) \comp b_X,
    \end{gather}
    where $\gamma_{X\otimes Y}:X^\vee\otimes Y^\vee \to (X\otimes Y)^\vee$ are natural transformations induced by the rigidity structure. Notice that both $\Tr^L(f),\Tr^R(f)\in \End(\one)$.    
    A spherical category is a pivotal category for which $\Tr^L(f)=\Tr^R(f)$ for all morphisms $f\in \End(X)$.\,\footnote{In the context of multifusion categories, the concept of ``spherical'' can be directly extended. The essential requirement is the equality of the left and right traces as complex numbers. Notably, given that the left and right traces are morphisms in distinct components, a more in-depth analysis can be found in Refs.~\cite{cui2017state,etingof2016tensor}.}
    
    \item In a monoidal $\dagger$-category, given an object $X$ with a right dual $(X^{\vee},b_X,d_X)$, it inherently possesses a left dual $(X^{\vee},d_X^{\dagger},b_X^{\dagger})$, and the converse also holds true. In this case, we denote the dual object as $\bar{X}=X^*:=X^{\vee}={^{\vee}}X$.

    \item For a monoidal category $\EC$, we define $\EC^{\rm rev}$ such that $X\otimes^{\rm rev}Y:=Y\otimes X$. For a braided category $\EC$, $\overline{\EC}$ is defined with reversed braiding morphisms.
\end{itemize}


\begin{definition}
    A multifusion category over $\Cbb$ is a finite semisimple $\Cbb$-linear rigid monoidal category.
    A fusion category is a multifusion category whose tensor unit is simple.
\end{definition}

When $\ED$ is a multifusion category, since $\one$ may not be simple, it can be decomposed as $\one=\oplus_{i\in I} \one_i$ with $\one_i$ simple for all $i\in I$.
It can be proved that $\one_i\otimes \one_j\cong \delta_{i,j}\one_i$ and each $\one_i$ has left and right duals such that $\one_i^{\vee}\cong \one_i \cong {^{\vee}}\one_i$.
Let $\ED_{i,j}:=\one_i \otimes \ED \otimes \one_j$, then we have a decomposition (as additive category)
\begin{equation}
    \ED=\bigoplus_{i,j\in I} \ED_{i,j},
\end{equation}
which will be referred to as the canonical grading of $\ED$.
Every simple object in $\ED$ belongs to some $\ED_{i,j}$.
The tensor product maps $\ED_{i,j}\times \ED_{k,l}$ to zero when $j\neq k$ and to $\ED_{i,l}$ when $j=k$.
The categories $\ED_{i,i}$ are fusion categories with the tensor unit $\one_i$.
If $X\in \ED_{i,j}$ has left or right dual, then this dual must be in $\ED_{j,i}$.

It has been demonstrated in Refs.~\cite{Reutter2023uniqueness, ciamprone2023weak} that the unitary structure of every $C^*$ fusion category, if it exists, must be unique (up to natural isomorphisms). However, for $C^*$ multifusion categories, there may exist inequivalent unitary structures on them.

\begin{example}[UMFC $\mathsf{Mat}_n$]
    A typical example of UMFC is the one constructed from $n\times n$ matrices $\mathsf{Mat}_n=\oplus_{i,j\in I} \mathsf{M}_{i,j}$. The index set $I=\{1,\cdots,n\}$, and each component $\mathsf{M}_{i,j}$ contains exactly one simple object $\mathds{I}_{i,j}$ which can be regarded as the matrix basis $\mathds{I}_{i,j}=|i\rangle \langle j|$.
    The tensor product obeys the matrix multiplication rule $\mathds{I}_{i,j}\otimes \mathds{I}_{k,l}=\delta_{j,k}\mathds{I}_{i,l}$.
    Each diagonal component $\mathsf{M}_{i,i}$ is a UFC equivalent to $\mathsf{Vect}$, the category of all finite-dimensional vector spaces.
    The rigidity is given by $\mathds{I}_{i,j}^*=\mathds{I}_{j,i}$, and all structural isomorphisms are the identity map.
\end{example}

\vspace{1em}
\emph{Drinfeld center of multifusion category.} ---
The Drinfeld center $\mathcal{Z}(\ED)$ of a multifusion category $\ED$ is defined as follows. The objects of $\mathcal{Z}(\ED)$  are pairs $(X,c_{X,\bullet})$ with $c_{X,\bullet}=\{c_{X,Y}:X\otimes Y\xrightarrow{\sim} Y\otimes X\}_{Y\in \ED}$ called half-braiding.
A morphism $f$ between $(X,c_{X,\bullet})$ and $(Y,c_{Y,\bullet})$ is a morphism $f:X\to Y$ such that $(\id_Z\otimes f)\comp c_{X,Z}= c_{Y,Z}\comp (f\otimes \id_Z)$. The Drinfeld center $\mathcal{Z}(\ED)$ is a braided category. It is well known that the Drinfeld center $\mathcal{Z}(\ED)$ is equivalent to the $\ED|\ED$-bimodule functor category $\Fun_{\ED|\ED}(\ED,\ED)$ as braided monoidal categories \cite{Kong2018the}.
If $\ED$ is rigid, then $\mathcal{Z}(\ED)$ is also rigid. Similarly, if $\ED$ is pivotal, then $\mathcal{Z}(\ED)$ is also pivotal. Moreover, the following result holds:

\begin{lemma} \label{prop:UMFCcenter}For an indecomposable multifusion category $\ED=\oplus_{i,j\in I}\ED_{i,j}$, we have
\begin{enumerate}
    \item The Drinfeld center of the diagonal component $\ED_{i,i}$ are braided monoidal equivalent, \emph{viz.}, $\mathcal{Z}(\ED_{i,i})\simeq \mathcal{Z}(\ED_{j,j})$ for all $i,j\in I$. 
    \item The Drinfeld center $\mathcal{Z}(\ED)$ of $\ED$ is braided monoidal equivalent to $\mathcal{Z}(\ED_{i,i})$ for all $i \in I$.
\end{enumerate}
\end{lemma}

\begin{proof}
   The proofs of statements 1 and 2 can be found in Refs.~\cite[Theorem 2.4]{chang2015enriching} and \cite[Theorem 2.5.1]{Kong2018the}, respectively.
\end{proof}

Notice that the non-chiral topological phase is usually defined as a UMTC $\EB$ for which there exists a UFC $\EC$ such that $\EB\simeq \mathcal{Z}(\EC)$ as braided monoidal categories. From the above result and the fact that any UFC can be embedded into a UMFC (see Sec.~\ref{sec:DefectSN}), we see that the non-chiral topological phase $\EB$ can be equivalently defined as a UMTC $\EB$ for which there exists a UMFC $\ED$ such that $\EB\simeq \mathcal{Z}(\ED)$ as braided monoidal categories.

For a UMFC $\ED$, given that $\ED_{i,i}$ are all UFCs, their Drinfeld centers $\mathcal{Z}(\ED_{i,i})$ are UMTCs. A direct consequence of statement 2 in the lemma above is that the Drinfeld center $\mathcal{Z}(\ED)$ of an indecomposable UMFC $\ED$ is a UMTC.
As we delve further into our exploration, it becomes evident that the above result guarantees the utility of any given UMFC in constructing defective string-net models characterized by gapped boundaries, boundary defects, and other features.

\subsection{Weak Hopf symmetry for $2d$ topological phase}

The group symmetry of a quantum system can be regarded as a group action on the system's Hilbert space $\mathcal{H}$, i.e., for each $g\in G$, we have $U_g:\mathcal{H}\to \mathcal{H}$ satisfying $U_{gh}=U_gU_h$ and $U_e=\id_\mathcal{H}$.
Consider the group algebra $\mathbb{C}[G]$, the Hilbert space $\mathcal{H}$ is a $\mathbb{C}[G]$-module.
The ground state  $\Psi\in \mathcal{H}$ is invariant under the $\mathbb{C}[G]$-action in the sense that 
\begin{equation}
   ( \sum_{g\in G} c_g g) \triangleright \Psi = \sum_{g\in G} c_g \Psi =\varepsilon(\sum_{g\in G} c_g g) \Psi,
\end{equation}
where $\varepsilon: \mathbb{C}[G] \to \Cbb$ is the counit map of the group algebra.
This motivates the introduction of Hopf algebra symmetry, defined as $h\triangleright \Psi = \varepsilon(h) \Psi$ \cite{bais2003hopf}. However, this definition cannot be straightforwardly extended to the weak Hopf case, as the trivial representation exhibits a nontrivial algebra action. To address this challenge, in Ref.~\cite{Jia2023weak}, we introduce the concept of weak Hopf symmetry.
The weak Hopf symmetry plays a pivotal role in comprehending the string-net model. In the following, we will recall the definition of weak Hopf symmetry and discuss its properties.

Let us first introduce the weak Hopf algebra within the braided multifusion category. In this work, our primary focus will be on the case of the category of finite-dimensional complex vector spaces, denoted as $\Vect_{\Cbb}$. A \emph{weak bialgebra} $(W, \mu, \eta, \Delta, \varepsilon)$ is an object $W\in \Vect_{\Cbb}$ equipped with the following structure morphisms (here $\one=\Cbb$):
	\begin{itemize}
		\item An algebra structure $(W,\mu,\eta)$, where $\mu: W\otimes W \to W$ and  $\eta: \one \to W$  are linear morphisms called multiplication and unit morphisms respectively
		\begin{equation}
			\mu=\begin{aligned}
			\begin{tikzpicture}
				 \draw[black, line width=1.0pt]  (-0.5, 0) .. controls (-0.4, 1) and (0.4, 1) .. (0.5, 0);
				 \draw[black, line width=1.0pt]  (0,0.75)--(0,1.15);
				\end{tikzpicture}
			\end{aligned},\quad 
		\eta= \begin{aligned}
			\begin{tikzpicture}
				\draw[black, line width=1.0pt]  (0,0.75)--(0,1.5);
				\draw [fill = white](0, 0.73) circle (2pt);
			\end{tikzpicture}
		\end{aligned}\,\,\, .
		\end{equation}
		They satisfy 
		\begin{gather}
			\mu \comp (\mu \otimes \id) =\mu \comp (\id \otimes \mu), \quad  \mu \comp (\eta \otimes \id)=\id =\mu \comp (\id \otimes \eta). \\
	\begin{aligned}
			\begin{tikzpicture}
				\draw[black, line width=1.0pt]  (-0.5, 0) .. controls (-0.4, 0.8) and (0.4, 0.8) .. (0.5, -0.5);
				\draw[black, line width=1.0pt]  (-0.08,0.53)--(-0.08,0.93);
				\draw[black, line width=1.0pt]  (-0.8, -0.5) .. controls (-0.7, 0.19) and (-0.3, 0.19) .. (-0.2, -0.5);
			\end{tikzpicture}
		\end{aligned}
=	
	\begin{aligned}
	\begin{tikzpicture}
		\draw[black, line width=1.0pt]  (-0.5, -0.5) .. controls (-0.4, 0.8) and (0.4, 0.8) .. (0.5, 0);
		\draw[black, line width=1.0pt]  (0.08,0.53)--(0.08,0.93);
		\draw[black, line width=1.0pt]  (0.2, -0.5) .. controls (0.3, 0.19) and (0.7, 0.19) .. (0.8, -0.5);
	\end{tikzpicture}
\end{aligned},
	\quad 
			\id\,\,\,
	\begin{aligned}
		\begin{tikzpicture}
			\draw[black, line width=1.0pt] (0,-0.5)--(0,0.7);
		\end{tikzpicture}
	\end{aligned}
= \begin{aligned}
			\begin{tikzpicture}
				\draw[black, line width=1.0pt]   (0,0.3) arc (90:180:0.5);
				\draw [fill = white](-0.5, -0.2) circle (2pt);
				\draw[black, line width=1.0pt] (0,-0.5)--(0,0.7);
			\end{tikzpicture}
		\end{aligned}
	=
			 \begin{aligned}
		\begin{tikzpicture}
			\draw[black, line width=1.0pt]   (0.5,-0.2) arc (0:90:0.5);
			\draw [fill = white](0.5, -0.2) circle (2pt);
			\draw[black, line width=1.0pt] (0,-0.5)--(0,0.7);
		\end{tikzpicture}
	\end{aligned}\,\,\, .
	\end{gather}
		
		\item A coalgebra structure $(W,\Delta,\varepsilon)$, where $\Delta: W\to W\otimes W$ and $\varepsilon: W\to \one$ are called comultiplication and counit morphisms respectively
				\begin{equation}
			\Delta=\begin{aligned}
				\begin{tikzpicture}
					\draw[black, line width=1.0pt]  (-0.5, 1) .. controls (-0.4, 0) and (0.4, 0) .. (0.5, 1);
					\draw[black, line width=1.0pt]  (0,0.23)--(0,-0.23);
				\end{tikzpicture}
			\end{aligned},\quad 
			\varepsilon= \begin{aligned}
				\begin{tikzpicture}
					\draw[black, line width=1.0pt]  (0,0.75)--(0,1.5);
					\draw [fill = white](0, 1.5) circle (2pt);
				\end{tikzpicture}
			\end{aligned}\,\,\, .
		\end{equation}
		They satisfy
		\begin{gather}
			(\Delta \otimes \id)\comp \Delta = (\id \otimes \Delta )\comp \Delta,\quad  (\varepsilon \otimes \id) \comp \Delta = \id = (\id \otimes \varepsilon)\comp \Delta. \\
	\begin{aligned}
			\begin{tikzpicture}
				\draw[black, line width=1.0pt]  (-0.5, 0.8) .. controls (-0.4, 0) and (0.4, 0) .. (0.5, 1.3);
				\draw[black, line width=1.0pt]  (-0.1,-0.16)--(-0.1,0.24);
				\draw[black, line width=1.0pt]  (-0.8, 1.3) .. controls (-0.7, 0.62) and (-0.3, 0.62) .. (-0.2, 1.3);
			\end{tikzpicture}
		\end{aligned}
	=
		\begin{aligned}
		\begin{tikzpicture}
			\draw[black, line width=1.0pt]  (-0.5, 1.3) .. controls (-0.4, 0) and (0.4, 0) .. (0.5, 0.8);
			\draw[black, line width=1.0pt]  (0.1,-0.16)--(0.1,0.24);
			\draw[black, line width=1.0pt]  (0.2, 1.3) .. controls (0.3, 0.62) and (0.7, 0.62) .. (0.8, 1.3);
		\end{tikzpicture}
	\end{aligned},
	\quad 
		\id\,\,\,
			\begin{aligned}
			\begin{tikzpicture}
				\draw[black, line width=1.0pt] (0,-0.5)--(0,0.7);
			\end{tikzpicture}
		\end{aligned}
		=\begin{aligned}
			\begin{tikzpicture}
				\draw[black, line width=1.0pt]   (-0.5,0.5) arc (180:270:0.5);
				\draw [fill = white](-0.5, 0.5) circle (2pt);
				\draw[black, line width=1.0pt] (0,-0.5)--(0,0.7);
			\end{tikzpicture}
		\end{aligned}
		=
		\begin{aligned}
			\begin{tikzpicture}
				\draw[black, line width=1.0pt]   (0,0) arc (270:360:0.5);
				\draw [fill = white](0.5, 0.5) circle (2pt);
				\draw[black, line width=1.0pt] (0,-0.5)--(0,0.7);
			\end{tikzpicture}
		\end{aligned}.
	\end{gather}
	\end{itemize}
	
To establish a weak bialgebra, the aforementioned algebra and coalgebra structures need to adhere to specific consistency conditions. Specifically:
	\begin{enumerate}
		\item[(1)] The comultiplication preserves multiplication, viz., $\Delta\comp\mu = \mu_{W\otimes W}\comp (\Delta\otimes \Delta)= (\mu\otimes \mu)\comp (\id\otimes c_{W,W} \otimes \id)\comp(\Delta\otimes \Delta)$\,\footnote{The $c_{W,W}$ is the braiding morphism in $\Vect_{\Cbb}$, diagrammatically denoted as $c_{W,W}=\begin{aligned}\begin{tikzpicture}\braid[
			width=0.4cm,
			height=0.1cm,
			line width=0.3pt,
			style strands={1}{black},
			style strands={2}{black}] (Kevin)
			s_1^{-1} ;	\end{tikzpicture}\end{aligned}\,$;
		we will also adopt the notations $\mu^{\rm op}=\mu \comp c_{W,W}$ and $\Delta^{\rm op}=c_{W,W}\comp \Delta$ in what follows.},  
			\begin{equation}
			\begin{aligned}
			\begin{tikzpicture}
				\draw[black, line width=1.0pt]   (-0.5,0)..   controls (-0.4,0.8) and (0.4,0.8)             ..(0.5,0);
				\draw[black, line width=1.0pt] (0,0.6)--(0,1);
					\draw[black, line width=1.0pt]   (-0.5,1.6)..   controls (-0.4,0.8) and (0.4,0.8)             ..(0.5,1.6);
			\end{tikzpicture}
		\end{aligned}
	=
		\begin{aligned}
		\begin{tikzpicture}
		\draw[black, line width=1.0pt] (-0.1,-1) arc (180:360:0.3);
		\draw[black, line width=1.0pt] (0.5,0) arc (0:180:0.3);
		 \draw[black, line width=1.0pt] (-0.1,0)--(-0.1,-1);
		 	\draw[black, line width=1.0pt] (1.6,0) arc (0:180:0.3);
		 	\draw[black, line width=1.0pt] (1,-1) arc (180:360:0.3);
		 \draw[black, line width=1.0pt] (1.6,0)--(1.6,-1);
		 	 \draw[black, line width=1.0pt] (0.2,-1.3)--(0.2,-1.6);
		 	  \draw[black, line width=1.0pt] (0.2,0.3)--(0.2,0.6);
		 	  \draw[black, line width=1.0pt] (1.3,0.3)--(1.3,0.6);
		 	  \draw[black, line width=1.0pt] (1.3,-1.3)--(1.3,-1.6);
			\braid[
			width=0.5cm,
			height=0.5cm,
			line width=1.0pt,
			style strands={1}{black},
			style strands={2}{black}] (Kevin)
			s_1^{-1} ;
		\end{tikzpicture}
	\end{aligned}\,\,\,.
	\end{equation}

		\item[(2)]Compatibility of comultiplication and unit  $(\Delta\otimes \id)\comp \Delta \comp \eta=(\id \otimes \mu \otimes \id )\comp  (\Delta \otimes \Delta)  \comp(\eta \otimes \eta) =  (\id \otimes \mu \otimes \id )\comp (\id \otimes c_{W,W}\otimes \id)\comp (\Delta \otimes \Delta)\comp (\eta \otimes \eta)$, 
			\begin{equation}\label{eq:unitG}
				\begin{aligned}
			\begin{tikzpicture}
				\draw[black, line width=1.0pt] (-0.1,-1) arc (180:360:0.3);
				\draw[black, line width=1.0pt] (-0.1,-0.8)--(-0.1,-1);
				\draw[black, line width=1.0pt] (0.5,-1).. controls (0.6,-0.6)and (0.7,-0.6) ..(0.8,-0.5);
				\draw[black, line width=1.0pt] (-0.4,-0.5) arc (180:360:0.3);
				\draw[black, line width=1.0pt] (0.2,-1.3)--(0.2,-1.6);
				\draw [fill = white] (0.2,-1.6) circle (2pt);
			\end{tikzpicture}		\end{aligned}
=			
	\begin{aligned}	\begin{tikzpicture}
		\draw[black, line width=1.0pt] (-0.1,-0.8) arc (180:360:0.3);
		\draw[black, line width=1.0pt] (1,-0.8) arc (180:360:0.3);
		\draw[black, line width=1.0pt] (1,0) arc (0:180:0.25);
		\draw[black, line width=1.0pt] (-0.1,0.5)--(-0.1,-0.8);
        \draw[black, line width=1.0pt] (0.5,0)--(0.5,-0.8);
          \draw[black, line width=1.0pt] (1,0)--(1,-0.8);
		\draw[black, line width=1.0pt] (0.75,0.25)--(0.75,0.5);
		\draw[black, line width=1.0pt] (1.6,0.5)--(1.6,-0.8);
		\draw[black, line width=1.0pt] (1.3,-1.1)--(1.3,-1.35);
		\draw[black, line width=1.0pt] (0.2,-1.1)--(0.2,-1.35);
		\draw [fill = white](0.2,-1.35) circle (2pt);
		\draw [fill = white] (1.3,-1.35) circle (2pt);
	\end{tikzpicture}
\end{aligned}
=
			\begin{aligned}	\begin{tikzpicture}
				\draw[black, line width=1.0pt] (-0.1,-0.8) arc (180:360:0.3);
					\draw[black, line width=1.0pt] (1,-0.8) arc (180:360:0.3);
			\draw[black, line width=1.0pt] (1,0) arc (0:180:0.25);
				\draw[black, line width=1.0pt] (-0.1,0.5)--(-0.1,-0.8);
					\braid[
				width=0.5cm,
				height=0.3cm,
				line width=1.0pt,
				style strands={1}{black},
				style strands={2}{black}] (Kevin)
				s_1^{-1} ;
				\draw[black, line width=1.0pt] (0.75,0.25)--(0.75,0.5);
				\draw[black, line width=1.0pt] (1.6,0.5)--(1.6,-0.8);
	    	\draw[black, line width=1.0pt] (1.3,-1.1)--(1.3,-1.35);
	    	\draw[black, line width=1.0pt] (0.2,-1.1)--(0.2,-1.35);
	    		\draw [fill = white](0.2,-1.35) circle (2pt);
	    		\draw [fill = white] (1.3,-1.35) circle (2pt);
			\end{tikzpicture}
		\end{aligned}\,\,\,.
	\end{equation}
		
		\item[(3)]Compatibility of multiplication and counit $\varepsilon \comp \mu \comp (\id \otimes \mu) = (\varepsilon \otimes \varepsilon)\comp (\mu \otimes \mu)\comp (\id \otimes  \Delta \otimes\id)=(\varepsilon \otimes \varepsilon)\comp (\mu \otimes \mu)\comp (\id \otimes  \Delta^{\rm op} \otimes\id)$, 
					\begin{equation} \label{eq:weak-mult}
			\begin{aligned}
				\begin{tikzpicture}
					\draw[black, line width=1.0pt] (0.5,-0.8) arc (0:180:0.3);
					\draw[black, line width=1.0pt] (-0.1,-0.8)--(-0.1,-1);
					\draw[black, line width=1.0pt] (0.2,-0.5)--(0.2,-0.25);
				    \draw[black, line width=1.0pt]  (0.5,-0.8).. controls (0.6,-1.18)and (0.7,-1.25) ..(0.8,-1.3);
					\draw[black, line width=1.0pt] (0.2,-1.3) arc (0:180:0.3);
							\draw [fill = white] (0.2,-0.25) circle (2pt);
			\end{tikzpicture}		\end{aligned}
			=			
\begin{aligned}	\begin{tikzpicture}
		\draw[black, line width=1.0pt] (0.5,0) arc (0:180:0.3);
		\draw[black, line width=1.0pt] (1.6,0) arc (0:180:0.3);
		\draw[black, line width=1.0pt] (0.5,-0.8) arc (180:360:0.25);
		\draw[black, line width=1.0pt] (0.5,0)--(0.5,-0.8);
			\draw[black, line width=1.0pt] (1,0)--(1,-0.8);
		\draw[black, line width=1.0pt] (-0.1,0)--(-0.1,-1.3);
		\draw[black, line width=1.0pt] (1.6,0)--(1.6,-1.3);
		\draw[black, line width=1.0pt] (0.75,-1.05)--(0.75,-1.3);
		\draw[black, line width=1.0pt] (1.3,0.3)--(1.3,0.55);
		\draw[black, line width=1.0pt] (0.2,0.3)--(0.2,0.55);
		\draw [fill = white]  (0.2,0.55) circle (2pt);
		\draw [fill = white] (1.3,0.55) circle (2pt);
	\end{tikzpicture}
\end{aligned}
			=
			\begin{aligned}	\begin{tikzpicture}
					\draw[black, line width=1.0pt] (0.5,0) arc (0:180:0.3);
               	\draw[black, line width=1.0pt] (1.6,0) arc (0:180:0.3);
               	 	\draw[black, line width=1.0pt] (0.5,-0.8) arc (180:360:0.25);
					\braid[
					width=0.5cm,
					height=0.3cm,
					line width=1.0pt,
					style strands={1}{black},
					style strands={2}{black}] (Kevin)
					s_1^{-1} ;
					\draw[black, line width=1.0pt] (-0.1,0)--(-0.1,-1.3);
					\draw[black, line width=1.0pt] (1.6,0)--(1.6,-1.3);
					\draw[black, line width=1.0pt] (0.75,-1.05)--(0.75,-1.3);
					\draw[black, line width=1.0pt] (1.3,0.3)--(1.3,0.55);
					\draw[black, line width=1.0pt] (0.2,0.3)--(0.2,0.55);
					\draw [fill = white]  (0.2,0.55) circle (2pt);
					\draw [fill = white] (1.3,0.55) circle (2pt);
				\end{tikzpicture}
			\end{aligned}\,\,\,.
		\end{equation} 
	\end{enumerate}

A \emph{weak Hopf algebra} $(W,\mu,\eta,\Delta,\varepsilon,S)$ is a weak bialgebra  $(W,\mu,\eta,\Delta,\varepsilon)$ equipped with a morphism $S: W\to W$ called antipode which satisfies the following three conditions:
	\begin{enumerate}
		\item[(1)] $\mu \comp (\mathrm{id}_W \otimes S) \comp \Delta  =  (\varepsilon \otimes \mathrm{id}_W) \comp (\mu \otimes \mathrm{id}_W) \comp (\mathrm{id}_W \otimes c_{W, W}) \comp (\Delta \otimes \mathrm{id}_W) \comp (\eta \otimes \mathrm{id}_W)$, 
					\begin{equation}
			\begin{aligned}	
				\begin{tikzpicture}
					\draw[black, line width=1.0pt] (0.5,0) arc (0:180:0.3);
					\draw[black, line width=1.0pt] (-0.1,-0.8) arc (180:360:0.3);
					\draw[black, line width=1.0pt] (-0.1,0)--(-0.1,-0.8);
					\draw[black, line width=1.0pt] (0.5,-0.8)--(0.5,-0.65);
					\draw[black, line width=1.0pt] (0.5,0)--(0.5,-0.3);
					\draw[black, line width=1.0pt] (0.2,-1.1)--(0.2,-1.35);
					\draw[black, line width=1.0pt] (0.2,0.3)--(0.2,0.55);
					\draw (0.3,-0.3) rectangle (0.7,-0.65);
					\node (start) [at=(Kevin-1-s),yshift=-0.47cm] {$S$};
				\end{tikzpicture}
			\end{aligned}
			=
			\begin{aligned}	
				\begin{tikzpicture}
					\draw[black, line width=1.0pt] (0.5,0) arc (0:180:0.3);
					\draw[black, line width=1.0pt] (-0.1,-0.8) arc (180:360:0.3);
					\braid[
					width=0.5cm,
					height=0.3cm,
					line width=1.0pt,
					style strands={1}{black},
					style strands={2}{black}] (Kevin)
					s_1^{-1} ;
					\draw[black, line width=1.0pt] (-0.1,0)--(-0.1,-0.8);
					\draw[black, line width=1.0pt] (1,-0.8)--(1,-1.35);
					\draw[black, line width=1.0pt] (1,0)--(1,0.55);
					\draw[black, line width=1.0pt] (0.2,-1.1)--(0.2,-1.35);
					\draw[black, line width=1.0pt] (0.2,0.3)--(0.2,0.55);
					\draw [fill = white]  (0.2,0.55) circle (2pt);
					\draw [fill = white] (0.2,-1.35) circle (2pt);
				\end{tikzpicture}
			\end{aligned}\,\,\,.
		\end{equation}
        The map on the right-hand side will be denoted as $\varepsilon_L: W\to W$, whose image is denoted as $W_L$.

		\item[(2)] $\mu \comp (S \otimes \mathrm{id}_W) \comp \Delta = (\mathrm{id}_W \otimes \varepsilon) \comp (\mathrm{id}_W \otimes \mu) \comp (c_{W, W} \otimes \mathrm{id}_W) \comp (\mathrm{id}_W \otimes \Delta) \comp (\mathrm{id}_W \otimes \eta)$, 	
							\begin{equation}
			\begin{aligned}	
				\begin{tikzpicture}
					\draw[black, line width=1.0pt] (0.5,0) arc (0:180:0.3);
					\draw[black, line width=1.0pt] (-0.1,-0.8) arc (180:360:0.3);
					\draw[black, line width=1.0pt] (-0.1,-0.65)--(-0.1,-0.8);
					\draw[black, line width=1.0pt] (-0.1,0)--(-0.1,-0.3);
					\draw[black, line width=1.0pt] (0.5,-0.8)--(0.5,0);
					\draw[black, line width=1.0pt] (0.2,-1.1)--(0.2,-1.35);
					\draw[black, line width=1.0pt] (0.2,0.3)--(0.2,0.55);
					\draw (-0.3,-0.3) rectangle (0.1,-0.65);
					\node (start) [at=(Kevin-1-s),xshift=-0.58cm,yshift=-0.47cm] {$S$};
				\end{tikzpicture}
			\end{aligned}
			=
			\begin{aligned}	
				\begin{tikzpicture}
					\draw[black, line width=1.0pt] (1.6,0) arc (0:180:0.3);
					\draw[black, line width=1.0pt] (1.3,0.3)--(1.3,0.55);
					\draw[black, line width=1.0pt] (1.0,-0.8) arc (180:360:0.3);
					\draw[black, line width=1.0pt] (1.6,0)--(1.6,-0.8);
					\draw[black, line width=1.0pt] (1.3,-1.1)--(1.3,-1.35);
					\braid[
					width=0.5cm,
					height=0.3cm,
					line width=1.0pt,
					style strands={1}{black},
					style strands={2}{black}] (Kevin)
					s_1^{-1} ;
					\draw[black, line width=1.0pt] (0.5,0)--(0.5,0.55);
					\draw[black, line width=1.0pt] (0.5,-0.6)--(0.5,-1.35);
					\draw [fill = white] (1.3,0.55) circle (2pt);
					\draw [fill = white] (1.3,-1.35) circle (2pt);
				\end{tikzpicture}
			\end{aligned}\,\,\,.
		\end{equation}
			The map on the right-hand side will be denoted as $\varepsilon_R: W\to W$, whose image is denoted as $W_R$.

	   \item[(3)] $S  = \mu \comp (\mu \otimes \mathrm{id}_W) \comp (S \otimes \mathrm{id}_W \otimes S) \comp (\Delta \otimes \mathrm{id}_W) \comp \Delta$, 
	   \begin{equation}
	   		\begin{aligned}	
	   		\begin{tikzpicture}
	   			\draw[black, line width=1.0pt] (-0.1,-0.65)--(-0.1,-1.05);
	   			\draw[black, line width=1.0pt] (-0.1,0.1)--(-0.1,-0.3);
	   			\draw (-0.3,-0.3) rectangle (0.1,-0.65);
	   			\node (start) [at=(Kevin-1-s),xshift=-0.58cm,yshift=-0.47cm] {$S$};
	   		\end{tikzpicture}
	   	\end{aligned}
   	=
	   	\begin{aligned}	
	   	\begin{tikzpicture}
	   		\draw[black, line width=1.0pt] (0.5,0) arc (0:180:0.3);
	   		\draw[black, line width=1.0pt] (-0.1,-0.4) arc (180:360:0.3);
	   		\draw[black, line width=1.0pt] (0.5,-0.4)--(0.5,0);
	   		\draw[black, line width=1.0pt] (0.2,-0.7) arc (180:360:0.3);
	   		\draw[black, line width=1.0pt] (0.8,0.3) arc (0:180:0.3);
	   		\draw[black, line width=1.0pt] (0.8,0.3)--(0.8,0);
	   		\draw[black, line width=1.0pt] (0.8,-0.7)--(0.8,-0.4);
	   		\draw[black, line width=1.0pt] (0.5,0.6)--(0.5,0.85);
	   		\draw[black, line width=1.0pt] (0.5,-1.25)--(0.5,-1);
	   		\draw (-0.3,0) rectangle (0.1,-0.4);
	   		\node (start) [at=(Kevin-1-s),xshift=-0.58cm,yshift=-0.2cm] {$S$};
	   		\draw (0.6,-0.4) rectangle (1,0);
	   		\node (start) [at=(Kevin-1-s),xshift=0.35cm,yshift=-0.2cm] {$S$};
	   	\end{tikzpicture}
	   \end{aligned}\,\,\,.
	   \end{equation}
	\end{enumerate}

A complex weak Hopf algebra $W$ is termed simple if its underlying algebra $(W, \mu, \eta)$ possesses no nontrivial subalgebras. It is referred to as semisimple if its underlying algebra can be expressed as a direct sum of simple algebras.

A useful tool for investigating the weak Hopf algebra is the canonical pairing between $W$ and its dual $\hat{W}:=\Hom(W,\mathbb{C})$. For $\phi\in \hat{W}$ and $x\in W$, we define $\langle \phi, x\rangle:=\phi(x)$.
$\hat{W}$ has a canonical weak Hopf algebra structure with structure maps given by:
\begin{align}
& \langle \hat{\mu}(\varphi\otimes \psi),x\rangle=\langle\varphi\otimes \psi, \Delta(x)\rangle,\\
& \langle \hat{\eta} (1),x\rangle= \varepsilon(x),\, \text{i.e.,}\hat{1}=\varepsilon,\\
& \langle \hat{\Delta}(\varphi), x\otimes y \rangle=\langle \varphi, \mu(x\otimes y)\rangle,\\
& \hat{\varepsilon}(\varphi)=\langle \varphi, \eta(1)\rangle,\\
& \langle \hat{S}(\varphi),x\rangle =\langle \varphi, S(x)\rangle.
\end{align}
The opposite weak Hopf algebra $W^{\rm op}$ is defined as $(W, \mu^{\rm op}, \eta, \Delta, \varepsilon, S^{-1})$, and the co-opposite weak Hopf algebra $W^{\rm cop}$ is defined as $(W, \mu, \eta, \Delta^{\rm op}, \varepsilon, S^{-1})$. Note that the antipode is always invertible \cite{BOHM1998weak}.

\begin{definition}[Quantum double, see Ref.~\cite{Jia2023weak} for details] \label{def:QuantumDouble}
For a weak Hopf algebra $W\in \mathsf{Vect}_{\mathbb{C}}$, its quantum double is defined as the space $D(W)=(\hat{W}\otimes W)/J$ where $J$ is generated by
\begin{equation}
    \phi\otimes lh - \phi(l\rightharpoonup \varepsilon) \otimes h,\quad \phi\otimes rh - \phi(\varepsilon \leftharpoonup r)\otimes h,\quad l\in W_L,\;r\in W_R.
\end{equation}
$D(W)$ is equipped with the following weak Hopf algebra structure:
\begin{enumerate}
	\item[(1)] The multiplication $[\varphi \otimes h] [\psi \otimes g]=\sum_{(\psi),(h)} [\varphi \psi^{(2)} \otimes  h^{(2)} g]  \langle \psi^{(1)}, S^{-1}(h^{(3)}) \rangle \langle \psi^{(3)}, h^{(1)}\rangle$.
	\item[(2)] The unit $[\varepsilon \otimes 1_W]$.
	\item[(3)] The comultiplication $\Delta([\varphi \otimes  h])= \sum_{(\varphi), (h)} [\varphi^{(2)} \otimes h^{(1)}] \otimes [\varphi^{(1)} \otimes h^{(2)}] $.
	\item[(4)] The counit $\varepsilon ([\varphi \otimes h])= \langle \varphi, \varepsilon_R(S^{-1}(h))\rangle$.
	\item[(5)] The antipode $S([\varphi \otimes h])
	= \sum_{(\varphi), (h)} [\hat{S}^{-1}(\varphi^{(2)}) \otimes S(h^{(2)})]
	\langle  \varphi^{(1)},h^{(3)}\rangle
	\langle  \varphi^{(3)}, S^{-1}(h^{(1)})\rangle$.
\end{enumerate}
The quantum double has a canonical quasitriangular structure, ensuring that the representation category of $D(W)$ is braided. Notice that here and henceforth we adopt Sweedler's notation for comultiplication $\Delta(x) = \sum_{(x)}x^{(1)}\otimes x^{(2)}$, and for arrows $(a\rightharpoonup \phi)(h) = \phi(ha)$, $(\phi\leftharpoonup a)(h)=\phi(ah)$.  
\end{definition}

\begin{example}
The following are some examples of weak Hopf algebras:
\begin{enumerate}
    \item For every Hopf algebra $H$, its dynamical twist is a weak Hopf algebra.
    \item If $B$ is a separable algebra, then its enveloping algebra $B^{\rm op}\otimes B$ can be equipped with a weak Hopf algebra structure \cite{nikshych2003invariants}. 
\end{enumerate}
\end{example}


For a given weak Hopf symmetry $W$, the symmetry charges (or sectors) are defined as finite-dimensional irreducible representations of $W$ (left $W$-modules, the action will be denoted as ``$x\triangleright z$'').
The category of these symmetry charges is denoted as $\mathsf{Rep}(W)$, where the morphisms between two charges $X$ and $Y$ are $W$-linear homomorphisms.
The fusion between $X,Y\in \Rep(W)$ are defined as
\begin{equation}
    X\otimes Y:=\{x\in X\otimes_{\mathbb{C}} Y: \Delta(1)\triangleright x=x\}.
\end{equation}
Since $\Delta(1)$ is an idempotent, we see $X\otimes Y= \Delta(1)\triangleright (X\otimes_{\Cbb} Y)$.
The vacuum charge is the left counital subalgebra $\one=W_L=\varepsilon_L(W)$ with the left $W$-action given by $x\triangleright z=\varepsilon_L(xz)$.
For a charge $X$, its anti-charge is defined as $\bar{X}:=\Hom(X,\mathbb{C})$, with the left $W$ action given by $(x\triangleright f)(y)=f(S(x)y)$ for $f\in \bar{X}$.
The duality maps $b_X: \one \to X\otimes \bar{X}$ and $d_X: \bar{X}\otimes X \to \one$ are given by: $b_X(x)=x\,\triangleright \,(\sum_i v_i\otimes f_i)$ where $\{v_i\}$ and $\{f_i\}$ 
are bases of $X$ and $\bar{X}$ respectively which are dual to each other, and $d_X(\sum_j g_j\otimes z_j)=\sum_j\sum_{(1)} g_j(1^{(1)} z_j) 1^{(2)}$ where $\sum_j g_j\otimes z_j\in \bar{X}\otimes X$.
The weak Hopf algebra $W$ is called \emph{indecomposable} if and only if $\Rep(W)$ is indecomposable as a tensor category.

\begin{definition}[Weak Hopf symmetry] 
For a quantum system $(H, \mathcal{H})$, we assert that it possesses a weak Hopf symmetry if $\mathcal{H}$ forms a representation of $W$. 
This implies that we can establish a mapping from $g \in W$ to operators $U_g: \mathcal{H} \to \mathcal{H}$ such that $U_{gh} = U_gU_h$ and $U_1 = \operatorname{id}_{\mathcal{H}}$. We will also denote the action simply as $g \triangleright \Psi$.   A subspace $\mathcal{V}$ of $\mathcal{H}$ is called a vacuum space if and only if it is the ground state space of Hamiltonian $H$ and is isomorphic to the direct sum of trivial representations of $W$. 
\end{definition}

\begin{remark}
   In Ref.~\cite{Jia2023weak}, we introduce the weak Hopf symmetry for a given state $\Psi \in \mathcal{H}$. Recall that $\mathcal{H}$ is a $W$-module. $\Psi$ is called invariant under the action of $g\in W$ if $g\triangleright \Psi=\varepsilon_L(g) \triangleright\Psi$. This is a special case of the weak Hopf symmetry given above.
\end{remark}

\begin{lemma}\label{prop:WHAofmonoidal}
    For every multifusion category $\EC$, there is a (non-unique) corresponding weak Hopf algebra $W$ such that $\EC$ is equivalent to the representation category $\mathsf{Rep}(W)$ as UMFCs.
\end{lemma}

\begin{proof}
    See Refs.~\cite{szlachanyi2000finite}, \cite[Theorem 4.1]{ostrik2003module}, \cite[Corollary 2.22]{etingof2005fusion} for a detailed discussion.
\end{proof}

\begin{definition}[Connected weak Hopf algebra \cite{etingof2005fusion}]
If the vacuum sector of a weak Hopf algebra $W$ is irreducible, then $W$ is called connected (or pure). In this case, the representation category $\mathsf{Rep}(W)$ is a UFC.   
The vacuum sector $W_L$ is irreducible if and only if the intersection of $W_L$ and the center $Z(W)$ is $\mathbb{C}$.
\end{definition}

\begin{proposition}\label{prop:chargesymmetry}
We have the following two crucial results about the quantum double of a weak Hopf symmetry:
   \begin{enumerate}
       \item The quantum double $D(W)$ of any indecomposable weak Hopf algebra $W$ is a connected weak Hopf algebra.
       \item For every non-chiral topological phase (UMTC) $\EB$, there is a weak Hopf algebra $W$ such that $\EB$ can be realized as weak Hopf lattice gauge theory in the sense that $\EB\simeq \Rep(D(W))$. 
   \end{enumerate}  
\end{proposition}
\begin{proof}
1. Combining Lemma~\ref{prop:UMFCcenter} and the fact that $\mathcal{Z}(\Rep(W)) \simeq \Rep(D(W))$, we arrive at our conclusion.

2. Given that $\EB$ is non-chiral, there exists a multifusion category $\EC$ such that $\EB \simeq \mathcal{Z}(\EC)$. Utilizing Lemma~\ref{prop:WHAofmonoidal}, we can identify a weak Hopf algebra $W$ such that $\Rep(W) \simeq \EC$. Leveraging the fact $\mathcal{Z}(\Rep(W)) \simeq \Rep(D(W))$, we thus reach our conclusion.
\end{proof}

From a given weak Hopf algebra $W$, it is possible to construct a lattice gauge theory with \emph{gauge symmetry} $W$, wherein the topological excitations are classified by representations of $D(W)$ \cite{Jia2023weak}. Consequently, $D(W)$ can be viewed as a symmetry, which we will refer to as \emph{charge symmetry}.
As we will elaborate later, the gauge symmetry determines the string types in the multifusion string-net model, whereas the charge symmetry characterizes the topological charges of the model.
Proposition~\ref{prop:chargesymmetry} asserts that for non-chiral topological phases, the gauge symmetry and charge symmetry are distinct, they are related by a quantum double operation up to Morita equivalence. Table~\ref{tab:TubeAlgebra} is a summary of results regarding gauge symmetry and charge symmetry realized by the tube algebras, which will be clarified in Sec.~\ref{sec:tube}.


\subsection{Macroscopic theory of multifusion string-net}

\begin{table}[t]
\centering \small 
\begin{tabular} {|l|c|} 
\hline
 Bulk input data  & UMFC $\ED$   \\ 
   \hline
 Bulk input data  & $\ED$-module $\EM$  \\ 
   \hline  
 Bulk gauge symmetry  &

    $\text{boundary tube algebra}\,\,  W_{\rm gauge}\simeq_{\rm w. M.} \mathbf{Tube}({_{\ED}}\EM)$  \\ 
 \hline
 Bulk charge symmetry  & $\text{bulk tube algebra}\,\, W_{\rm charge}\simeq_{\rm s. M.} \mathbf{Tube}({_{\ED}}\ED_{\ED})$
   \\ 
 \hline
 Bulk phase & $\mathcal{Z}(\ED) \simeq_{\otimes,\rm br} \Rep(W_{\rm charge})\simeq_{\otimes,\rm br}  \Rep(D(W_{\rm gauge}))$
    \\ 
 \hline
\end{tabular}
\caption{The weak Hopf symmetries of the multifusion string-net model. We use ``$\simeq_{\rm w. M.}$'' to denote weak Morita equivalence, and ``$\simeq_{\rm s. M.}$'' to denote strong Morita equivalence; ``$\simeq_{\otimes}$'' to denote monoidal equivalence, and ``$\simeq_{\otimes,\rm br}$'' to denote braided monoidal equivalence.
\label{tab:TubeAlgebra}}
\end{table}

Now, let us proceed to develop the macroscopic theory of multifusion string-net for the $2d$ bulk, $1d$ gapped boundary, and $1d$ gapped domain wall.
This subsection serves as a summary of the forthcoming results, offering an overview of the broader picture that will be elucidated later in this work.

\vspace{1em}
\emph{Bulk theory.} ---
For the $2d$ bulk phase, the multifusion string-net model with input UMFC $\EC$  is equivalent to the lattice gauge theory with a weak Hopf gauge symmetry $W$ such that $\EC \simeq \Rep(W)$ in the sense that they have the same topological excitations $\EP\simeq \mathcal{Z}(\EC)\simeq \Rep(D(W))$.
For lattice gauge theory (denoted as $\mathsf{QD}$, abbreviation of ``quantum double''), we have
\begin{equation}
    \text{finite group}~\mathsf{QD} \subset \text{Hopf}~\mathsf{QD} \subset  \text{connected weak Hopf}~\mathsf{QD}\subset  \text{weak Hopf}~\mathsf{QD}.
\end{equation}
For string-net theory (denoted as $\mathsf{SN}$), we have
\begin{equation}
    \text{Levin-Wen}~\mathsf{SN} \subset \text{multifusion}~\mathsf{SN}.
\end{equation}
When $\EC$ is a UFC, the gauge symmetry $W$ must be a connected weak Hopf algebra, which encompasses all Hopf algebras as examples.
In other words, lattice gauge theories with connected weak Hopf gauge symmetry, including Hopf lattice gauge theories, correspond to Levin-Wen string-net models.
Using results in the previous subsection, we have
\begin{equation}
    \text{multifusion}~\mathsf{SN} \Leftrightarrow  \text{weak Hopf}~\mathsf{QD}, \quad  \text{Levin-Wen}~\mathsf{SN} \Leftrightarrow  \text{connected weak Hopf}~\mathsf{QD}.
\end{equation}
This correspondence has been carefully discussed in Refs.~\cite{Buerschaper2013a,meusburger2017kitaev,jia2023boundary} for Hopf symmetry.
General lattice gauge theories with (not necessarily connected) weak Hopf gauge symmetry correspond to multifusion string-net models, the lattice realization of this correspondence will be discussed in our future work.

The correspondence between the multifusion string-net model and weak Hopf lattice gauge theory enables us to comprehend the multifusion string-net through the lens of the weak Hopf lattice gauge theory established in our previous work~\cite{Jia2023weak}.
A crucial fact of weak Hopf symmetries for a multifusion string-net is that they are not unique.

\begin{definition}
    For $2d$ topological phases, two weak Hopf gauge symmetries $W$ and $W'$ are called weak Morita equivalent if and only if $\mathcal{Z}(\Rep(W))\simeq \mathcal{Z}(\Rep(W'))$ as UMFCs.
    Two weak Hopf charge symmetries $V$ and $V'$ are called strong Morita equivalent if and only if $\Rep(V)\simeq \Rep(V')$ as braided monoidal categories.
\end{definition}

\begin{proposition}
    For $2d$ topological phases, two weak Hopf gauge symmetries $W$ and $W'$ are weak Morita equivalent if and only if $\Rep(D(W))\simeq \Rep(D(W'))$ as braided monoidal categories, viz., $D(W)$ and $D(W')$ are strong Morita equivalent.
\end{proposition}
\begin{proof}
    This arises directly from the fact that $\mathcal{Z}(\Rep(A))\simeq \Rep(D(A))$ as braided monoidal categories for a weak Hopf algebra $A$.
\end{proof}

With the groundwork laid out above, we are prepared to summarize our main findings regarding the bulk phase of the multifusion string-net model:

\begin{theorem} \label{thm:UMFC_SN}
For a multifusion string-net $\mathsf{SN}_{\ED}$ with an input UMFC $\ED$, its gauge symmetry is described by a weak Hopf symmetry $W$ in the sense that $\ED\simeq \mathsf{Rep}(W)$ as multifusion categories. 
The weak Hopf symmetries of a multifusion string-net are not unique, and all these symmetries are Morita equivalent in the sense that $\mathsf{Rep}(W)\simeq \Rep(W')\simeq \ED$ as UMFCs.
The topological excitations of the multifusion string-net model are given by $\mathcal{Z}(\ED)$, which is a UMTC. The multifusion string-net model has the same topological excitations as the weak Hopf lattice gauge theory with gauge symmetry given by $W$, for which the topological excitations are given by the representation category $\mathsf{Rep}(D(W))$ of the quantum double $D(W)$ of $W$. More precisely, this is characterized by $\mathcal{Z}(\ED) \simeq \mathsf{Rep}(D(W))$ as UMTCs.
This means that the charge symmetry of the multifusion string-net model is given by $D(W)$.
Noticeably, the charge symmetries are not unique; rather, they are Morita equivalent to each other.
\end{theorem}

\begin{remark}
   Determining the explicit weak Hopf gauge and charge symmetries for a multifusion string-net is indeed a challenging endeavor. Here we give two approaches:
\begin{enumerate}
    \item \emph{Tannaka-Krein reconstruction from fiber functor.} --- For any UMFC $\ED$, there exists a suitable algebra $A$ over $\mathbb{C}$ such that we can equip $\ED$ with an exact faithful monoidal functor $F: \ED\to {_A}\Mod_A$ (called fiber functor over $\ED$). The $W=\operatorname{End}(F)$ consisting of natural transformations from $F$ to itself is an algebra with product given by composition. It can be shown that $W$ is a weak Hopf symmetry and $\Rep(W)\simeq \ED$ as UMFCs \cite{szlachanyi2000finite}. This means that $W$ is a weak Hopf gauge symmetry for the multifusion string-net $\mathsf{SN}_{\ED}$.
    The charge symmetry of $\mathsf{SN}_{\ED}$ is thus given by $D(W)$.
    \item \emph{Tube algebra.} --- A more physically intuitive approach is grounded in tube algebra, which will be thoroughly discussed in Sec.~\ref{sec:tube}. The boundary tube algebra is a weak Hopf symmetry for any gapped boundary of multifusion string-net. If the bulk UMFC is $\ED$, we choose the smooth boundary ${_{\ED}}\ED$ and construct the boundary tube algebra $T=\mathbf{Tube}({_{\ED}}\ED)$ (see Sec.~\ref{sec:bdtheory} for details).
    The boundary phase is given by the representation category of boundary tube algebra $\Rep(T)$. 
    Then using the boundary-bulk correspondence, the bulk phase is given by $\mathcal{Z}(\Rep(T))$. This means that we can regard boundary tube algebra $T$ as the weak Hopf gauge symmetry for the multifusion string-net $\mathsf{SN}_{\ED}$.
    The weak Hopf charge symmetry is given by the quantum double $D(T)$. Notice that there is another method to construct the weak Hopf charge symmetry directly from the bulk tube algebra $T'=\mathbf{Tube}({_{\ED}}\ED_{\ED})$, which is strongly Morita equivalent to $D(T)$.
    This aspect will be explored in detail in the subsequent sections.
\end{enumerate}

\end{remark}

\vspace{1em}
\emph{Gapped boundary theory.} ---
Utilizing the correspondence between multifusion string-net and lattice gauge theories, we can extend the notion of gauge symmetry and charge symmetry to the $1d$ gapped boundary. For a $2d$ multifusion string-net with input UMFC $\ED=\Rep(W)$, the boundary gauge symmetry is characterized by the $W$-comodule algebra $\mathfrak{A}$ (or equivalently, $W$-module algebra) \cite{Jia2023weak}.
This will be discussed in detail in Sec.~\ref{sec:bdtheory}.
The category $_{\mathfrak{A}}\Mod$ of $\mathfrak{A}$-modules is a module category over $\ED$.
Via the Kitaev-Kong construction, we obtain a boundary lattice model for multifusion string-net (see Sec.~\ref{sec:bdtheory} for details).
In general, this gapped boundary phase is a multifusion phase, which is different from the Levin-Wen string-net model.
If the gapped boundary phase is a multifusion phase, it will be referred to as unstable; if it is a fusion phase, it will be referred to as stable.

The gapped boundary theory of the Levin-Wen string-net model can be macroscopically described by anyon condensation theory. 
For the gapped boundary of a topological phase with topological excitations characterized by a UMTC $\EP$, the gapped boundary is determined by a Lagrangian algebra $\Acal\in \EP$ \cite{Bais2002,bais2003hopf,Bais2009,Kong2014,eliens2010anyon,burnell2018anyon}.
The boundary phase is given by the category of $\Acal$-modules in $\EP$, which will be denoted as $\EP_{\Acal}:={_{\Acal}}\Mod(\EP)$.
Since $\EP_{\Acal}$ is a $1d$ phase, the topological excitations can only fuse but not braid, thus, $\EP_{\Acal}$ is a monoidal category.
It can be proved that the Lagrangian algebra is always simple, thus the boundary phase is always a fusion phase.
For the multifusion string-net, we need to consider the anyon condensation that is described by a disconnected algebra (thus not a Lagrangian algebra) $\mathcal{B}$ over a Lagrangian algebra $\EA$.
The boundary phase is given by $\EP_{\mathcal{B}|\mathcal{B}}:={_{\Bcal}}\Mod_{\Bcal}(\EP)$, the category of $\mathcal{B}|\mathcal{B}$-bimodules in $\EP$. See Sec.~\ref{sec:bdtheory} for details.

From a symmetry point of view, for the gapped boundary of multifusion string-net $\EC=\Rep(W)$, its boundary gauge symmetry is given by a $W$-comodule algebra $\mathfrak{A}$. The category $_{\mathfrak{A}}\Mod$  of $\mathfrak{A}$-modules is a $\EC$-module category.
The boundary charge symmetry is still a weak Hopf algebra $K$ for which $\EP_{\Bcal|\Bcal}\simeq \Rep(K)$ as monoidal categories.
The $K$ and $\mathfrak{A}$ is related by a smash product \cite{jia2023boundary,Jia2023weak}, and $K$
 is strongly Morita equivalent to the boundary tube algebra.

\begin{theorem}
There are two equivalent macroscopic descriptions of the gapped boundary theory for multifusion string-net:
\begin{enumerate}
    \item   Weak Hopf symmetry breaking:  For gapped boundary of the multifusion string-net model $\mathsf{SN}_{\ED}$ with $\ED=\mathsf{Rep}(W)$ for the gauge weak Hopf symmetry $W$, the boundary gauge symmetry is characterized by a $W$-comodule algebra $\mathfrak{A}$. The boundary string labels are given by irreducible $\mathfrak{A}$-modules. 
    The boundary phase is given by the category of $\mathsf{Rep}(W)$-module functors $\mathsf{Fun}_{\mathsf{Rep}(W)}(_{\mathfrak{A}}\Mod,_{\mathfrak{A}}\Mod)$, which is equivalent to the category $_{\mathfrak{A}}\Mod_{\mathfrak{A}}^W$ of $W$-equivariant $\mathfrak{A}|\mathfrak{A}$-bimodules. The boundary weak Hopf charge symmetry is given by the boundary tube algebra.
    \item Anyon condensation: The gapped boundary is determined by a pair of algebras $(\Acal,\Bcal)$ in the bulk phase  $\EP\simeq \mathsf{Rep}(D(W))\simeq \mathcal{Z}(\ED)$, for which $\Acal$ is a Lagrangian algebra and there is an algebra homomorphism $\zeta_{\Acal,\Bcal}:\Acal \to \Bcal$ (this makes $\Bcal$ an algebra over $\Acal$). This $\Bcal$ is in general disconnected, and the boundary phase is given by the category $\EP_{\Bcal|\Bcal}$ of $\Bcal|\Bcal$-bimodules in $\EP$.
\end{enumerate}
\end{theorem}

\begin{remark}
   A gapped boundary can be regarded as a gapped domain wall that separates a multifusion string-net with the trivial phase $\Hilb$.
\end{remark}

\vspace{1em}
\emph{Gapped domain wall theory.} ---
Consider $2d$ topologically ordered phases characterized by UMTCs $\EP_1$ and $\EP_2$. Via the folding trick, a gapped domain wall separating two phases can be transformed into a gapped boundary of the phase $\EP_1\boxtimes \overline{\EP_2}$, where $\overline{\EP_2}$ is the same category with $\EP_2$ but equipped with the anti-braiding and ``$\boxtimes$'' is the Deligne tensor product (physically, it means stacking two topological phases together).
Subsequently, we can employ the gapped boundary theory established above to derive the macroscopic description of gapped domain walls.

From anyon condensation point of view, the gapped domain walls are classified by the Lagrangian algebras $\Acal$ in $\EP_1\boxtimes \overline{\EP_2}$ and an algebra $\Bcal$ in $\EP_1\boxtimes \overline{\EP_2}$ together with an algebra homomorphism $\zeta_{\Acal,\Bcal}:\Acal \to \Bcal$.
The domain wall phase is given by $(\EP_1\boxtimes \overline{\EP_2})_{\Bcal|\Bcal}:={_{\Bcal}}\Mod_{\Bcal}(\EP_1\boxtimes \overline{\EP_2})$, the category of $\Bcal|\Bcal$-bimodules in $\EP_1\boxtimes \overline{\EP_2}$.

From a weak Hopf symmetry point of view, for the multifusion string-net model with two UMFCs $\EC_1=\Rep(W_1)$ and $\EC_2=\Rep(W_2)$, the corresponding bulk phases are Drinfeld centers $\EP_1=\mathcal{Z}(\EC_1)$ and $\EP_2=\mathcal{Z}(\EC_2)$.
The domain wall is characterized by a $\EC_1|\EC_2$-bimodule category $\EM$.
Via folding trick, $\EM$ can be regarded as a $\EC_1\boxtimes \EC_2^{\rm rev}$-left module category, meaning that it can be regarded as a gapped boundary of the multifusion string-net with input UMFC as  $\EC_1\boxtimes \EC_2^{\rm rev}$.
Thus we can apply the the boundary theory to obtain the theory of gapped domain wall.
For this case, the corresponding bulk gauge symmetry is $W_1\otimes W_2^{\rm cop}$, and the domain wall gauge symmetry is a $W_1\otimes W_2^{\rm cop}$-comodule algebra.
This is equivalent to saying that the domain wall weak Hopf gauge symmetry is characterized by a $W_1|W_2$-bicomodule algebra $\mathfrak{A}$.
The charge symmetry of the gapped domain wall is given by the domain wall tube algebra (see Sec.~\ref{sec:walltheory}).

\begin{theorem}
   There are two equivalent macroscopic descriptions of the gapped domain wall that separates two multifusion string-net bulks with input UMFCs $\ED_1=\mathsf{Rep}(W_1)$ and  $\ED_2=\mathsf{Rep}(W_2)$ respectively: 
\begin{enumerate}
    \item Weak Hopf symmetry breaking: For a gapped domain wall with weak Hopf gauge symmetry given by a $W_1|W_2$-bicomodule algebra $\mathfrak{A}$, the category ${_{\mathfrak{A}}}\Mod$ of left $A$-modules is a $\Rep(W_1)|\Rep(W_2)$-bimodule category. The domain wall excitations are given by $\Fun_{\Rep(W_1)|\Rep(W_2)}({_{\mathfrak{A}}}\Mod,{_{\mathfrak{A}}}\Mod)$ which is equivalent to ${_{\mathfrak{A}}}\Mod_{\mathfrak{A}}^{W_1\otimes W_2^{\rm cop}}$, the category of $W_1\otimes W_2^{\rm cop}$-equivariant $\mathfrak{A}|\mathfrak{A}$-bimodules. The charge weak Hopf symmetry is given by the domain wall tube algebra.
    \item Anyon condensation: Two bulk topological phases are given by $\EP_i\simeq \Rep(D(W_i))\simeq \mathcal{Z}(\ED_i)$ with $i=1,2$. The gapped domain wall is characterized by a Lagrangian algebra $\Acal$ in $\EP_1\boxtimes \overline{\EP_2}$ together with an algebra $\Bcal$ in $\EP_1\boxtimes \overline{\EP_2}$ such that there is an algebra homomorphism $\zeta_{\Acal,\Bcal}:\Acal\to \Bcal$. The domain wall excitations are given by  $(\EP_1\boxtimes \overline{\EP_2})_{\Bcal|\Bcal}:={_{\Bcal}}\Mod_{\Bcal}(\EP_1\boxtimes \overline{\EP_2})$, the category of $\Bcal|\Bcal$-bimodules in $\EP_1\boxtimes \overline{\EP_2}$.
\end{enumerate}
\end{theorem}




\section{Generalized multifusion string-net model}
\label{Sec:SNmulti}
In this section, we will provide a thorough analysis of the generalized multifusion string-net model.
While we acknowledge that this aspect is well-established among experts, a comprehensive discussion of the multifusion string-net model has been notably absent, except for a succinct mathematical outline provided in Ref.~\cite{chang2015enriching}. In this section, we endeavor to offer a more elaborate treatment of the subject.

\subsection{Input data of multifusion string-net model}

Before delving into the intricacies of string-net ground states and lattice Hamiltonians, it is prudent to establish a foundation by elucidating the pertinent input data that will be utilized.
The input data of multifusion string-net model is a UMFC $\ED=\oplus_{i,j\in I}\ED_{i,j}=\Rep(W)$ for some weak Hopf algebra $W$. It is worth noting that the construction of a string-net model can also be achieved by addressing the self-consistency conditions associated with certain input data.
Throughout this study, we consistently presume the presence of a UMFC as our input data.

\vspace{1em}
\emph{Label set and involution.} --- 
Let $\Irr(\ED) = 
\{X_{i,j}^k : k=1, \cdots, n_{i,j}, i, j \in I\}$ denote the collection of equivalence classes pertaining to simple objects in the category $\ED$. A label set $L_{\ED}$ is the range of an injective map $L:\Irr(\ED) \to \mathbb{N}$.
We will not distinguish $\Irr(\ED)$ and  $L_{\ED}$ in this work and will use them interchangeably whenever there is no risk of ambiguity.
Notice that $L_{\ED}$ has a grading structure induced by the grading structure of $\ED$, $L_{\ED}=\sqcup_{i,j}L_{\ED_{i,j}}$.
For UMFC $\ED=\oplus_{i,j\in I}\ED_{i,j}$, we denote the label set of all objects $\one_i$, $i\in I$, as $L_{\ED}^0$.
The rigidity structure induces an involution on $L_{\ED}$, $*:L_{\ED} \to L_{\ED}$.
For $X\in L_{\ED}^0$, $X^*=X$; for $X\in L_{\ED_{i,j}}$, $X^*\in L_{\ED_{j,i}}$.

\vspace{1em}
\emph{Fusion rule and quantum dimension.} --- 
The fusion rule of $\ED$ is a map $N:L_{\ED} \times L_{\ED} \times L_{\ED} \to \mathbb{N}$, which is induced by the fusion structure of $\ED$, $X\otimes Y=\oplus_{Z\in \Irr(\ED)} N_{X,Y}^Z Z$.
The fusion rule is associative in the sense that $\sum_{Z\in \Irr(\ED)} N_{X,Y}^ZN_{Z,W}^V=\sum_{Z\in \Irr(\ED)} N_{X,Z}^V N_{Y,W}^Z$.
The fusion rule is not commutative in general: $N_{X,Y}^Z\neq N_{Y,X}^Z$.
For an object $X$, if we regard $N_{X,Y}^Z$ as a matrix with entries $[N_X]_Y^Z=N_{X,Y}^Z$, then from the Frobenius-Perron theorem \cite{etingof2005fusion} we know that the largest eigenvalue is always positive; this number is defined as the quantum dimension of $X$ and denoted as $d_X$. The total quantum dimension of $\ED$ is defined as $d_{\ED}=\sum_{X\in \Irr(\ED)} d_X^2$. $\ED$ is called multiplicity-free if the fusion coefficient $N_{X,Y}^Z\in \{0,1\}$ for all simple objects $X, Y, Z$.

\vspace{1em}
\emph{Local normalization and local evaluation.} ---
In a diagrammatic representation, each edge is represented as a line directing upwards, which is similar to anyon worldline. For fusion/splitting, there are associated vector spaces $V_c^{ab}=\Hom_\ED(c,a\otimes b)$ and $V_{ab}^c=\Hom_\ED(a\otimes b,c)$ spanned by the fusion/splitting vertex vectors. To normalize these local vertex vectors, we can introduce a factor $Y_{c}^{ab}$ in the following manner:
\begin{align} 
	(Y_{c}^{ab})^{-1/2}
	\begin{aligned}
		\begin{tikzpicture}
			\draw[-latex,line width=.6pt,black] (0,-0.5) -- (0,-.1);
			\draw[line width=.6pt,black] (0,-0.4) -- (0,0);
			\draw[-latex,line width=.6pt,black](0,0)--(-0.3,0.3);
			\draw[line width=.6pt,black](-0.1,0.1)--(-0.4,0.4);
			\draw[-latex,line width=.6pt,black](0,0)--(0.3,0.3);
			\draw[line width=.6pt,black](0.1,0.1)--(0.4,0.4);
			\node[ line width=0.6pt, dashed, draw opacity=0.5] (a) at (-0.5,0.6){$a$};
			\node[ line width=0.6pt, dashed, draw opacity=0.5] (a) at (0.5,0.6){$b$};
			\node[ line width=0.6pt, dashed, draw opacity=0.5] (a) at (0,-0.7){$c$};
			\node[ line width=0.6pt, dashed, draw opacity=0.5] (a) at (0.3,-0.1){$\alpha$};
		\end{tikzpicture}
	\end{aligned}
	&=|c\to a,b;\alpha\rangle, \label{eq:vertexvec1} \\
	 (Y_{c}^{ab})^{-1/2}
	\begin{aligned}
		\begin{tikzpicture}
			\draw[-latex,line width=.6pt,black] (0.4,-0.4) -- (0.1,-0.1);
			\draw[line width=.6pt,black] (0.3,-0.3) -- (0,0);
			\draw[-latex,line width=.6pt,black] (-0.4,-0.4) -- (-0.1,-0.1);
			\draw[line width=.6pt,black] (-0.3,-0.3) -- (0,0);
			\draw[-latex,line width=.6pt,black](0,0)--(0,0.3);
			\draw[line width=.6pt,black](0,0.1)--(0,0.5);
			\node[ line width=0.6pt, dashed, draw opacity=0.5] (a) at (0.4,-0.6){$b$};
			\node[ line width=0.6pt, dashed, draw opacity=0.5] (a) at (0,0.7){$c$};
			\node[ line width=0.6pt, dashed, draw opacity=0.5] (a) at (-0.4,-0.6){$a$};
			\node[ line width=0.6pt, dashed, draw opacity=0.5] (a) at (0.3,0.1){$\beta$};
		\end{tikzpicture}
	\end{aligned}
	&=\langle a,b\to c;\beta|. \label{eq:vertexvec2}
\end{align}
Notice that $\dim V^{ab}_{c}=\dim V_{ab}^{c} =N_{ab}^c$.
We can introduce a pairing between $V_{c}^{ab}$ and $V_{ab}^c$ by 
\begin{equation}\label{eq:inner}
	\langle c\to a,b;\alpha | a,b\to c';\beta \rangle = \delta_{c,c'}\delta_{\alpha,\beta}.
\end{equation}
Diagrammatically, this is characterized by the local evaluation
\begin{equation} 
\text{loop move:}\quad	\begin{aligned}
		\begin{tikzpicture}
			\draw[-latex,line width=.6pt,black](0,0)--(0,0.4);
			\draw[line width=.6pt,black](0,0.1)--(0,0.5);
			\draw[line width=.6pt,black](0,-0.3) circle (0.3);
			\draw[-latex,line width=.6pt,black](-0.3,-0.3)--(-.3,-.2);
			\draw[-latex,line width=.6pt,black](0.3,-0.3)--(.3,-.2);
			\draw[-latex,line width=.6pt,black](0,-1.1)--(0,-0.7);
			\draw[line width=.6pt,black](0,-1.1)--(0,-0.6);
			\node[ line width=0.6pt, dashed, draw opacity=0.5] (a) at (0.6,-0.3){$b$};
			\node[ line width=0.6pt, dashed, draw opacity=0.5] (a) at (0,0.75){$c'$};
			\node[ line width=0.6pt, dashed, draw opacity=0.5] (a) at (0,-1.3){$c$};
			\node[ line width=0.6pt, dashed, draw opacity=0.5] (a) at (-0.6,-0.3){$a$};
			\node[ line width=0.6pt, dashed, draw opacity=0.5] (a) at (0.3,0.2){$\beta$};
			\node[ line width=0.6pt, dashed, draw opacity=0.5] (a) at (-0.3,-0.8){$\alpha$};
		\end{tikzpicture}
	\end{aligned}
	=
	\delta_{c,c'}\delta_{\alpha,\beta} {Y_c^{ab}}
	\begin{aligned}
		\begin{tikzpicture}
			\draw[line width=.6pt,black](0,-1.1)--(0,0.5);
			\draw[-latex,line width=.6pt,black](0,-0.3)--(0,0);
			\node[ line width=0.6pt, dashed, draw opacity=0.5] (a) at (0,0.7){$c$};
		\end{tikzpicture}
	\end{aligned}. \label{eq:loopev}
\end{equation}
This means that we read the left-hand side of Eq.~\eqref{eq:inner} from the right to the left and read the diagrams from the bottom to the top.
Similarly, the completeness of the fusion basis can be expressed as
\begin{equation}
	\mathds{I}=\sum_{c,\alpha} |c\to a,b;\alpha\rangle \langle a,b \to c;\alpha|,
\end{equation}
where $\mathds{I}$ is identity over the linear space $\Hom_{\ED}(a\otimes b,a\otimes b)$.
Diagrammatically, this corresponds to
 \begin{equation} \label{eq:paraev}
	\text{parallel move:}\quad\begin{aligned}
		\begin{tikzpicture}
			\draw[line width=.6pt,black](0,-1.1)--(0,0.5);
			\draw[-latex,line width=.6pt,black](0,-0.3)--(0,0);
			\node[ line width=0.6pt, dashed, draw opacity=0.5] (a) at (0,0.7){$a$};
			\draw[line width=.6pt,black](0.6,-1.1)--(0.6,0.5);
			\draw[-latex,line width=.6pt,black](0.6,-0.3)--(0.6,0);
			\node[ line width=0.6pt, dashed, draw opacity=0.5] (a) at (0.6,0.7){$b$};
		\end{tikzpicture}
	\end{aligned}
	=
	\sum_{c,\alpha}\frac{1}{Y_c^{ab}}
	\begin{aligned}
		\begin{tikzpicture}
			\draw[-latex,line width=.6pt,black] (0.4,-0.4) -- (0.1,-0.1);
			\draw[line width=.6pt,black] (0.3,-0.3) -- (0,0);
			\draw[-latex,line width=.6pt,black] (-0.4,-0.4) -- (-0.1,-0.1);
			\draw[line width=.6pt,black] (-0.3,-0.3) -- (0,0);
			\draw[-latex,line width=.6pt,black](0,0)--(0,0.45);
			\draw[line width=.6pt,black](0,0.1)--(0,0.7);
			\draw[line width=.6pt,black](-0.4,1.1)--(0,0.7);
			\draw[line width=.6pt,black](0.4,1.1)--(0,0.7);
			\draw[-latex,line width=.6pt,black](0,0.7)--(0.3,1);
			\draw[line width=.6pt,black](0.4,1.1)--(0,0.7);
			\draw[-latex,line width=.6pt,black](0,0.7)--(-0.3,1);
			\node[ line width=0.6pt, dashed, draw opacity=0.5] (a) at (0.4,-0.6){$b$};
			\node[ line width=0.6pt, dashed, draw opacity=0.5] (a) at (-0.4,-0.6){$a$};
			\node[ line width=0.6pt, dashed, draw opacity=0.5] (a) at (0.3,0.1){$\alpha$};
			\node[ line width=0.6pt, dashed, draw opacity=0.5] (a) at (0.3,0.6){$\alpha$};
			\node[ line width=0.6pt, dashed, draw opacity=0.5] (a) at (-0.2,0.3){$c$};
			\node[ line width=0.6pt, dashed, draw opacity=0.5] (a) at (0.4,1.3){$b$};
			\node[ line width=0.6pt, dashed, draw opacity=0.5] (a) at (-0.4,1.3){$a$};
		\end{tikzpicture}
	\end{aligned}.
\end{equation}
Typically, the choice for $Y^{ab}_c$ is made as $(d_ad_b/d_c)^{1/2}$. It is noteworthy that this corresponds to a specific gauge selection.

\vspace{1em}
\emph{F-symbols and evaluation.} --- 
The associativity isomorphism $a: (i\otimes j)\otimes k \to i \otimes (j\otimes k)$ induces a linear map $F_{l}^{ijk}:\Hom_{\ED}(l,(i\otimes j)\otimes k)\to \Hom_{\ED}(l,i\otimes (j\otimes k))$. Similarly, $a^{\dagger}$ induces $(F_{l}^{ijk})^{\dagger}$. Since $\ED$ is unitary, we have $(F_{l}^{ijk})^{\dagger}=(F_{l}^{ijk})^{-1}$.  Diagrammatically, these maps are called F-moves:
\begin{gather}
	\begin{aligned}
		\begin{tikzpicture}
			\draw[-latex,line width=.6pt,black] (0,-0.5) -- (0,-.1);
			\draw[line width=.6pt,black] (0,-0.4) -- (0,0);
			\draw[-latex,line width=.6pt,black](0,0)--(-0.3,0.3);
			\draw[line width=.6pt,black](-0.1,0.1)--(-0.4,0.4);
			\draw[-latex,line width=.6pt,black](-0.4,0.4)--(-0.7,0.7);
			\draw[line width=.6pt,black](-0.4,0.4)--(-0.8,0.8);
			\draw[-latex,line width=.6pt,black](-0.4,0.4)--(-0.1,0.7);
			\draw[line width=.6pt,black](-0.4,0.4)--(-0,0.8);
			\draw[-latex,line width=.6pt,black](0,0)--(0.7,0.7);
			\draw[line width=.6pt,black](0.1,0.1)--(0.8,0.8);
			\node[ line width=0.6pt, dashed, draw opacity=0.5] (a) at (-0.8,1){$i$};
			\node[ line width=0.6pt, dashed, draw opacity=0.5] (a) at (0,1){$j$};
			\node[ line width=0.6pt, dashed, draw opacity=0.5] (a) at (0.8,1){$k$};
			\node[ line width=0.6pt, dashed, draw opacity=0.5] (a) at (0,-0.7){$l$};
			\node[ line width=0.6pt, dashed, draw opacity=0.5] (a) at (0.3,-0.1){$\alpha$};
			\node[ line width=0.6pt, dashed, draw opacity=0.5] (a) at (-0.7,0.3){$\beta$};
			\node[ line width=0.6pt, dashed, draw opacity=0.5] (a) at (-0.4,0){$m$};
		\end{tikzpicture}
	\end{aligned}
	 =\sum_{n,\mu,\nu} [F^{ijk}_{l}]^{n\mu\nu}_{m\alpha\beta}
	\begin{aligned}
		\begin{tikzpicture}
			\draw[-latex,line width=.6pt,black] (0,-0.5) -- (0,-.1);
			\draw[line width=.6pt,black] (0,-0.4) -- (0,0);
			\draw[-latex,line width=.6pt,black](0,0)--(-0.7,0.7);
			\draw[line width=.6pt,black](0,0)--(-0.8,0.8);
			\draw[-latex,line width=.6pt,black](0.4,0.4)--(0.1,0.7);
			\draw[line width=.6pt,black](0.4,0.4)--(0,0.8);
			\draw[-latex,line width=.6pt,black](0,0)--(0.7,0.7);
			\draw[-latex,line width=.6pt,black](0,0)--(0.3,0.3);
			\draw[line width=.6pt,black](0.1,0.1)--(0.8,0.8);
			\node[ line width=0.6pt, dashed, draw opacity=0.5] (a) at (-0.8,1){$i$};
			\node[ line width=0.6pt, dashed, draw opacity=0.5] (a) at (0,1){$j$};
			\node[ line width=0.6pt, dashed, draw opacity=0.5] (a) at (0.8,1){$k$};
			\node[ line width=0.6pt, dashed, draw opacity=0.5] (a) at (0,-0.7){$l$};
			\node[ line width=0.6pt, dashed, draw opacity=0.5] (a) at (0.3,-0.1){$\mu$};
			\node[ line width=0.6pt, dashed, draw opacity=0.5] (a) at (0,0.3){$n$};
				\node[ line width=0.6pt, dashed, draw opacity=0.5] (a) at (0.6,0.3){$\nu$};
		\end{tikzpicture}
	\end{aligned}, \label{eq:F-move1} \\
 \begin{aligned}
		\begin{tikzpicture}
			\draw[-latex,line width=.6pt,black] (0,-0.5) -- (0,-.1);
			\draw[line width=.6pt,black] (0,-0.4) -- (0,0);
			\draw[-latex,line width=.6pt,black](0,0)--(-0.7,0.7);
			\draw[line width=.6pt,black](0,0)--(-0.8,0.8);
			\draw[-latex,line width=.6pt,black](0.4,0.4)--(0.1,0.7);
			\draw[line width=.6pt,black](0.4,0.4)--(0,0.8);
			\draw[-latex,line width=.6pt,black](0,0)--(0.7,0.7);
			\draw[-latex,line width=.6pt,black](0,0)--(0.3,0.3);
			\draw[line width=.6pt,black](0.1,0.1)--(0.8,0.8);
			\node[ line width=0.6pt, dashed, draw opacity=0.5] (a) at (-0.8,1){$i$};
			\node[ line width=0.6pt, dashed, draw opacity=0.5] (a) at (0,1){$j$};
			\node[ line width=0.6pt, dashed, draw opacity=0.5] (a) at (0.8,1){$k$};
			\node[ line width=0.6pt, dashed, draw opacity=0.5] (a) at (0,-0.7){$l$};
			\node[ line width=0.6pt, dashed, draw opacity=0.5] (a) at (0.3,-0.1){$\mu$};
			\node[ line width=0.6pt, dashed, draw opacity=0.5] (a) at (0,0.3){$n$};
				\node[ line width=0.6pt, dashed, draw opacity=0.5] (a) at (0.6,0.3){$\nu$};
		\end{tikzpicture}
	\end{aligned}
	 =\sum_{m,\alpha,\beta} [(F^{ijk}_{l})^{-1}]_{n\mu\nu}^{m\alpha\beta}
	\begin{aligned}
        \begin{tikzpicture}
			\draw[-latex,line width=.6pt,black] (0,-0.5) -- (0,-.1);
			\draw[line width=.6pt,black] (0,-0.4) -- (0,0);
			\draw[-latex,line width=.6pt,black](0,0)--(-0.3,0.3);
			\draw[line width=.6pt,black](-0.1,0.1)--(-0.4,0.4);
			\draw[-latex,line width=.6pt,black](-0.4,0.4)--(-0.7,0.7);
			\draw[line width=.6pt,black](-0.4,0.4)--(-0.8,0.8);
			\draw[-latex,line width=.6pt,black](-0.4,0.4)--(-0.1,0.7);
			\draw[line width=.6pt,black](-0.4,0.4)--(-0,0.8);
			\draw[-latex,line width=.6pt,black](0,0)--(0.7,0.7);
			\draw[line width=.6pt,black](0.1,0.1)--(0.8,0.8);
			\node[ line width=0.6pt, dashed, draw opacity=0.5] (a) at (-0.8,1){$i$};
			\node[ line width=0.6pt, dashed, draw opacity=0.5] (a) at (0,1){$j$};
			\node[ line width=0.6pt, dashed, draw opacity=0.5] (a) at (0.8,1){$k$};
			\node[ line width=0.6pt, dashed, draw opacity=0.5] (a) at (0,-0.7){$l$};
			\node[ line width=0.6pt, dashed, draw opacity=0.5] (a) at (0.3,-0.1){$\alpha$};
			\node[ line width=0.6pt, dashed, draw opacity=0.5] (a) at (-0.7,0.3){$\beta$};
			\node[ line width=0.6pt, dashed, draw opacity=0.5] (a) at (-0.4,0){$m$};
		\end{tikzpicture}
	\end{aligned}, \label{eq:F-move2}  \\
 \begin{aligned}
		\begin{tikzpicture}
			\draw[-latex,line width=.6pt,black] (0,0) -- (0,0.4);
			\draw[line width=.6pt,black] (0,0.5) -- (0,0.1);
			\draw[-latex,line width=.6pt,black](-0.4,-0.4)--(-0.1,-0.1);
			\draw[line width=.6pt,black](-0.3,-0.3)--(0,0);
			\draw[-latex,line width=.6pt,black](-0.8,-0.8)--(-0.5,-0.5);
			\draw[line width=.6pt,black](-0.7,-0.7)--(-0.4,-0.4);
			\draw[-latex,line width=.6pt,black](0,-0.8)--(-0.3,-0.5);
			\draw[line width=.6pt,black](-0.4,-0.4)--(-0.1,-0.7);
			\draw[-latex,line width=.6pt,black](0.8,-0.8)--(0.5,-0.5);
			\draw[line width=.6pt,black](0,0)--(0.7,-0.7);
			\node[ line width=0.6pt, dashed, draw opacity=0.5] (a) at (-0.8,-1){$i$};
			\node[ line width=0.6pt, dashed, draw opacity=0.5] (a) at (0,-1){$j$};
			\node[ line width=0.6pt, dashed, draw opacity=0.5] (a) at (0.8,-1){$k$};
			\node[ line width=0.6pt, dashed, draw opacity=0.5] (a) at (0,0.7){$l$};
			\node[ line width=0.6pt, dashed, draw opacity=0.5] (a) at (0.3,0.1){$\alpha$};
			\node[ line width=0.6pt, dashed, draw opacity=0.5] (a) at (-0.7,-0.3){$\beta$};
			\node[ line width=0.6pt, dashed, draw opacity=0.5] (a) at (-0.4,0){$m$};
		\end{tikzpicture}
	\end{aligned}
	 =\sum_{n,\mu,\nu} [F_{ijk}^{l}]^{n\mu\nu}_{m\alpha\beta}
	\begin{aligned}
		\begin{tikzpicture}
			\draw[-latex,line width=.6pt,black] (0,0) -- (0,0.4);
			\draw[line width=.6pt,black] (0,0.5) -- (0,0.1);
			\draw[-latex,line width=.6pt,black](-0.8,-0.8)--(-0.5,-0.5);
			\draw[line width=.6pt,black](-0.7,-0.7)--(0,0);
			\draw[-latex,line width=.6pt,black](0,-0.8)--(0.3,-0.5);
			\draw[line width=.6pt,black](0.4,-0.4)--(0,-0.8);
			\draw[-latex,line width=.6pt,black](0.8,-0.8)--(0.5,-0.5);
			\draw[line width=.6pt,black](0,0)--(0.7,-0.7);
			\node[ line width=0.6pt, dashed, draw opacity=0.5] (a) at (-0.8,-1){$i$};
			\node[ line width=0.6pt, dashed, draw opacity=0.5] (a) at (0,-1){$j$};
			\node[ line width=0.6pt, dashed, draw opacity=0.5] (a) at (0.8,-1){$k$};
			\node[ line width=0.6pt, dashed, draw opacity=0.5] (a) at (0,0.7){$l$};
			\node[ line width=0.6pt, dashed, draw opacity=0.5] (a) at (0.3,0.1){$\mu$};
			\node[ line width=0.6pt, dashed, draw opacity=0.5] (a) at (0.65,-0.3){$\nu$};
			\node[ line width=0.6pt, dashed, draw opacity=0.5] (a) at (0.05,-0.35){$n$};
		\end{tikzpicture}
	\end{aligned}, \label{eq:F-move3}  \\
 \begin{aligned}
        \begin{tikzpicture}
			\draw[-latex,line width=.6pt,black] (0,0) -- (0,0.4);
			\draw[line width=.6pt,black] (0,0.5) -- (0,0.1);
			\draw[-latex,line width=.6pt,black](-0.8,-0.8)--(-0.5,-0.5);
			\draw[line width=.6pt,black](-0.7,-0.7)--(0,0);
			\draw[-latex,line width=.6pt,black](0,-0.8)--(0.3,-0.5);
			\draw[line width=.6pt,black](0.4,-0.4)--(0,-0.8);
			\draw[-latex,line width=.6pt,black](0.8,-0.8)--(0.5,-0.5);
			\draw[line width=.6pt,black](0,0)--(0.7,-0.7);
			\node[ line width=0.6pt, dashed, draw opacity=0.5] (a) at (-0.8,-1){$i$};
			\node[ line width=0.6pt, dashed, draw opacity=0.5] (a) at (0,-1){$j$};
			\node[ line width=0.6pt, dashed, draw opacity=0.5] (a) at (0.8,-1){$k$};
			\node[ line width=0.6pt, dashed, draw opacity=0.5] (a) at (0,0.7){$l$};
			\node[ line width=0.6pt, dashed, draw opacity=0.5] (a) at (0.3,0.1){$\mu$};
			\node[ line width=0.6pt, dashed, draw opacity=0.5] (a) at (0.65,-0.3){$\nu$};
			\node[ line width=0.6pt, dashed, draw opacity=0.5] (a) at (0.05,-0.35){$n$};
		\end{tikzpicture}
	\end{aligned}
	 =\sum_{m,\alpha,\beta} [(F_{ijk}^{l})^{-1}]_{n\mu\nu}^{m\alpha\beta}
	\begin{aligned}
		\begin{tikzpicture}
			\draw[-latex,line width=.6pt,black] (0,0) -- (0,0.4);
			\draw[line width=.6pt,black] (0,0.5) -- (0,0.1);
			\draw[-latex,line width=.6pt,black](-0.4,-0.4)--(-0.1,-0.1);
			\draw[line width=.6pt,black](-0.3,-0.3)--(0,0);
			\draw[-latex,line width=.6pt,black](-0.8,-0.8)--(-0.5,-0.5);
			\draw[line width=.6pt,black](-0.7,-0.7)--(-0.4,-0.4);
			\draw[-latex,line width=.6pt,black](0,-0.8)--(-0.3,-0.5);
			\draw[line width=.6pt,black](-0.4,-0.4)--(-0.1,-0.7);
			\draw[-latex,line width=.6pt,black](0.8,-0.8)--(0.5,-0.5);
			\draw[line width=.6pt,black](0,0)--(0.7,-0.7);
			\node[ line width=0.6pt, dashed, draw opacity=0.5] (a) at (-0.8,-1){$i$};
			\node[ line width=0.6pt, dashed, draw opacity=0.5] (a) at (0,-1){$j$};
			\node[ line width=0.6pt, dashed, draw opacity=0.5] (a) at (0.8,-1){$k$};
			\node[ line width=0.6pt, dashed, draw opacity=0.5] (a) at (0,0.7){$l$};
			\node[ line width=0.6pt, dashed, draw opacity=0.5] (a) at (0.3,0.1){$\alpha$};
			\node[ line width=0.6pt, dashed, draw opacity=0.5] (a) at (-0.7,-0.3){$\beta$};
			\node[ line width=0.6pt, dashed, draw opacity=0.5] (a) at (-0.4,0){$m$};
		\end{tikzpicture}
	\end{aligned}.   \label{eq:F-move4} 
\end{gather}
The associator in UMFC adheres to the pentagon relation, as depicted in Fig.~\ref{fig:pentagon}. In explicit terms, this relationship is expressed as
\begin{gather}
    \sum_{\mu}[F_a^{mkl}]_{n\alpha\beta}^{b\mu\nu}[F_a^{ijb}]_{m\mu\gamma}^{c\sigma\zeta} = \sum_{d,\lambda,\tau,\nu} [F_n^{ijk}]^{d\lambda\tau}_{m\beta\gamma}[F_a^{idl}]^{c\sigma\eta}_{n\alpha\lambda}[F_c^{jkl}]_{d\tau\eta}^{b\zeta\nu}. 
\end{gather}
Since $\ED$ is unitary, the F-matrix can also be chosen as a unitary matrix
\begin{equation} 
    (F_{l}^{ijk})^{\dagger}=(F_{l}^{ijk})^{-1}.
\end{equation}
The F-symbols exhibit a mirror conjugate symmetry, a property pivotal for the construction of the Levin-Wen model with a boundary. By drawing parallels to the methods outlined in Ref.~\cite{hong2009symmetrization}, we can demonstrate that
\begin{equation} \label{eq:F-symbol-delta}
    \sum_{n,\mu,\nu} [F_l^{ijk}]_{m\alpha\beta}^{n\mu\nu}[F^l_{ijk}]_{m'\alpha'\beta'}^{n\mu\nu}=\delta_{m,m'}\delta_{\alpha,\alpha'}\delta_{\beta,\beta'}.
\end{equation}

\begin{figure}[t]
		\centering
		\includegraphics[width=10.5cm]{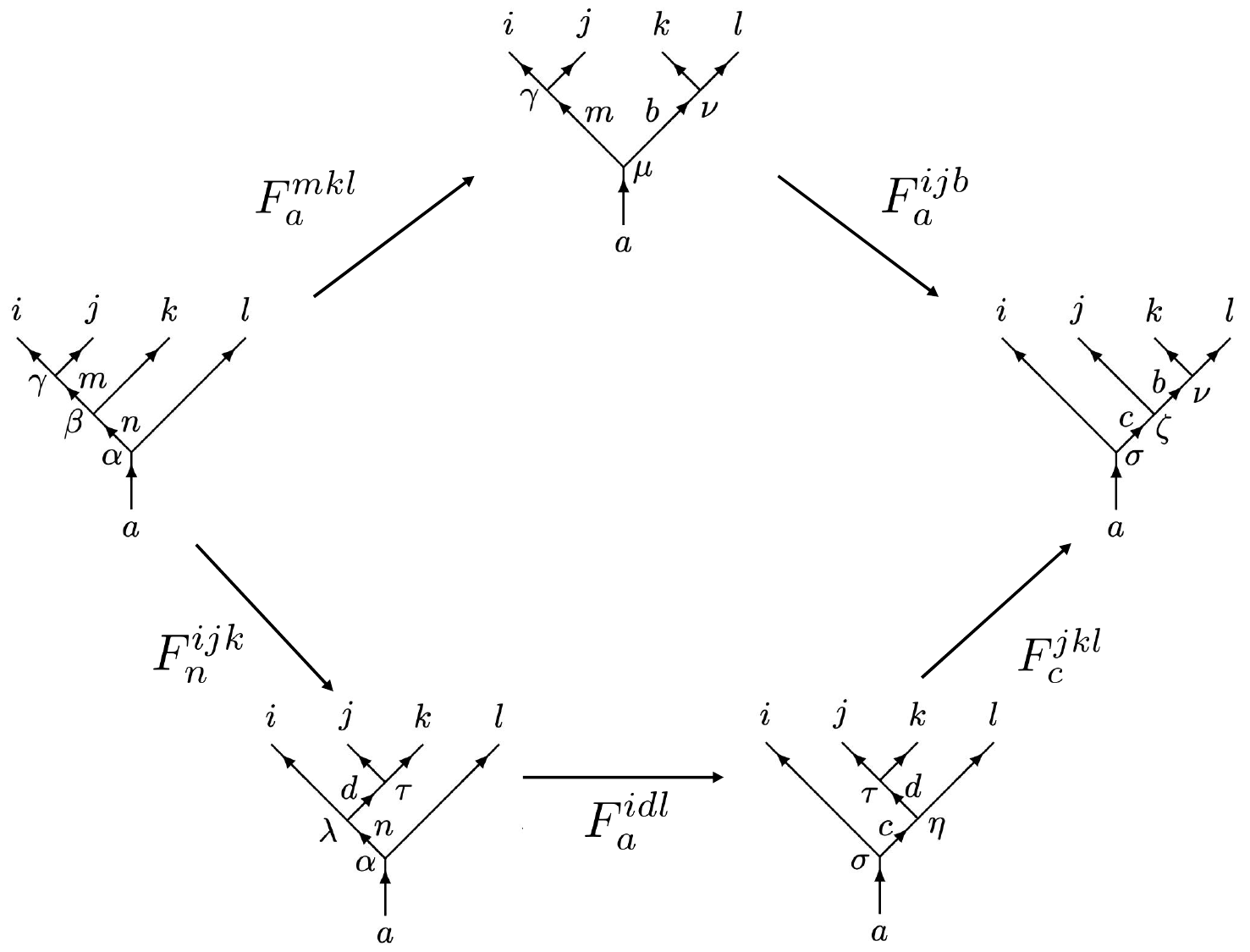}
		\caption{Pentagon relation for the F-symbols. \label{fig:pentagon}}
\end{figure}

\begin{proposition}
 The loop move, parallel move, and F-move, collectively known as topological local moves, are equivalent to the Pachner moves.
\end{proposition}

\begin{proof}
  This can be demonstrated in a manner similar to the fusion category case. See Refs.~\cite{Hu2012ground,hu2018boundary}.
\end{proof}

\emph{Markov trace.} --- 
Consider a general morphism  $X\in \Hom_{\ED}$ that has the same input and output objects, we have the following string diagram:
\begin{equation}
A =	\begin{aligned}
		\begin{tikzpicture}
			\draw[-latex,line width=.6pt,black] (-1.2,-0.5) -- (-1.2,-.1);
			\draw[line width=.6pt,black] (-1.2,-0.5) -- (-1.2,0);
			\node[ line width=0.6pt, dashed, draw opacity=0.5] (a) at (-1.2,-0.7){$a_1$};
			\draw[-latex,line width=.6pt,black] (-0.7,-0.5) -- (-0.7,-.1);
			\draw[line width=.6pt,black] (-0.7,-0.5) -- (-0.7,0);
			\node[ line width=0.6pt, dashed, draw opacity=0.5] (a) at (-0.7,-0.7){$a_2$};
			\draw[-latex,line width=.6pt,black] (1.2,-0.5) -- (1.2,-.1);
			\draw[line width=.6pt,black] (1.2,-0.5) -- (1.2,0);
			\node[ line width=0.6pt, dashed, draw opacity=0.5] (a) at (1.2,-0.7){$a_n$};
			\draw[line width=0.6 pt, fill=gray, fill opacity=0.2] 
			(-1.5,0) -- (1.5,0) -- (1.5,0.5) -- (-1.5,0.5) -- cycle; 
			\node[ line width=0.6pt, dashed, draw opacity=0.5] (a) at (0,0.25){$X$};
			\node[ line width=0.6pt, dashed, draw opacity=0.5] (a) at (0.3,-0.5){$\cdots$};
			\draw[-latex,line width=.6pt,black] (-1.2,0.5) -- (-1.2,0.85);
			\draw[line width=.6pt,black] (-1.2,0.5) -- (-1.2,1);
			\node[ line width=0.6pt, dashed, draw opacity=0.5] (a) at (-1.2,1.2){$a_1$};
			\draw[-latex,line width=.6pt,black] (1.2,0.5) -- (1.2,0.85);
			\draw[line width=.6pt,black] (1.2,0.5) -- (1.2,1);
			\node[ line width=0.6pt, dashed, draw opacity=0.5] (a) at (1.2,1.2){$a_n$};
			\draw[-latex,line width=.6pt,black] (-0.7,0.5) -- (-0.7,0.85);
			\draw[line width=.6pt,black] (-0.7,0.5) -- (-0.7,1);
			\node[ line width=0.6pt, dashed, draw opacity=0.5] (a) at (-0.7,1.2){$a_2$};
			\node[ line width=0.6pt, dashed, draw opacity=0.5] (a) at (0.3,1.0){$\cdots$};
		\end{tikzpicture}
	\end{aligned}\in V_{a_1\cdots a_n}^{a_1\cdots a_n},
\end{equation}
where all the coefficients and different inner vertex morphism labels are contained in the box.
Its \emph{Markov trace} is defined as follows
\begin{align}
	\tilde{\Tr} \left(  \begin{aligned}
		\begin{tikzpicture}
			\draw[-latex,line width=.6pt,black] (-1.2,-0.5) -- (-1.2,-.1);
			\draw[line width=.6pt,black] (-1.2,-0.5) -- (-1.2,0);
			\node[ line width=0.6pt, dashed, draw opacity=0.5] (a) at (-1.2,-0.7){$a_1$};
			\draw[-latex,line width=.6pt,black] (-0.7,-0.5) -- (-0.7,-.1);
			\draw[line width=.6pt,black] (-0.7,-0.5) -- (-0.7,0);
			\node[ line width=0.6pt, dashed, draw opacity=0.5] (a) at (-0.7,-0.7){$a_2$};
			\draw[-latex,line width=.6pt,black] (1.2,-0.5) -- (1.2,-.1);
			\draw[line width=.6pt,black] (1.2,-0.5) -- (1.2,0);
			\node[ line width=0.6pt, dashed, draw opacity=0.5] (a) at (1.2,-0.7){$a_n$};
			\draw[line width=0.6 pt, fill=gray, fill opacity=0.2] 
			(-1.5,0) -- (1.5,0) -- (1.5,0.5) -- (-1.5,0.5) -- cycle; 
			\node[ line width=0.6pt, dashed, draw opacity=0.5] (a) at (0,0.25){$X$};
			\node[ line width=0.6pt, dashed, draw opacity=0.5] (a) at (0.3,-0.5){$\cdots$};
			\draw[-latex,line width=.6pt,black] (-1.2,0.5) -- (-1.2,0.85);
			\draw[line width=.6pt,black] (-1.2,0.5) -- (-1.2,1);
			\node[ line width=0.6pt, dashed, draw opacity=0.5] (a) at (-1.2,1.2){$a_1$};
			\draw[-latex,line width=.6pt,black] (1.2,0.5) -- (1.2,0.85);
			\draw[line width=.6pt,black] (1.2,0.5) -- (1.2,1);
			\node[ line width=0.6pt, dashed, draw opacity=0.5] (a) at (1.2,1.2){$a_n$};
			\draw[-latex,line width=.6pt,black] (-0.7,0.5) -- (-0.7,0.85);
			\draw[line width=.6pt,black] (-0.7,0.5) -- (-0.7,1);
			\node[ line width=0.6pt, dashed, draw opacity=0.5] (a) at (-0.7,1.2){$a_2$};
			\node[ line width=0.6pt, dashed, draw opacity=0.5] (a) at (0.3,1.0){$\cdots$};
		\end{tikzpicture}
	\end{aligned} \right)
=\quad  \begin{aligned}
	\begin{tikzpicture}
		\draw[-latex,line width=.6pt,black] (-1.2,-0.5) -- (-1.2,-.1);
		\draw[line width=.6pt,black] (-1.2,-0.5) -- (-1.2,0);
		\node[ line width=0.6pt, dashed, draw opacity=0.5] (a) at (-1.5,-0.4){$a_1$};
		\draw[-latex,line width=.6pt,black] (-0.7,-0.5) -- (-0.7,-.1);
		\draw[line width=.6pt,black] (-0.7,-0.5) -- (-0.7,0);
		\node[ line width=0.6pt, dashed, draw opacity=0.5] (a) at (-0.3,-0.4){$a_2$};
		\draw[-latex,line width=.6pt,black] (1.2,-0.5) -- (1.2,-.1);
		\draw[line width=.6pt,black] (1.2,-0.5) -- (1.2,0);
		\node[ line width=0.6pt, dashed, draw opacity=0.5] (a) at (0.8,-0.4){$a_n$};
		\draw[line width=0.6 pt, fill=gray, fill opacity=0.2] 
		(-1.5,0) -- (1.5,0) -- (1.5,0.5) -- (-1.5,0.5) -- cycle; 
		\node[ line width=0.6pt, dashed, draw opacity=0.5] (a) at (0,0.25){$X$};
		\node[ line width=0.6pt, dashed, draw opacity=0.5] (a) at (0.3,-0.5){$\cdots$};
		\draw[-latex,line width=.6pt,black] (-1.2,0.5) -- (-1.2,0.85);
		\draw[line width=.6pt,black] (-1.2,0.5) -- (-1.2,1);
		\node[ line width=0.6pt, dashed, draw opacity=0.5] (a) at (-1.5,0.9){$a_1$};
		\draw[-latex,line width=.6pt,black] (1.2,0.5) -- (1.2,0.85);
		\draw[line width=.6pt,black] (1.2,0.5) -- (1.2,1);
		\node[ line width=0.6pt, dashed, draw opacity=0.5] (a) at (0.8,0.9){$a_n$};
		\draw[-latex,line width=.6pt,black] (-0.7,0.5) -- (-0.7,0.85);
		\draw[line width=.6pt,black] (-0.7,0.5) -- (-0.7,1);
		\node[ line width=0.6pt, dashed, draw opacity=0.5] (a) at (-0.4,0.9){$a_2$};
		\node[ line width=0.6pt, dashed, draw opacity=0.5] (a) at (0.3,1.0){$\cdots$};
		\draw[line width=.6pt,black] (1.2,1.0) -- (1.2,1.2)--(1.7,1.2)--(1.7,-0.7)--(1.2,-0.7)--(1.2,-0.5);
		\draw[line width=.6pt,black] (-0.7,1.0) -- (-0.7,1.4)--(1.9,1.4)--(1.9,-0.9)--(-0.7,-0.9)--(-0.7,-0.5);
		\draw[line width=.6pt,black] (-1.2,1.0) -- (-1.2,1.6)--(2.1,1.6)--(2.1,-1.1)--(-1.2,-1.1)--(-1.2,-0.5);
	\end{tikzpicture}
\end{aligned}\quad. \label{eq:Markovtrace}
\end{align}

\begin{definition}[Input data]
    The input data for the generalized multifusion string-net are:
    \begin{enumerate}
        \item String type: $\Irr(\ED)$;
        \item Fusion rule: $N_{a,b}^c$\,\footnote{Recall that the quantum dimensions $d_a$'s are determined by the fusion rule.};
        \item Local normalization factor: $Y_{c}^{ab}$.
        \item F-symbols: $F_{l}^{ijk}$, $F_{ijk}^l$.
    \end{enumerate}  
\end{definition}

\begin{remark}[Tetrahedral symmetry]
The original Levin-Wen string-net model requires that F-symbols satisfy the tetrahedral symmetry \cite{Levin2005,chang2015enriching}. 
However, it is known that not all UMFCs allow F-symbols that have tetrahedral symmetry \cite{turaev2016quantum,hong2009symmetrization,fuchs2023tetrahedral}. For the unimodal ribbon category, Turaev proved that it is always possible to have tetrahedral symmetric F-symbols \cite{turaev2016quantum}. Recently, a more comprehensive understanding of the tetrahedral symmetry is provided in Ref.~\cite{fuchs2023tetrahedral}.
\end{remark}



\subsection{Generalized multifusion string-net ground states}

Imagine a connected oriented trivalent lattice $\Sigma$, called a string-net\,\footnote{It is worth noting that bivalent vertices are permissible, provided one of the edges is designated as vacant edge $\one \in \ED$. Given that vacant edges can be arbitrarily added or removed, this does not alter the model's fundamentals.}, situated on a spatial $2d$-manifold. 
We will assume that the string-net forms a cellulation of the given surface and denote the sets of vertices, edges and faces as $C^k(\Sigma)$, $k=0,1,2$ respectively.
A \emph{generalized string-net} is characterized by edges that are oriented in either an upward or downward direction\,\footnote{Using the Whitney embedding theorem, the string-net can be embedded into $\mathbb{R}^N$, allowing for the definition of the up or down direction in relation to a specified axis.
It is also noteworthy that, we only consider the spatial manifold $\Sigma$, for which a trivalent lattice can always be identified on it.
}. See Fig.~\ref{fig:StringNet} for an illustration.
In this lattice, the physical degrees of freedom are located on its edges and vertices.
We label edges with simple objects in $\ED$ and vertices with morphisms in $\ED$.
The Hilbert space $\mathcal{H}[\Sigma, \ED]$ is spanned by all configurations of the labels on edges and vertices.

In a graphical representation, each bulk edge is labeled by elements in $\Irr(\ED)$, and if we reverse the direction of one edge labeled by $k$, then the label should be replaced by $k^*$. 
The vacant edge, which is usually drawn as dotted line or  is invisible, is labeled by $\mathds{1}_i$. To summarize: 
\begin{gather}\label{eq:edgelabel}
 \mathrm{bulk\,  edge:}
 \begin{aligned}
		\begin{tikzpicture}
			\draw[line width=.6pt,black](0,-1.1)--(0,0.5);
			\draw[-latex,line width=.6pt,black](0,-0.3)--(0,0);
			\node[ line width=0.6pt, dashed, draw opacity=0.5] (a) at (0,0.8){$k$};
		\end{tikzpicture}
	\end{aligned}
 =
         \begin{aligned}
		\begin{tikzpicture}
			\draw[line width=.6pt,black](0,-1.1)--(0,0.5);
			\draw[-latex,line width=.6pt,black](0,0)--(0,-0.3);
			\node[ line width=0.6pt, dashed, draw opacity=0.5] (a) at (0,0.8){$k^*$};
		\end{tikzpicture}
	\end{aligned}=\begin{aligned}
		\begin{tikzpicture}
			\draw[line width=.6pt,black](0,-1.1)--(0,0.5);
			\draw[-latex,line width=.6pt,black](0,-0.3)--(0,0);
			\node[ line width=0.6pt, dashed, draw opacity=0.5] (a) at (0,0.8){$k^{**}$};
		\end{tikzpicture}
	\end{aligned},\quad 
  \text{vacant\, edge}:
  \begin{aligned}
		\begin{tikzpicture}
			\draw[dotted,line width=.8pt,black](0,-1.1)--(0,0.5);
			\node[ line width=0.6pt, dashed, draw opacity=0.5] (a) at (0,0.8){$\mathbf{1}_i$};
		\end{tikzpicture}
	\end{aligned}.
\end{gather}
It is important to observe that in multifusion string-net models, vacant edges are represented by a set of edge labels denoted as $\mathbf{1}_i$, where $i \in I$. Unlike the case of original fusion string-nets, the addition or removal of vacant edges is constrained. Specifically, the condition $\one_i\otimes X \neq 0$ holds only when $X\in \Irr(\ED_{i,j})$. Thus, the inclusion or exclusion of vacant edges is subject to the grading structure imposed by the UMFC $\ED$.

The trivalent vertices are labeled by morphisms in $\ED$.
Graphically the vertices look like
\begin{gather}
\text{bulk\, vertex:}	\begin{aligned}
		\begin{tikzpicture}
			\draw[-latex,line width=.6pt,black] (0,-0.5) -- (0,-.1);
			\draw[line width=.6pt,black] (0,-0.4) -- (0,0);
			\draw[-latex,line width=.6pt,black](0,0)--(-0.3,0.3);
			\draw[line width=.6pt,black](-0.1,0.1)--(-0.4,0.4);
			\draw[-latex,line width=.6pt,black](0,0)--(0.3,0.3);
			\draw[line width=.6pt,black](0.1,0.1)--(0.4,0.4);
			\node[ line width=0.6pt, dashed, draw opacity=0.5] (a) at (-0.5,0.6){$a$};
			\node[ line width=0.6pt, dashed, draw opacity=0.5] (a) at (0.5,0.6){$b$};
			\node[ line width=0.6pt, dashed, draw opacity=0.5] (a) at (0,-0.7){$c$};
			\node[ line width=0.6pt, dashed, draw opacity=0.5] (a) at (0.3,-0.1){$\alpha$};
		\end{tikzpicture}
	\end{aligned},
 	\begin{aligned}
		\begin{tikzpicture}
			\draw[-latex,line width=.6pt,black] (0.4,-0.4) -- (0.1,-0.1);
			\draw[line width=.6pt,black] (0.3,-0.3) -- (0,0);
			\draw[-latex,line width=.6pt,black] (-0.4,-0.4) -- (-0.1,-0.1);
			\draw[line width=.6pt,black] (-0.3,-0.3) -- (0,0);
			\draw[-latex,line width=.6pt,black](0,0)--(0,0.3);
			\draw[line width=.6pt,black](0,0.1)--(0,0.5);
			\node[ line width=0.6pt, dashed, draw opacity=0.5] (a) at (0.4,-0.6){$b$};
			\node[ line width=0.6pt, dashed, draw opacity=0.5] (a) at (0,0.7){$c$};
			\node[ line width=0.6pt, dashed, draw opacity=0.5] (a) at (-0.4,-0.6){$a$};
			\node[ line width=0.6pt, dashed, draw opacity=0.5] (a) at (0.3,0.1){$\beta$};
		\end{tikzpicture}
	\end{aligned}.
\end{gather}
Notice that here for $a\in \ED_{i,j}$ and $b\in \ED_{k,l}$, the morphism is nonzero only if $j=k$ and $c\in \ED_{i,l}$. 
When $c=\one_i$, $a\in \ED_{i,j}$ and $b=a^*\in \ED_{j,i}$, the vertex becomes bivalent.

\begin{definition}[Inner product of vertex morphisms] \label{def:innerproduct}
For the vertex morphisms $\mu, \nu \in V_{c}^{ab}$ we define their Markov inner product via stacking two vertex diagrams corresponding to $\mu$ and $\nu^{\dagger}$ together 
\begin{equation}
 \begin{aligned}
		\begin{tikzpicture}
			\draw[-latex,line width=.6pt,black](0,0)--(0,0.4);
			\draw[line width=.6pt,black](0,0.1)--(0,0.5);
			\draw[line width=.6pt,black](0,-0.3) circle (0.3);
			\draw[-latex,line width=.6pt,black](-0.3,-0.3)--(-.3,-.2);
			\draw[-latex,line width=.6pt,black](0.3,-0.3)--(.3,-.2);
			\draw[-latex,line width=.6pt,black](0,-1.1)--(0,-0.7);
			\draw[line width=.6pt,black](0,-1.1)--(0,-0.6);
			\node[ line width=0.6pt, dashed, draw opacity=0.5] (a) at (0.6,-0.3){$b$};
			\node[ line width=0.6pt, dashed, draw opacity=0.5] (a) at (0,0.7){$c$};
			\node[ line width=0.6pt, dashed, draw opacity=0.5] (a) at (0,-1.3){$c$};
			\node[ line width=0.6pt, dashed, draw opacity=0.5] (a) at (-0.6,-0.3){$a$};
			\node[ line width=0.6pt, dashed, draw opacity=0.5] (a) at (0.3,0.2){$\nu^{\dagger}$};
			\node[ line width=0.6pt, dashed, draw opacity=0.5] (a) at (-0.3,-0.8){$\mu$};
		\end{tikzpicture}
	\end{aligned}  = \langle \nu|\mu\rangle_M  \id_c\;.
\end{equation}
The Markov trace provided in Eq.~\eqref{eq:Markovtrace} can be applied to both sides of the preceding expression, resulting in an equivalent definition. This justifies the term ``Markov inner product''.
In juxtaposition with the standard inner product presented in Eqs.~\eqref{eq:inner} and \eqref{eq:loopev}, the sole discrepancy lies in a fixed factor incorporating the quantum dimensions associated with the edge labels.
We will interchangeably use these terms, following the convention that the vertex labels $\alpha\in V_{c}^{ab}$, with $\beta\in V_{ab}^c$ (treated as $\beta=\gamma^{\dagger}$, where $\gamma\in V_{c}^{ab}$), are chosen to ensure mutual duality.
\end{definition}

We can introduce a tensor product structure for generalized multifusion string-net as follows. The total Hilbert space is defined as $\mathcal{H}_{\mathrm{tot}}[\Sigma, \ED]=\otimes_{v\in{C}^0(\Sigma)}\mathcal{H}_v$, where ${C}^k(\Sigma)$ denotes the set of $k$-cells in a lattice, and $\mathcal{H}_v = \oplus_{a,b,c\in\Irr(\ED)}V_c^{ab}$ or $\oplus_{a,b,c\in\Irr(\ED)}V^c_{ab}$, depending on the shape of $v$. 
We can introduce edge projectors $E_e$ on the space $\mathcal{H}_{v_1}\otimes \mathcal{H}_{v_2}$, where $v_1$ and $v_2$ represent adjacent vertices connected by $e$, in the following manner:
\begin{equation}
    E_e (|a,b\rightarrow c; \alpha\rangle \otimes |i\rightarrow j,k; \beta\rangle)= \delta_{c,i} |a,b\rightarrow c; \alpha\rangle \otimes |c\rightarrow j,k; \beta\rangle.
\end{equation}
We call a string-net labeling $|S\rangle$ stable if the labels on incident edges match with each other, \emph{viz.}, $|S\rangle \in (\otimes_{e\in C^1(\Sigma)} E_e) \mathcal{H}_{\mathrm{tot}}[\Sigma, \ED]$. The space generated by all stable labeling is the physical space $\mathcal{H}[\Sigma, \ED] \hookrightarrow \mathcal{H}_{\mathrm{tot}}[\Sigma, \ED]$.

A fully labeled string-net is viewed as a specific composition of morphisms in $\ED$, with its domain being the tensor product of all incoming edge labels, and its codomain being the tensor product of all outgoing edge labels. In contrast, a partially labeled string-net symbolizes the projection of the ground state subspace onto the subspace of $\mathcal{H}[\Sigma, \ED]$ where the states maintain fixed labels on the labeled vertices and edges.

Since a fully labeled string-net can be regarded as a composition of morphisms in the multifusion category $\ED$, the graphical calculus of multifusion category $\ED$ can be applied in the Hilbert space. 
It turns out that the topological features of the string-net gapped boundary wavefunction can be captured by some local evaluations of the string-net with boundary, which can be regarded as a state renormalization procedure. 
When applying state renormalization, the lattice on which the string-net is defined will change.
Nevertheless, we still can hope that the ground state space remains unchanged, since string-net state is a fixed-point state under renormalization. 
This leads us to  compare the ground state spaces corresponding to different graphs. 
Consider two graphs $\Sigma_1$ and $\Sigma_2$ which only differ in some local area and the corresponding total Hilbert spaces $\mathcal{H}_1$ and $\mathcal{H}_2$ satisfy $\dim \mathcal{H}_1\geq \dim \mathcal{H}_2$; if two lattices can be renormalized into the same fixed-point class, the respective invariant ground state subspaces $\mathcal{K}_1$ and $\mathcal{K}_2$ are isomorphic.
An evaluation $\mathrm{ev}:\mathcal{H}_1\to\mathcal{H}_2$ is a linear map which is identity map on the overlapping part of two graphs and is unitary when restricted on $\mathcal{K}_1$. The evaluation reduces some freedom of the system thus represents the renormalization procedure of the states.

The equations \eqref{eq:loopev}, \eqref{eq:paraev}, and the F-moves specified in Eqs.~\eqref{eq:F-move1}-\eqref{eq:F-move4} are referred to as \emph{topological local moves}.
By employing these topological local moves, we can construct string-net ground states characterized by their topological nature, representing fixed-point states.

The ground state can be expressed as $|\Psi\rangle=\sum_X\Psi(X)|X\rangle$, where the $|X\rangle$'s represent string-net configurations and the $\Psi(X)$'s are the corresponding complex coefficients.
Notice that $|X\rangle$ contains all the edge labels $a_e \in \Irr(\ED)$ with $e\in C^1({\Sigma})$ and vertex labels $\alpha_v\in \Hom_{\ED}$ with $v\in C^0(\Sigma)$; a more explicit notation should be $|X(a_e,\alpha_v);e\in C^1({\Sigma}),v\in C^0(\Sigma)\rangle$. A string-net configuration represents a composition of vertex morphisms in $\ED$, \emph{viz.}, $X(a_e,\alpha_v)=\alpha_n\comp \cdots \comp \alpha_1\in \Hom_{\ED}$.
The dual vector $\langle X(a_e,\alpha_v)|$ can be regarded as the dual morphism of $X(a_e,\alpha_v)$, \emph{viz.}, $X(a_e,\alpha_v)^{\dagger}=\alpha_1^{\dagger}\comp \cdots \comp \alpha_n^{\dagger}$.
The Definition~\ref{def:innerproduct} of vertex morphisms can thus be generalized to the string-net configurations:

\begin{definition}[Inner product of generalized multifusion string-net configurations] \label{def:innerSN}
For two string-net configurations $X(a_e,\alpha_v)$ and $Y(b_e,\beta_v)$ in $\Hom_{\ED}(i_1\otimes\cdots \otimes i_s,j_1\otimes\cdots \otimes j_s)$, we define their inner product via stacking the morphisms  $X(a_e,\alpha_v)$ and $Y(b_e,\beta_v)^{\dagger}=Y(b_e,\beta_v^{\dagger})$.
The Markov inner product is defined by
\begin{equation}
   \begin{aligned}
		\begin{tikzpicture}
			\draw[-latex,line width=.6pt,black] (-1.2,-0.5) -- (-1.2,-.1);
			\draw[line width=.6pt,black] (-1.2,-0.5) -- (-1.2,0);
			\node[ line width=0.6pt, dashed, draw opacity=0.5] (a) at (-1.2,-0.7){$i_1$};
			\draw[-latex,line width=.6pt,black] (-0.7,-0.5) -- (-0.7,-.1);
			\draw[line width=.6pt,black] (-0.7,-0.5) -- (-0.7,0);
			\node[ line width=0.6pt, dashed, draw opacity=0.5] (a) at (-0.7,-0.7){$i_2$};
			\draw[-latex,line width=.6pt,black] (1.2,-0.5) -- (1.2,-.1);
			\draw[line width=.6pt,black] (1.2,-0.5) -- (1.2,0);
			\node[ line width=0.6pt, dashed, draw opacity=0.5] (a) at (1.2,-0.7){$i_s$};
			\draw[line width=0.6 pt, fill=gray, fill opacity=0.2] 
			(-2,0) -- (2,0) -- (2,0.5) -- (-2,0.5) -- cycle; 
			\node[ line width=0.6pt, dashed, draw opacity=0.5] (a) at (0,0.25){$Y(b_e,\beta_v)^{\dagger} \comp X(a_e,\alpha_v)$};
			\node[ line width=0.6pt, dashed, draw opacity=0.5] (a) at (0.3,-0.5){$\cdots$};
			\draw[-latex,line width=.6pt,black] (-1.2,0.5) -- (-1.2,0.85);
			\draw[line width=.6pt,black] (-1.2,0.5) -- (-1.2,1);
			\node[ line width=0.6pt, dashed, draw opacity=0.5] (a) at (-1.2,1.2){$i_1$};
			\draw[-latex,line width=.6pt,black] (1.2,0.5) -- (1.2,0.85);
			\draw[line width=.6pt,black] (1.2,0.5) -- (1.2,1);
			\node[ line width=0.6pt, dashed, draw opacity=0.5] (a) at (1.2,1.2){$i_s$};
			\draw[-latex,line width=.6pt,black] (-0.7,0.5) -- (-0.7,0.85);
			\draw[line width=.6pt,black] (-0.7,0.5) -- (-0.7,1);
			\node[ line width=0.6pt, dashed, draw opacity=0.5] (a) at (-0.7,1.2){$i_2$};
			\node[ line width=0.6pt, dashed, draw opacity=0.5] (a) at (0.3,1.0){$\cdots$};
		\end{tikzpicture}
	\end{aligned}=\langle Y|X\rangle_M  \id_{i_1\otimes \cdots \otimes i_s}. 
\end{equation}
We can take Markov trace for both sides of the above expression to obtain an equivalent definition.
\end{definition}

Thus to determine the state, we only need to determine $\Psi(X)$.
This can be done by fixing the convention that $\Psi(\text{vac})=1$ and using the local evaluations to transform all the configurations into the vacant string-net.
Here by vacant string-net we mean a morphism from $\one$ to $\one$ in the UMFC $\ED$. Indeed, the vacant configuration in the UMFC context is nuanced due to the complexity introduced by the fact that $\one$ is no longer a simple object.
Taking a disk as an example (see Fig.~\ref{fig:StringNetBd}), a fully labeled string-net can be regarded as a morphism from $\one_j$ to $\one_j$. Every string-net configuration can be transformed into the vacant configuration via the local moves.

\begin{figure}[t]
		\centering
		\includegraphics[width=11cm]{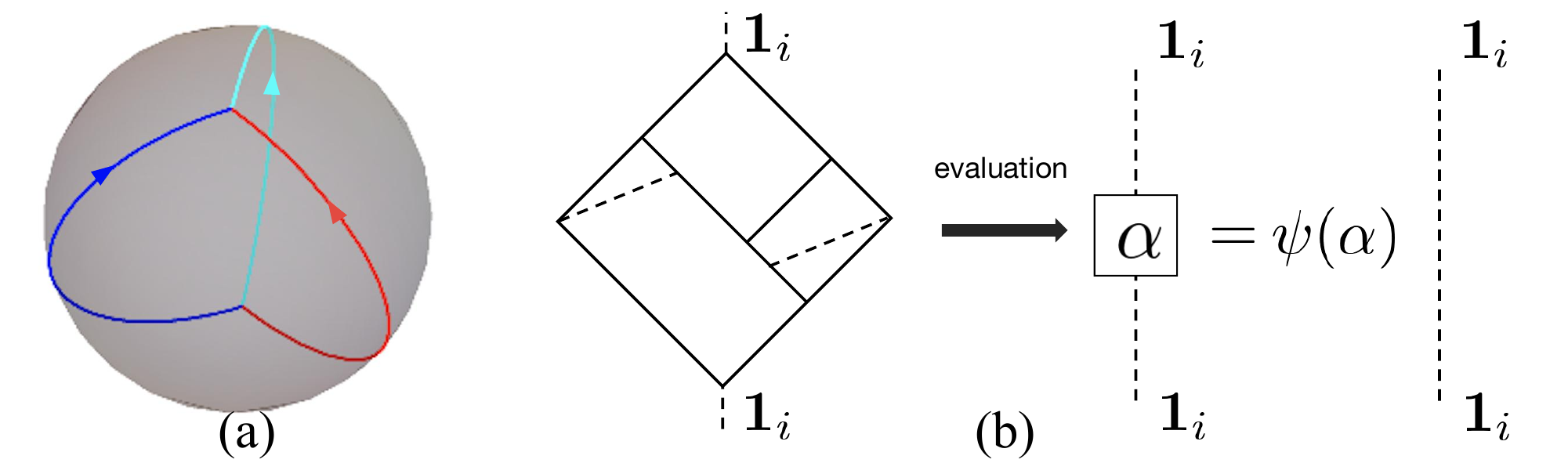}
		\caption{Depiction of the generalized multifusion string-net. \label{fig:StringNet}}
\end{figure}

\subsection{Generalized multifusion string-net lattice Hamiltonian}

Having introduced the construction of a generalized multifusion string-net ground state, this subsection will focus on the development of the corresponding lattice Hamiltonian.

For a given vertex in the generalized multifusion string-net, we define $\delta_{a,b\to c}=1$ if $V_{ab}^c\neq 0$; otherwise, we set $\delta_{a,b\to c}=0$. Similarly $\delta_{c\to a,b}$ is defined. Using these, we introduce a vertex projector $Q_v$ as follows:
\begin{equation}
    Q_v \big{|}  \begin{aligned}
		\begin{tikzpicture}
			\draw[-latex,line width=.6pt,black] (0,-0.5) -- (0,-.1);
			\draw[line width=.6pt,black] (0,-0.4) -- (0,0);
			\draw[-latex,line width=.6pt,black](0,0)--(-0.3,0.3);
			\draw[line width=.6pt,black](-0.1,0.1)--(-0.4,0.4);
			\draw[-latex,line width=.6pt,black](0,0)--(0.3,0.3);
			\draw[line width=.6pt,black](0.1,0.1)--(0.4,0.4);
			\node[ line width=0.6pt, dashed, draw opacity=0.5] (a) at (-0.5,0.6){$a$};
			\node[ line width=0.6pt, dashed, draw opacity=0.5] (a) at (0.5,0.6){$b$};
			\node[ line width=0.6pt, dashed, draw opacity=0.5] (a) at (0,-0.7){$c$};
			\node[ line width=0.6pt, dashed, draw opacity=0.5] (a) at (0.3,-0.1){$\alpha$};
		\end{tikzpicture}
	\end{aligned} \big{\rangle} =\delta_{c\to a,b} \,\big{|}\begin{aligned}
		\begin{tikzpicture}
			\draw[-latex,line width=.6pt,black] (0,-0.5) -- (0,-.1);
			\draw[line width=.6pt,black] (0,-0.4) -- (0,0);
			\draw[-latex,line width=.6pt,black](0,0)--(-0.3,0.3);
			\draw[line width=.6pt,black](-0.1,0.1)--(-0.4,0.4);
			\draw[-latex,line width=.6pt,black](0,0)--(0.3,0.3);
			\draw[line width=.6pt,black](0.1,0.1)--(0.4,0.4);
			\node[ line width=0.6pt, dashed, draw opacity=0.5] (a) at (-0.5,0.6){$a$};
			\node[ line width=0.6pt, dashed, draw opacity=0.5] (a) at (0.5,0.6){$b$};
			\node[ line width=0.6pt, dashed, draw opacity=0.5] (a) at (0,-0.7){$c$};
			\node[ line width=0.6pt, dashed, draw opacity=0.5] (a) at (0.3,-0.1){$\alpha$};
		\end{tikzpicture}
	\end{aligned}\big{\rangle}\,.
\end{equation}
The condition $Q_v=1$ corresponds to satisfying Gauss's law at the vertex $v$ in an analogy with lattice gauge theory.

The face operator $B_f$ is more complicated. It is a linear combination of some basic face operators $B^k_f$:
\begin{equation}
    B_f=\sum_{k\in \Irr(\ED)}w_k B^k_f,
\end{equation}
where $w_k=Y^{k^* k}_{\one}/\sum_{l\in \Irr(\ED)} d_l^2$. By selecting the gauge as $Y^{ab}_c=(d_ad_b/d_c)^{1/2}$, we observe that this aligns with the gauge choice conventionally made in the original string-net model.
The face operator $B_f^k$ acts on the face of the string-net by inserting a $k$-loop into the face $f$, which graphically looks like
\begin{gather}
    B_f^k\, \Big{|} 
    \begin{aligned}
        \begin{tikzpicture}
            \draw[line width=.6pt,black] (0,0) -- (0.87,0.5);
            \draw[line width=.6pt,black] (0.87,1.5) -- (0.87,0.5);
            \draw[line width=.6pt,black] (0.87,1.5) -- (0,2);
            \draw[line width=.6pt,black] (0,0) -- (-0.87,0.5);
            \draw[line width=.6pt,black] (-0.87,1.5) -- (-0.87,0.5);
            \draw[line width=.6pt,black] (-0.87,1.5) -- (0,2);
            \draw[-latex,line width=.6pt,black] (0,0) -- (0.5742,0.33);
            \draw[-latex,line width=.6pt,black] (0.87,0.5) -- (0.87,1.1);
            \draw[-latex,line width=.6pt,black] (0.87,1.5) -- (0.2958,1.83);
            \draw[-latex,line width=.6pt,black] (0,0) -- (-0.5742,0.33);
            \draw[-latex,line width=.6pt,black] (-0.87,0.5) -- (-0.87,1.1);
            \draw[-latex,line width=.6pt,black] (-0.87,1.5) -- (-0.2958,1.83);
            \draw[line width=.6pt,black] (0,-0.6) -- (0,0);
            \draw[line width=.6pt,black] (0,2) -- (0,2.6);
            \draw[line width=.6pt,black] (0.87,0.5) -- (1.5225,0.125);
            \draw[line width=.6pt,black] (0.87,1.5) -- (1.5225,1.875);
            \draw[line width=.6pt,black] (-0.87,0.5) -- (-1.5225,0.125);
            \draw[line width=.6pt,black] (-0.87,1.5) -- (-1.5225,1.875);
            \draw[-latex,line width=.6pt,black] (0,-0.6) -- (0,-0.2);
            \draw[-latex,line width=.6pt,black] (0,2) -- (0,2.4);
            \draw[-latex,line width=.6pt,black] (1.5225,0.125) -- (1.1,0.3678);
            \draw[-latex,line width=.6pt,black] (0.87,1.5) -- (1.2925,1.7428);
            \draw[-latex,line width=.6pt,black] (-1.5225,0.125) -- (-1.1,0.3678);
            \draw[-latex,line width=.6pt,black] (-0.87,1.5) -- (-1.2925,1.7428);
            \node[ line width=0.6pt, dashed, draw opacity=0.5] (a) at (1.2,1){$j_1$};
            \node[ line width=0.6pt, dashed, draw opacity=0.5] (a) at (0.6,2.1){$j_2$};
            \node[ line width=0.6pt, dashed, draw opacity=0.5] (a) at (-0.6,2.05){$j_3$};
            \node[ line width=0.6pt, dashed, draw opacity=0.5] (a) at (-1.2,1){$j_4$};
            \node[ line width=0.6pt, dashed, draw opacity=0.5] (a) at (-0.55,0){$j_5$};
            \node[ line width=0.6pt, dashed, draw opacity=0.5] (a) at (0.55,0){$j_6$};
            \node[ line width=0.6pt, dashed, draw opacity=0.5] (a) at (1.8,2){$i_1$};
            \node[ line width=0.6pt, dashed, draw opacity=0.5] (a) at (0,2.9){$i_2$};
            \node[ line width=0.6pt, dashed, draw opacity=0.5] (a) at (-1.8,2){$i_3$};
            \node[ line width=0.6pt, dashed, draw opacity=0.5] (a) at (-1.8,0){$i_4$};
            \node[ line width=0.6pt, dashed, draw opacity=0.5] (a) at (0,-0.9){$i_5$};
            \node[ line width=0.6pt, dashed, draw opacity=0.5] (a) at (1.8,0){$i_6$};
            \node[ line width=0.6pt, dashed, draw opacity=0.5] (a) at (0.65,1.35){$\alpha_1$};
            \node[ line width=0.6pt, dashed, draw opacity=0.5] (a) at (0,1.7){$\alpha_2$};
            \node[ line width=0.6pt, dashed, draw opacity=0.5] (a) at (-0.55,1.35){$\alpha_3$};
            \node[ line width=0.6pt, dashed, draw opacity=0.5] (a) at (-0.55,0.55){$\alpha_4$};
            \node[ line width=0.6pt, dashed, draw opacity=0.5] (a) at (0,0.3){$\alpha_5$};
            \node[ line width=0.6pt, dashed, draw opacity=0.5] (a) at (0.62,0.65){$\alpha_6$};
        \end{tikzpicture}
    \end{aligned}
    \Big{\rangle} = 
    \Big{|}
    \begin{aligned}
        \begin{tikzpicture}
            \draw[line width=.6pt,black] (0,0) -- (0.87,0.5);
            \draw[line width=.6pt,black] (0.87,1.5) -- (0.87,0.5);
            \draw[line width=.6pt,black] (0.87,1.5) -- (0,2);
            \draw[line width=.6pt,black] (0,0) -- (-0.87,0.5);
            \draw[line width=.6pt,black] (-0.87,1.5) -- (-0.87,0.5);
            \draw[line width=.6pt,black] (-0.87,1.5) -- (0,2);
            \draw[-latex,line width=.6pt,black] (0,0) -- (0.5742,0.33);
            \draw[-latex,line width=.6pt,black] (0.87,0.5) -- (0.87,1.1);
            \draw[-latex,line width=.6pt,black] (0.87,1.5) -- (0.2958,1.83);
            \draw[-latex,line width=.6pt,black] (0,0) -- (-0.5742,0.33);
            \draw[-latex,line width=.6pt,black] (-0.87,0.5) -- (-0.87,1.1);
            \draw[-latex,line width=.6pt,black] (-0.87,1.5) -- (-0.2958,1.83);
            \draw[line width=.6pt,black] (0,-0.6) -- (0,0);
            \draw[line width=.6pt,black] (0,2) -- (0,2.6);
            \draw[line width=.6pt,black] (0.87,0.5) -- (1.5225,0.125);
            \draw[line width=.6pt,black] (0.87,1.5) -- (1.5225,1.875);
            \draw[line width=.6pt,black] (-0.87,0.5) -- (-1.5225,0.125);
            \draw[line width=.6pt,black] (-0.87,1.5) -- (-1.5225,1.875);
            \draw[-latex,line width=.6pt,black] (0,-0.6) -- (0,-0.2);
            \draw[-latex,line width=.6pt,black] (0,2) -- (0,2.4);
            \draw[-latex,line width=.6pt,black] (1.5225,0.125) -- (1.1,0.3678);
            \draw[-latex,line width=.6pt,black] (0.87,1.5) -- (1.2925,1.7428);
            \draw[-latex,line width=.6pt,black] (-1.5225,0.125) -- (-1.1,0.3678);
            \draw[-latex,line width=.6pt,black] (-0.87,1.5) -- (-1.2925,1.7428);
            \node[ line width=0.6pt, dashed, draw opacity=0.5] (a) at (1.2,1){$j_1$};
            \node[ line width=0.6pt, dashed, draw opacity=0.5] (a) at (0.6,2.1){$j_2$};
            \node[ line width=0.6pt, dashed, draw opacity=0.5] (a) at (-0.6,2.05){$j_3$};
            \node[ line width=0.6pt, dashed, draw opacity=0.5] (a) at (-1.2,1){$j_4$};
            \node[ line width=0.6pt, dashed, draw opacity=0.5] (a) at (-0.55,0){$j_5$};
            \node[ line width=0.6pt, dashed, draw opacity=0.5] (a) at (0.55,0){$j_6$};
            \node[ line width=0.6pt, dashed, draw opacity=0.5] (a) at (1.8,2){$i_1$};
            \node[ line width=0.6pt, dashed, draw opacity=0.5] (a) at (0,2.9){$i_2$};
            \node[ line width=0.6pt, dashed, draw opacity=0.5] (a) at (-1.8,2){$i_3$};
            \node[ line width=0.6pt, dashed, draw opacity=0.5] (a) at (-1.8,0){$i_4$};
            \node[ line width=0.6pt, dashed, draw opacity=0.5] (a) at (0,-0.9){$i_5$};
            \node[ line width=0.6pt, dashed, draw opacity=0.5] (a) at (1.8,0){$i_6$};
            \node[ line width=0.6pt, dashed, draw opacity=0.5] (a) at (0.65,1.35){$\alpha_1$};
            \node[ line width=0.6pt, dashed, draw opacity=0.5] (a) at (0,1.7){$\alpha_2$};
            \node[ line width=0.6pt, dashed, draw opacity=0.5] (a) at (-0.55,1.35){$\alpha_3$};
            \node[ line width=0.6pt, dashed, draw opacity=0.5] (a) at (-0.55,0.55){$\alpha_4$};
            \node[ line width=0.6pt, dashed, draw opacity=0.5] (a) at (0,0.3){$\alpha_5$};
            \node[ line width=0.6pt, dashed, draw opacity=0.5] (a) at (0.62,0.65){$\alpha_6$};
            \draw[line width=0.8pt,black] (0.35,1.05) .. controls +(0,0.4) and +(0,0.4) .. (-0.35,1.05);
            \draw[line width=0.8pt,black] (0.35,0.95) .. controls +(0,-0.4) and +(0,-0.4) .. (-0.35,0.95);
            \draw[line width=0.8pt,black](0.35,1.05) -- (0.35,0.95);
            \draw[line width=0.8pt,black](-0.35,1.05) -- (-0.35,0.95);
            \draw[-latex,line width=0.8pt,black](0.35,0.9) -- (0.35,1.12);
            \node[ line width=0.1pt, dashed, draw opacity=0.5] (a) at (0.1,1){$k$};
        \end{tikzpicture}
    \end{aligned}
    \Big{\rangle}\;.
\end{gather}
It is noteworthy that, commonly, a hexagon is considered as an illustrative example, while the face operator is designed to act on a face with an arbitrary number of edges.
The condition $B_f=1$ can be conceptualized as ensuring that no flux penetrates the face $f$.

\begin{proposition}
 The local stabilizers $Q_v$ and $B_f$ functions are projectors and mutually commute. Consequently, the Hamiltonian of the generalized multifusion string-net\,\footnote{The Hamiltonian is formulated on the space $\mathcal{H}[\Sigma,\ED]$, encompassing string-nets with stable labelings. To emphasize the tensor product structure, one can define the Hamiltonian in the total space $\mathcal{H}_{\rm tot}[\Sigma,\ED]$, wherein the Hamiltonian incorporates edge projectors as well: $H=-J_e\sum_e E_e -J_v\sum_v Q_v-J_f\sum_f B_f$.} ($J_v,J_f>0$)
 \begin{equation}
     H=-J_v\sum_v Q_v-J_f\sum_f B_f
 \end{equation}
 being a local commutative projector (LCP) Hamiltonian, exhibits a gap in the thermodynamic limit.
\end{proposition}

\begin{proof}
By definition, we see that the $Q_v$'s are projectors and $[Q_v,Q_{v'}]=0$ for all distinct vertices $v$ and $v'$.

Using the same method as in Refs.~\cite{hong2009symmetrization,Hahn2020generalized}, we can show that \begin{equation}\label{eq:Bfdual}
    (B_f^k)^{\dagger}=B^{k^*}_f.
\end{equation}
Notice that the main trick here is to use the inner product that we introduced in Definitions ~\ref{def:innerproduct} and \ref{def:innerSN}. It also satisfies $B_f^kB_f^l = \sum_n N^n_{kl}B_f^n$. In fact, this follows from the following 
\begin{align}
    \begin{aligned}
        \begin{tikzpicture}
            \draw[line width=.6pt,black] (0.2,1.45) .. controls +(0,0.2) and +(0,0.2) .. (-0.2,1.45) (0.55,1.45) .. controls +(0,0.6) and +(0,0.6) .. (-0.55,1.45);
            \draw[line width=.6pt,black] (0.2,0.95) .. controls +(0,-0.2) and +(0,-0.2) .. (-0.2,0.95) (0.55,0.95) .. controls +(0,-0.6) and +(0,-0.6) .. (-0.55,0.95);
             \draw[line width=.6pt,black] (-0.2,1.45) .. controls +(-0.0,0) and +(-0.0,0) .. (-0.2,0.95) (-0.55,1.45) .. controls +(-0.0,0) and +(-0.0,0) .. (-0.55,0.95);
             \draw[line width=.6pt,black] (0.2,1.45) .. controls +(-0.0,0) and +(-0.0,0) .. (0.2,0.95) (0.55,1.45) .. controls +(-0.0,0) and +(-0.0,0) .. (0.55,0.95);
             \draw[-latex,line width=.6pt,black] (0.55,1.2) -- (0.55,1.3); 
              \draw[-latex,line width=.6pt,black] (0.2,1.2) -- (0.2,1.3); 
            \node[ line width=0.6pt, dashed, draw opacity=0.5] (a) at (0,1.2){$k$};
            \node[ line width=0.6pt, dashed, draw opacity=0.5] (a) at (0.38,1.2){$l$};
        \end{tikzpicture}
    \end{aligned} 
     = \sum_{n,\alpha} \frac{1}{Y_{kl}^n}\; 
    \begin{aligned}
        \begin{tikzpicture}
             \draw[line width=.6pt,black] (0.375,1.34) -- (0.375,1.05);
            \draw[line width=.6pt,black] (0.2,1.45) .. controls +(0,-0.15) and +(0,-0.15) .. (0.55,1.45) (0.2,0.95) .. controls +(0,0.15) and +(0,0.15) .. (0.55,0.95); 
            \draw[line width=.6pt,black] (0.2,1.55) .. controls +(0,0.2) and +(0,0.2) .. (-0.2,1.55) (0.55,1.55) .. controls +(0,0.6) and +(0,0.6) .. (-0.55,1.55);
            \draw[line width=.6pt,black] (0.2,1.45) -- (0.2,1.55);
            \draw[line width=.6pt,black] (0.55,1.45) -- (0.55,1.55);
            \draw[line width=.6pt,black] (0.2,0.85) .. controls +(0,-0.2) and +(0,-0.2) .. (-0.2,0.85) (0.55,0.85) .. controls +(0,-0.6) and +(0,-0.6) .. (-0.55,0.85);
            \draw[line width=.6pt,black] (0.2,0.95) -- (0.2,0.85);
            \draw[line width=.6pt,black] (0.55,0.95) -- (0.55,0.85);
             \draw[line width=.6pt,black] (-0.2,1.55) .. controls +(-0.0,0) and +(-0.0,0) .. (-0.2,0.85) (-0.55,1.55) .. controls +(-0.0,0) and +(-0.0,0) .. (-0.55,0.85);
             \draw[-latex,line width=.6pt,black] (-0.55,1.2) -- (-0.55,1.1); 
              \draw[-latex,line width=.6pt,black] (-0.2,1.2) -- (-0.2,1.1); 
               \draw[-latex,line width=.6pt,black] (0.375,1.2) -- (0.375,1.3);  
            \node[ line width=0.6pt, dashed, draw opacity=0.5] (a) at (0.59,1.2){$n$};
            \node[ line width=0.6pt, dashed, draw opacity=0.5] (a) at (0.06,0.95){$k$};
            \node[ line width=0.6pt, dashed, draw opacity=0.5] (a) at (0.67,0.94){$l$};
            \node[ line width=0.6pt, dashed, draw opacity=0.5] (a) at (0.375,1.5){$\alpha$};
            \node[ line width=0.6pt, dashed, draw opacity=0.5] (a) at (0.375,0.9){$\alpha$};
        \end{tikzpicture}
    \end{aligned} 
    = \sum_{n,\alpha} \frac{1}{Y_{kl}^n}\; 
    \begin{aligned}
        \begin{tikzpicture}
            \draw[line width=.6pt,black] (0.35,1.45) .. controls +(0,0.5) and +(0,0.5) .. (-0.35,1.45);
            \draw[line width=.6pt,black]  (0.35,0.95) .. controls +(0,-0.5) and +(0,-0.5) .. (-0.35,0.95);
             \draw[line width=.6pt,black] (-0.35,1.45) .. controls +(-0.0,0) and +(-0.0,0) .. (-0.35,0.95);
             \draw[line width=0.6pt,black] (0.35,1.2) circle (0.25);
             \draw[-latex,line width=.6pt,black] (0.1,1.2) -- (0.1,1.3); 
             \draw[-latex,line width=.6pt,black] (0.6,1.2) -- (0.6,1.3); 
            \draw[-latex,line width=.6pt,black] (-0.35,1.2) -- (-0.35,1.1);  
            \node[ line width=0.6pt, dashed, draw opacity=0.5] (a) at (0.17,0.85){$n$};
            \node[ line width=0.6pt, dashed, draw opacity=0.5] (a) at (-0.08,1.2){$k$};
            \node[ line width=0.6pt, dashed, draw opacity=0.5] (a) at (0.45,1.2){$l$};
            \node[ line width=0.6pt, dashed, draw opacity=0.5] (a) at (0.52,1.58){$\alpha$};
            \node[ line width=0.6pt, dashed, draw opacity=0.5] (a) at (0.52,0.83){$\alpha$};
        \end{tikzpicture}
    \end{aligned} 
    =  \sum_{\alpha,n} \; 
    \begin{aligned}
        \begin{tikzpicture}
            \draw[line width=.6pt,black] (0.45,1.35) .. controls +(0,0.6) and +(0,0.6) .. (-0.45,1.35);
            \draw[line width=.6pt,black]  (0.45,1.05) .. controls +(0,-0.6) and +(0,-0.6) .. (-0.45,1.05);
             \draw[line width=.6pt,black] (-0.45,1.35) .. controls +(-0.0,0) and +(-0.0,0) .. (-0.45,1.05) (0.45,1.35) .. controls +(-0.0,0) and +(-0.0,0) .. (0.45,1.05);
             \draw[-latex,line width=.6pt,black] (0.45,1.2) -- (0.45,1.3); 
            \node[ line width=0.6pt, dashed, draw opacity=0.5] (a) at (0.2,1.2){$n$};
        \end{tikzpicture}
    \end{aligned}\;,
\end{align}
where the first identity is from Eq.~\eqref{eq:paraev}, and the third from Eq.~\eqref{eq:loopev}. Using these facts one can show that the $B_f$'s are projectors. 
\end{proof}

\begin{proposition}
    The generalized multifusion string-net ground state $|\Psi\rangle$ constructed in the previous subsection is invariant under the action of $Q_v$ and $B_f$ for all vertices $v$ and faces $f$. 
\end{proposition}

\begin{proof}
The condition $Q_v |\Psi\rangle=|\Psi\rangle$ is clear from the definition.
To show that $|\Psi\rangle$ is stabilized by all $B_f$'s, we claim that
\begin{equation}
    B_f^a=Y^{aa^*}_{\one_i} 
\end{equation}
for $a\in \ED_{i,j}$. Then from the definition of $B_f$ we derive $B_f |\Psi\rangle=|\Psi\rangle$.
To show the claim, consider the string-net configuration $\langle S|$, we see that 
\begin{equation}
    \langle S|B_f^a|\Psi\rangle= \langle B_f^{a^*} S |\Psi\rangle = Y^{aa^*}_{\one_i}  \langle S |\Psi\rangle,
\end{equation}
where we have used Eqs.~\eqref{eq:loopev} and \eqref{eq:Bfdual}. This proves the claim. 
\end{proof}

\subsection{Topological ground state degeneracy}

The simplest topological observable of the topological ordered phase is the ground state degeneracy.
For the original Levin-Wen string-net model, we know that it depends on the topology of the spatial manifold and the input UFC.
The ground state degeneracies for the Levin-Wen string-net model have been computed in prior works \cite{Hu2012ground, Vidal2022partition, ritzzwilling2023topological}.

A parallel observation holds for the multifusion string-net model. The ground state degeneracy in the generalized multifusion string-net model remains a topological invariant, resilient to topological local moves.
Using the local stabilizers, the ground state degeneracy can be expressed as
\begin{equation}
    \operatorname{GSD}= \Tr \left(\prod_v Q_v\right) \left(\prod_f B_f\right).
\end{equation}
For a given spatial manifold, the calculation of the $\operatorname{GSD}$ involves evaluating it through the consideration of the simplest cellulation of the manifold and using the fact that $\operatorname{GSD}$ is a topological invariant.

Moore-Seiberg formula \cite{Moore1989,ritzzwilling2023topological} of $\operatorname{GSD}$ also holds for the generalized multifusion string-net model (with an input indecomposable UMFC $\ED$), since the topological excitations still form a UMTC (see Theorem~\ref{thm:UMFC_SN}).
Consider an orientable $g$-genus surface $\Sigma_{g,k}$ with $k$ punctures, we can assign anyons (simple objects in $\mathcal{Z}(\ED)$) to the punctures of the surface, denoted as $X_1,\cdots,X_k$. The following is an example with 6 anyons on a $3$-genus surface:
\begin{equation*}
\begin{tikzpicture}[
  tqft,
  every outgoing boundary component/.style={fill=blue!50},
  outgoing boundary component 3/.style={fill=none,draw=black},
  every incoming boundary component/.style={fill=green!50},
  every lower boundary component/.style={draw,black},
  every upper boundary component/.style={draw,black},
  cobordism/.style={fill=gray!40},
  cobordism edge/.style={draw,black},
  genus=3,
  hole 2/.style={draw},
  view from=incoming,
  anchor=between incoming 1 and 2
]
\pic[name=a,tqft,
    incoming boundary components=5,
    skip incoming boundary components={2,4},
    outgoing boundary components=5,
    skip outgoing boundary components={2,3,5},
    offset=-.5];
\end{tikzpicture}
\end{equation*}
This surface can be decomposed into pant surfaces that are subsequently glued together. It is essential to note that this decomposition is not unique in general, but the associativity of fusion of anyons guarantees that the final result for GSD is the same:
\begin{equation}\label{eq:GSDTQFT}
    \operatorname{GSD}[\Sigma_{g,k},X_1,\cdots,X_k]=\sum_{Y\in \Irr(\mathcal{Z}(\ED))}\left(\prod_{j=1}^k S_{X_j,Y}\right)S_{\one,Y}^{2-2g-k},
\end{equation}
where $S$ is the S-matrix, $\one$ is the vacuum charge (tensor unit of $\mathcal{Z}(\ED)$), $g$ is the number of genus, and $k$ is the number of punctures.
When there are no punctures, we obtain the expression of the GSD for multifusion string-net model:
\begin{equation}
    \operatorname{GSD}[\Sigma_{g}]=\operatorname{GSD}[\Sigma_{g,k},X_1=\cdots=X_k=\one]=\sum_{Y\in \Irr(\mathcal{Z}(\ED))}S_{\one,Y}^{2-2g}.
\end{equation}
It is crucial to observe that the expression \eqref{eq:GSDTQFT} is not applicable in scenarios involving excitations. This limitation arises from the existence of additional fusion channels, as detailed in \cite{ritzzwilling2023topological}.

\section{Weak Hopf tube algebra and topological excitations}
\label{sec:tube}

In this section, we present an approach to investigate and classify the topological excitations in a multifusion string-net model.
We will use the string diagram calculus to construct a tube algebra and show that the tube algebra is a $C^*$ weak Hopf algebra.
This can be regarded as a generalization of Ocneanu's tube algebra \cite{ocneanu1994chirality,ocneanu2001operator}.
It is worth noting that our approach differs from the Q-algebra approach for Levin-Wen string-net proposed in Ref.~\cite{lan2014topological}, where the Q-algebra lacks a coalgebra structure.
Our construction of the bulk tube algebra exhibits another interesting property: the Kitaev-Kong boundary tube algebra (denoted as $\tilde{\mathbf{L}}({_{\EC}}\EM)$ in our notation, see Sec.~\ref{sec:bdtheory} for a detailed discussion) constructed in Ref.~\cite{Kitaev2012boundary} can be naturally embedded into it. Furthermore, intriguingly, we can consider its dual weak Hopf algebra as a mirror reflection of the diagram corresponding to $\tilde{\mathbf{L}}({_{\EC}}\EM)$. This reflection reveals that their crossed multiplication aligns well with the bulk tube algebra.
Notice that hereinafter, whenever necessary to stress the base UMFC category, we will denote a left (resp. right) $\ED$-module category by $_{\ED}\EM$ (resp. $\EM_{\ED}$), and a $\EC|\ED$-bimodule category by ${_{\EC}}\EM_{\ED}$.

The ground state space of multifusion string-net model consists of translational invariant constant energy density states which satisfy stabilizer conditions: $Q_v|\Psi\rangle=|\Psi\rangle$ and $B_f|\Psi\rangle=|\Psi\rangle$ for all $v\in C^0(\Sigma)$ and $f\in C^2(\Sigma)$.
For a local connected region of string-net, if the stabilizers in this region do not stabilize the state, then we say that there are excitations inside the region.
We call a state $|\Phi\rangle$ whereby $Q_v|\Psi\rangle=0$ an electric charge excitation; we call a state $|\Phi\rangle$ whereby $B_f|\Psi\rangle=0$ a magnetic flux excitation; if $Q_v|\Psi\rangle=0=B_f|\Psi\rangle$, we call them dyonic charge excitations.
Due to their topological nature, we can explore how to shrink or enlarge the excited region. This perspective leads us to the concept of the tube algebra, which establishes connections between the excited region and the unexcited region.
We will show that the excitations are modules of the tube algebra.
The coalgebra action can also be understood naturally in this framework.

\subsection{Bulk tube algebra as weak Hopf algebra}

\begin{figure}[t]
		\centering
		\includegraphics[width=12cm]{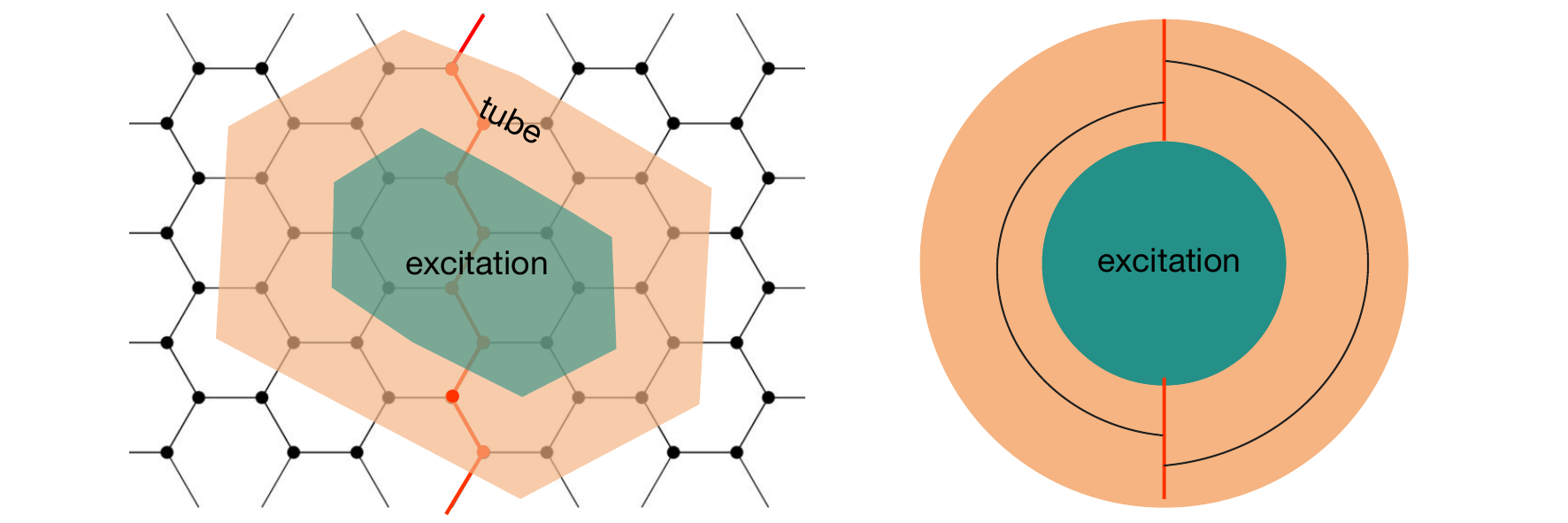}
		\caption{Depiction of the tube algebra. \label{fig:tube}}
\end{figure}

For a generalized string-net $\Sigma$, consider a disk region $\Gamma \subset \Sigma$ which consists of a connected set of vertices. This means that we always cut the disk region for a string-net along the edges, 
resulting in a boundary comprising a collection of faces. It is crucial to recall that, in constructing our generalized multifusion string-net, a tensor product structure is assigned to each edge, with the edge projector ensuring the alignment of the two labels on the same edge.
The boundary $\partial \Gamma$ of $\Gamma$ are edges that connect the vertices inside $\Gamma$ and that outside $\Gamma$.
See Fig.~\ref{fig:tube} for an illustration.
For $e \in \partial \Gamma$, we denote $\partial_i e$ and $\partial_o e$ the vertices of $e$ inside and outside $\Gamma$, respectively. A tube region $\operatorname{Tube}_{\Gamma}$ surrounding $\Gamma$ consists of vertices $\partial_o e$ for all $e \in \partial \Gamma$, along with vertices that form the shortest connected path in $\Sigma$ containing all $\partial_o e$ for $e \in \partial \Gamma$.
The boundary edges of $\operatorname{Tube}_{\Gamma}$  can be divided into two groups: 
$\partial_i\!\operatorname{Tube}_{\Gamma}$ consisting of all edges that connect $\Gamma$, and 
$\partial_o\!\operatorname{Tube}_{\Gamma}$ consisting of all edges that outside the $\Gamma \cup \operatorname{Tube}_{\Gamma}$.

Consider the disk region $\Gamma \subset \Sigma$, we have Hamiltonian for this disk
\begin{equation}
    H[\Gamma]=-\sum_{v\in \Gamma} Q_v -\sum_{f\in \Gamma} B_f.
\end{equation}
Notice that the boundary faces of $\Gamma$ are not contained in the above Hamiltonian, thus $H[\Gamma]$ is an open-boundary Hamiltonian.
This implies that we have a natural decomposition of the ground state space for $H[\Gamma]$: $\mathcal{V}_{\rm GS}[\Gamma]=\bigoplus_{C_{\rm bd}} \mathcal{V}_{\rm GS}[\Gamma]_{C_{\rm bd}}$ with ``$C_{\rm bd}$'' represents the open-boundary conditions.
When we consider the topological properties of the excitations in $\Gamma$, a crucial observation is that the disk region is unchanged under gluing a tube:
\begin{equation*}
    \begin{aligned}
      \begin{tikzpicture}[
  tqft,
  every outgoing boundary component/.style={fill=blue!50},
  outgoing boundary component 3/.style={fill=none,draw=red},
  every incoming boundary component/.style={fill=green!50},
  every lower boundary component/.style={draw=black},
  every upper boundary component/.style={draw,black},
  cobordism/.style={fill=orange!50},
  cobordism edge/.style={draw},
  view from=incoming,
  cobordism height=1cm,
]
\begin{scope}[every node/.style={rotate=90}]
\pic[name=a,
  tqft,
  incoming boundary components=1,
  outgoing boundary components=1
  ];
\end{scope}
\end{tikzpicture}
  \end{aligned}  
    \begin{aligned}
      \begin{tikzpicture}[
  tqft,
  every outgoing boundary component/.style={fill=blue!50},
  outgoing boundary component 3/.style={fill=none,draw=black},
  every incoming boundary component/.style={fill=green!50},
  every lower boundary component/.style={draw=black},
  every upper boundary component/.style={draw,black},
  cobordism/.style={fill=teal!80},
  cobordism edge/.style={draw},
  view from=incoming,
  cobordism height=3cm,
]
\begin{scope}[every node/.style={rotate=90}]
\pic[name=a,
  tqft,
  incoming boundary components=1,
  outgoing boundary components=0
  ];
\end{scope}
\end{tikzpicture}
  \end{aligned}  
  \overset{\text{gluing}}{\rightarrow}
      \begin{aligned}
      \begin{tikzpicture}[
  tqft,
  every outgoing boundary component/.style={fill=blue!50},
  outgoing boundary component 3/.style={fill=none,draw=red},
  every incoming boundary component/.style={fill=green!50},
  every lower boundary component/.style={draw=black},
  every upper boundary component/.style={draw,black},
  cobordism/.style={fill=teal!80},
  cobordism edge/.style={draw},
  view from=incoming,
  cobordism height=5cm,
]
\begin{scope}[every node/.style={rotate=90}]
\pic[name=a,
  tqft,
  incoming boundary components=1,
  outgoing boundary components=0
  ];
\end{scope}
\end{tikzpicture}
  \end{aligned}  
\end{equation*}
This means that we can topologically enlarge the disk $\Gamma$ to $\Gamma'=\Gamma\cup \operatorname{Tube}_{\Gamma}$.
The excitations for $H[\Gamma']$ and $H[\Gamma]$ are topologically equivalent.
From the algebraic point of view, this means that there is an action of tube over the disk.
Consider a general tube $\operatorname{Tube}$, it has two boundaries, both of the boundaries consist of collections of boundary faces. We can also construct a Hamiltonian for the tube
\begin{equation}
    H[\operatorname{Tube}]=-\sum_{v\in \operatorname{Tube}}Q_v -\sum_{f\in \operatorname{Tube}} B_f,
\end{equation}
where by $v\in \operatorname{Tube}$ and $f\in \operatorname{Tube}$ we mean the vertices and faces that fully contained in $\operatorname{Tube}$.
The tube ground state space will also have decomposition depending on the choices of open-boundary conditions for two boundaries $\partial_0\!\operatorname{Tube}$ and $\partial_1\!\operatorname{Tube}$:
\begin{equation}
    \mathcal{V}_{\rm GS}[\operatorname{Tube}]=\bigoplus_{\partial_0\!\operatorname{Tube},\,\partial_1\!\operatorname{Tube}}\mathcal{V}_{\rm GS}[\operatorname{Tube}]_{\partial_0\!\operatorname{Tube},\,\partial_1\!\operatorname{Tube}}.
\end{equation}

\begin{equation*}
  \begin{aligned}
      \begin{tikzpicture}[
  tqft,
  every outgoing boundary component/.style={fill=blue!50},
  outgoing boundary component 3/.style={fill=none,draw=red},
  every incoming boundary component/.style={fill=green!50},
  every lower boundary component/.style={draw=black},
  every upper boundary component/.style={draw,purple},
  cobordism/.style={fill=orange!50},
  cobordism edge/.style={draw},
  view from=incoming,
  cobordism height=1cm,
]
\begin{scope}[every node/.style={rotate=90}]
\pic[name=a,
  tqft,
  incoming boundary components=1,
  outgoing boundary components=1
  ];
\end{scope}
\end{tikzpicture}
  \end{aligned}  
    \begin{aligned}
      \begin{tikzpicture}[
  tqft,
  every outgoing boundary component/.style={fill=blue!50},
  outgoing boundary component 3/.style={fill=none,draw=red},
  every incoming boundary component/.style={fill=green!50},
  every lower boundary component/.style={draw=black},
  every upper boundary component/.style={draw,purple},
  cobordism/.style={fill=orange!50},
  cobordism edge/.style={draw},
  view from=incoming,
  cobordism height=1cm,
]
\begin{scope}[every node/.style={rotate=90}]
\pic[name=a,
  tqft,
  incoming boundary components=1,
  outgoing boundary components=1
  ];
\end{scope}
\end{tikzpicture}
  \end{aligned}  
  \overset{\text{gluing}}{\rightarrow}
      \begin{aligned}
      \begin{tikzpicture}[
  tqft,
  every outgoing boundary component/.style={fill=blue!50},
  outgoing boundary component 3/.style={fill=none,draw=red},
  every incoming boundary component/.style={fill=green!50},
  every lower boundary component/.style={draw=black},
  every upper boundary component/.style={draw,purple},
  cobordism/.style={fill=orange!50},
  cobordism edge/.style={draw},
  view from=incoming,
  cobordism height=1cm,
]
\begin{scope}[every node/.style={rotate=90}]
\pic[name=a,
  tqft,
  incoming boundary components=1,
  outgoing boundary components=1
  ];
  \pic[name=a,
  tqft,
  incoming boundary components=1,
  outgoing boundary components=1
  ];
\end{scope}
\end{tikzpicture}
  \end{aligned}  
\end{equation*}

Consider a domain wall that separates two multifusion string-net phases with input UMFCs $\EC$ and $\ED$, respectively. Based on the Kitaev-Kong construction \cite{Kitaev2012boundary}, we know that the domain wall is characterized by a $\EC|\ED$-bimodule category $\EM$.
Suppose that there is a disk region $\Gamma$ across the domain wall where the stabilizer conditions do not hold for vertices and faces inside the region. Due to the topological nature of the string-net model, we can deform the tube region and $\Gamma$  using the topological local moves such that the tube region only consists of domain wall vertices  and the boundary of excited region only consists of two domain wall edges (see Fig.~\ref{fig:tube}).
Each tube vertex corresponds to a morphism space in the bimodule category. By taking the tensor product of these vertex spaces while applying the edge stability condition, we naturally obtain a vector space. This vector space will be denoted as $\mathbf{Tube}(_{\EC}\EM_{\ED})$.
We will demonstrate that this vector space forms an algebra, and the wall excitations are modules over this algebra.
Let us first consider a special case of trivial domain wall (viewing $\ED$ as a $\ED|\ED$-bimodule).
The topological excitations of the multifusion string-net model can be regarded as a point defect over trivial domain walls.

\begin{definition}[Bulk tube algebra]
The bulk tube algebra $\mathbf{Tube}(_{\ED}\ED_{\ED})$ is spanned by the following basis (up to planar isotopy):
\begin{equation}
   \left\{  
   \begin{aligned}
        \begin{tikzpicture}
             \draw[line width=.6pt,black] (0,0.5)--(0,1.5);
             \draw[line width=.6pt,black] (0,-0.5)--(0,-1.5);
             \draw[red] (0,0.8) arc[start angle=90, end angle=270, radius=0.8];
             \draw[blue] (0,1.3) arc[start angle=90, end angle=-90, radius=1.3];
            \node[ line width=0.6pt, dashed, draw opacity=0.5] (a) at (0,1.7){$h$};
             \node[ line width=0.6pt, dashed, draw opacity=0.5] (a) at (0,-1.7){$c$};
            \node[ line width=0.6pt, dashed, draw opacity=0.5] (a) at (-1,0){$a$};
            \node[ line width=0.6pt, dashed, draw opacity=0.5] (a) at (1.5,0){$b$};
            \node[ line width=0.6pt, dashed, draw opacity=0.5] (a) at (-0.2,-1){$d$};
            \node[ line width=0.6pt, dashed, draw opacity=0.5] (a) at (-0.4,-1.3){$\mu$};
            \node[ line width=0.6pt, dashed, draw opacity=0.5] (a) at (0.2,-0.8){$\nu$};
            \node[ line width=0.6pt, dashed, draw opacity=0.5] (a) at (0,-0.3){$e$};
            \node[ line width=0.6pt, dashed, draw opacity=0.5] (a) at (0,0.3){$f$};
            \node[ line width=0.6pt, dashed, draw opacity=0.5] (a) at (-0.2,1){$g$};
            \node[ line width=0.6pt, dashed, draw opacity=0.5] (a) at (-0.4,1.3){$\gamma$};
            \node[ line width=0.6pt, dashed, draw opacity=0.5] (a) at (0.2,0.8){$\zeta$};
        \end{tikzpicture}
    \end{aligned}
   :\quad  a,\cdots,h\in \Irr(\ED), \mu,\nu,\gamma,\zeta\in \Hom_{\ED} 
    \right\}.
    \label{eq:tubebasis}
\end{equation}
Note that the arrows have been omitted to avoid clutter in the equation; all edges are assumed to be directed upwards. Also, observe that these basis elements must adhere to the graded fusion rule of $\ED$, namely, the following fusions must be nonzero: $a\otimes e$, $d\otimes b$, $a\otimes f$, $g\otimes b$. 
\end{definition}

It is clear that the dimension of $\mathbf{Tube}({_{\ED}\ED_{\ED}})$ is equal to 
\begin{equation}
    \dim \mathbf{Tube}({_{\ED}\ED_{\ED}}) = \sum_{a,\cdots,h\in\Irr(\ED)} N_{c}^{db}N_{d}^{ae}N_{af}^gN_{gb}^h.
\end{equation}

\begin{proposition}
   For any given UMFC $\ED$, the tube algebra $\mathbf{Tube}({_{\ED}\ED_{\ED}})$ is an algebra with the following algebra structure. 
   \begin{itemize}
       \item The unit is given by
       \begin{equation}
           1=\sum_{a,b} \;\begin{aligned}
        \begin{tikzpicture}
             \draw[line width=.6pt,black] (0,0.5)--(0,1.5);
             \draw[line width=.6pt,black] (0,-0.5)--(0,-1.5);
             \draw[red, dotted] (0,0.8) arc[start angle=90, end angle=270, radius=0.8];
             \draw[blue,dotted] (0,1.3) arc[start angle=90, end angle=-90, radius=1.3];
            \node[ line width=0.6pt, dashed, draw opacity=0.5] (a) at (-0.2,-1){$a$};
            \node[ line width=0.6pt, dashed, draw opacity=0.5] (a) at (-0.2,1){$b$};
        \end{tikzpicture}
    \end{aligned}\;,
    \end{equation}
    where the dotted line represents the tensor unit $\one$. It is worth noting that from Eq.~\eqref{eq:tubebasis}, these dotted lines indeed represent $\one_i$, where $i\in I$ depending on the choices of $a,b\in \Irr(\ED)$. For example, if $a,b\in \ED_{i,j}$, the left dotted line represents $\one_i$ and the right dotted line represents $\one_j$.
    \item The multiplication is of the form  
    \begin{align}
       &\mu\left( \begin{aligned}
        \begin{tikzpicture}
             \draw[line width=.6pt,black] (0,0.5)--(0,1.5);
             \draw[line width=.6pt,black] (0,-0.5)--(0,-1.5);
             \draw[red] (0,0.8) arc[start angle=90, end angle=270, radius=0.8];
             \draw[blue] (0,1.3) arc[start angle=90, end angle=-90, radius=1.3];
            \node[ line width=0.6pt, dashed, draw opacity=0.5] (a) at (0,1.7){$h$};
             \node[ line width=0.6pt, dashed, draw opacity=0.5] (a) at (0,-1.7){$c$};
            \node[ line width=0.6pt, dashed, draw opacity=0.5] (a) at (-1,0){$a$};
            \node[ line width=0.6pt, dashed, draw opacity=0.5] (a) at (1.5,0){$b$};
            \node[ line width=0.6pt, dashed, draw opacity=0.5] (a) at (-0.2,-1){$d$};
            \node[ line width=0.6pt, dashed, draw opacity=0.5] (a) at (-0.4,-1.3){$\mu$};
            \node[ line width=0.6pt, dashed, draw opacity=0.5] (a) at (0.2,-0.8){$\nu$};
            \node[ line width=0.6pt, dashed, draw opacity=0.5] (a) at (0,-0.3){$e$};
            \node[ line width=0.6pt, dashed, draw opacity=0.5] (a) at (0,0.3){$f$};
            \node[ line width=0.6pt, dashed, draw opacity=0.5] (a) at (-0.2,1){$g$};
            \node[ line width=0.6pt, dashed, draw opacity=0.5] (a) at (-0.4,1.3){$\gamma$};
            \node[ line width=0.6pt, dashed, draw opacity=0.5] (a) at (0.2,0.8){$\zeta$};
        \end{tikzpicture}
    \end{aligned}\otimes \begin{aligned}
        \begin{tikzpicture}
             \draw[line width=.6pt,black] (0,0.5)--(0,1.5);
             \draw[line width=.6pt,black] (0,-0.5)--(0,-1.5);
             \draw[red] (0,0.8) arc[start angle=90, end angle=270, radius=0.8];
             \draw[blue] (0,1.3) arc[start angle=90, end angle=-90, radius=1.3];
            \node[ line width=0.6pt, dashed, draw opacity=0.5] (a) at (0,1.7){$h'$};
             \node[ line width=0.6pt, dashed, draw opacity=0.5] (a) at (0,-1.7){$c'$};
            \node[ line width=0.6pt, dashed, draw opacity=0.5] (a) at (-1,0){$a'$};
            \node[ line width=0.6pt, dashed, draw opacity=0.5] (a) at (1.5,0){$b'$};
            \node[ line width=0.6pt, dashed, draw opacity=0.5] (a) at (-0.2,-1){$d'$};
            \node[ line width=0.6pt, dashed, draw opacity=0.5] (a) at (-0.4,-1.4){$\mu'$};
            \node[ line width=0.6pt, dashed, draw opacity=0.5] (a) at (0.2,-0.8){$\nu'$};
            \node[ line width=0.6pt, dashed, draw opacity=0.5] (a) at (0,-0.3){$e'$};
            \node[ line width=0.6pt, dashed, draw opacity=0.5] (a) at (0,0.3){$f'$};
            \node[ line width=0.6pt, dashed, draw opacity=0.5] (a) at (-0.2,1){$g'$};
            \node[ line width=0.6pt, dashed, draw opacity=0.5] (a) at (-0.4,1.3){$\gamma'$};
            \node[ line width=0.6pt, dashed, draw opacity=0.5] (a) at (0.2,0.8){$\zeta'$};
        \end{tikzpicture}
    \end{aligned}\right)=\delta_{f,h'}\delta_{e,c'}\;\begin{aligned}
        \begin{tikzpicture}
             \draw[line width=.6pt,black] (0,0.5)--(0,2.5);
             \draw[line width=.6pt,black] (0,-0.5)--(0,-2.5);
             \draw[red] (0,0.8) arc[start angle=90, end angle=270, radius=0.8];
             \draw[blue] (0,1.3) arc[start angle=90, end angle=-90, radius=1.3];
            \node[ line width=0.6pt, dashed, draw opacity=0.5] (a) at (0.2,1.5){$f$};
             \node[ line width=0.6pt, dashed, draw opacity=0.5] (a) at (0.2,-1.5){$e$};
            \node[ line width=0.6pt, dashed, draw opacity=0.5] (a) at (-1,0){$a'$};
            \node[ line width=0.6pt, dashed, draw opacity=0.5] (a) at (1.5,0){$b'$};
            \node[ line width=0.6pt, dashed, draw opacity=0.5] (a) at (-0.2,-1){$d'$};
            \node[ line width=0.6pt, dashed, draw opacity=0.5] (a) at (-0.2,-1.3){$\mu'$};
            \node[ line width=0.6pt, dashed, draw opacity=0.5] (a) at (0.2,-0.8){$\nu'$};
            \node[ line width=0.6pt, dashed, draw opacity=0.5] (a) at (0,-0.3){$e'$};
            \node[ line width=0.6pt, dashed, draw opacity=0.5] (a) at (0,0.3){$f'$};
            \node[ line width=0.6pt, dashed, draw opacity=0.5] (a) at (-0.2,1){$g'$};
            \node[ line width=0.6pt, dashed, draw opacity=0.5] (a) at (-0.4,1.3){$\gamma'$};
            \node[ line width=0.6pt, dashed, draw opacity=0.5] (a) at (0.2,0.8){$\zeta'$};
            \draw[red] (0,1.6) arc[start angle=90, end angle=270, radius=1.6];
            \node[ line width=0.6pt, dashed, draw opacity=0.5] (a) at (-0.3,1.8){$\zeta$};
          \node[ line width=0.6pt, dashed, draw opacity=0.5] (a) at (-0.3,-1.8){$\nu$};
           \draw[blue] (0,2.1) arc[start angle=90, end angle=-90, radius=2.1];
            \node[ line width=0.6pt, dashed, draw opacity=0.5] (a) at (-1.3,0.4){$a$};
                \node[ line width=0.6pt, dashed, draw opacity=0.5] (a) at (1.8,0.4){$b$};
         \node[ line width=0.6pt, dashed, draw opacity=0.5] (a) at (0.2,1.8){$g$};
          \node[ line width=0.6pt, dashed, draw opacity=0.5] (a) at (0.2,-1.8){$d$};
        \node[ line width=0.6pt, dashed, draw opacity=0.5] (a) at (0.2,2.3){$\gamma$};
        \node[ line width=0.6pt, dashed, draw opacity=0.5] (a) at (0.2,-2.3){$\mu$};
       \node[ line width=0.6pt, dashed, draw opacity=0.5] (a) at (-0.3,2.5){$h$};
       \node[ line width=0.6pt, dashed, draw opacity=0.5] (a) at (-0.3,-2.5){$c$};
        \end{tikzpicture}
    \end{aligned}\nonumber\\
    =&\sum [(F_{ag'b'}^g)^{-1}]_{f \gamma' \zeta}^{i\rho \sigma}[(F^{ad'b'}_d)^{-1}]_{e \mu' \nu }^{j \omega \tau} [(F_{aa'f'}^i)^{-1}]_{g' \zeta' \rho}^{k\alpha \beta} [(F^{aa'e'}_j)^{-1}]_{d' \omega \nu'}^{k \alpha \beta' } \nonumber\\
    &\quad \quad \quad [F_{ib'b}^h]_{g\sigma\gamma}^{s \xi \lambda}[F_c^{jb'b}]_{d\tau\mu}^{s\xi \lambda'}\delta_{f,h'}\delta_{e,c'} \sqrt{\frac{d_ad_{a'}}{d_k}} \sqrt{\frac{d_bd_{b'}}{d_s}}   \begin{aligned}
        \begin{tikzpicture}
             \draw[line width=.6pt,black] (0,0.5)--(0,1.5);
             \draw[line width=.6pt,black] (0,-0.5)--(0,-1.5);
             \draw[red] (0,0.8) arc[start angle=90, end angle=270, radius=0.8];
             \draw[blue] (0,1.3) arc[start angle=90, end angle=-90, radius=1.3];
            \node[ line width=0.6pt, dashed, draw opacity=0.5] (a) at (0,1.7){$h$};
             \node[ line width=0.6pt, dashed, draw opacity=0.5] (a) at (0,-1.7){$c$};
            \node[ line width=0.6pt, dashed, draw opacity=0.5] (a) at (-1,0){$k$};
            \node[ line width=0.6pt, dashed, draw opacity=0.5] (a) at (1.5,0){$s$};
            \node[ line width=0.6pt, dashed, draw opacity=0.5] (a) at (-0.2,-1){$j$};
            \node[ line width=0.6pt, dashed, draw opacity=0.5] (a) at (-0.3,-1.4){$\lambda'$};
            \node[ line width=0.6pt, dashed, draw opacity=0.5] (a) at (0.2,-0.8){$\beta'$};
            \node[ line width=0.6pt, dashed, draw opacity=0.5] (a) at (0,-0.3){$e'$};
            \node[ line width=0.6pt, dashed, draw opacity=0.5] (a) at (0,0.3){$f'$};
            \node[ line width=0.6pt, dashed, draw opacity=0.5] (a) at (-0.2,1){$i$};
            \node[ line width=0.6pt, dashed, draw opacity=0.5] (a) at (-0.4,1.3){$\lambda$};
            \node[ line width=0.6pt, dashed, draw opacity=0.5] (a) at (0.2,0.8){$\beta$};
        \end{tikzpicture}
    \end{aligned}, \label{eq:tube-mult}
     \end{align}
     where the summation is taken over all the upper indices. 
   \end{itemize}
\end{proposition}

\begin{proof}
From the definition, it is evident that multiplication is associative, as the diagrams for both $(\bullet \cdot \bullet)\cdot \bullet$ and $\bullet\cdot (\bullet\cdot \bullet)$ are the same. The unit property follows directly from the definition.
\end{proof}

To simplify the notation, we will also represent the basis diagram in Eq.~\eqref{eq:tubebasis} as $[a;b]_{c,d,e;\mu,\nu}^{f,g,h;\zeta,\gamma}$.
In the same vein, we denote the diagram on the right-hand side of the multiplication as $[a,a';b',b]_{c,d,c',d',e';\mu,\nu,\mu',\nu'}^{f',g',h',g,h;\zeta',\gamma',\zeta,\gamma}$.
For the more general case, a similar notation applies.
Since the excited region has two boundary edges, we can denote it as $[\Psi]_x^y\in \Hom_{\ED}(x,y)$.
The space $\mathbf{Disk}_{_{\ED}\ED_{\ED}}=\bigoplus_{x,y\in \Irr(\ED)}\Hom_{\ED}(x,y)$ forms a left module over $\mathbf{Tube}({_{\ED}\ED_{\ED}})$ with the action given by
\begin{equation}\label{eq:TubeModule}
       \begin{aligned}
        \begin{tikzpicture}
             \draw[line width=.6pt,black] (0,0.5)--(0,1.5);
             \draw[line width=.6pt,black] (0,-0.5)--(0,-1.5);
             \draw[red] (0,0.8) arc[start angle=90, end angle=270, radius=0.8];
             \draw[blue] (0,1.3) arc[start angle=90, end angle=-90, radius=1.3];
            \node[ line width=0.6pt, dashed, draw opacity=0.5] (a) at (0,1.7){$h$};
             \node[ line width=0.6pt, dashed, draw opacity=0.5] (a) at (0,-1.7){$c$};
            \node[ line width=0.6pt, dashed, draw opacity=0.5] (a) at (-1,0){$a$};
            \node[ line width=0.6pt, dashed, draw opacity=0.5] (a) at (1.5,0){$b$};
            \node[ line width=0.6pt, dashed, draw opacity=0.5] (a) at (-0.2,-1){$d$};
            \node[ line width=0.6pt, dashed, draw opacity=0.5] (a) at (-0.4,-1.3){$\mu$};
            \node[ line width=0.6pt, dashed, draw opacity=0.5] (a) at (0.2,-0.8){$\nu$};
            \node[ line width=0.6pt, dashed, draw opacity=0.5] (a) at (0,-0.3){$e$};
            \node[ line width=0.6pt, dashed, draw opacity=0.5] (a) at (0,0.3){$f$};
            \node[ line width=0.6pt, dashed, draw opacity=0.5] (a) at (-0.2,1){$g$};
            \node[ line width=0.6pt, dashed, draw opacity=0.5] (a) at (-0.4,1.3){$\gamma$};
            \node[ line width=0.6pt, dashed, draw opacity=0.5] (a) at (0.2,0.8){$\zeta$};
        \end{tikzpicture}
    \end{aligned}\triangleright 
    \begin{aligned}
        \begin{tikzpicture}
             \draw[line width=.6pt,black] (0,0)--(0,0.6);
             \draw [draw=black,fill=gray, fill opacity=0.2] (-0.3,0) rectangle (0.3,-0.6);
             \draw[line width=.6pt,black] (0,-0.6)--(0,-1.2);
              \node[ line width=0.6pt, dashed, draw opacity=0.5] (a) at (0,-.3){$\Psi$};
             \node[ line width=0.6pt, dashed, draw opacity=0.5] (a) at (0.2,-.9){$x$};
               \node[ line width=0.6pt, dashed, draw opacity=0.5] (a) at (0.2,.3){$y$};
        \end{tikzpicture}
    \end{aligned}
    =\frac{\delta_{x,e}\delta_{y,f}}{\sqrt{d_ed_f}}
     \begin{aligned}
        \begin{tikzpicture}
             \draw[line width=.6pt,black] (0,0.3)--(0,1.5);
             \draw[line width=.6pt,black] (0,-0.3)--(0,-1.5);
             \draw[red] (0,0.8) arc[start angle=90, end angle=270, radius=0.8];
             \draw[blue] (0,1.3) arc[start angle=90, end angle=-90, radius=1.3];
            \node[ line width=0.6pt, dashed, draw opacity=0.5] (a) at (0,1.7){$h$};
             \node[ line width=0.6pt, dashed, draw opacity=0.5] (a) at (0,-1.7){$c$};
            \node[ line width=0.6pt, dashed, draw opacity=0.5] (a) at (-1,0){$a$};
            \node[ line width=0.6pt, dashed, draw opacity=0.5] (a) at (1.5,0){$b$};
            \node[ line width=0.6pt, dashed, draw opacity=0.5] (a) at (-0.2,-1){$d$};
            \node[ line width=0.6pt, dashed, draw opacity=0.5] (a) at (-0.4,-1.3){$\mu$};
            \node[ line width=0.6pt, dashed, draw opacity=0.5] (a) at (0.2,-0.8){$\nu$};
            \node[ line width=0.6pt, dashed, draw opacity=0.5] (a) at (-0.2,-0.6){$e$};
            \node[ line width=0.6pt, dashed, draw opacity=0.5] (a) at (-0.2,0.5){$f$};
            \node[ line width=0.6pt, dashed, draw opacity=0.5] (a) at (-0.2,1){$g$};
            \node[ line width=0.6pt, dashed, draw opacity=0.5] (a) at (-0.4,1.3){$\gamma$};
            \node[ line width=0.6pt, dashed, draw opacity=0.5] (a) at (0.2,0.8){$\zeta$};
             \draw [draw=black,fill=gray, fill opacity=0.2] (-0.3,-0.3) rectangle (0.3,0.3);
             \node[ line width=0.6pt, dashed, draw opacity=0.5] (a) at (0,0){$\Psi$};
        \end{tikzpicture}
    \end{aligned}\in \mathbf{Disk}_{_{\ED}\ED_{\ED}}. 
\end{equation}
It can be easily verified that this is indeed a module action by using the topological local moves.
We observe that the physical properties of a given region $\Gamma$ are encoded in the boundary conditions of the region. State renormalization allows us to treat the excited region as a left module over the tube algebra. The topological charges are classified by the irreducible representations of the tube algebra. We can also examine the relationship between the excited region before and after the topological local moves; as we will show later, their corresponding tube algebras are Morita equivalent.

Before we discuss the representation theory of the tube algebra, let us first construct the coalgebra structure for the tube algebra. Suppose there are two excited regions on the domain wall; we can consider the action of the tube algebra over these two regions.
Using parallel evaluation, we obtain the following:
\begin{equation}
    \begin{aligned}
        \begin{tikzpicture}
             \draw[line width=.6pt,black] (0,0)--(0,1.5);
             \draw [draw=black,fill=gray, fill opacity=0.2] (-0.3,0) rectangle (0.3,-0.6);
              \draw[line width=.6pt,black] (0,-0.6)--(0,-1.5);
             \draw [draw=black,fill=gray, fill opacity=0.2] (-0.3,1.5) rectangle (0.3,2.1);  
            \draw[line width=.6pt,black] (0,2.1)--(0,3);
            \draw[red] (0,2.4) arc[start angle=90, end angle=270, radius=1.65];
             \draw[blue] (0,2.7) arc[start angle=90, end angle=-90, radius=1.95];
            \node[ line width=0.6pt, dashed, draw opacity=0.5] (a) at (0.2,0.8){$i$};
            \node[ line width=0.6pt, dashed, draw opacity=0.5] (a) at (-1.35,0.8){$a$};
              \node[ line width=0.6pt, dashed, draw opacity=0.5] (a) at (1.65,0.8){$b$};
            \node[ line width=0.6pt, dashed, draw opacity=0.5] (a) at (0,-1.8){$c$};
            \node[ line width=0.6pt, dashed, draw opacity=0.5] (a) at (0.25,-0.8){$e$};
              \node[ line width=0.6pt, dashed, draw opacity=0.5] (a) at (-0.2,-1.1){$d$};
            \node[ line width=0.6pt, dashed, draw opacity=0.5] (a) at (0.2,-1.4){$\mu$};
                \node[ line width=0.6pt, dashed, draw opacity=0.5] (a) at (0.15,-1){$\nu$};
               \node[ line width=0.6pt, dashed, draw opacity=0.5] (a) at (0,-.3){$\Phi$}; 
                  \node[ line width=0.6pt, dashed, draw opacity=0.5] (a) at (0,1.8){$\Psi$}; 
           \node[ line width=0.6pt, dashed, draw opacity=0.5] (a) at (0.25,2.3){$f$};             \node[ line width=0.6pt, dashed, draw opacity=0.5] (a) at (-0.25,2.3){$\zeta$};  
               \node[ line width=0.6pt, dashed, draw opacity=0.5] (a) at (-0.3,2.6){$g$};
                    \node[ line width=0.6pt, dashed, draw opacity=0.5] (a) at (0.25,2.7){$\gamma$};  
        \node[ line width=0.6pt, dashed, draw opacity=0.5] (a) at (0,3.2){$h$};  
        \end{tikzpicture}
    \end{aligned}
    =
    \sum_{j,k,\rho,\sigma} \sqrt{ \frac{d_j}{d_ad_i}}\sqrt{ \frac{d_k}{d_jd_b}}
       \begin{aligned}
        \begin{tikzpicture}
             \draw[line width=.6pt,black] (0,0)--(0,1.5);
             \draw [draw=black,fill=gray, fill opacity=0.2] (-0.3,0) rectangle (0.3,-0.6);
              \draw[line width=.6pt,black] (0,-0.6)--(0,-1.5);
             \draw [draw=black,fill=gray, fill opacity=0.2] (-0.3,1.5) rectangle (0.3,2.1);  
            \draw[line width=.6pt,black] (0,2.1)--(0,3);
            \draw[red] (0,2.4) arc[start angle=90, end angle=270, radius=0.6];
             \draw[blue] (0,2.7) arc[start angle=90, end angle=-90, radius=0.9];
            \draw[red] (0,0.3) arc[start angle=90, end angle=270, radius=0.6];
    \draw[blue] (0,0.6) arc[start angle=90, end angle=-90, radius=0.9];
             \node[ line width=0.6pt, dashed, draw opacity=0.5] (a) at (0,-1.8){$c$};
            \node[ line width=0.6pt, dashed, draw opacity=0.5] (a) at (0.25,-0.8){$e$};
              \node[ line width=0.6pt, dashed, draw opacity=0.5] (a) at (-0.2,-1.1){$d$};
            \node[ line width=0.6pt, dashed, draw opacity=0.5] (a) at (0.2,-1.4){$\mu$};
                \node[ line width=0.6pt, dashed, draw opacity=0.5] (a) at (0.15,-1){$\nu$};
               \node[ line width=0.6pt, dashed, draw opacity=0.5] (a) at (0,-.3){$\Phi$}; 
                  \node[ line width=0.6pt, dashed, draw opacity=0.5] (a) at (0,1.8){$\Psi$}; 
      \node[ line width=0.6pt, dashed, draw opacity=0.5] (a) at (-0.45,-.3){$a$}; 
                  \node[ line width=0.6pt, dashed, draw opacity=0.5] (a) at (0.7,1.8){$b$};
  \node[ line width=0.6pt, dashed, draw opacity=0.5] (a) at (-0.45,1.8){$a$};
      \node[ line width=0.6pt, dashed, draw opacity=0.5] (a) at (0.7,-.3){$b$};   
        \node[ line width=0.6pt, dashed, draw opacity=0.5] (a) at (0.25,2.3){$f$};             \node[ line width=0.6pt, dashed, draw opacity=0.5] (a) at (-0.25,2.3){$\zeta$};  
               \node[ line width=0.6pt, dashed, draw opacity=0.5] (a) at (-0.3,2.6){$g$};
                    \node[ line width=0.6pt, dashed, draw opacity=0.5] (a) at (0.25,2.7){$\gamma$};  
        \node[ line width=0.6pt, dashed, draw opacity=0.5] (a) at (0,3.2){$h$};  
       \node[ line width=0.6pt, dashed, draw opacity=0.5] (a) at (0.25,1.2){$\sigma$};  
        \node[ line width=0.6pt, dashed, draw opacity=0.5] (a) at (-0.2,1.35){$i$}; 
          \node[ line width=0.6pt, dashed, draw opacity=0.5] (a) at (0.25,0.8){$\rho$};  
              \node[ line width=0.6pt, dashed, draw opacity=0.5] (a) at (-0.25,0.5){$\rho$};
                       \node[ line width=0.6pt, dashed, draw opacity=0.5] (a) at (0.25,0.35){$\sigma$}; 
             \node[ line width=0.6pt, dashed, draw opacity=0.5] (a) at (-0.2,1){$j$};               \node[ line width=0.6pt, dashed, draw opacity=0.5] (a) at (-0.4,0.75){$k$};   \node[ line width=0.6pt, dashed, draw opacity=0.5] (a) at (-0.1,0.4){$j$}; 
               \node[ line width=0.6pt, dashed, draw opacity=0.5] (a) at (0.1,0.15){$i$};  
        \end{tikzpicture}
    \end{aligned}. \label{eq:VerticalFusion}
\end{equation}
This naturally gives us the comultiplication structure, as follows.

\begin{proposition} \label{prop:tube-coalgebra}
 For any given UMFC $\ED$, the tube algebra $\mathbf{Tube}({_{\ED}\ED_{\ED}})$ is a coalgebra with the following coalgebra structure. 
\begin{itemize}
    \item The counit is of the form: 
    \begin{equation}
    \varepsilon\left(\begin{aligned}\begin{tikzpicture}
             \draw[line width=.6pt,black] (0,0.5)--(0,1.5);
             \draw[line width=.6pt,black] (0,-0.5)--(0,-1.5);
             \draw[red] (0,0.8) arc[start angle=90, end angle=270, radius=0.8];
             \draw[blue] (0,1.3) arc[start angle=90, end angle=-90, radius=1.3];
            \node[ line width=0.6pt, dashed, draw opacity=0.5] (a) at (0,1.7){$h$};
             \node[ line width=0.6pt, dashed, draw opacity=0.5] (a) at (0,-1.7){$c$};
            \node[ line width=0.6pt, dashed, draw opacity=0.5] (a) at (-1,0){$a$};
            \node[ line width=0.6pt, dashed, draw opacity=0.5] (a) at (1.5,0){$b$};
            \node[ line width=0.6pt, dashed, draw opacity=0.5] (a) at (-0.2,-1){$d$};
            \node[ line width=0.6pt, dashed, draw opacity=0.5] (a) at (-0.4,-1.3){$\mu$};
            \node[ line width=0.6pt, dashed, draw opacity=0.5] (a) at (0.2,-0.8){$\nu$};
            \node[ line width=0.6pt, dashed, draw opacity=0.5] (a) at (0,-0.3){$e$};
            \node[ line width=0.6pt, dashed, draw opacity=0.5] (a) at (0,0.3){$f$};
            \node[ line width=0.6pt, dashed, draw opacity=0.5] (a) at (-0.2,1){$g$};
            \node[ line width=0.6pt, dashed, draw opacity=0.5] (a) at (-0.4,1.3){$\gamma$};
            \node[ line width=0.6pt, dashed, draw opacity=0.5] (a) at (0.2,0.8){$\zeta$};
        \end{tikzpicture}
    \end{aligned}\right)
    =\frac{\delta_{e,f}\delta_{c,h}}{d_h}
    \begin{aligned}\begin{tikzpicture}
             \draw[line width=.6pt,black] (0,-1.8)--(0,1.8);
             \draw[red] (0,0.8) arc[start angle=90, end angle=270, radius=0.8];
             \draw[blue] (0,1.3) arc[start angle=90, end angle=-90, radius=1.3];
             \draw[black] (0,1.8) arc[start angle=90, end angle=-90, radius=1.8];
            \node[ line width=0.6pt, dashed, draw opacity=0.5] (a) at (1.6,0){$h$};
            \node[ line width=0.6pt, dashed, draw opacity=0.5] (a) at (-0.6,0){$a$};
            \node[ line width=0.6pt, dashed, draw opacity=0.5] (a) at (1.1,0){$b$};
            \node[ line width=0.6pt, dashed, draw opacity=0.5] (a) at (-0.2,-1){$d$};
            \node[ line width=0.6pt, dashed, draw opacity=0.5] (a) at (-0.4,-1.3){$\mu$};
            \node[ line width=0.6pt, dashed, draw opacity=0.5] (a) at (0.2,-0.8){$\nu$};
            \node[ line width=0.6pt, dashed, draw opacity=0.5] (a) at (0.3,0){$f$};
            \node[ line width=0.6pt, dashed, draw opacity=0.5] (a) at (-0.2,1){$g$};
            \node[ line width=0.6pt, dashed, draw opacity=0.5] (a) at (-0.4,1.3){$\gamma$};
            \node[ line width=0.6pt, dashed, draw opacity=0.5] (a) at (0.2,0.8){$\zeta$};
        \end{tikzpicture}
    \end{aligned}=\delta_{e,f}\delta_{c,h}\delta_{d,g}\delta_{\nu,\zeta}\delta_{\mu,\gamma} \sqrt{\frac{d_a d_f d_b}{d_h}}.
    \end{equation}
    \item The comultiplication is given by
    \begin{equation}
    \Delta\left(\begin{aligned}\begin{tikzpicture}
             \draw[line width=.6pt,black] (0,0.5)--(0,1.5);
             \draw[line width=.6pt,black] (0,-0.5)--(0,-1.5);
             \draw[red] (0,0.8) arc[start angle=90, end angle=270, radius=0.8];
             \draw[blue] (0,1.3) arc[start angle=90, end angle=-90, radius=1.3];
            \node[ line width=0.6pt, dashed, draw opacity=0.5] (a) at (0,1.7){$h$};
             \node[ line width=0.6pt, dashed, draw opacity=0.5] (a) at (0,-1.7){$c$};
            \node[ line width=0.6pt, dashed, draw opacity=0.5] (a) at (-1,0){$a$};
            \node[ line width=0.6pt, dashed, draw opacity=0.5] (a) at (1.5,0){$b$};
            \node[ line width=0.6pt, dashed, draw opacity=0.5] (a) at (-0.2,-1){$d$};
            \node[ line width=0.6pt, dashed, draw opacity=0.5] (a) at (-0.4,-1.3){$\mu$};
            \node[ line width=0.6pt, dashed, draw opacity=0.5] (a) at (0.2,-0.8){$\nu$};
            \node[ line width=0.6pt, dashed, draw opacity=0.5] (a) at (0,-0.3){$e$};
            \node[ line width=0.6pt, dashed, draw opacity=0.5] (a) at (0,0.3){$f$};
            \node[ line width=0.6pt, dashed, draw opacity=0.5] (a) at (-0.2,1){$g$};
            \node[ line width=0.6pt, dashed, draw opacity=0.5] (a) at (-0.4,1.3){$\gamma$};
            \node[ line width=0.6pt, dashed, draw opacity=0.5] (a) at (0.2,0.8){$\zeta$};
        \end{tikzpicture}
    \end{aligned}\right)
    =\sum_{i,j,k,\rho,\sigma} \sqrt{ \frac{d_j}{d_ad_i}}\sqrt{ \frac{d_k}{d_jd_b}}
    \begin{aligned}\begin{tikzpicture}
             \draw[line width=.6pt,black] (0,0.5)--(0,1.5);
             \draw[line width=.6pt,black] (0,-0.5)--(0,-1.5);
             \draw[red] (0,0.8) arc[start angle=90, end angle=270, radius=0.8];
             \draw[blue] (0,1.3) arc[start angle=90, end angle=-90, radius=1.3];
            \node[ line width=0.6pt, dashed, draw opacity=0.5] (a) at (0,1.7){$h$};
             \node[ line width=0.6pt, dashed, draw opacity=0.5] (a) at (0,-1.7){$k$};
            \node[ line width=0.6pt, dashed, draw opacity=0.5] (a) at (-1,0){$a$};
            \node[ line width=0.6pt, dashed, draw opacity=0.5] (a) at (1.5,0){$b$};
            \node[ line width=0.6pt, dashed, draw opacity=0.5] (a) at (-0.2,-1){$j$};
            \node[ line width=0.6pt, dashed, draw opacity=0.5] (a) at (-0.3,-1.4){$\rho$};
            \node[ line width=0.6pt, dashed, draw opacity=0.5] (a) at (0.2,-0.8){$\sigma$};
            \node[ line width=0.6pt, dashed, draw opacity=0.5] (a) at (0,-0.3){$i$};
            \node[ line width=0.6pt, dashed, draw opacity=0.5] (a) at (0,0.3){$f$};
            \node[ line width=0.6pt, dashed, draw opacity=0.5] (a) at (-0.2,1){$g$};
            \node[ line width=0.6pt, dashed, draw opacity=0.5] (a) at (-0.4,1.3){$\gamma$};
            \node[ line width=0.6pt, dashed, draw opacity=0.5] (a) at (0.2,0.8){$\zeta$};
        \end{tikzpicture}
    \end{aligned} 
    \otimes 
    \begin{aligned}\begin{tikzpicture}
             \draw[line width=.6pt,black] (0,0.5)--(0,1.5);
             \draw[line width=.6pt,black] (0,-0.5)--(0,-1.5);
             \draw[red] (0,0.8) arc[start angle=90, end angle=270, radius=0.8];
             \draw[blue] (0,1.3) arc[start angle=90, end angle=-90, radius=1.3];
            \node[ line width=0.6pt, dashed, draw opacity=0.5] (a) at (0,1.7){$k$};
             \node[ line width=0.6pt, dashed, draw opacity=0.5] (a) at (0,-1.7){$c$};
            \node[ line width=0.6pt, dashed, draw opacity=0.5] (a) at (-1,0){$a$};
            \node[ line width=0.6pt, dashed, draw opacity=0.5] (a) at (1.5,0){$b$};
            \node[ line width=0.6pt, dashed, draw opacity=0.5] (a) at (-0.2,-1){$d$};
            \node[ line width=0.6pt, dashed, draw opacity=0.5] (a) at (-0.4,-1.3){$\mu$};
            \node[ line width=0.6pt, dashed, draw opacity=0.5] (a) at (0.2,-0.8){$\nu$};
            \node[ line width=0.6pt, dashed, draw opacity=0.5] (a) at (0,-0.3){$e$};
            \node[ line width=0.6pt, dashed, draw opacity=0.5] (a) at (0,0.3){$i$};
            \node[ line width=0.6pt, dashed, draw opacity=0.5] (a) at (-0.2,1){$j$};
            \node[ line width=0.6pt, dashed, draw opacity=0.5] (a) at (-0.4,1.3){$\rho$};
            \node[ line width=0.6pt, dashed, draw opacity=0.5] (a) at (0.2,0.8){$\sigma$};
        \end{tikzpicture}
\end{aligned}\;.\label{eq:TubeComultiplication}
    \end{equation}
In the diagrammatic representation, if we omit the summation over $i$ and the coefficient factors, this means that we insert an $i$ domain wall edge between $e$ and $f$ and then perform the parallel move for $a$ and $i$, then for $b$ and $j$.
    
\end{itemize}
\end{proposition}

\begin{proof}
The coassociativity and counit conditions are easy to verify from the definition. For example, applying $(\id\otimes\varepsilon)$ to Eq.~\eqref{eq:TubeComultiplication} yields 
\begin{align*}
    \sum_{i,j,k,\rho,\sigma} \sqrt{ \frac{d_j}{d_ad_i}}\sqrt{ \frac{d_k}{d_jd_b}}
    \begin{aligned}\begin{tikzpicture}
             \draw[line width=.6pt,black] (0,0.5)--(0,1.5);
             \draw[line width=.6pt,black] (0,-0.5)--(0,-1.5);
             \draw[red] (0,0.8) arc[start angle=90, end angle=270, radius=0.8];
             \draw[blue] (0,1.3) arc[start angle=90, end angle=-90, radius=1.3];
            \node[ line width=0.6pt, dashed, draw opacity=0.5] (a) at (0,1.7){$h$};
             \node[ line width=0.6pt, dashed, draw opacity=0.5] (a) at (0,-1.7){$k$};
            \node[ line width=0.6pt, dashed, draw opacity=0.5] (a) at (-1,0){$a$};
            \node[ line width=0.6pt, dashed, draw opacity=0.5] (a) at (1.5,0){$b$};
            \node[ line width=0.6pt, dashed, draw opacity=0.5] (a) at (-0.2,-1){$j$};
            \node[ line width=0.6pt, dashed, draw opacity=0.5] (a) at (-0.3,-1.4){$\rho$};
            \node[ line width=0.6pt, dashed, draw opacity=0.5] (a) at (0.2,-0.8){$\sigma$};
            \node[ line width=0.6pt, dashed, draw opacity=0.5] (a) at (0,-0.3){$i$};
            \node[ line width=0.6pt, dashed, draw opacity=0.5] (a) at (0,0.3){$f$};
            \node[ line width=0.6pt, dashed, draw opacity=0.5] (a) at (-0.2,1){$g$};
            \node[ line width=0.6pt, dashed, draw opacity=0.5] (a) at (-0.4,1.3){$\gamma$};
            \node[ line width=0.6pt, dashed, draw opacity=0.5] (a) at (0.2,0.8){$\zeta$};
        \end{tikzpicture}
    \end{aligned} \delta_{c,k}\delta_{e,i}\delta_{d,j}\delta_{\nu,\sigma}\delta_{\mu,\rho}\sqrt{\frac{d_ad_ed_b}{d_k}}\;  = \; \begin{aligned}\begin{tikzpicture}
             \draw[line width=.6pt,black] (0,0.5)--(0,1.5);
             \draw[line width=.6pt,black] (0,-0.5)--(0,-1.5);
             \draw[red] (0,0.8) arc[start angle=90, end angle=270, radius=0.8];
             \draw[blue] (0,1.3) arc[start angle=90, end angle=-90, radius=1.3];
            \node[ line width=0.6pt, dashed, draw opacity=0.5] (a) at (0,1.7){$h$};
             \node[ line width=0.6pt, dashed, draw opacity=0.5] (a) at (0,-1.7){$c$};
            \node[ line width=0.6pt, dashed, draw opacity=0.5] (a) at (-1,0){$a$};
            \node[ line width=0.6pt, dashed, draw opacity=0.5] (a) at (1.5,0){$b$};
            \node[ line width=0.6pt, dashed, draw opacity=0.5] (a) at (-0.2,-1){$d$};
            \node[ line width=0.6pt, dashed, draw opacity=0.5] (a) at (-0.4,-1.3){$\mu$};
            \node[ line width=0.6pt, dashed, draw opacity=0.5] (a) at (0.2,-0.8){$\nu$};
            \node[ line width=0.6pt, dashed, draw opacity=0.5] (a) at (0,-0.3){$e$};
            \node[ line width=0.6pt, dashed, draw opacity=0.5] (a) at (0,0.3){$f$};
            \node[ line width=0.6pt, dashed, draw opacity=0.5] (a) at (-0.2,1){$g$};
            \node[ line width=0.6pt, dashed, draw opacity=0.5] (a) at (-0.4,1.3){$\gamma$};
            \node[ line width=0.6pt, dashed, draw opacity=0.5] (a) at (0.2,0.8){$\zeta$};
        \end{tikzpicture}
    \end{aligned}\;. 
\end{align*}
One proves the other conditions similarly. 
\end{proof}

\begin{remark} \label{rmk:take-counit}
A direct observation is that when taking $\varepsilon$ for the product in equality~\eqref{eq:tube-mult}, it does not matter to complete the second or third diagram (up to a factor $1/d_h$). More precisely, to compute $\varepsilon(XY)$ for $X=[a;b]_{c,d,e;\mu,\nu}^{f,g,h;\zeta,\gamma}$ and $Y = [a';b']_{c',d',e';\mu',\nu'}^{f',g',h';\zeta',\gamma'}$, we can first stack the diagrams, and then complete the resulting diagram in the following sense:
\begin{align}
    \delta_{f,h'}\delta_{e,c'}\delta_{e',f'}\delta_{c,h} \frac{1}{d_h}\;
    \begin{aligned}
        \begin{tikzpicture}
            \draw[black] (0,2.8) arc[start angle=90, end angle=-90, radius=2.8];
             \draw[line width=.6pt,black] (0,0)--(0,2.8);
             \draw[line width=.6pt,black] (0,0)--(0,-2.8);
             \draw[red] (0,0.8) arc[start angle=90, end angle=270, radius=0.8];
             \draw[blue] (0,1.3) arc[start angle=90, end angle=-90, radius=1.3];
            \node[ line width=0.6pt, dashed, draw opacity=0.5] (a) at (0.2,1.5){$f$};
             \node[ line width=0.6pt, dashed, draw opacity=0.5] (a) at (0.2,-1.5){$e$};
            \node[ line width=0.6pt, dashed, draw opacity=0.5] (a) at (-1,0){$a'$};
            \node[ line width=0.6pt, dashed, draw opacity=0.5] (a) at (1.5,0){$b'$};
            \node[ line width=0.6pt, dashed, draw opacity=0.5] (a) at (-0.2,-1){$d'$};
            \node[ line width=0.6pt, dashed, draw opacity=0.5] (a) at (-0.2,-1.3){$\mu'$};
            \node[ line width=0.6pt, dashed, draw opacity=0.5] (a) at (0.2,-0.8){$\nu'$};
            \node[ line width=0.6pt, dashed, draw opacity=0.5] (a) at (0.3,0){$e'$};
            \node[ line width=0.6pt, dashed, draw opacity=0.5] (a) at (-0.2,1){$g'$};
            \node[ line width=0.6pt, dashed, draw opacity=0.5] (a) at (-0.4,1.3){$\gamma'$};
            \node[ line width=0.6pt, dashed, draw opacity=0.5] (a) at (0.2,0.8){$\zeta'$};
            \draw[red] (0,1.6) arc[start angle=90, end angle=270, radius=1.6];
            \node[ line width=0.6pt, dashed, draw opacity=0.5] (a) at (-0.3,1.8){$\zeta$};
          \node[ line width=0.6pt, dashed, draw opacity=0.5] (a) at (-0.3,-1.8){$\nu$};
           \draw[blue] (0,2.1) arc[start angle=90, end angle=-90, radius=2.1];
            \node[ line width=0.6pt, dashed, draw opacity=0.5] (a) at (-1.3,0.4){$a$};
                \node[ line width=0.6pt, dashed, draw opacity=0.5] (a) at (1.8,0.4){$b$};
         \node[ line width=0.6pt, dashed, draw opacity=0.5] (a) at (0.2,1.8){$g$};
          \node[ line width=0.6pt, dashed, draw opacity=0.5] (a) at (0.2,-1.8){$d$};
        \node[ line width=0.6pt, dashed, draw opacity=0.5] (a) at (0.2,2.3){$\gamma$};
        \node[ line width=0.6pt, dashed, draw opacity=0.5] (a) at (0.2,-2.3){$\mu$};
       \node[ line width=0.6pt, dashed, draw opacity=0.5] (a) at (2.6,0){$h$};
        \end{tikzpicture}
    \end{aligned}\;.\label{eq:eps_prod}
\end{align}
Indeed, by definition, the counit of $XY$ is computed by completing 
the third diagram in \eqref{eq:tube-mult}: 
    \begin{align*}
         &\sum  [(F_{ag'b'}^g)^{-1}]_{f \gamma' \zeta}^{i\rho \sigma}[(F^{ad'b'}_d)^{-1}]_{e \mu' \nu }^{j \omega \tau} [(F_{aa'f'}^i)^{-1}]_{g' \zeta' \rho}^{k\alpha \beta} [(F^{aa'e'}_j)^{-1}]_{d' \omega \nu'}^{k \alpha\beta'}\\
        &\quad \quad [F_{ib'b}^h]_{g\sigma\gamma}^{s \xi \lambda}[F_c^{jb'b}]_{d\tau\mu}^{s\xi\lambda'}\delta_{f,h'}\delta_{e,c'}\delta_{e',f'}\delta_{c,h}\delta_{i,j}\delta_{\beta,\beta'}\delta_{\lambda,\lambda'}\sqrt{\frac{d_ad_{a'}d_{e'}d_{b'}d_b}{d_h}}  \\
         = & \sum  [(F_{ag'b'}^d)^{-1}]_{f \gamma' \zeta}^{i\rho \sigma}[(F^{ad'b'}_d)^{-1}]_{e  \mu' \nu }^{i \omega \tau} [(F_{aa'e'}^i)^{-1}]_{g' \zeta' \rho}^{k\alpha \beta} [(F^{aa'e'}_i)^{-1}]_{d' \omega \nu'}^{k \alpha \beta}\\
        &\quad \quad [F_{ib'b}^c]_{g\sigma\gamma}^{s \xi \lambda}[F_c^{ib'b}]_{d\tau\mu}^{s\xi \lambda}\delta_{f,h'}\delta_{e,c'}\delta_{e',f'}\delta_{c,h}\delta_{d,g}\sqrt{\frac{d_ad_{a'}d_{e'}d_{b'}d_b}{d_h}}  \\
         = & \sum [(F_{ag'b'}^d)^{-1}]_{f \gamma' \zeta}^{i\rho \sigma}[(F^{ad'b'}_d)^{-1}]_{e  \mu' \nu}^{i \omega \tau} \delta_{d',g'}\delta_{\omega,\zeta'}\delta_{\nu',\rho}\\
        &\quad \quad \delta_{\sigma,\tau}\delta_{\mu,\gamma}\delta_{f,h'}\delta_{e,c'}\delta_{e',f'}\delta_{c,h}\delta_{d,g}\sqrt{\frac{d_ad_{a'}d_{e'}d_{b'}d_b}{d_h}}  \\
        = & \sum [(F_{ad'b'}^d)^{-1}]_{f \gamma' \zeta}^{i\nu' \sigma}[(F^{ad'b'}_d)^{-1}]_{e  \mu'\nu}^{i \zeta'\sigma} \delta_{d',g'}\\
        &\quad \quad \delta_{\mu,\gamma}\delta_{f,h'}\delta_{e,c'}\delta_{e',f'}\delta_{c,h}\delta_{d,g}\sqrt{\frac{d_ad_{a'}d_{e'}d_{b'}d_b}{d_h}}  \\
        = & \; \; \delta_{e,f}\delta_{\nu,\zeta}\delta_{\mu',\gamma'} \delta_{\nu',\zeta'}\delta_{d',g'} \delta_{\mu,\gamma}\delta_{f,h'}\delta_{e,c'}\delta_{e',f'}\delta_{c,h}\delta_{d,g}\sqrt{\frac{d_ad_{a'}d_{e'}d_{b'}d_b}{d_h}},
\end{align*}
which is equal to the value of \eqref{eq:eps_prod}. Here we have used Eq.~\eqref{eq:F-symbol-delta}. This fact can be generalized to compute the counit of product of arbitrary number of diagrams.  We will use this fact frequently in the sequel. 
\end{remark}

Notice that for the tube algebra, we have $\Delta(1) \neq 1 \otimes 1$; instead, we have $\Delta(1) = \sum_{k} 1^{(1)}_k \otimes 1^{(2)}_k$:
\begin{equation} \label{eq:coprod-unit}
   \Delta(1)= \sum_{x,y,z\in \Irr(\ED)}
   \begin{aligned}
        \begin{tikzpicture}
             \draw[line width=.6pt,black] (0,0.2)--(0,0.6);
             \draw[line width=.6pt,black] (0,-0.2)--(0,-0.6);
            \node[ line width=0.6pt, dashed, draw opacity=0.5] (a) at (-0.2,0.3){$x$};
            \node[ line width=0.6pt, dashed, draw opacity=0.5] (a) at (-0.2,-0.3){$z$};
        \end{tikzpicture}
    \end{aligned} \otimes  \begin{aligned}
        \begin{tikzpicture}
             \draw[line width=.6pt,black] (0,0.2)--(0,0.6);
             \draw[line width=.6pt,black] (0,-0.2)--(0,-0.6);
                \node[ line width=0.6pt, dashed, draw opacity=0.5] (a) at (-0.2,0.3){$z$};
            \node[ line width=0.6pt, dashed, draw opacity=0.5] (a) at (-0.2,-0.3){$y$};
        \end{tikzpicture}
    \end{aligned}\,\, .
\end{equation}
(As usual, we omit the edges labeled by the tensor unit $\mathbf{1}$.) This indicates that the tube algebra should have a weak bialgebra structure, and indeed it does.

\begin{proposition}
For any given UMFC $\ED$, the tube algebra $\mathbf{Tube}({_{\ED}\ED_{\ED}})$ constitutes a weak bialgebra with the previously given structure morphisms $\mu,1,\Delta,\varepsilon$.
\end{proposition}

\begin{proof}
    Consider $X=[a;b]_{c,d,e;\mu,\nu}^{f,g,h;\zeta,\gamma}, Y = [a';b']_{c',d',e';\mu',\nu'}^{f',g',h';\zeta',\gamma'}$ and $Z=[a'';b'']_{c'',d'',e'';\mu'',\nu''}^{f'',g'',h'';\zeta'',\gamma''}$.
    
    1. We verify that $\Delta$ is multiplicative: $\Delta(XY)=\Delta(X)\Delta(Y)$. By Proposition~\ref{prop:tube-coalgebra}, to compute the comultiplication, one can insert a wall labeled by $i\in\Irr(\ED)$ and take the summation over all such labels. Thus $\Delta(XY)$ is equal to 
    \begin{align*}
    &\sum_{i}\sum_{j,\rho}\sum_{k,\sigma}\sum_{s,\tau}\sum_{t,\theta} \sqrt{\frac{d_j}{d_{a'}d_i}}\sqrt{\frac{d_k}{d_{b'}d_j}}\sqrt{\frac{d_s}{d_{a}d_k}}\sqrt{\frac{d_t}{d_{b}d_s}} \delta_{f,h'}\delta_{e,c'}\\
    & \quad \quad \quad\quad \quad \quad\quad \quad \quad \begin{aligned}
        \begin{tikzpicture}
             \draw[line width=.6pt,black] (0,0.5)--(0,2.5);
             \draw[line width=.6pt,black] (0,-0.5)--(0,-2.5);
             \draw[red] (0,0.8) arc[start angle=90, end angle=270, radius=0.8];
             \draw[blue] (0,1.3) arc[start angle=90, end angle=-90, radius=1.3];
            \node[ line width=0.6pt, dashed, draw opacity=0.5] (a) at (0.2,1.5){$f$};
             \node[ line width=0.6pt, dashed, draw opacity=0.5] (a) at (0.2,-1.5){$k$};
            \node[ line width=0.6pt, dashed, draw opacity=0.5] (a) at (-1,0){$a'$};
            \node[ line width=0.6pt, dashed, draw opacity=0.5] (a) at (1.5,0){$b'$};
            \node[ line width=0.6pt, dashed, draw opacity=0.5] (a) at (-0.2,-1){$j$};
            \node[ line width=0.6pt, dashed, draw opacity=0.5] (a) at (-0.2,-1.3){$\sigma$};
            \node[ line width=0.6pt, dashed, draw opacity=0.5] (a) at (0.2,-0.8){$\rho$};
            \node[ line width=0.6pt, dashed, draw opacity=0.5] (a) at (0,-0.3){$i$};
            \node[ line width=0.6pt, dashed, draw opacity=0.5] (a) at (0,0.3){$f'$};
            \node[ line width=0.6pt, dashed, draw opacity=0.5] (a) at (-0.2,1){$g'$};
            \node[ line width=0.6pt, dashed, draw opacity=0.5] (a) at (-0.4,1.3){$\gamma'$};
            \node[ line width=0.6pt, dashed, draw opacity=0.5] (a) at (0.2,0.8){$\zeta'$};
            \draw[red] (0,1.6) arc[start angle=90, end angle=270, radius=1.6];
            \node[ line width=0.6pt, dashed, draw opacity=0.5] (a) at (-0.3,1.8){$\zeta$};
          \node[ line width=0.6pt, dashed, draw opacity=0.5] (a) at (-0.3,-1.8){$\tau$};
           \draw[blue] (0,2.1) arc[start angle=90, end angle=-90, radius=2.1];
            \node[ line width=0.6pt, dashed, draw opacity=0.5] (a) at (-1.3,0.4){$a$};
                \node[ line width=0.6pt, dashed, draw opacity=0.5] (a) at (1.8,0.4){$b$};
         \node[ line width=0.6pt, dashed, draw opacity=0.5] (a) at (0.2,1.8){$g$};
          \node[ line width=0.6pt, dashed, draw opacity=0.5] (a) at (0.2,-1.8){$s$};
        \node[ line width=0.6pt, dashed, draw opacity=0.5] (a) at (0.2,2.3){$\gamma$};
        \node[ line width=0.6pt, dashed, draw opacity=0.5] (a) at (0.2,-2.3){$\theta$};
       \node[ line width=0.6pt, dashed, draw opacity=0.5] (a) at (-0.3,2.5){$h$};
       \node[ line width=0.6pt, dashed, draw opacity=0.5] (a) at (-0.3,-2.5){$t$};
        \end{tikzpicture}
    \end{aligned} \otimes 
    \begin{aligned}
        \begin{tikzpicture}
             \draw[line width=.6pt,black] (0,0.5)--(0,2.5);
             \draw[line width=.6pt,black] (0,-0.5)--(0,-2.5);
             \draw[red] (0,0.8) arc[start angle=90, end angle=270, radius=0.8];
             \draw[blue] (0,1.3) arc[start angle=90, end angle=-90, radius=1.3];
            \node[ line width=0.6pt, dashed, draw opacity=0.5] (a) at (0.2,1.5){$k$};
             \node[ line width=0.6pt, dashed, draw opacity=0.5] (a) at (0.2,-1.5){$e$};
            \node[ line width=0.6pt, dashed, draw opacity=0.5] (a) at (-1,0){$a'$};
            \node[ line width=0.6pt, dashed, draw opacity=0.5] (a) at (1.5,0){$b'$};
            \node[ line width=0.6pt, dashed, draw opacity=0.5] (a) at (-0.2,-1){$d'$};
            \node[ line width=0.6pt, dashed, draw opacity=0.5] (a) at (-0.2,-1.3){$\mu'$};
            \node[ line width=0.6pt, dashed, draw opacity=0.5] (a) at (0.2,-0.8){$\nu'$};
            \node[ line width=0.6pt, dashed, draw opacity=0.5] (a) at (0,-0.3){$e'$};
            \node[ line width=0.6pt, dashed, draw opacity=0.5] (a) at (0,0.3){$i$};
            \node[ line width=0.6pt, dashed, draw opacity=0.5] (a) at (-0.2,1){$j$};
            \node[ line width=0.6pt, dashed, draw opacity=0.5] (a) at (-0.4,1.3){$\sigma$};
            \node[ line width=0.6pt, dashed, draw opacity=0.5] (a) at (0.2,0.8){$\rho$};
            \draw[red] (0,1.6) arc[start angle=90, end angle=270, radius=1.6];
            \node[ line width=0.6pt, dashed, draw opacity=0.5] (a) at (-0.3,1.8){$\tau$};
          \node[ line width=0.6pt, dashed, draw opacity=0.5] (a) at (-0.3,-1.8){$\nu$};
           \draw[blue] (0,2.1) arc[start angle=90, end angle=-90, radius=2.1];
            \node[ line width=0.6pt, dashed, draw opacity=0.5] (a) at (-1.3,0.4){$a$};
                \node[ line width=0.6pt, dashed, draw opacity=0.5] (a) at (1.8,0.4){$b$};
         \node[ line width=0.6pt, dashed, draw opacity=0.5] (a) at (0.2,1.8){$s$};
          \node[ line width=0.6pt, dashed, draw opacity=0.5] (a) at (0.2,-1.8){$d$};
        \node[ line width=0.6pt, dashed, draw opacity=0.5] (a) at (0.2,2.3){$\theta$};
        \node[ line width=0.6pt, dashed, draw opacity=0.5] (a) at (0.2,-2.3){$\mu$};
       \node[ line width=0.6pt, dashed, draw opacity=0.5] (a) at (-0.3,2.5){$t$};
       \node[ line width=0.6pt, dashed, draw opacity=0.5] (a) at (-0.3,-2.5){$c$};
        \end{tikzpicture}
    \end{aligned} \\
     = & \sum_{k,s,t,\tau,\theta}\sqrt{\frac{d_s}{d_{a}d_k}}\sqrt{\frac{d_t}{d_{b}d_s}}\sum_{i,j,k',\rho,\sigma} \sqrt{\frac{d_j}{d_{a'}d_i}}\sqrt{\frac{d_{k'}}{d_{b'}d_j}}  \\
    & \left( \begin{aligned}
        \begin{tikzpicture}
             \draw[line width=.6pt,black] (0,0.5)--(0,1.5);
             \draw[line width=.6pt,black] (0,-0.5)--(0,-1.5);
             \draw[red] (0,0.8) arc[start angle=90, end angle=270, radius=0.8];
             \draw[blue] (0,1.3) arc[start angle=90, end angle=-90, radius=1.3];
            \node[ line width=0.6pt, dashed, draw opacity=0.5] (a) at (0,1.7){$h$};
             \node[ line width=0.6pt, dashed, draw opacity=0.5] (a) at (0,-1.7){$t$};
            \node[ line width=0.6pt, dashed, draw opacity=0.5] (a) at (-1,0){$a$};
            \node[ line width=0.6pt, dashed, draw opacity=0.5] (a) at (1.5,0){$b$};
            \node[ line width=0.6pt, dashed, draw opacity=0.5] (a) at (-0.2,-1){$s$};
            \node[ line width=0.6pt, dashed, draw opacity=0.5] (a) at (-0.4,-1.3){$\theta$};
            \node[ line width=0.6pt, dashed, draw opacity=0.5] (a) at (0.2,-0.8){$\tau$};
            \node[ line width=0.6pt, dashed, draw opacity=0.5] (a) at (0,-0.3){$k$};
            \node[ line width=0.6pt, dashed, draw opacity=0.5] (a) at (0,0.3){$f$};
            \node[ line width=0.6pt, dashed, draw opacity=0.5] (a) at (-0.2,1){$g$};
            \node[ line width=0.6pt, dashed, draw opacity=0.5] (a) at (-0.4,1.3){$\gamma$};
            \node[ line width=0.6pt, dashed, draw opacity=0.5] (a) at (0.2,0.8){$\zeta$};
        \end{tikzpicture}
    \end{aligned} \begin{aligned}
        \begin{tikzpicture}
             \draw[line width=.6pt,black] (0,0.5)--(0,1.5);
             \draw[line width=.6pt,black] (0,-0.5)--(0,-1.5);
             \draw[red] (0,0.8) arc[start angle=90, end angle=270, radius=0.8];
             \draw[blue] (0,1.3) arc[start angle=90, end angle=-90, radius=1.3];
            \node[ line width=0.6pt, dashed, draw opacity=0.5] (a) at (0,1.7){$h'$};
             \node[ line width=0.6pt, dashed, draw opacity=0.5] (a) at (0,-1.7){$k'$};
            \node[ line width=0.6pt, dashed, draw opacity=0.5] (a) at (-1,0){$a'$};
            \node[ line width=0.6pt, dashed, draw opacity=0.5] (a) at (1.5,0){$b'$};
            \node[ line width=0.6pt, dashed, draw opacity=0.5] (a) at (-0.2,-1){$j$};
            \node[ line width=0.6pt, dashed, draw opacity=0.5] (a) at (-0.4,-1.4){$\sigma$};
            \node[ line width=0.6pt, dashed, draw opacity=0.5] (a) at (0.2,-0.8){$\rho$};
            \node[ line width=0.6pt, dashed, draw opacity=0.5] (a) at (0,-0.3){$i$};
            \node[ line width=0.6pt, dashed, draw opacity=0.5] (a) at (0,0.3){$f'$};
            \node[ line width=0.6pt, dashed, draw opacity=0.5] (a) at (-0.2,1){$g'$};
            \node[ line width=0.6pt, dashed, draw opacity=0.5] (a) at (-0.4,1.3){$\gamma'$};
            \node[ line width=0.6pt, dashed, draw opacity=0.5] (a) at (0.2,0.8){$\zeta'$};
        \end{tikzpicture}
    \end{aligned}\right) 
    \otimes 
    \left( \begin{aligned}
        \begin{tikzpicture}
             \draw[line width=.6pt,black] (0,0.5)--(0,1.5);
             \draw[line width=.6pt,black] (0,-0.5)--(0,-1.5);
             \draw[red] (0,0.8) arc[start angle=90, end angle=270, radius=0.8];
             \draw[blue] (0,1.3) arc[start angle=90, end angle=-90, radius=1.3];
            \node[ line width=0.6pt, dashed, draw opacity=0.5] (a) at (0,1.7){$t$};
             \node[ line width=0.6pt, dashed, draw opacity=0.5] (a) at (0,-1.7){$c$};
            \node[ line width=0.6pt, dashed, draw opacity=0.5] (a) at (-1,0){$a$};
            \node[ line width=0.6pt, dashed, draw opacity=0.5] (a) at (1.5,0){$b$};
            \node[ line width=0.6pt, dashed, draw opacity=0.5] (a) at (-0.2,-1){$d$};
            \node[ line width=0.6pt, dashed, draw opacity=0.5] (a) at (-0.4,-1.3){$\mu$};
            \node[ line width=0.6pt, dashed, draw opacity=0.5] (a) at (0.2,-0.8){$\nu$};
            \node[ line width=0.6pt, dashed, draw opacity=0.5] (a) at (0,-0.3){$e$};
            \node[ line width=0.6pt, dashed, draw opacity=0.5] (a) at (0,0.3){$k$};
            \node[ line width=0.6pt, dashed, draw opacity=0.5] (a) at (-0.2,1){$s$};
            \node[ line width=0.6pt, dashed, draw opacity=0.5] (a) at (-0.4,1.3){$\theta$};
            \node[ line width=0.6pt, dashed, draw opacity=0.5] (a) at (0.2,0.8){$\tau$};
        \end{tikzpicture}
    \end{aligned} \begin{aligned}
        \begin{tikzpicture}
             \draw[line width=.6pt,black] (0,0.5)--(0,1.5);
             \draw[line width=.6pt,black] (0,-0.5)--(0,-1.5);
             \draw[red] (0,0.8) arc[start angle=90, end angle=270, radius=0.8];
             \draw[blue] (0,1.3) arc[start angle=90, end angle=-90, radius=1.3];
            \node[ line width=0.6pt, dashed, draw opacity=0.5] (a) at (0,1.7){$k'$};
             \node[ line width=0.6pt, dashed, draw opacity=0.5] (a) at (0,-1.7){$c'$};
            \node[ line width=0.6pt, dashed, draw opacity=0.5] (a) at (-1,0){$a'$};
            \node[ line width=0.6pt, dashed, draw opacity=0.5] (a) at (1.5,0){$b'$};
            \node[ line width=0.6pt, dashed, draw opacity=0.5] (a) at (-0.2,-1){$d'$};
            \node[ line width=0.6pt, dashed, draw opacity=0.5] (a) at (-0.4,-1.4){$\mu'$};
            \node[ line width=0.6pt, dashed, draw opacity=0.5] (a) at (0.2,-0.8){$\nu'$};
            \node[ line width=0.6pt, dashed, draw opacity=0.5] (a) at (0,-0.3){$e'$};
            \node[ line width=0.6pt, dashed, draw opacity=0.5] (a) at (0,0.3){$i$};
            \node[ line width=0.6pt, dashed, draw opacity=0.5] (a) at (-0.2,1){$j$};
            \node[ line width=0.6pt, dashed, draw opacity=0.5] (a) at (-0.4,1.3){$\sigma$};
            \node[ line width=0.6pt, dashed, draw opacity=0.5] (a) at (0.2,0.8){$\rho$};
        \end{tikzpicture}
    \end{aligned}\right),
    \end{align*}
    which is $\Delta(X)\Delta(Y)$. 
    
    2. The first condition of the weak comultiplicativity of the unit is verified as follows: 
    \begin{align*}
        (\Delta(1)\otimes 1)\cdot &(1\otimes \Delta(1)) \\
        &= \left(\sum_{x,y,k,z,w} \begin{aligned}
        \begin{tikzpicture}
             \draw[line width=.6pt,black] (0,0.2)--(0,0.6);
             \draw[line width=.6pt,black] (0,-0.2)--(0,-0.6);
            \node[ line width=0.6pt, dashed, draw opacity=0.5] (a) at (-0.2,0.3){$x$};
            \node[ line width=0.6pt, dashed, draw opacity=0.5] (a) at (-0.2,-0.3){$k$};
        \end{tikzpicture}
    \end{aligned} \otimes  \begin{aligned}
        \begin{tikzpicture}
             \draw[line width=.6pt,black] (0,0.2)--(0,0.6);
             \draw[line width=.6pt,black] (0,-0.2)--(0,-0.6);
                \node[ line width=0.6pt, dashed, draw opacity=0.5] (a) at (-0.2,0.3){$k$};
            \node[ line width=0.6pt, dashed, draw opacity=0.5] (a) at (-0.2,-0.3){$y$};
        \end{tikzpicture}
    \end{aligned} \otimes  \begin{aligned}
        \begin{tikzpicture}
             \draw[line width=.6pt,black] (0,0.2)--(0,0.6);
             \draw[line width=.6pt,black] (0,-0.2)--(0,-0.6);
                \node[ line width=0.6pt, dashed, draw opacity=0.5] (a) at (-0.2,0.3){$z$};
            \node[ line width=0.6pt, dashed, draw opacity=0.5] (a) at (-0.2,-0.3){$w$};
        \end{tikzpicture}
    \end{aligned} \right)\cdot 
    \left(\sum_{x',y',k',z',w'} \begin{aligned}
        \begin{tikzpicture}
             \draw[line width=.6pt,black] (0,0.2)--(0,0.6);
             \draw[line width=.6pt,black] (0,-0.2)--(0,-0.6);
            \node[ line width=0.6pt, dashed, draw opacity=0.5] (a) at (-0.2,0.3){$x'$};
            \node[ line width=0.6pt, dashed, draw opacity=0.5] (a) at (-0.2,-0.3){$y'$};
        \end{tikzpicture}
    \end{aligned} \otimes  \begin{aligned}
        \begin{tikzpicture}
             \draw[line width=.6pt,black] (0,0.2)--(0,0.6);
             \draw[line width=.6pt,black] (0,-0.2)--(0,-0.6);
                \node[ line width=0.6pt, dashed, draw opacity=0.5] (a) at (-0.2,0.3){$z'$};
            \node[ line width=0.6pt, dashed, draw opacity=0.5] (a) at (-0.2,-0.3){$k'$};
        \end{tikzpicture}
    \end{aligned} \otimes  \begin{aligned}
        \begin{tikzpicture}
             \draw[line width=.6pt,black] (0,0.2)--(0,0.6);
             \draw[line width=.6pt,black] (0,-0.2)--(0,-0.6);
                \node[ line width=0.6pt, dashed, draw opacity=0.5] (a) at (-0.2,0.3){$k'$};
            \node[ line width=0.6pt, dashed, draw opacity=0.5] (a) at (-0.2,-0.3){$w'$};
        \end{tikzpicture}
    \end{aligned} \right) \\ 
    &= \sum_{x,y,k,z,w}\sum_{x',y',k',z',w'}\delta_{x,x'}\delta_{k,y'}\delta_{k,z'}\delta_{y,k'}\delta_{z,k'}\delta_{w,w'} \begin{aligned}
        \begin{tikzpicture}
             \draw[line width=.6pt,black] (0,0.2)--(0,0.6);
             \draw[line width=.6pt,black] (0,-0.2)--(0,-0.6);
            \node[ line width=0.6pt, dashed, draw opacity=0.5] (a) at (-0.2,0.3){$x$};
            \node[ line width=0.6pt, dashed, draw opacity=0.5] (a) at (-0.2,-0.3){$k$};
        \end{tikzpicture}
    \end{aligned} \otimes  \begin{aligned}
        \begin{tikzpicture}
             \draw[line width=.6pt,black] (0,0.2)--(0,0.6);
             \draw[line width=.6pt,black] (0,-0.2)--(0,-0.6);
                \node[ line width=0.6pt, dashed, draw opacity=0.5] (a) at (-0.2,0.3){$k$};
            \node[ line width=0.6pt, dashed, draw opacity=0.5] (a) at (-0.2,-0.3){$y$};
        \end{tikzpicture}
    \end{aligned} \otimes  \begin{aligned}
        \begin{tikzpicture}
             \draw[line width=.6pt,black] (0,0.2)--(0,0.6);
             \draw[line width=.6pt,black] (0,-0.2)--(0,-0.6);
                \node[ line width=0.6pt, dashed, draw opacity=0.5] (a) at (-0.2,0.3){$z$};
            \node[ line width=0.6pt, dashed, draw opacity=0.5] (a) at (-0.2,-0.3){$w$};
        \end{tikzpicture}
    \end{aligned}  \\
    & = \sum_{x,w,k,z}
        \begin{aligned}
        \begin{tikzpicture}
             \draw[line width=.6pt,black] (0,0.2)--(0,0.6);
             \draw[line width=.6pt,black] (0,-0.2)--(0,-0.6);
            \node[ line width=0.6pt, dashed, draw opacity=0.5] (a) at (-0.2,0.3){$x$};
            \node[ line width=0.6pt, dashed, draw opacity=0.5] (a) at (-0.2,-0.3){$k$};
        \end{tikzpicture}
    \end{aligned} \otimes  \begin{aligned}
        \begin{tikzpicture}
             \draw[line width=.6pt,black] (0,0.2)--(0,0.6);
             \draw[line width=.6pt,black] (0,-0.2)--(0,-0.6);
                \node[ line width=0.6pt, dashed, draw opacity=0.5] (a) at (-0.2,0.3){$k$};
            \node[ line width=0.6pt, dashed, draw opacity=0.5] (a) at (-0.2,-0.3){$z$};
        \end{tikzpicture}
    \end{aligned} \otimes  \begin{aligned}
        \begin{tikzpicture}
             \draw[line width=.6pt,black] (0,0.2)--(0,0.6);
             \draw[line width=.6pt,black] (0,-0.2)--(0,-0.6);
                \node[ line width=0.6pt, dashed, draw opacity=0.5] (a) at (-0.2,0.3){$z$};
            \node[ line width=0.6pt, dashed, draw opacity=0.5] (a) at (-0.2,-0.3){$w$};
        \end{tikzpicture}
    \end{aligned} \\
    & = (\Delta\otimes \id)\comp\Delta(1).  
    \end{align*} 
    Similarly one can verify the other condition $(1\otimes \Delta(1))\cdot (\Delta(1)\otimes 1)= (\Delta\otimes \id)\comp\Delta(1)$. 

    3. It remains to verify that the counit is weak multiplicative, namely, the following identities hold:  $$\varepsilon(XYZ)=\sum_{(Y)}\varepsilon(XY^{(1)})\varepsilon(Y^{(2)}Z)=\sum_{(Y)}\varepsilon(XY^{(2)})\varepsilon(Y^{(1)}Z).$$
    Indeed, according to Remark~\ref{rmk:take-counit}, one computes  
    \begin{align*}
        \varepsilon(XYZ)&=\delta_{f,h'}\delta_{e,c'}\delta_{f',h''}\delta_{e',c''}\delta_{c,h}\delta_{e'',f''}\frac{1}{d_h}\begin{aligned}
        \begin{tikzpicture}
             \draw[line width=.6pt,black] (0,0)--(0,3.6);
             \draw[line width=.6pt,black] (0,0)--(0,-3.6);
             \draw[red] (0,0.8) arc[start angle=90, end angle=270, radius=0.8];
             \draw[blue] (0,1.3) arc[start angle=90, end angle=-90, radius=1.3];
            \node[ line width=0.6pt, dashed, draw opacity=0.5] (a) at (-0.2,2.2){$f$};
             \node[ line width=0.6pt, dashed, draw opacity=0.5] (a) at (-0.2,-2.2){$e$};
             \node[ line width=0.6pt, dashed, draw opacity=0.5] (a) at (3.1,0){$b$};
             \node[ line width=0.6pt, dashed, draw opacity=0.5] (a) at (2.3,0){$b'$};
            \node[ line width=0.6pt, dashed, draw opacity=0.5] (a) at (1.6,0){$b''$};
            \node[ line width=0.6pt, dashed, draw opacity=0.5] (a) at (0.25,-1.8){$d'$};
            \node[ line width=0.6pt, dashed, draw opacity=0.5] (a) at (0.25,-1){$d''$};
            \node[ line width=0.6pt, dashed, draw opacity=0.5] (a) at (-0.3,-1.9){$\mu'$};
            \node[ line width=0.6pt, dashed, draw opacity=0.5] (a) at (0.3,-0.6){$\nu''$};
            \node[ line width=0.6pt, dashed, draw opacity=0.5] (a) at (0.25,-1.5){$\nu'$};
            \node[ line width=0.6pt, dashed, draw opacity=0.5] (a) at (-0.3,-1.1){$\mu''$};
            \node[ line width=0.6pt, dashed, draw opacity=0.5] (a) at (-0.2,-1.35){$e'$};
            \node[ line width=0.6pt, dashed, draw opacity=0.5] (a) at (0.3,0){$e''$};
            \node[ line width=0.6pt, dashed, draw opacity=0.5] (a) at (-0.2,1.3){$f'$};
            \node[ line width=0.6pt, dashed, draw opacity=0.5] (a) at (0.2,1.85){$g'$};
            \node[ line width=0.6pt, dashed, draw opacity=0.5] (a) at (-0.2,1.9){$\gamma'$};
            \node[ line width=0.6pt, dashed, draw opacity=0.5] (a) at (-0.2,1){$\gamma''$};
            \node[ line width=0.6pt, dashed, draw opacity=0.5] (a) at (0.3,1){$g''$};
            \node[ line width=0.6pt, dashed, draw opacity=0.5] (a) at (0.3,0.6){$\zeta''$};
            \node[ line width=0.6pt, dashed, draw opacity=0.5] (a) at (0.3,1.55){$\zeta'$};
            \draw[red] (0,1.6) arc[start angle=90, end angle=270, radius=1.6];
            \node[ line width=0.6pt, dashed, draw opacity=0.5] (a) at (0.2,2.4){$\zeta$};
          \node[ line width=0.6pt, dashed, draw opacity=0.5] (a) at (0.2,-2.3){$\nu$};
           \draw[blue] (0,2.1) arc[start angle=90, end angle=-90, radius=2.1];
            \node[ line width=0.6pt, dashed, draw opacity=0.5] (a) at (-2.6,0){$a$};
            \node[ line width=0.6pt, dashed, draw opacity=0.5] (a) at (-1.8,0){$a'$};
            \node[ line width=0.6pt, dashed, draw opacity=0.5] (a) at (-1,0){$a''$};
         \node[ line width=0.6pt, dashed, draw opacity=0.5] (a) at (-0.2,2.6){$g$};
          \node[ line width=0.6pt, dashed, draw opacity=0.5] (a) at (0.2,-2.6){$d$};
        \node[ line width=0.6pt, dashed, draw opacity=0.5] (a) at (-0.2,3){$\gamma$};
        \node[ line width=0.6pt, dashed, draw opacity=0.5] (a) at (-0.2,-2.9){$\mu$};
       \node[ line width=0.6pt, dashed, draw opacity=0.5] (a) at (3.8,0){$h$};
        \draw[red] (0,2.4) arc[start angle=90, end angle=270, radius=2.4];
       \draw[blue] (0,2.9) arc[start angle=90, end angle=-90, radius=2.9];
       \draw[black] (0,3.6) arc[start angle=90, end angle=-90, radius=3.6];
        \end{tikzpicture}
        \end{aligned} \\
        & = \delta_{f,h'}\delta_{e,c'}\delta_{f',h''}\delta_{e',c''}\delta_{c,h}\delta_{e'',f''} \delta_{\nu,\zeta}\delta_{\mu,\gamma}\delta_{\nu',\zeta'}\delta_{\mu',\gamma'}\delta_{\nu'',\zeta''}\delta_{\mu'',\gamma''}\\
        &\quad\quad \delta_{d'',g''}\delta_{e',f'}\delta_{d',g'}\delta_{e,f}\delta_{d,g}\sqrt{\frac{d_ad_{a'}d_{a''}d_{e''}d_bd_{b'}d_{b''}}{d_h}}. 
    \end{align*}
    On the other hand, one computes 
    \begin{align*}
        \sum_{(Y)}\varepsilon(XY^{(1)})\varepsilon(Y^{(2)}Z) & = \sum_{i,j,k,\rho,\sigma}\sqrt{\frac{d_j}{d_{a'}d_i}}\sqrt{\frac{d_k}{d_jd_{b'}}} \delta_{f,h'}\delta_{e,k}\delta_{i,h''}\delta_{e',c''}\delta_{i,f'}\delta_{c,h}\delta_{e'',f''}\delta_{c',k}\frac{1}{d_h}\frac{1}{d_k}\\
        & \begin{aligned}
        \begin{tikzpicture}
             \draw[line width=.6pt,black] (0,0)--(0,2.9);
             \draw[line width=.6pt,black] (0,0)--(0,-2.9);
             \draw[red] (0,0.8) arc[start angle=90, end angle=270, radius=0.8];
             \draw[blue] (0,1.3) arc[start angle=90, end angle=-90, radius=1.3];
            \node[ line width=0.6pt, dashed, draw opacity=0.5] (a) at (0.2,1.5){$f$};
             \node[ line width=0.6pt, dashed, draw opacity=0.5] (a) at (0.2,-1.5){$e$};
            \node[ line width=0.6pt, dashed, draw opacity=0.5] (a) at (-1,0){$a'$};
            \node[ line width=0.6pt, dashed, draw opacity=0.5] (a) at (1.5,0){$b'$};
            \node[ line width=0.6pt, dashed, draw opacity=0.5] (a) at (-0.2,-1){$j$};
            \node[ line width=0.6pt, dashed, draw opacity=0.5] (a) at (-0.2,-1.4){$\rho$};
            \node[ line width=0.6pt, dashed, draw opacity=0.5] (a) at (0.2,-0.8){$\sigma$};
            \node[ line width=0.6pt, dashed, draw opacity=0.5] (a) at (0.15,0){$i$};
            \node[ line width=0.6pt, dashed, draw opacity=0.5] (a) at (-0.2,1){$g'$};
            \node[ line width=0.6pt, dashed, draw opacity=0.5] (a) at (-0.4,1.3){$\gamma'$};
            \node[ line width=0.6pt, dashed, draw opacity=0.5] (a) at (0.2,0.8){$\zeta'$};
            \draw[red] (0,1.6) arc[start angle=90, end angle=270, radius=1.6];
            \node[ line width=0.6pt, dashed, draw opacity=0.5] (a) at (-0.3,1.8){$\zeta$};
          \node[ line width=0.6pt, dashed, draw opacity=0.5] (a) at (-0.3,-1.8){$\nu$};
           \draw[blue] (0,2.1) arc[start angle=90, end angle=-90, radius=2.1];
           \draw[black] (0,2.9) arc[start angle=90, end angle=-90, radius=2.9];
            \node[ line width=0.6pt, dashed, draw opacity=0.5] (a) at (-1.8,0){$a$};
                \node[ line width=0.6pt, dashed, draw opacity=0.5] (a) at (2.3,0){$b$};
         \node[ line width=0.6pt, dashed, draw opacity=0.5] (a) at (0.2,1.8){$g$};
          \node[ line width=0.6pt, dashed, draw opacity=0.5] (a) at (0.2,-1.8){$d$};
        \node[ line width=0.6pt, dashed, draw opacity=0.5] (a) at (-0.2,2.3){$\gamma$};
        \node[ line width=0.6pt, dashed, draw opacity=0.5] (a) at (-0.2,-2.3){$\mu$};
       \node[ line width=0.6pt, dashed, draw opacity=0.5] (a) at (3.1,0){$h$};
        \end{tikzpicture}
    \end{aligned}
    \begin{aligned}
        \begin{tikzpicture}
             \draw[line width=.6pt,black] (0,0)--(0,2.9);
             \draw[line width=.6pt,black] (0,0)--(0,-2.9);
             \draw[red] (0,0.8) arc[start angle=90, end angle=270, radius=0.8];
             \draw[blue] (0,1.3) arc[start angle=90, end angle=-90, radius=1.3];
            \node[ line width=0.6pt, dashed, draw opacity=0.5] (a) at (0.2,1.5){$i$};
             \node[ line width=0.6pt, dashed, draw opacity=0.5] (a) at (0.2,-1.5){$e'$};
            \node[ line width=0.6pt, dashed, draw opacity=0.5] (a) at (-1,0){$a''$};
            \node[ line width=0.6pt, dashed, draw opacity=0.5] (a) at (1.5,0){$b''$};
            \node[ line width=0.6pt, dashed, draw opacity=0.5] (a) at (-0.2,-1){$d''$};
            \node[ line width=0.6pt, dashed, draw opacity=0.5] (a) at (-0.2,-1.3){$\mu''$};
            \node[ line width=0.6pt, dashed, draw opacity=0.5] (a) at (0.35,-0.8){$\nu''$};
            \node[ line width=0.6pt, dashed, draw opacity=0.5] (a) at (0.3,0){$e''$};
            \node[ line width=0.6pt, dashed, draw opacity=0.5] (a) at (-0.2,1){$g''$};
            \node[ line width=0.6pt, dashed, draw opacity=0.5] (a) at (-0.4,1.3){$\gamma''$};
            \node[ line width=0.6pt, dashed, draw opacity=0.5] (a) at (0.2,0.8){$\zeta''$};
            \draw[red] (0,1.6) arc[start angle=90, end angle=270, radius=1.6];
            \node[ line width=0.6pt, dashed, draw opacity=0.5] (a) at (-0.2,1.8){$\sigma$};
          \node[ line width=0.6pt, dashed, draw opacity=0.5] (a) at (-0.3,-1.8){$\nu'$};
           \draw[blue] (0,2.1) arc[start angle=90, end angle=-90, radius=2.1];
           \draw[black] (0,2.9) arc[start angle=90, end angle=-90, radius=2.9];
            \node[ line width=0.6pt, dashed, draw opacity=0.5] (a) at (-1.8,0){$a'$};
                \node[ line width=0.6pt, dashed, draw opacity=0.5] (a) at (2.3,0){$b'$};
         \node[ line width=0.6pt, dashed, draw opacity=0.5] (a) at (0.2,1.8){$j$};
          \node[ line width=0.6pt, dashed, draw opacity=0.5] (a) at (0.2,-1.8){$d'$};
        \node[ line width=0.6pt, dashed, draw opacity=0.5] (a) at (-0.2,2.3){$\rho$};
        \node[ line width=0.6pt, dashed, draw opacity=0.5] (a) at (-0.2,-2.3){$\mu'$};
       \node[ line width=0.6pt, dashed, draw opacity=0.5] (a) at (3.1,0){$k$};
        \end{tikzpicture}
    \end{aligned} \\ 
    & = \sum_{i,j,k,\rho,\sigma}\sqrt{\frac{d_j}{d_{a'}d_i}}\sqrt{\frac{d_k}{d_jd_{b'}}} \delta_{f,h'}\delta_{e,k}\delta_{i,h''}\delta_{e',c''}\delta_{i,f'}\delta_{c,h}\delta_{e'',f''}\delta_{c',k} \frac{1}{d_h}\frac{1}{d_k}\\
    & \quad  \quad \delta_{\sigma,\zeta'}\delta_{\rho,\gamma'}\delta_{\nu,\zeta}\delta_{\mu,\gamma} \delta_{j,g'}\delta_{e,f}\delta_{d,g}\sqrt{d_ad_{a'}d_id_bd_{b'}d_h} \\
    &\quad \quad \delta_{\nu'',\zeta''}\delta_{\mu'',\gamma''}\delta_{\nu',\sigma}\delta_{\mu',\rho}\delta_{d'',g''}\delta_{e',i}\delta_{d',j} \sqrt{d_{a'}d_{a''}d_{e''}d_{b'}d_{b''}d_k} \\
    & = \delta_{f,h'}\delta_{e,c'}\delta_{f',h''}\delta_{e',c''}\delta_{c,h}\delta_{e'',f''} \delta_{\nu,\zeta}\delta_{\mu,\gamma}\delta_{\nu',\zeta'}\delta_{\mu',\gamma'}\delta_{\nu'',\zeta''}\delta_{\mu'',\gamma''}\\
    &\quad \quad \delta_{d'',g''}\delta_{e',f'}\delta_{d',g'}\delta_{e,f}\delta_{d,g}\sqrt{\frac{d_ad_{a'}d_{a''}d_{e''}d_bd_{b'}d_{b''}}{d_h}}, 
    \end{align*}
    as desired. Similarly, one has $\varepsilon(XYZ) = \sum_{(Y)}\varepsilon(XY^{(2)})\varepsilon(Y^{(1)}Z)$. 
\end{proof}

Before we proceed to construct the weak Hopf algebra structure of the tube algebra, it is worth discussing the left and right counital maps. Recall that by definition, the left and the right counital maps are 
\begin{equation*}
    \varepsilon_L(X) = \sum_{(1)}\varepsilon(1^{(1)}\cdot X)1^{(2)}\quad\text{and}\quad \varepsilon_R(X) = \sum_{(1)}1^{(1)}\varepsilon(X\cdot 1^{(2)}),
\end{equation*}
respectively. Using the coproduct of the unit given by Eq.~\eqref{eq:coprod-unit},  it follows immediately that  
\begin{align}
          &\varepsilon_L\left( 
          \begin{aligned}\begin{tikzpicture}
             \draw[line width=.6pt,black] (0,0.5)--(0,1.5);
             \draw[line width=.6pt,black] (0,-0.5)--(0,-1.5);
             \draw[red] (0,0.8) arc[start angle=90, end angle=270, radius=0.8];
             \draw[blue] (0,1.3) arc[start angle=90, end angle=-90, radius=1.3];
            \node[ line width=0.6pt, dashed, draw opacity=0.5] (a) at (0,1.7){$h$};
             \node[ line width=0.6pt, dashed, draw opacity=0.5] (a) at (0,-1.7){$c$};
            \node[ line width=0.6pt, dashed, draw opacity=0.5] (a) at (-1,0){$a$};
            \node[ line width=0.6pt, dashed, draw opacity=0.5] (a) at (1.5,0){$b$};
            \node[ line width=0.6pt, dashed, draw opacity=0.5] (a) at (-0.2,-1){$d$};
            \node[ line width=0.6pt, dashed, draw opacity=0.5] (a) at (-0.4,-1.3){$\mu$};
            \node[ line width=0.6pt, dashed, draw opacity=0.5] (a) at (0.2,-0.8){$\nu$};
            \node[ line width=0.6pt, dashed, draw opacity=0.5] (a) at (0,-0.3){$e$};
            \node[ line width=0.6pt, dashed, draw opacity=0.5] (a) at (0,0.3){$f$};
            \node[ line width=0.6pt, dashed, draw opacity=0.5] (a) at (-0.2,1){$g$};
            \node[ line width=0.6pt, dashed, draw opacity=0.5] (a) at (-0.4,1.3){$\gamma$};
            \node[ line width=0.6pt, dashed, draw opacity=0.5] (a) at (0.2,0.8){$\zeta$};
        \end{tikzpicture}
    \end{aligned}
          \right) =  \delta_{e,f}\delta_{c,h}\delta_{d,g}\delta_{\nu,\zeta}\delta_{\mu,\gamma} \sqrt{\frac{d_a d_f d_b}{d_h}}\;\sum_y
          \;\,\begin{aligned}
        \begin{tikzpicture}
             \draw[line width=.6pt,black] (0,0.5)--(0,1.5);
             \draw[line width=.6pt,black] (0,-0.5)--(0,-1.5);
             \draw[red, dotted] (0,0.8) arc[start angle=90, end angle=270, radius=0.8];
             \draw[blue,dotted] (0,1.3) arc[start angle=90, end angle=-90, radius=1.3];
            \node[ line width=0.6pt, dashed, draw opacity=0.5] (a) at (-0.2,-1){$y$};
            \node[ line width=0.6pt, dashed, draw opacity=0.5] (a) at (-0.2,1){$c$};
        \end{tikzpicture}
    \end{aligned}\;, \\
    &\varepsilon_R\left( 
          \begin{aligned}\begin{tikzpicture}
             \draw[line width=.6pt,black] (0,0.5)--(0,1.5);
             \draw[line width=.6pt,black] (0,-0.5)--(0,-1.5);
             \draw[red] (0,0.8) arc[start angle=90, end angle=270, radius=0.8];
             \draw[blue] (0,1.3) arc[start angle=90, end angle=-90, radius=1.3];
            \node[ line width=0.6pt, dashed, draw opacity=0.5] (a) at (0,1.7){$h$};
             \node[ line width=0.6pt, dashed, draw opacity=0.5] (a) at (0,-1.7){$c$};
            \node[ line width=0.6pt, dashed, draw opacity=0.5] (a) at (-1,0){$a$};
            \node[ line width=0.6pt, dashed, draw opacity=0.5] (a) at (1.5,0){$b$};
            \node[ line width=0.6pt, dashed, draw opacity=0.5] (a) at (-0.2,-1){$d$};
            \node[ line width=0.6pt, dashed, draw opacity=0.5] (a) at (-0.4,-1.3){$\mu$};
            \node[ line width=0.6pt, dashed, draw opacity=0.5] (a) at (0.2,-0.8){$\nu$};
            \node[ line width=0.6pt, dashed, draw opacity=0.5] (a) at (0,-0.3){$e$};
            \node[ line width=0.6pt, dashed, draw opacity=0.5] (a) at (0,0.3){$f$};
            \node[ line width=0.6pt, dashed, draw opacity=0.5] (a) at (-0.2,1){$g$};
            \node[ line width=0.6pt, dashed, draw opacity=0.5] (a) at (-0.4,1.3){$\gamma$};
            \node[ line width=0.6pt, dashed, draw opacity=0.5] (a) at (0.2,0.8){$\zeta$};
        \end{tikzpicture}
    \end{aligned}\right) =  \delta_{e,f}\delta_{c,h}\delta_{d,g}\delta_{\nu,\zeta}\delta_{\mu,\gamma} \sqrt{\frac{d_a d_f d_b}{d_h}}\;\sum_x
    \;\,\begin{aligned}
        \begin{tikzpicture}
             \draw[line width=.6pt,black] (0,0.5)--(0,1.5);
             \draw[line width=.6pt,black] (0,-0.5)--(0,-1.5);
             \draw[red, dotted] (0,0.8) arc[start angle=90, end angle=270, radius=0.8];
             \draw[blue,dotted] (0,1.3) arc[start angle=90, end angle=-90, radius=1.3];
            \node[ line width=0.6pt, dashed, draw opacity=0.5] (a) at (-0.2,-1){$f$};
            \node[ line width=0.6pt, dashed, draw opacity=0.5] (a) at (-0.2,1){$x$};
        \end{tikzpicture}
    \end{aligned}\;.
\end{align}

Now we are ready to introduce the antipode. Together with the weak bialgebra structure above, the tube algebra is endowed with a weak Hopf algebra structure. 

\begin{proposition}
  For any given UMFC $\ED$, the tube algebra $\mathbf{Tube}({_{\ED}\ED_{\ED}})$ constitutes a weak Hopf algebra with the following antipode operation $S$:
      \begin{equation} \label{eq:tube-antipode}
          S\left( 
          \begin{aligned}\begin{tikzpicture}
             \draw[line width=.6pt,black] (0,0.5)--(0,1.5);
             \draw[line width=.6pt,black] (0,-0.5)--(0,-1.5);
             \draw[red] (0,0.8) arc[start angle=90, end angle=270, radius=0.8];
             \draw[blue] (0,1.3) arc[start angle=90, end angle=-90, radius=1.3];
            \node[ line width=0.6pt, dashed, draw opacity=0.5] (a) at (0,1.7){$h$};
             \node[ line width=0.6pt, dashed, draw opacity=0.5] (a) at (0,-1.7){$c$};
            \node[ line width=0.6pt, dashed, draw opacity=0.5] (a) at (-1,0){$a$};
            \node[ line width=0.6pt, dashed, draw opacity=0.5] (a) at (1.5,0){$b$};
            \node[ line width=0.6pt, dashed, draw opacity=0.5] (a) at (-0.2,-1){$d$};
            \node[ line width=0.6pt, dashed, draw opacity=0.5] (a) at (-0.4,-1.3){$\mu$};
            \node[ line width=0.6pt, dashed, draw opacity=0.5] (a) at (0.2,-0.8){$\nu$};
            \node[ line width=0.6pt, dashed, draw opacity=0.5] (a) at (0,-0.3){$e$};
            \node[ line width=0.6pt, dashed, draw opacity=0.5] (a) at (0,0.3){$f$};
            \node[ line width=0.6pt, dashed, draw opacity=0.5] (a) at (-0.2,1){$g$};
            \node[ line width=0.6pt, dashed, draw opacity=0.5] (a) at (-0.4,1.3){$\gamma$};
            \node[ line width=0.6pt, dashed, draw opacity=0.5] (a) at (0.2,0.8){$\zeta$};
        \end{tikzpicture}
    \end{aligned}
          \right) =\frac{d_f}{d_h}\;\begin{aligned}\begin{tikzpicture}
             \draw[line width=.6pt,black] (0,0.5)--(0,1.5);
             \draw[line width=.6pt,black] (0,-0.5)--(0,-1.5);
             \draw[red] (0,1.3) arc[start angle=90, end angle=270, radius=1.3];
             \draw[blue] (0,0.8) arc[start angle=90, end angle=-90, radius=0.8];
            \node[ line width=0.6pt, dashed, draw opacity=0.5] (a) at (0,1.7){$e$};
             \node[ line width=0.6pt, dashed, draw opacity=0.5] (a) at (0,-1.7){$f$};
            \node[ line width=0.6pt, dashed, draw opacity=0.5] (a) at (-1,0){$\bar{a}$};
            \node[ line width=0.6pt, dashed, draw opacity=0.5] (a) at (1,0){$\bar{b}$};
            \node[ line width=0.6pt, dashed, draw opacity=0.5] (a) at (-0.2,-1){$g$};
            \node[ line width=0.6pt, dashed, draw opacity=0.5] (a) at (0.2,-1.3){$\zeta$};
            \node[ line width=0.6pt, dashed, draw opacity=0.5] (a) at (-0.2,-0.7){$\gamma$};
            \node[ line width=0.6pt, dashed, draw opacity=0.5] (a) at (0,-0.3){$h$};
            \node[ line width=0.6pt, dashed, draw opacity=0.5] (a) at (0,0.3){$c$};
            \node[ line width=0.6pt, dashed, draw opacity=0.5] (a) at (-0.2,1){$d$};
            \node[ line width=0.6pt, dashed, draw opacity=0.5] (a) at (-0.4,1.4){$\nu$};
            \node[ line width=0.6pt, dashed, draw opacity=0.5] (a) at (0.3,1.0){$\mu$};
        \end{tikzpicture}
    \end{aligned}\;.
      \end{equation}

\end{proposition}

\begin{proof}
For $X=[a;b]_{c,d,e;\mu,\nu}^{f,g,h;\zeta,\gamma}$, we need to verify three identities:
\begin{gather}
    \sum_{(X)}X^{(1)}S(X^{(2)}) = \varepsilon_L(X), \label{eq:antipode-1}\\
    \sum_{(X)}S(X^{(1)})X^{(2)} = \varepsilon_R(X),\label{eq:antipode-2}\\ \sum_{(X)}S(X^{(1)})X^{(2)}S(X^{(3)}) = S(X). \label{eq:antipode-3}
\end{gather}
The left-hand side of \eqref{eq:antipode-1} is equal to 
\begin{align*}
      & \sum_{i,j,k,\rho,\sigma} \sqrt{\frac{d_j}{d_ad_i}}\sqrt{\frac{d_k}{d_jd_b}} \frac{d_i}{d_k}\begin{aligned}\begin{tikzpicture}
             \draw[line width=.6pt,black] (0,0.5)--(0,1.5);
             \draw[line width=.6pt,black] (0,-0.5)--(0,-1.5);
             \draw[red] (0,0.8) arc[start angle=90, end angle=270, radius=0.8];
             \draw[blue] (0,1.3) arc[start angle=90, end angle=-90, radius=1.3];
            \node[ line width=0.6pt, dashed, draw opacity=0.5] (a) at (0,1.7){$h$};
             \node[ line width=0.6pt, dashed, draw opacity=0.5] (a) at (0,-1.7){$k$};
            \node[ line width=0.6pt, dashed, draw opacity=0.5] (a) at (-1,0){$a$};
            \node[ line width=0.6pt, dashed, draw opacity=0.5] (a) at (1.5,0){$b$};
            \node[ line width=0.6pt, dashed, draw opacity=0.5] (a) at (-0.2,-1){$j$};
            \node[ line width=0.6pt, dashed, draw opacity=0.5] (a) at (-0.2,-1.4){$\rho$};
            \node[ line width=0.6pt, dashed, draw opacity=0.5] (a) at (0.2,-0.8){$\sigma$};
            \node[ line width=0.6pt, dashed, draw opacity=0.5] (a) at (0,-0.3){$i$};
            \node[ line width=0.6pt, dashed, draw opacity=0.5] (a) at (0,0.3){$f$};
            \node[ line width=0.6pt, dashed, draw opacity=0.5] (a) at (-0.2,1){$g$};
            \node[ line width=0.6pt, dashed, draw opacity=0.5] (a) at (-0.2,1.3){$\gamma$};
            \node[ line width=0.6pt, dashed, draw opacity=0.5] (a) at (0.2,0.8){$\zeta$};
        \end{tikzpicture}
    \end{aligned}  \; 
    \begin{aligned}
    \begin{tikzpicture}
             \draw[line width=.6pt,black] (0,0.5)--(0,1.5);
             \draw[line width=.6pt,black] (0,-0.5)--(0,-1.5);
             \draw[blue] (0,0.8) arc[start angle=90, end angle=-90, radius=0.8];
             \draw[red] (0,1.3) arc[start angle=90, end angle=270, radius=1.3];
            \node[ line width=0.6pt, dashed, draw opacity=0.5] (a) at (0,1.7){$e$};
             \node[ line width=0.6pt, dashed, draw opacity=0.5] (a) at (0,-1.7){$i$};
            \node[ line width=0.6pt, dashed, draw opacity=0.5] (a) at (-1.1,0){$\bar{a}$};
            \node[ line width=0.6pt, dashed, draw opacity=0.5] (a) at (1,0){$\bar{b}$};
            \node[ line width=0.6pt, dashed, draw opacity=0.5] (a) at (0.2,-1){$j$};
            \node[ line width=0.6pt, dashed, draw opacity=0.5] (a) at (0.2,-1.4){$\sigma$};
            \node[ line width=0.6pt, dashed, draw opacity=0.5] (a) at (-0.2,-0.8){$\rho$};
            \node[ line width=0.6pt, dashed, draw opacity=0.5] (a) at (0,-0.3){$k$};
            \node[ line width=0.6pt, dashed, draw opacity=0.5] (a) at (0,0.3){$c$};
            \node[ line width=0.6pt, dashed, draw opacity=0.5] (a) at (0.2,1){$d$};
            \node[ line width=0.6pt, dashed, draw opacity=0.5] (a) at (0.2,1.3){$\nu$};
            \node[ line width=0.6pt, dashed, draw opacity=0.5] (a) at (-0.2,0.8){$\mu$};
        \end{tikzpicture}
    \end{aligned} \\
    = & \sum_{i,j,k,\rho,\sigma} \sqrt{\frac{d_i}{d_ad_j}}\sqrt{\frac{d_j}{d_kd_b}}\,\delta_{e,f} \begin{aligned}
    \begin{tikzpicture}
             \draw[line width=.6pt,black] (0,0.5)--(0,2.5);
             \draw[line width=.6pt,black] (0,-0.5)--(0,-2.5);
             \draw[blue] (0,0.8) arc[start angle=90, end angle=-90, radius=0.8];
             \draw[red] (0,1.3) arc[start angle=90, end angle=270, radius=1.3];
             \draw[red] (0,1.8) arc[start angle=90, end angle=270, radius=1.8];
             \draw[blue] (0,2.3) arc[start angle=90, end angle=-90, radius=2.3];
            \node[ line width=0.6pt, dashed, draw opacity=0.5] (a) at (-0.2,1.5){$e$};
             \node[ line width=0.6pt, dashed, draw opacity=0.5] (a) at (-0.2,-1.5){$i$};
            \node[ line width=0.6pt, dashed, draw opacity=0.5] (a) at (-1.1,0){$\bar{a}$};
            \node[ line width=0.6pt, dashed, draw opacity=0.5] (a) at (1,0){$\bar{b}$};
            \node[ line width=0.6pt, dashed, draw opacity=0.5] (a) at (0.2,-1){$j$};
            \node[ line width=0.6pt, dashed, draw opacity=0.5] (a) at (0.2,-1.4){$\sigma$};
            \node[ line width=0.6pt, dashed, draw opacity=0.5] (a) at (-0.2,-0.8){$\rho$};
            \node[ line width=0.6pt, dashed, draw opacity=0.5] (a) at (0,-0.3){$k$};
            \node[ line width=0.6pt, dashed, draw opacity=0.5] (a) at (0,0.3){$c$};
            \node[ line width=0.6pt, dashed, draw opacity=0.5] (a) at (0.2,1){$d$};
            \node[ line width=0.6pt, dashed, draw opacity=0.5] (a) at (0.2,1.3){$\nu$};
            \node[ line width=0.6pt, dashed, draw opacity=0.5] (a) at (-0.2,0.8){$\mu$};
            \node[ line width=0.6pt, dashed, draw opacity=0.5] (a) at (-2,0){$a$};
            \node[ line width=0.6pt, dashed, draw opacity=0.5] (a) at (2.1,0){$b$};
            \node[ line width=0.6pt, dashed, draw opacity=0.5] (a) at (0,2.7){$h$};
            \node[ line width=0.6pt, dashed, draw opacity=0.5] (a) at (0,-2.7){$k$};
            \node[ line width=0.6pt, dashed, draw opacity=0.5] (a) at (-0.2,2.2){$\gamma$};
            \node[ line width=0.6pt, dashed, draw opacity=0.5] (a) at (0.2,2){$g$};
            \node[ line width=0.6pt, dashed, draw opacity=0.5] (a) at (0.2,1.63){$\zeta$}; 
            \node[ line width=0.6pt, dashed, draw opacity=0.5] (a) at (0.2,-1.7){$\sigma$};
            \node[ line width=0.6pt, dashed, draw opacity=0.5] (a) at (0.2,-2){$j$};
            \node[ line width=0.6pt, dashed, draw opacity=0.5] (a) at (-0.2,-2.2){$\rho$};
        \end{tikzpicture}
    \end{aligned} \\
    = &\quad \sum_{k}\delta_{e,f} \begin{aligned}
    \begin{tikzpicture}
             \draw[line width=.6pt,black] (0,-0.6)--(0,2.6);
             \draw[line width=.6pt,black] (0,-1)--(0,-2);
             \draw[red] (0,1.7) arc[start angle=90, end angle=270, radius=0.7];
             \draw[blue] (0,2.3) arc[start angle=90, end angle=-90, radius=1.3];
            \node[ line width=0.6pt, dashed, draw opacity=0.5] (a) at (-0.2,1){$f$};
            \node[ line width=0.6pt, dashed, draw opacity=0.5] (a) at (-0.2,-1.5){$k$};
            \node[ line width=0.6pt, dashed, draw opacity=0.5] (a) at (-0.2,-0.5){$c$};
            \node[ line width=0.6pt, dashed, draw opacity=0.5] (a) at (0.2,0.1){$d$};
            \node[ line width=0.6pt, dashed, draw opacity=0.5] (a) at (0.2,0.5){$\nu$};
            \node[ line width=0.6pt, dashed, draw opacity=0.5] (a) at (-0.2,-0.2){$\mu$};
            \node[ line width=0.6pt, dashed, draw opacity=0.5] (a) at (-0.9,1){$a$};
            \node[ line width=0.6pt, dashed, draw opacity=0.5] (a) at (1.5,1){$b$};
            \node[ line width=0.6pt, dashed, draw opacity=0.5] (a) at (0,2.85){$h$};
            \node[ line width=0.6pt, dashed, draw opacity=0.5] (a) at (-0.2,2.4){$\gamma$};
            \node[ line width=0.6pt, dashed, draw opacity=0.5] (a) at (0.2,2){$g$};
            \node[ line width=0.6pt, dashed, draw opacity=0.5] (a) at (0.2,1.63){$\zeta$}; 
        \end{tikzpicture}
    \end{aligned} = \sum_k \delta_{e,f}\delta_{d,g}\delta_{\nu,\zeta}\delta_{c,h}\delta_{\mu,\gamma}\sqrt{\frac{d_ad_f}{d_d}}\sqrt{\frac{d_dd_b}{d_h}} \;\;\begin{aligned}
        \begin{tikzpicture}
             \draw[line width=.6pt,black] (0,0.5)--(0,1.5);
             \draw[line width=.6pt,black] (0,-0.5)--(0,-1.5);
             \draw[red, dotted] (0,0.8) arc[start angle=90, end angle=270, radius=0.8];
             \draw[blue,dotted] (0,1.3) arc[start angle=90, end angle=-90, radius=1.3];
            \node[ line width=0.6pt, dashed, draw opacity=0.5] (a) at (-0.2,-1){$k$};
            \node[ line width=0.6pt, dashed, draw opacity=0.5] (a) at (-0.2,1){$h$};
        \end{tikzpicture}
    \end{aligned}\;,
\end{align*}
which is just $\varepsilon_L(X)$. Here in the second equality we used the parallel move twice, namely, we first apply $\displaystyle \sum_{i,\sigma}\sqrt{\frac{d_i}{d_ad_j}}
	\begin{aligned}
		\begin{tikzpicture}
			\draw[-latex,line width=.6pt,black] (0.4,-0.4) -- (0.1,-0.1);
			\draw[line width=.6pt,black] (0.3,-0.3) -- (0,0);
			\draw[-latex,line width=.6pt,red] (-0.4,-0.4) -- (-0.1,-0.1);
			\draw[line width=.6pt,red] (-0.3,-0.3) -- (0,0);
			\draw[-latex,line width=.6pt,black](0,0)--(0,0.45);
			\draw[line width=.6pt,black](0,0.1)--(0,0.7);
			\draw[line width=.6pt,red](-0.4,1.1)--(0,0.7);
			\draw[line width=.6pt,black](0.4,1.1)--(0,0.7);
			\draw[-latex,line width=.6pt,black](0,0.7)--(0.3,1);
			\draw[line width=.6pt,black](0.4,1.1)--(0,0.7);
			\draw[-latex,line width=.6pt,red](0,0.7)--(-0.3,1);
			\node[ line width=0.6pt, dashed, draw opacity=0.5] (a) at (0.4,-0.6){$j$};
			\node[ line width=0.6pt, dashed, draw opacity=0.5] (a) at (-0.4,-0.6){$a$};
			\node[ line width=0.6pt, dashed, draw opacity=0.5] (a) at (0.3,0.1){$\sigma$};
			\node[ line width=0.6pt, dashed, draw opacity=0.5] (a) at (0.3,0.6){$\sigma$};
			\node[ line width=0.6pt, dashed, draw opacity=0.5] (a) at (-0.2,0.3){$i$};
			\node[ line width=0.6pt, dashed, draw opacity=0.5] (a) at (0.4,1.3){$j$};
			\node[ line width=0.6pt, dashed, draw opacity=0.5] (a) at (-0.4,1.3){$a$};
		\end{tikzpicture}
	\end{aligned} = \begin{aligned}
		\begin{tikzpicture}
			\draw[line width=.6pt,red](0,-1.1)--(0,0.5);
			\draw[-latex,line width=.6pt,red](0,-0.3)--(0,0);
			\node[ line width=0.6pt, dashed, draw opacity=0.5] (a) at (0,0.7){$a$};
			\draw[line width=.6pt,black](0.6,-1.1)--(0.6,0.5);
			\draw[-latex,line width=.6pt,black](0.6,-0.3)--(0.6,0);
			\node[ line width=0.6pt, dashed, draw opacity=0.5] (a) at (0.6,0.7){$j$};
		\end{tikzpicture}
	\end{aligned}$ to separate $a$ and $j$, which forms the red half-loop up to planar isotopy, and then apply this move again to $k$ and $b$. Similarly, one also has Eq.~\eqref{eq:antipode-2}. Finally, one computes $\sum_{(X)}S(X^{(1)})X^{(2)}S(X^{(3)}) = \sum_{(X)}S(X^{(1)})\varepsilon_L(X^{(2)})$ as follows:
\begin{align*}
    \sum_{i,j,k,\rho,\sigma} &\sqrt{ \frac{d_j}{d_ad_i}}\sqrt{ \frac{d_k}{d_jd_b}} \sum_y \delta_{i,e}\delta_{k,c}\delta_{j,d}\delta_{\sigma,\nu}\delta_{\rho,\mu}\sqrt{\frac{d_ad_bd_e}{d_c}} \,\frac{d_f}{d_h}  
        \;\begin{aligned}\begin{tikzpicture}
             \draw[line width=.6pt,black] (0,0.5)--(0,1.5);
             \draw[line width=.6pt,black] (0,-0.5)--(0,-1.5);
             \draw[red] (0,1.3) arc[start angle=90, end angle=270, radius=1.3];
             \draw[blue] (0,0.8) arc[start angle=90, end angle=-90, radius=0.8];
            \node[ line width=0.6pt, dashed, draw opacity=0.5] (a) at (0,1.7){$i$};
             \node[ line width=0.6pt, dashed, draw opacity=0.5] (a) at (0,-1.7){$f$};
            \node[ line width=0.6pt, dashed, draw opacity=0.5] (a) at (-1,0){$\bar{a}$};
            \node[ line width=0.6pt, dashed, draw opacity=0.5] (a) at (1,0){$\bar{b}$};
            \node[ line width=0.6pt, dashed, draw opacity=0.5] (a) at (-0.2,-1){$g$};
            \node[ line width=0.6pt, dashed, draw opacity=0.5] (a) at (0.2,-1.3){$\zeta$};
            \node[ line width=0.6pt, dashed, draw opacity=0.5] (a) at (-0.2,-0.7){$\sigma$};
            \node[ line width=0.6pt, dashed, draw opacity=0.5] (a) at (0,-0.3){$h$};
            \node[ line width=0.6pt, dashed, draw opacity=0.5] (a) at (0,0.3){$k$};
            \node[ line width=0.6pt, dashed, draw opacity=0.5] (a) at (-0.2,1){$j$};
            \node[ line width=0.6pt, dashed, draw opacity=0.5] (a) at (-0.4,1.4){$\sigma$};
            \node[ line width=0.6pt, dashed, draw opacity=0.5] (a) at (0.3,1.0){$\rho$};
        \end{tikzpicture}
    \end{aligned} \; 
    \begin{aligned}
        \begin{tikzpicture}
             \draw[line width=.6pt,black] (0,0.5)--(0,1.5);
             \draw[line width=.6pt,black] (0,-0.5)--(0,-1.5);
             \draw[red, dotted] (0,0.8) arc[start angle=90, end angle=270, radius=0.8];
             \draw[blue,dotted] (0,1.3) arc[start angle=90, end angle=-90, radius=1.3];
            \node[ line width=0.6pt, dashed, draw opacity=0.5] (a) at (0,-1.7){$y$};
            \node[ line width=0.6pt, dashed, draw opacity=0.5] (a) at (0,1.7){$c$};
        \end{tikzpicture}
    \end{aligned} \\
    & = \sum_{i,j,k,\rho,\sigma} \sqrt{\frac{d_j}{d_ad_i}}\sqrt{ \frac{d_k}{d_jd_b}} \delta_{i,e}\delta_{k,c}\delta_{j,d}\delta_{\sigma,\nu}\delta_{\rho,\mu}\sqrt{\frac{d_ad_bd_e}{d_c}} \,\frac{d_f}{d_h}  
        \;\begin{aligned}\begin{tikzpicture}
             \draw[line width=.6pt,black] (0,0.5)--(0,1.5);
             \draw[line width=.6pt,black] (0,-0.5)--(0,-1.5);
             \draw[red] (0,1.3) arc[start angle=90, end angle=270, radius=1.3];
             \draw[blue] (0,0.8) arc[start angle=90, end angle=-90, radius=0.8];
            \node[ line width=0.6pt, dashed, draw opacity=0.5] (a) at (0,1.7){$i$};
             \node[ line width=0.6pt, dashed, draw opacity=0.5] (a) at (0,-1.7){$f$};
            \node[ line width=0.6pt, dashed, draw opacity=0.5] (a) at (-1,0){$\bar{a}$};
            \node[ line width=0.6pt, dashed, draw opacity=0.5] (a) at (1,0){$\bar{b}$};
            \node[ line width=0.6pt, dashed, draw opacity=0.5] (a) at (-0.2,-1){$g$};
            \node[ line width=0.6pt, dashed, draw opacity=0.5] (a) at (0.2,-1.3){$\zeta$};
            \node[ line width=0.6pt, dashed, draw opacity=0.5] (a) at (-0.2,-0.7){$\sigma$};
            \node[ line width=0.6pt, dashed, draw opacity=0.5] (a) at (0,-0.3){$h$};
            \node[ line width=0.6pt, dashed, draw opacity=0.5] (a) at (0,0.3){$k$};
            \node[ line width=0.6pt, dashed, draw opacity=0.5] (a) at (-0.2,1){$j$};
            \node[ line width=0.6pt, dashed, draw opacity=0.5] (a) at (-0.4,1.4){$\sigma$};
            \node[ line width=0.6pt, dashed, draw opacity=0.5] (a) at (0.3,1.0){$\rho$};
        \end{tikzpicture}
    \end{aligned}\;,
\end{align*}
which is equal to $S(X)$. This finishes the proof. 
\end{proof}

\begin{remark} \label{rmk:antipode}
    For $X=[a;b]_{c,d,e;\mu,\nu}^{f,g,h;\zeta,\gamma}$, when computing $S^2(X)$, one first needs to rewrite the right-hand side of Eq.~\eqref{eq:tube-antipode} by two F-moves in terms of basis diagrams, and then apply $S$ to the basis diagrams. But it can be shown that this process can be achieved by doing a similar operation directly as that in the definition of $S$. Specifically, one has 
    \begin{equation} 
          S\left( \begin{aligned}\begin{tikzpicture}
             \draw[line width=.6pt,black] (0,0.5)--(0,1.5);
             \draw[line width=.6pt,black] (0,-0.5)--(0,-1.5);
             \draw[red] (0,1.3) arc[start angle=90, end angle=270, radius=1.3];
             \draw[blue] (0,0.8) arc[start angle=90, end angle=-90, radius=0.8];
            \node[ line width=0.6pt, dashed, draw opacity=0.5] (a) at (0,1.7){$e$};
             \node[ line width=0.6pt, dashed, draw opacity=0.5] (a) at (0,-1.7){$f$};
            \node[ line width=0.6pt, dashed, draw opacity=0.5] (a) at (-1,0){$\bar{a}$};
            \node[ line width=0.6pt, dashed, draw opacity=0.5] (a) at (1,0){$\bar{b}$};
            \node[ line width=0.6pt, dashed, draw opacity=0.5] (a) at (-0.2,-1){$g$};
            \node[ line width=0.6pt, dashed, draw opacity=0.5] (a) at (0.2,-1.3){$\zeta$};
            \node[ line width=0.6pt, dashed, draw opacity=0.5] (a) at (-0.2,-0.7){$\gamma$};
            \node[ line width=0.6pt, dashed, draw opacity=0.5] (a) at (0,-0.3){$h$};
            \node[ line width=0.6pt, dashed, draw opacity=0.5] (a) at (0,0.3){$c$};
            \node[ line width=0.6pt, dashed, draw opacity=0.5] (a) at (-0.2,1){$d$};
            \node[ line width=0.6pt, dashed, draw opacity=0.5] (a) at (-0.4,1.4){$\nu$};
            \node[ line width=0.6pt, dashed, draw opacity=0.5] (a) at (0.3,1.0){$\mu$};
        \end{tikzpicture}
    \end{aligned}
          \right) =\frac{d_c}{d_e}\; \begin{aligned}\begin{tikzpicture}
             \draw[line width=.6pt,black] (0,0.5)--(0,1.5);
             \draw[line width=.6pt,black] (0,-0.5)--(0,-1.5);
             \draw[red] (0,0.8) arc[start angle=90, end angle=270, radius=0.8];
             \draw[blue] (0,1.3) arc[start angle=90, end angle=-90, radius=1.3];
            \node[ line width=0.6pt, dashed, draw opacity=0.5] (a) at (0,1.7){$h$};
             \node[ line width=0.6pt, dashed, draw opacity=0.5] (a) at (0,-1.7){$c$};
            \node[ line width=0.6pt, dashed, draw opacity=0.5] (a) at (-1,0){$a$};
            \node[ line width=0.6pt, dashed, draw opacity=0.5] (a) at (1.5,0){$b$};
            \node[ line width=0.6pt, dashed, draw opacity=0.5] (a) at (-0.2,-1){$d$};
            \node[ line width=0.6pt, dashed, draw opacity=0.5] (a) at (-0.4,-1.3){$\mu$};
            \node[ line width=0.6pt, dashed, draw opacity=0.5] (a) at (0.2,-0.8){$\nu$};
            \node[ line width=0.6pt, dashed, draw opacity=0.5] (a) at (0,-0.3){$e$};
            \node[ line width=0.6pt, dashed, draw opacity=0.5] (a) at (0,0.3){$f$};
            \node[ line width=0.6pt, dashed, draw opacity=0.5] (a) at (-0.2,1){$g$};
            \node[ line width=0.6pt, dashed, draw opacity=0.5] (a) at (-0.4,1.3){$\gamma$};
            \node[ line width=0.6pt, dashed, draw opacity=0.5] (a) at (0.2,0.8){$\zeta$};
        \end{tikzpicture}
    \end{aligned}\;.
      \end{equation}
    Therefore, one has $S^2(X) = (d_fd_c/d_hd_e)X$ by inspecting the diagrams. This means that $S^2\neq \id$ in general, which is different from the Hopf algebra case. 
    For a general weak Hopf algebra $W$, there exists a grouplike element $\xi\in W$ such that $S^2(h)=\xi h\xi^{-1}$ holds for all $h\in W$.  For the tube algebra, we can write down a grouplike element which realizes this condition. In fact, this grouplike element $\xi$ of $\mathbf{Tube}({_{\ED}}\ED_{\ED})$ and its inverse are given by 
    \begin{equation}
    \xi = \sum_{x,y} \frac{d_y}{d_x} 
    \begin{aligned}
    \begin{tikzpicture}
         \draw[line width=.6pt,black] (0,0.2)--(0,0.6);
         \draw[line width=.6pt,black] (0,-0.2)--(0,-0.6);
            \node[ line width=0.6pt, dashed, draw opacity=0.5] (a) at (-0.2,0.3){$x$};
        \node[ line width=0.6pt, dashed, draw opacity=0.5] (a) at (-0.2,-0.3){$y$};
    \end{tikzpicture}
\end{aligned}\;, \quad 
\xi^{-1} = \sum_{x',y'} \frac{d_{x'}}{d_{y'}} 
    \begin{aligned}
    \begin{tikzpicture}
         \draw[line width=.6pt,black] (0,0.2)--(0,0.6);
         \draw[line width=.6pt,black] (0,-0.2)--(0,-0.6);
            \node[ line width=0.6pt, dashed, draw opacity=0.5] (a) at (-0.2,0.3){$x'$};
        \node[ line width=0.6pt, dashed, draw opacity=0.5] (a) at (-0.2,-0.3){$y'$};
    \end{tikzpicture}
\end{aligned}\;. 
\end{equation}
That is because 
    \begin{align*}
        \xi X\xi^{-1}=\sum_{x,y}\sum_{x',y'}\delta_{x,h}\delta_{y,c}\frac{d_y}{d_x}\delta_{x',f}\delta_{y',e}\frac{d_{x'}}{d_{y'}}\begin{aligned}\begin{tikzpicture}
             \draw[line width=.6pt,black] (0,0.5)--(0,1.5);
             \draw[line width=.6pt,black] (0,-0.5)--(0,-1.5);
             \draw[red] (0,0.8) arc[start angle=90, end angle=270, radius=0.8];
             \draw[blue] (0,1.3) arc[start angle=90, end angle=-90, radius=1.3];
            \node[ line width=0.6pt, dashed, draw opacity=0.5] (a) at (0,1.7){$h$};
             \node[ line width=0.6pt, dashed, draw opacity=0.5] (a) at (0,-1.7){$c$};
            \node[ line width=0.6pt, dashed, draw opacity=0.5] (a) at (-1,0){$a$};
            \node[ line width=0.6pt, dashed, draw opacity=0.5] (a) at (1.5,0){$b$};
            \node[ line width=0.6pt, dashed, draw opacity=0.5] (a) at (-0.2,-1){$d$};
            \node[ line width=0.6pt, dashed, draw opacity=0.5] (a) at (-0.4,-1.3){$\mu$};
            \node[ line width=0.6pt, dashed, draw opacity=0.5] (a) at (0.2,-0.8){$\nu$};
            \node[ line width=0.6pt, dashed, draw opacity=0.5] (a) at (0,-0.3){$e$};
            \node[ line width=0.6pt, dashed, draw opacity=0.5] (a) at (0,0.3){$f$};
            \node[ line width=0.6pt, dashed, draw opacity=0.5] (a) at (-0.2,1){$g$};
            \node[ line width=0.6pt, dashed, draw opacity=0.5] (a) at (-0.4,1.3){$\gamma$};
            \node[ line width=0.6pt, dashed, draw opacity=0.5] (a) at (0.2,0.8){$\zeta$};
        \end{tikzpicture}
    \end{aligned}  = \frac{d_cd_f}{d_hd_e}\begin{aligned}\begin{tikzpicture}
             \draw[line width=.6pt,black] (0,0.5)--(0,1.5);
             \draw[line width=.6pt,black] (0,-0.5)--(0,-1.5);
             \draw[red] (0,0.8) arc[start angle=90, end angle=270, radius=0.8];
             \draw[blue] (0,1.3) arc[start angle=90, end angle=-90, radius=1.3];
            \node[ line width=0.6pt, dashed, draw opacity=0.5] (a) at (0,1.7){$h$};
             \node[ line width=0.6pt, dashed, draw opacity=0.5] (a) at (0,-1.7){$c$};
            \node[ line width=0.6pt, dashed, draw opacity=0.5] (a) at (-1,0){$a$};
            \node[ line width=0.6pt, dashed, draw opacity=0.5] (a) at (1.5,0){$b$};
            \node[ line width=0.6pt, dashed, draw opacity=0.5] (a) at (-0.2,-1){$d$};
            \node[ line width=0.6pt, dashed, draw opacity=0.5] (a) at (-0.4,-1.3){$\mu$};
            \node[ line width=0.6pt, dashed, draw opacity=0.5] (a) at (0.2,-0.8){$\nu$};
            \node[ line width=0.6pt, dashed, draw opacity=0.5] (a) at (0,-0.3){$e$};
            \node[ line width=0.6pt, dashed, draw opacity=0.5] (a) at (0,0.3){$f$};
            \node[ line width=0.6pt, dashed, draw opacity=0.5] (a) at (-0.2,1){$g$};
            \node[ line width=0.6pt, dashed, draw opacity=0.5] (a) at (-0.4,1.3){$\gamma$};
            \node[ line width=0.6pt, dashed, draw opacity=0.5] (a) at (0.2,0.8){$\zeta$};
        \end{tikzpicture}
    \end{aligned} = S^2(X). 
    \end{align*}

\end{remark}

A $*$-weak Hopf algebra is a weak Hopf algebra $W$ equipped with a $C^*$ structure $*:W\to W$ such that $\Delta$ is a $*$-homomorphism. In other words,
\begin{equation}
(x^*)^* = x,\,  (x+y)^* = x^* + y^*, \, (xy)^* = y^*x^*, \, (\alpha x)^* = \bar{\alpha} x^*,
\end{equation}
and
\begin{equation}
\Delta(x)^* = \Delta(x^*),
\end{equation}
hold for all $x, y \in W$ and $\alpha\in \mathbb{C}$. Additionally, it is shown that $S(x^*) = S^{-1}(x)^*$. Note that the antipode $S$ is invertible. 
The term $C^*$ weak Hopf algebra is used if there exists a fully faithful $*$-representation $\rho: W\to \mathbf{B}(\mathcal{H})$ for some operator space over a Hilbert space $\mathcal{H}$. Notice that a $C^*$ weak Hopf algebra is semisimple.
Let us now demonstrate that we can establish a $C^*$ structure over the tube algebra.

\begin{proposition}
For any given UMFC $\ED$, the tube algebra $\mathbf{Tube}({_{\ED}\ED_{\ED}})$  is a $C^*$
weak Hopf algebra. The $*$-operation is given by
\begin{equation}
        \left( \begin{aligned}\begin{tikzpicture}
             \draw[line width=.6pt,black] (0,0.5)--(0,1.5);
             \draw[line width=.6pt,black] (0,-0.5)--(0,-1.5);
             \draw[red] (0,0.8) arc[start angle=90, end angle=270, radius=0.8];
             \draw[blue] (0,1.3) arc[start angle=90, end angle=-90, radius=1.3];
            \node[ line width=0.6pt, dashed, draw opacity=0.5] (a) at (0,1.7){$h$};
             \node[ line width=0.6pt, dashed, draw opacity=0.5] (a) at (0,-1.7){$c$};
            \node[ line width=0.6pt, dashed, draw opacity=0.5] (a) at (-1,0){$a$};
            \node[ line width=0.6pt, dashed, draw opacity=0.5] (a) at (1.5,0){$b$};
            \node[ line width=0.6pt, dashed, draw opacity=0.5] (a) at (-0.2,-1){$d$};
            \node[ line width=0.6pt, dashed, draw opacity=0.5] (a) at (-0.4,-1.3){$\mu$};
            \node[ line width=0.6pt, dashed, draw opacity=0.5] (a) at (0.2,-0.8){$\nu$};
            \node[ line width=0.6pt, dashed, draw opacity=0.5] (a) at (0,-0.3){$e$};
            \node[ line width=0.6pt, dashed, draw opacity=0.5] (a) at (0,0.3){$f$};
            \node[ line width=0.6pt, dashed, draw opacity=0.5] (a) at (-0.2,1){$g$};
            \node[ line width=0.6pt, dashed, draw opacity=0.5] (a) at (-0.4,1.3){$\gamma$};
            \node[ line width=0.6pt, dashed, draw opacity=0.5] (a) at (0.2,0.8){$\zeta$};
        \end{tikzpicture}
    \end{aligned} \right)^*= \frac{d_e}{d_c} \;
\begin{aligned}\begin{tikzpicture}
             \draw[line width=.6pt,black] (0,0.5)--(0,1.5);
             \draw[line width=.6pt,black] (0,-0.5)--(0,-1.5);
             \draw[red] (0,1.3) arc[start angle=90, end angle=270, radius=1.3];
             \draw[blue] (0,0.8) arc[start angle=90, end angle=-90, radius=0.8];
            \node[ line width=0.6pt, dashed, draw opacity=0.5] (a) at (0,1.7){$f$};
             \node[ line width=0.6pt, dashed, draw opacity=0.5] (a) at (0,-1.7){$e$};
            \node[ line width=0.6pt, dashed, draw opacity=0.5] (a) at (-1,0){$\bar{a}$};
            \node[ line width=0.6pt, dashed, draw opacity=0.5] (a) at (1,0){$\bar{b}$};
            \node[ line width=0.6pt, dashed, draw opacity=0.5] (a) at (-0.2,-1){$d$};
            \node[ line width=0.6pt, dashed, draw opacity=0.5] (a) at (0.2,-1.3){$\nu$};
            \node[ line width=0.6pt, dashed, draw opacity=0.5] (a) at (-0.2,-0.7){$\mu$};
            \node[ line width=0.6pt, dashed, draw opacity=0.5] (a) at (0,-0.3){$c$};
            \node[ line width=0.6pt, dashed, draw opacity=0.5] (a) at (0,0.3){$h$};
            \node[ line width=0.6pt, dashed, draw opacity=0.5] (a) at (-0.2,1){$g$};
            \node[ line width=0.6pt, dashed, draw opacity=0.5] (a) at (-0.4,1.4){$\zeta$};
            \node[ line width=0.6pt, dashed, draw opacity=0.5] (a) at (0.3,1.0){$\gamma$};
        \end{tikzpicture}
    \end{aligned}\;. \label{eq:star-operation}
\end{equation}
The morphisms on the right-hand side should be interpreted as the dual morphisms corresponding to the respective vertices.
\end{proposition}

\begin{proof}
    Fix $X=[a;b]_{c,d,e;\mu,\nu}^{f,g,h;\zeta,\gamma}$ and $Y = [a';b']_{c',d',e';\mu',\nu'}^{f',g',h';\zeta',\gamma'}$ in $\mathbf{Tube}({_{\ED}\ED_{\ED}})$. Note that the $*$-operation in Eq.~\eqref{eq:star-operation} is extended antilinearly, thus $(\alpha X)^* = \bar{\alpha}X^*$ for $\alpha\in \mathbb{C}$. It is clear that $(X^*)^* = X$, $(X+Y)^*=X^*+Y^*$. It remains to show that $(XY)^*=Y^*X^*$ and $\Delta(X)^* = \Delta(X^*)$. The first identity holds because of the following
    \begin{align*}
        \frac{d_{e'}}{d_{c'}}\frac{d_e}{d_c}\begin{aligned}\begin{tikzpicture}
             \draw[line width=.6pt,black] (0,0.5)--(0,1.5);
             \draw[line width=.6pt,black] (0,-0.5)--(0,-1.5);
             \draw[red] (0,1.3) arc[start angle=90, end angle=270, radius=1.3];
             \draw[blue] (0,0.8) arc[start angle=90, end angle=-90, radius=0.8];
            \node[ line width=0.6pt, dashed, draw opacity=0.5] (a) at (0,1.7){$f'$};
             \node[ line width=0.6pt, dashed, draw opacity=0.5] (a) at (0,-1.7){$e'$};
            \node[ line width=0.6pt, dashed, draw opacity=0.5] (a) at (-1,0){$\bar{a'}$};
            \node[ line width=0.6pt, dashed, draw opacity=0.5] (a) at (1,0){$\bar{b'}$};
            \node[ line width=0.6pt, dashed, draw opacity=0.5] (a) at (-0.2,-1){$d'$};
            \node[ line width=0.6pt, dashed, draw opacity=0.5] (a) at (0.2,-1.3){$\nu'$};
            \node[ line width=0.6pt, dashed, draw opacity=0.5] (a) at (-0.2,-0.7){$\mu'$};
            \node[ line width=0.6pt, dashed, draw opacity=0.5] (a) at (0,-0.3){$c'$};
            \node[ line width=0.6pt, dashed, draw opacity=0.5] (a) at (0,0.3){$h'$};
            \node[ line width=0.6pt, dashed, draw opacity=0.5] (a) at (-0.2,1){$g'$};
            \node[ line width=0.6pt, dashed, draw opacity=0.5] (a) at (-0.4,1.4){$\zeta'$};
            \node[ line width=0.6pt, dashed, draw opacity=0.5] (a) at (0.3,1.0){$\gamma'$};
        \end{tikzpicture}
    \end{aligned} 
    \begin{aligned}\begin{tikzpicture}
             \draw[line width=.6pt,black] (0,0.5)--(0,1.5);
             \draw[line width=.6pt,black] (0,-0.5)--(0,-1.5);
             \draw[red] (0,1.3) arc[start angle=90, end angle=270, radius=1.3];
             \draw[blue] (0,0.8) arc[start angle=90, end angle=-90, radius=0.8];
            \node[ line width=0.6pt, dashed, draw opacity=0.5] (a) at (0,1.7){$f$};
             \node[ line width=0.6pt, dashed, draw opacity=0.5] (a) at (0,-1.7){$e$};
            \node[ line width=0.6pt, dashed, draw opacity=0.5] (a) at (-1,0){$\bar{a}$};
            \node[ line width=0.6pt, dashed, draw opacity=0.5] (a) at (1,0){$\bar{b}$};
            \node[ line width=0.6pt, dashed, draw opacity=0.5] (a) at (-0.2,-1){$d$};
            \node[ line width=0.6pt, dashed, draw opacity=0.5] (a) at (0.2,-1.3){$\nu$};
            \node[ line width=0.6pt, dashed, draw opacity=0.5] (a) at (-0.2,-0.7){$\mu$};
            \node[ line width=0.6pt, dashed, draw opacity=0.5] (a) at (0,-0.3){$c$};
            \node[ line width=0.6pt, dashed, draw opacity=0.5] (a) at (0,0.3){$h$};
            \node[ line width=0.6pt, dashed, draw opacity=0.5] (a) at (-0.2,1){$g$};
            \node[ line width=0.6pt, dashed, draw opacity=0.5] (a) at (-0.4,1.4){$\zeta$};
            \node[ line width=0.6pt, dashed, draw opacity=0.5] (a) at (0.3,1.0){$\gamma$};
        \end{tikzpicture}
    \end{aligned}
    =\frac{d_{e'}}{d_{c}}\delta_{f,h'}\delta_{e,c'}\begin{aligned}
        \begin{tikzpicture}
             \draw[line width=.6pt,black] (0,0.5)--(0,2.5);
             \draw[line width=.6pt,black] (0,-0.5)--(0,-2.5);
             \draw[blue] (0,0.8) arc[start angle=90, end angle=-90, radius=0.8];
             \draw[red] (0,1.3) arc[start angle=90, end angle=270, radius=1.3];
            \node[ line width=0.6pt, dashed, draw opacity=0.5] (a) at (0.2,1.5){$f$};
            \node[ line width=0.6pt, dashed, draw opacity=0.5] (a) at (0.2,-1.5){$e$};
            \node[ line width=0.6pt, dashed, draw opacity=0.5] (a) at (-1.5,0){$\bar{a}$};
            \node[ line width=0.6pt, dashed, draw opacity=0.5] (a) at (1,0){$\bar{b}$};
            \node[ line width=0.6pt, dashed, draw opacity=0.5] (a) at (-0.2,-1){$d$};
            \node[ line width=0.6pt, dashed, draw opacity=0.5] (a) at (-0.2,-1.3){$\nu$};
            \node[ line width=0.6pt, dashed, draw opacity=0.5] (a) at (0.2,-0.8){$\mu$};
            \node[ line width=0.6pt, dashed, draw opacity=0.5] (a) at (0,-0.3){$c$};
            \node[ line width=0.6pt, dashed, draw opacity=0.5] (a) at (0,0.3){$h$};
            \node[ line width=0.6pt, dashed, draw opacity=0.5] (a) at (-0.2,1){$g$};
            \node[ line width=0.6pt, dashed, draw opacity=0.5] (a) at (-0.4,1.3){$\zeta$};
            \node[ line width=0.6pt, dashed, draw opacity=0.5] (a) at (0.2,0.8){$\gamma$};
            \draw[blue] (0,1.6) arc[start angle=90, end angle=-90, radius=1.6];
          \node[ line width=0.6pt, dashed, draw opacity=0.5] (a) at (-0.3,1.8){$\gamma'$};
          \node[ line width=0.6pt, dashed, draw opacity=0.5] (a) at (-0.3,-1.8){$\mu'$};
           \draw[red] (0,2.1) arc[start angle=90, end angle=270, radius=2.1];
        \node[ line width=0.6pt, dashed, draw opacity=0.5] (a) at (-2.3,0){$\bar{a'}$};
       \node[ line width=0.6pt, dashed, draw opacity=0.5] (a) at (1.9,0){$\bar{b'}$};
       \node[ line width=0.6pt, dashed, draw opacity=0.5] (a) at (0.2,1.8){$g'$};
       \node[ line width=0.6pt, dashed, draw opacity=0.5] (a) at (0.2,-1.8){$d'$};
       \node[ line width=0.6pt, dashed, draw opacity=0.5] (a) at (0.2,2.3){$\zeta'$};
       \node[ line width=0.6pt, dashed, draw opacity=0.5] (a) at (0.2,-2.3){$\nu'$};
       \node[ line width=0.6pt, dashed, draw opacity=0.5] (a) at (-0.3,2.5){$f'$};
       \node[ line width=0.6pt, dashed, draw opacity=0.5] (a) at (-0.3,-2.5){$e'$};
        \end{tikzpicture}
    \end{aligned},
    \end{align*}
    where the left-hand side is $Y^*X^*$ and the right-hand side is $(XY)^*$\,\footnote{This comes from a similar observation as that in Remark~\ref{rmk:take-counit}. To operate the $*$-operation on the diagram of the product, one does not need to perform the F-move first; instead, performing a similar operation as in Eq.~\eqref{eq:star-operation} on the entire diagram is sufficient.}. For the second identity, using Eq.~\eqref{eq:TubeComultiplication}, we have 
    \begin{align*}
       \Delta(X)^*&= \sum_{i,j,k,\rho,\sigma} \sqrt{ \frac{d_j}{d_ad_i}}\sqrt{ \frac{d_k}{d_jd_b}}\frac{d_i}{d_k}\frac{d_e}{d_c}\;
\begin{aligned}\begin{tikzpicture}
             \draw[line width=.6pt,black] (0,0.5)--(0,1.5);
             \draw[line width=.6pt,black] (0,-0.5)--(0,-1.5);
             \draw[red] (0,1.3) arc[start angle=90, end angle=270, radius=1.3];
             \draw[blue] (0,0.8) arc[start angle=90, end angle=-90, radius=0.8];
            \node[ line width=0.6pt, dashed, draw opacity=0.5] (a) at (0,1.7){$f$};
             \node[ line width=0.6pt, dashed, draw opacity=0.5] (a) at (0,-1.7){$i$};
            \node[ line width=0.6pt, dashed, draw opacity=0.5] (a) at (-1,0){$\bar{a}$};
            \node[ line width=0.6pt, dashed, draw opacity=0.5] (a) at (1,0){$\bar{b}$};
            \node[ line width=0.6pt, dashed, draw opacity=0.5] (a) at (-0.2,-1){$j$};
            \node[ line width=0.6pt, dashed, draw opacity=0.5] (a) at (0.2,-1.3){$\sigma$};
            \node[ line width=0.6pt, dashed, draw opacity=0.5] (a) at (-0.2,-0.7){$\rho$};
            \node[ line width=0.6pt, dashed, draw opacity=0.5] (a) at (0,-0.3){$k$};
            \node[ line width=0.6pt, dashed, draw opacity=0.5] (a) at (0,0.3){$h$};
            \node[ line width=0.6pt, dashed, draw opacity=0.5] (a) at (-0.2,1){$g$};
            \node[ line width=0.6pt, dashed, draw opacity=0.5] (a) at (-0.4,1.4){$\zeta$};
            \node[ line width=0.6pt, dashed, draw opacity=0.5] (a) at (0.3,1.0){$\gamma$};
        \end{tikzpicture}
    \end{aligned} \otimes 
    \begin{aligned}\begin{tikzpicture}
             \draw[line width=.6pt,black] (0,0.5)--(0,1.5);
             \draw[line width=.6pt,black] (0,-0.5)--(0,-1.5);
             \draw[red] (0,1.3) arc[start angle=90, end angle=270, radius=1.3];
             \draw[blue] (0,0.8) arc[start angle=90, end angle=-90, radius=0.8];
            \node[ line width=0.6pt, dashed, draw opacity=0.5] (a) at (0,1.7){$i$};
             \node[ line width=0.6pt, dashed, draw opacity=0.5] (a) at (0,-1.7){$e$};
            \node[ line width=0.6pt, dashed, draw opacity=0.5] (a) at (-1,0){$\bar{a}$};
            \node[ line width=0.6pt, dashed, draw opacity=0.5] (a) at (1,0){$\bar{b}$};
            \node[ line width=0.6pt, dashed, draw opacity=0.5] (a) at (-0.2,-1){$d$};
            \node[ line width=0.6pt, dashed, draw opacity=0.5] (a) at (0.2,-1.3){$\nu$};
            \node[ line width=0.6pt, dashed, draw opacity=0.5] (a) at (-0.2,-0.7){$\mu$};
            \node[ line width=0.6pt, dashed, draw opacity=0.5] (a) at (0,-0.3){$c$};
            \node[ line width=0.6pt, dashed, draw opacity=0.5] (a) at (0,0.3){$k$};
            \node[ line width=0.6pt, dashed, draw opacity=0.5] (a) at (-0.2,1){$j$};
            \node[ line width=0.6pt, dashed, draw opacity=0.5] (a) at (-0.4,1.4){$\sigma$};
            \node[ line width=0.6pt, dashed, draw opacity=0.5] (a) at (0.3,1.0){$\rho$};
        \end{tikzpicture}
    \end{aligned} \\
    & = \frac{d_e}{d_c}\sum_{i,j,k,\rho,\sigma} \sqrt{ \frac{d_j}{d_bd_k}}\sqrt{ \frac{d_i}{d_jd_a}}\;
\begin{aligned}\begin{tikzpicture}
             \draw[line width=.6pt,black] (0,0.5)--(0,1.5);
             \draw[line width=.6pt,black] (0,-0.5)--(0,-1.5);
             \draw[red] (0,1.3) arc[start angle=90, end angle=270, radius=1.3];
             \draw[blue] (0,0.8) arc[start angle=90, end angle=-90, radius=0.8];
            \node[ line width=0.6pt, dashed, draw opacity=0.5] (a) at (0,1.7){$f$};
             \node[ line width=0.6pt, dashed, draw opacity=0.5] (a) at (0,-1.7){$i$};
            \node[ line width=0.6pt, dashed, draw opacity=0.5] (a) at (-1,0){$\bar{a}$};
            \node[ line width=0.6pt, dashed, draw opacity=0.5] (a) at (1,0){$\bar{b}$};
            \node[ line width=0.6pt, dashed, draw opacity=0.5] (a) at (-0.2,-1){$j$};
            \node[ line width=0.6pt, dashed, draw opacity=0.5] (a) at (0.2,-1.3){$\sigma$};
            \node[ line width=0.6pt, dashed, draw opacity=0.5] (a) at (-0.2,-0.7){$\rho$};
            \node[ line width=0.6pt, dashed, draw opacity=0.5] (a) at (0,-0.3){$k$};
            \node[ line width=0.6pt, dashed, draw opacity=0.5] (a) at (0,0.3){$h$};
            \node[ line width=0.6pt, dashed, draw opacity=0.5] (a) at (-0.2,1){$g$};
            \node[ line width=0.6pt, dashed, draw opacity=0.5] (a) at (-0.4,1.4){$\zeta$};
            \node[ line width=0.6pt, dashed, draw opacity=0.5] (a) at (0.3,1.0){$\gamma$};
        \end{tikzpicture}
    \end{aligned} \otimes 
    \begin{aligned}\begin{tikzpicture}
             \draw[line width=.6pt,black] (0,0.5)--(0,1.5);
             \draw[line width=.6pt,black] (0,-0.5)--(0,-1.5);
             \draw[red] (0,1.3) arc[start angle=90, end angle=270, radius=1.3];
             \draw[blue] (0,0.8) arc[start angle=90, end angle=-90, radius=0.8];
            \node[ line width=0.6pt, dashed, draw opacity=0.5] (a) at (0,1.7){$i$};
             \node[ line width=0.6pt, dashed, draw opacity=0.5] (a) at (0,-1.7){$e$};
            \node[ line width=0.6pt, dashed, draw opacity=0.5] (a) at (-1,0){$\bar{a}$};
            \node[ line width=0.6pt, dashed, draw opacity=0.5] (a) at (1,0){$\bar{b}$};
            \node[ line width=0.6pt, dashed, draw opacity=0.5] (a) at (-0.2,-1){$d$};
            \node[ line width=0.6pt, dashed, draw opacity=0.5] (a) at (0.2,-1.3){$\nu$};
            \node[ line width=0.6pt, dashed, draw opacity=0.5] (a) at (-0.2,-0.7){$\mu$};
            \node[ line width=0.6pt, dashed, draw opacity=0.5] (a) at (0,-0.3){$c$};
            \node[ line width=0.6pt, dashed, draw opacity=0.5] (a) at (0,0.3){$k$};
            \node[ line width=0.6pt, dashed, draw opacity=0.5] (a) at (-0.2,1){$j$};
            \node[ line width=0.6pt, dashed, draw opacity=0.5] (a) at (-0.4,1.4){$\sigma$};
            \node[ line width=0.6pt, dashed, draw opacity=0.5] (a) at (0.3,1.0){$\rho$};
        \end{tikzpicture}
    \end{aligned}=\Delta(X^*), 
    \end{align*}
    where the last equality follows by inserting a $k$ wall in $X^*$, performing the parallel moves, and taking summation over $k\in \Irr(\ED)$. 
\end{proof}

\begin{remark}
  When applying $*$-operation to the right-hand side of Eq.~\eqref{eq:star-operation}, we can do the similar operation directly as that in Eq.~\eqref{eq:star-operation} without doing F-moves. More precisely, this means 
  \begin{equation}
        \left( \begin{aligned}\begin{tikzpicture}
             \draw[line width=.6pt,black] (0,0.5)--(0,1.5);
             \draw[line width=.6pt,black] (0,-0.5)--(0,-1.5);
             \draw[red] (0,1.3) arc[start angle=90, end angle=270, radius=1.3];
             \draw[blue] (0,0.8) arc[start angle=90, end angle=-90, radius=0.8];
            \node[ line width=0.6pt, dashed, draw opacity=0.5] (a) at (0,1.7){$f$};
             \node[ line width=0.6pt, dashed, draw opacity=0.5] (a) at (0,-1.7){$e$};
            \node[ line width=0.6pt, dashed, draw opacity=0.5] (a) at (-1,0){$\bar{a}$};
            \node[ line width=0.6pt, dashed, draw opacity=0.5] (a) at (1,0){$\bar{b}$};
            \node[ line width=0.6pt, dashed, draw opacity=0.5] (a) at (-0.2,-1){$d$};
            \node[ line width=0.6pt, dashed, draw opacity=0.5] (a) at (0.2,-1.3){$\nu$};
            \node[ line width=0.6pt, dashed, draw opacity=0.5] (a) at (-0.2,-0.7){$\mu$};
            \node[ line width=0.6pt, dashed, draw opacity=0.5] (a) at (0,-0.3){$c$};
            \node[ line width=0.6pt, dashed, draw opacity=0.5] (a) at (0,0.3){$h$};
            \node[ line width=0.6pt, dashed, draw opacity=0.5] (a) at (-0.2,1){$g$};
            \node[ line width=0.6pt, dashed, draw opacity=0.5] (a) at (-0.4,1.4){$\zeta$};
            \node[ line width=0.6pt, dashed, draw opacity=0.5] (a) at (0.3,1.0){$\gamma$};
        \end{tikzpicture}
    \end{aligned}
    \right)^*= \frac{d_c}{d_e} \;
    \begin{aligned}\begin{tikzpicture}
             \draw[line width=.6pt,black] (0,0.5)--(0,1.5);
             \draw[line width=.6pt,black] (0,-0.5)--(0,-1.5);
             \draw[red] (0,0.8) arc[start angle=90, end angle=270, radius=0.8];
             \draw[blue] (0,1.3) arc[start angle=90, end angle=-90, radius=1.3];
            \node[ line width=0.6pt, dashed, draw opacity=0.5] (a) at (0,1.7){$h$};
             \node[ line width=0.6pt, dashed, draw opacity=0.5] (a) at (0,-1.7){$c$};
            \node[ line width=0.6pt, dashed, draw opacity=0.5] (a) at (-1,0){$a$};
            \node[ line width=0.6pt, dashed, draw opacity=0.5] (a) at (1.5,0){$b$};
            \node[ line width=0.6pt, dashed, draw opacity=0.5] (a) at (-0.2,-1){$d$};
            \node[ line width=0.6pt, dashed, draw opacity=0.5] (a) at (-0.4,-1.3){$\mu$};
            \node[ line width=0.6pt, dashed, draw opacity=0.5] (a) at (0.2,-0.8){$\nu$};
            \node[ line width=0.6pt, dashed, draw opacity=0.5] (a) at (0,-0.3){$e$};
            \node[ line width=0.6pt, dashed, draw opacity=0.5] (a) at (0,0.3){$f$};
            \node[ line width=0.6pt, dashed, draw opacity=0.5] (a) at (-0.2,1){$g$};
            \node[ line width=0.6pt, dashed, draw opacity=0.5] (a) at (-0.4,1.3){$\gamma$};
            \node[ line width=0.6pt, dashed, draw opacity=0.5] (a) at (0.2,0.8){$\zeta$};
        \end{tikzpicture}
    \end{aligned} \;.
\end{equation}
  Actually we have used this fact while we showed that $*$-operation is involutive. It follows that $S(S(X^*)^*) = X$, namely $S(X^*) = S^{-1}(X)^*$, holds for any $X\in \mathbf{Tube}({_{\ED}\ED_{\ED}})$, which is consistent with the property for general weak Hopf algebras.
\end{remark}

\begin{remark}
It is worth noting that the coefficients chosen in the above definitions are not unique. We can apply local basis transformations and obtain new structure coefficients from the given ones.
\end{remark}

\subsection{Tube algebra as a crossed product}
\label{sec:TubeCross}

Since we are working on the trivial domain wall $_{\ED}\ED_{\ED}$, which forms a bimodule category over itself, if we disregard the right module structure, we obtain the following subalgebra:
\begin{equation}
 \begin{aligned}
     \mathbf{L}({_{\ED}\ED})
 \end{aligned}  =\operatorname{span}\left\{  
   \begin{aligned}
        \begin{tikzpicture}
             \draw[line width=.6pt,black] (0,0.5)--(0,1.5);
             \draw[line width=.6pt,black] (0,-0.5)--(0,-1.5);
             \draw[red] (0,0.8) arc[start angle=90, end angle=270, radius=0.8];
             \draw[blue,dotted] (0,1.3) arc[start angle=90, end angle=-90, radius=1.3];
            \node[ line width=0.6pt, dashed, draw opacity=0.5] (a) at (-1,0){$a$};
            \node[ line width=0.6pt, dashed, draw opacity=0.5] (a) at (1.5,0){$\one$};
            \node[ line width=0.6pt, dashed, draw opacity=0.5] (a) at (-0.2,-1.3){$d$};
            \node[ line width=0.6pt, dashed, draw opacity=0.5] (a) at (0.2,-0.8){$\nu$};
            \node[ line width=0.6pt, dashed, draw opacity=0.5] (a) at (0,-0.3){$e$};
            \node[ line width=0.6pt, dashed, draw opacity=0.5] (a) at (0,0.3){$f$};
            \node[ line width=0.6pt, dashed, draw opacity=0.5] (a) at (-0.2,1.3){$g$};
            \node[ line width=0.6pt, dashed, draw opacity=0.5] (a) at (0.2,0.8){$\zeta$};
        \end{tikzpicture}
    \end{aligned}
   :\quad  a,\cdots,g\in \Irr(\ED), \nu,\zeta\in \Hom_{\ED} 
    \right\}.
    \label{eq:tubebasis-L}
\end{equation}
From the discussion in the previous subsection, it is clear that  $\mathbf{L}({_{\ED}\ED})$ is a $C^*$ weak Hopf algebra.
It is isomorphic to the boundary tube algebra that will be discussed in Sec.~\ref{sec:bdtheory}.  The boundary tube algebra is originally constructed in Ref.~\cite{Kitaev2012boundary}. Note that the structure coefficient we choose here is slightly different from that of Ref.~\cite{Kitaev2012boundary}; also, refer to Ref.~\cite{bridgeman2023invertible} for a discussion, where the algebra is named annular algebra.
As will be discussed later in Sec.~\ref{sec:bdtheory}, this boundary tube algebra gives the boundary excitations (regard $\ED$ as a left $\ED$-module category): $\Fun_{\ED}(\ED,\ED) \simeq \Rep(\tilde{\mathbf{L}}({_{\ED}\ED}))$; see Sec.~\ref{sec:bdtheory} for the definition of $\tilde{\mathbf{L}}({_{\ED}\ED})$. Eq.~\eqref{eq:tubebasis-L} provides a natural embedding of boundary tube algebra to the bulk tube algebra: $\tilde{\mathbf{L}}({_{\ED}\ED})\cong \mathbf{L}({_{\ED}\ED}) \hookrightarrow \mathbf{Tube}({_{\ED}}\ED_{\ED})$.

Analogously, disregarding the left module structure, we obtain the following $C^*$ weak Hopf subalgebra:
\begin{equation}
 \begin{aligned}
     \mathbf{R}(\ED_{\ED})
 \end{aligned}  =\operatorname{span}\left\{  
   \begin{aligned}
        \begin{tikzpicture}
             \draw[line width=.6pt,black] (0,0.5)--(0,1.5);
             \draw[line width=.6pt,black] (0,-0.5)--(0,-1.5);
             \draw[red,dotted] (0,0.8) arc[start angle=90, end angle=270, radius=0.8];
             \draw[blue] (0,1.3) arc[start angle=90, end angle=-90, radius=1.3];
            \node[ line width=0.6pt, dashed, draw opacity=0.5] (a) at (0,1.7){$h$};
             \node[ line width=0.6pt, dashed, draw opacity=0.5] (a) at (0,-1.7){$c$};
            \node[ line width=0.6pt, dashed, draw opacity=0.5] (a) at (-1,0){$\one$};
            \node[ line width=0.6pt, dashed, draw opacity=0.5] (a) at (1.5,0){$b$};
            \node[ line width=0.6pt, dashed, draw opacity=0.5] (a) at (-0.2,-1.3){$\mu$};
            \node[ line width=0.6pt, dashed, draw opacity=0.5] (a) at (0,-0.3){$e$};
            \node[ line width=0.6pt, dashed, draw opacity=0.5] (a) at (0,0.3){$f$};
            \node[ line width=0.6pt, dashed, draw opacity=0.5] (a) at (-0.2,1.2){$\gamma$};
        \end{tikzpicture}
    \end{aligned}
   :\quad  b,\cdots,h\in \Irr(\ED), \mu,\gamma,\in \Hom_{\ED} 
    \right\}.
    \label{eq:tubebasis-R}
\end{equation}
We also have the embedding:  $\tilde{\mathbf{R}}({\ED_{\ED}})\cong \mathbf{R}({\ED_{\ED}}) \hookrightarrow \mathbf{Tube}({_{\ED}}\ED_{\ED})$; see Sec.~\ref{sec:bdtheory} for the definition of the right boundary tube algebra $\tilde{\mathbf{R}}({\ED_{\ED}})$.

We can introduce a crossed product between $\mathbf{R}(\ED_{\ED})$ and $\mathbf{L}(_{\ED}\ED)$. It is noteworthy to emphasize the strong similarity with the concept of the quantum double of a given weak Hopf algebra $W$ (see Definition~\ref{def:QuantumDouble}): $D(W)=(\hat{W}^{\rm cop}\otimes W)/J$, where the quantum double is derived from $W$ and its dual $\hat{W}$ through a crossed product.
Here, the crossed product is defined by 
\begin{equation}
    \Join \left(  
           \begin{aligned}
        \begin{tikzpicture}
             \draw[line width=.6pt,black] (0,0.5)--(0,1.5);
             \draw[line width=.6pt,black] (0,-0.5)--(0,-1.5);
             \draw[red,dotted] (0,0.8) arc[start angle=90, end angle=270, radius=0.8];
             \draw[blue] (0,1.3) arc[start angle=90, end angle=-90, radius=1.3];
            \node[ line width=0.6pt, dashed, draw opacity=0.5] (a) at (0,1.7){$h$};
             \node[ line width=0.6pt, dashed, draw opacity=0.5] (a) at (0,-1.7){$c$};
            \node[ line width=0.6pt, dashed, draw opacity=0.5] (a) at (-1,0){$\one$};
            \node[ line width=0.6pt, dashed, draw opacity=0.5] (a) at (1.5,0){$b$};
            \node[ line width=0.6pt, dashed, draw opacity=0.5] (a) at (0,-0.3){$d$};
            \node[ line width=0.6pt, dashed, draw opacity=0.5] (a) at (-0.4,-1.3){$\mu$};
            \node[ line width=0.6pt, dashed, draw opacity=0.5] (a) at (0,0.3){$g$};
            \node[ line width=0.6pt, dashed, draw opacity=0.5] (a) at (-0.4,1.3){$\gamma$};
        \end{tikzpicture}
    \end{aligned}
    \otimes 
    \begin{aligned}
        \begin{tikzpicture}
             \draw[line width=.6pt,black] (0,0.5)--(0,1.5);
             \draw[line width=.6pt,black] (0,-0.5)--(0,-1.5);
             \draw[red] (0,0.8) arc[start angle=90, end angle=270, radius=0.8];
             \draw[blue,dotted] (0,1.3) arc[start angle=90, end angle=-90, radius=1.3];
            \node[ line width=0.6pt, dashed, draw opacity=0.5] (a) at (-1,0){$a$};
            \node[ line width=0.6pt, dashed, draw opacity=0.5] (a) at (1.5,0){$\one$};
            \node[ line width=0.6pt, dashed, draw opacity=0.5] (a) at (-0.2,-1.3){$d'$};
            \node[ line width=0.6pt, dashed, draw opacity=0.5] (a) at (0.2,-0.8){$\nu$};
            \node[ line width=0.6pt, dashed, draw opacity=0.5] (a) at (0,-0.3){$e$};
            \node[ line width=0.6pt, dashed, draw opacity=0.5] (a) at (0,0.3){$f$};
            \node[ line width=0.6pt, dashed, draw opacity=0.5] (a) at (-0.2,1.3){$g'$};
            \node[ line width=0.6pt, dashed, draw opacity=0.5] (a) at (0.2,0.8){$\zeta$};
        \end{tikzpicture}
    \end{aligned}
   \right)
    = \delta_{g,g'}\delta_{d,d'}\begin{aligned}\begin{tikzpicture}
             \draw[line width=.6pt,black] (0,0.5)--(0,1.5);
             \draw[line width=.6pt,black] (0,-0.5)--(0,-1.5);
             \draw[red] (0,0.8) arc[start angle=90, end angle=270, radius=0.8];
             \draw[blue] (0,1.3) arc[start angle=90, end angle=-90, radius=1.3];
            \node[ line width=0.6pt, dashed, draw opacity=0.5] (a) at (0,1.7){$h$};
             \node[ line width=0.6pt, dashed, draw opacity=0.5] (a) at (0,-1.7){$c$};
            \node[ line width=0.6pt, dashed, draw opacity=0.5] (a) at (-1,0){$a$};
            \node[ line width=0.6pt, dashed, draw opacity=0.5] (a) at (1.5,0){$b$};
            \node[ line width=0.6pt, dashed, draw opacity=0.5] (a) at (-0.2,-1){$d$};
            \node[ line width=0.6pt, dashed, draw opacity=0.5] (a) at (-0.4,-1.3){$\mu$};
            \node[ line width=0.6pt, dashed, draw opacity=0.5] (a) at (0.2,-0.8){$\nu$};
            \node[ line width=0.6pt, dashed, draw opacity=0.5] (a) at (0,-0.3){$e$};
            \node[ line width=0.6pt, dashed, draw opacity=0.5] (a) at (0,0.3){$f$};
            \node[ line width=0.6pt, dashed, draw opacity=0.5] (a) at (-0.2,1){$g$};
            \node[ line width=0.6pt, dashed, draw opacity=0.5] (a) at (-0.4,1.3){$\gamma$};
            \node[ line width=0.6pt, dashed, draw opacity=0.5] (a) at (0.2,0.8){$\zeta$};
        \end{tikzpicture}
    \end{aligned}.
\end{equation}
Notice the crucial order here: if we take $\mathbf{R}(\ED_{\ED}) \Join \mathbf{L}(_{\ED}\ED)$, the multiplication is inherited directly from the tube diagrams, obviating the necessity for additional F-move operations.
It is easy to check that for $X\in \mathbf{R}(\ED_{\ED})$, $Y\in\mathbf{L}(_{\ED}\ED)$, $\Delta(X\Join Y)=\Delta(X)\Join \Delta(Y)$.
Using the natural embedding of the boundary tube algebras into the bulk tube algebra, the above construction implies that we can obtain the bulk tube algebra by taking crossed product of two boundary tube algebras $\tilde{\mathbf{L}}({_{\ED}\ED})$ and $\tilde{\mathbf{R}}({\ED_{\ED}})$.
To summarize, we have the following. 

\begin{proposition}
    The bulk tube algebra $\mathbf{Tube}({_{\ED}}\ED_{\ED})$  is isomorphic to the crossed product $\tilde{\mathbf{R}}(\ED_{\ED}) \Join \tilde{\mathbf{L}}({_{\ED}}\ED)$.
\end{proposition}

\begin{remark}
Noticeably, we can also view $\tilde{\mathbf{L}}({_{\ED}\ED})$ as the dual of $\tilde{\mathbf{R}}({\ED_{\ED}})$, and vice versa, by embedding them into bulk diagrams. This will be illustrated in detail in Sec.~\ref{sec:bdtheory}.    
\end{remark}

\subsection{Topological excitation and representation of tube algebra}
\label{subsec:FuntorExcitation}

It is pertinent to remember that our tube algebra is derived from the trivial domain wall $_{\ED}\ED_{\ED}$. From the macroscopic theory of multifusion string-net, we know that the excitations associated with domain walls are defined by the category of $\ED|\ED$-bimodule functors, denoted as $\Fun_{\ED|\ED}(\ED,\ED)$ (it is braided monoidal equivalent to the Drinfeld center $\mathcal{Z}(\ED)$).
The tensor product of two functors $F,G$ is defined as $F\otimes G:=G\comp F$, namely, for $M\in {_{\ED}}\ED_{\ED}$, we have $(F\otimes G)(M)=G\comp F(M)$.

For a bulk excitation $F\in \Fun_{\ED|\ED}(\ED,\ED)$, we can construct a corresponding lattice model in the following way. We introduce the graded vector spaces
\begin{equation}
    V_F:=\bigoplus_{x,y\in \Irr({_{\ED}}\ED_{\ED})} \Hom_{{_{\ED}}\ED_{\ED}}(F(x),y),\quad
    V^F:=\bigoplus_{x,y\in \Irr({_{\ED}}\ED_{\ED})} \Hom_{{_{\ED}}\ED_{\ED}}(x,F(y)).
\end{equation}
The corresponding bases of the two spaces can be diagrammatically represented as follows:
\begin{equation}
    \begin{aligned}
    \begin{tikzpicture}
     \fill [Apricot] (-1.5,-0.1) rectangle (0,1.5);
      \fill [GreenYellow] (1.5,-0.1) rectangle (0,1.5);
        \draw[line width=1.4pt,black] (0,-0.1) -- (0,1.5); 
        \draw[line width=1.4pt,black] (0,-0.1) -- (0,1.5); 
        \draw[-latex,line width=1.4pt,black] (0,0.3) -- (0,0.5);
        \draw[-latex,line width=1.4pt,black] (0,1) -- (0,1.2);
       \draw[line width=0.6pt,red] (0,0.7) -- (0.6,0.7); 
        \draw[-latex,line width=0.6pt,red] (0.6,0.7) -- (0.2,0.7); 
        \node[ line width=0.6pt, dashed, draw opacity=0.5] (a) at (0.4,0.9){$F$};
        \node[ line width=0.6pt, dashed, draw opacity=0.5] (a) at (0.2,-0.1){$x$};
        \node[ line width=0.6pt, dashed, draw opacity=0.5] (a) at (0.2,1.5){$y$};
        \node[ line width=0.6pt, dashed, draw opacity=0.5] (a) at (-0.25,0.8){$\alpha$};
    \end{tikzpicture}
\end{aligned}\in V_F,\quad
    \begin{aligned}
    \begin{tikzpicture}
   \fill [Apricot] (-1.5,-0.1) rectangle (0,1.5);
      \fill [GreenYellow] (1.5,-0.1) rectangle (0,1.5);
        \draw[line width=1.4pt,black] (0,-0.1) -- (0,1.5); 
        \draw[-latex,line width=1.4pt,black] (0,0.3) -- (0,0.5);
        \draw[-latex,line width=1.4pt,black] (0,1) -- (0,1.2);
       \draw[line width=0.6pt,red] (0,0.7) -- (0.6,0.7); 
        \draw[-latex,line width=0.6pt,red] (0,0.7) -- (0.4,0.7); 
        \node[ line width=0.6pt, dashed, draw opacity=0.5] (a) at (0.4,0.9){$F$};
        \node[ line width=0.6pt, dashed, draw opacity=0.5] (a) at (0.2,-0.1){$x$};
        \node[ line width=0.6pt, dashed, draw opacity=0.5] (a) at (0.2,1.5){$y$};
        \node[ line width=0.6pt, dashed, draw opacity=0.5] (a) at (-0.25,0.8){$\beta$};
    \end{tikzpicture}
\end{aligned} \in V^F,
\end{equation}
where we represent the functor as a directed dangling edge.
To construct the lattice Hamiltonian, we consider the region surrounding the excitation:
\begin{equation}
      \begin{aligned}
    \begin{tikzpicture}
   \fill [Apricot] (-2,-1) rectangle (0,3);
      \fill [GreenYellow] (2,-1) rectangle (0,3);
       \draw[line width=0.6pt,red] (0,0.9) -- (0.6,0.9); 
         \draw[line width=1.4pt,black] (0,-1) -- (0,3); 
        \draw[line width=1.4pt,black] (0,1.5) -- (1.6,1); 
        \draw[line width=1.4pt,black] (1.6,0) -- (1.6,1); 
        \draw[line width=1.4pt,black] (1.6,1) -- (2,1.4); 
        \draw[line width=1.4pt,black] (1.6,0) -- (0,-0.5); 
        \draw[line width=1.4pt,black] (1.6,0) -- (2,-0.4); 
        \draw[line width=1.4pt,black] (0,0.5) -- (-1.6,1); 
        \draw[line width=1.4pt,black] (-1.6,2) -- (-1.6,1); 
        \draw[line width=1.4pt,black] (-1.6,2) -- (-2,2.4); 
        \draw[line width=1.4pt,black] (-1.6,2) -- (0,2.5); 
        \draw[line width=1.4pt,black] (-1.6,1) -- (-2,0.6); 
        \node[ line width=0.6pt, dashed, draw opacity=0.5] (a) at (0.8,0.7){$F$};
        \node[ line width=0.6pt, dashed, draw opacity=0.5] (a) at (0.2,-0.1){$6$};
        \node[ line width=0.6pt, dashed, draw opacity=0.5] (a) at (0.2,1.7){$3$};
        \node[ line width=0.6pt, dashed, draw opacity=0.5] (a) at (-0.25,1){$9$};
       \node[ line width=0.6pt, dashed, draw opacity=0.5] (a) at (-0.25,0.2){$7$};
       \node[ line width=0.6pt, dashed, draw opacity=0.5] (a) at (0.25,2.6){$2$}; 
       \node[ line width=0.6pt, dashed, draw opacity=0.5] (a) at (-1.6,2.3){$1$}; 
      \node[ line width=0.6pt, dashed, draw opacity=0.5] (a) at (-1.6,0.4){$8$}; 
      \node[ line width=0.6pt, dashed, draw opacity=0.5] (a) at (1.6,1.3){$4$}; 
      \node[ line width=0.6pt, dashed, draw opacity=0.5] (a) at (1.6,-0.3){$5$}; 
     \node[ line width=0.6pt, dashed, draw opacity=0.5] (a) at (-0.8,1.6){$f_L$};
     \node[ line width=0.6pt, dashed, draw opacity=0.5] (a) at (0.8,0.2){$f_R$}; 
    \end{tikzpicture}
\end{aligned}\label{eq:DefectFace}
\end{equation}
We can introduce the vertex stabilizer $Q_{v_9}:V_F\to \tilde{V}_F$ as a projector whose support space is the direct sum of the spaces such that $\Hom_{{_{\ED}}\ED_{\ED}}(F(x),y)\neq 0$, which we denote as $\tilde{V}_F$ (a similar definition works for $\tilde{V}^F$). The face operators $B_{f_L}, B_{f_R}$ can be defined in a similar way as that for the bulk. Thus, we obtain the Hamiltonian:
\begin{equation}
    H[\Sigma_{\rm ext},F]=-\sum_{i=1}^9 Q_{v_i}-B_{f_L}-B_{f_R}.
\end{equation}

It is clear that $V_F$ and $V^F$ form modules over the tube algebra, where the action is defined as in Eq.~\eqref{eq:TubeModule}.
The tube algebra assumes the role of a local algebra, with the domain wall excitations serving as representations of this algebra.

\subsection{Morita theory for the tube algebra}
\label{sec:TubeMorita}
In this subsection, we denote the tube algebra $\mathbf{Tube}({_{\ED}}\ED_{\ED})$ as $\mathbf{T}^{0,0;0,0}$. The significance of the superscripts will be elucidated later. Building upon our earlier discussion, we observe that the topological excitations of the multifusion string-net model are delineated by the representation category of the tube algebra $\mathbf{T}^{0,0;0,0}$. It is noteworthy that $\mathbf{T}^{0,0;0,0}$ represents the most basic tube obtained through local moves. This prompts an inquiry into the relationship between the algebras before and after these topological local moves.
The answer is that they are \emph{Morita equivalent}.

Recall that two algebras $\mathcal{A}$ and $\mathcal{B}$ are called  Morita equivalent if and only if their module categories $_{\Acal}\Mod$ and $_{\Bcal}\Mod$ are equivalent.
It can be proved that $\mathcal{A}$ and $\mathcal{B}$  are 
Morita equivalent if and only if there exist $\Acal|\Bcal$-bimodule $M$ and $\Bcal|\Acal$-bimodule $N$ such that $M\otimes_{\Bcal}N\cong \Acal$ as $\Acal|\Acal$-bimodules and $N\otimes_{\Acal}M\cong \Bcal$ as $\Bcal|\Bcal$-bimodules.

Let us consider the general tube space $\mathbf{T}^{m_0,m_1;n_0,n_1}$ spanned by the tube string-net configurations   :
\begin{equation}
 \mathbf{T}^{m_0,m_1;n_0,n_1}=\operatorname{span} \left\{      \begin{aligned}
        \begin{tikzpicture}
             \draw[line width=.6pt,black] (0,0.5)--(0,1.5);
             \draw[line width=.6pt,black] (0,-0.5)--(0,-1.5);
             \draw[red,line width=.6pt] (120:0.8) -- (120:1.1);
             \draw[red,line width=.6pt] (150:0.8) -- (150:0.5);
            \draw[red,line width=.6pt] (180:0.8) -- (180:1.1);
            \draw[red,line width=.6pt] (210:0.8) -- (210:0.5);
          \draw[red,line width=.6pt] (240:0.8) -- (240:1.1);
             \draw[red] (0,0.8) arc[start angle=90, end angle=270, radius=0.8];
             \draw[blue] (0,1.3) arc[start angle=90, end angle=-90, radius=1.3];
            \node[red, line width=0.6pt, dashed, draw opacity=0.5] (a) at (225:0.6)
            {\small{$\cdot$}};
            \node[red, line width=0.6pt, dashed, draw opacity=0.5] (a) at (240:0.6){\small{$\cdot$}};
          \node[red, line width=0.6pt, dashed, draw opacity=0.5] (a) at (255:0.6){\small{$\cdot$}};
             \node[red, line width=0.6pt, dashed, draw opacity=0.5] (a) at (195:1)
            {\small{$\cdot$}};
            \node[red, line width=0.6pt, dashed, draw opacity=0.5] (a) at (210:1){\small{$\cdot$}};
          \node[red, line width=0.6pt, dashed, draw opacity=0.5] (a) at (225:1){\small{$\cdot$}};
            \draw[blue,line width=.6pt] (300:1.3) -- (300:1.6);
             \draw[blue,line width=.6pt] (330:1) -- (330:1.3);
            \draw[blue,line width=.6pt] (360:1.3) -- (360:1.6);
            \draw[blue,line width=.6pt] (30:1) -- (30:1.3);
             \draw[blue,line width=.6pt] (60:1.6) -- (60:1.3);
          \node[blue, line width=0.6pt, dashed, draw opacity=0.5] (a) at (75:1)
            {\small{$\cdot$}};
            \node[blue, line width=0.6pt, dashed, draw opacity=0.5] (a) at (60:1){\small{$\cdot$}};
          \node[blue, line width=0.6pt, dashed, draw opacity=0.5] (a) at (45:1){\small{$\cdot$}};
         \node[blue, line width=0.6pt, dashed, draw opacity=0.5] (a) at (45:1.45)
            {\small{$\cdot$}};
            \node[blue, line width=0.6pt, dashed, draw opacity=0.5] (a) at (30:1.45){\small{$\cdot$}};
          \node[blue, line width=0.6pt, dashed, draw opacity=0.5] (a) at (15:1.45){\small{$\cdot$}};
          \node[blue, line width=0.6pt, dashed, draw opacity=0.5] (a) at (30:1.85)
            {$n_1$};
          \node[blue, line width=0.6pt, dashed, draw opacity=0.5] (a) at (60:0.65)
            {$n_0$};
              \node[red, line width=0.6pt, dashed, draw opacity=0.5] (a) at (210:1.35)
            {$m_1$};
          \node[red, line width=0.6pt, dashed, draw opacity=0.5] (a) at (250:0.35)
            {$m_0$};
        \end{tikzpicture}
    \end{aligned}: \text{edge}\in \Irr,\text{vertex}\in \Hom\right\},
\end{equation}
where $m_0$ and $m_1$ represent the number of red internal and external edges, $n_0$ and $n_1$ represent the number of blue internal and external edges.
An algebraic structure can be endowed upon $\mathbf{T}^{m,m;n,n}$ by merging two basis diagrams and applying topological local moves. More precisely, the multiplication is given as 
\begin{align}
    \mu&\left(\begin{aligned}
        \begin{tikzpicture}
             \draw[line width=.6pt,black] (0,0.5)--(0,1.5);
             \draw[line width=.6pt,black] (0,-0.5)--(0,-1.5);
             \draw[red,line width=.6pt] (120:0.8) -- (120:1.1);
             \draw[red,line width=.6pt] (150:0.8) -- (150:0.5);
            \draw[red,line width=.6pt] (180:0.8) -- (180:1.1);
            \draw[red,line width=.6pt] (210:0.8) -- (210:0.5);
          \draw[red,line width=.6pt] (240:0.8) -- (240:1.1);
             \draw[red] (0,0.8) arc[start angle=90, end angle=270, radius=0.8];
             \draw[blue] (0,1.3) arc[start angle=90, end angle=-90, radius=1.3];
            \node[red, line width=0.6pt, dashed, draw opacity=0.5] (a) at (225:0.6)
            {\small{$\cdot$}};
            \node[red, line width=0.6pt, dashed, draw opacity=0.5] (a) at (240:0.6){\small{$\cdot$}};
          \node[red, line width=0.6pt, dashed, draw opacity=0.5] (a) at (255:0.6){\small{$\cdot$}};
             \node[red, line width=0.6pt, dashed, draw opacity=0.5] (a) at (195:1)
            {\small{$\cdot$}};
            \node[red, line width=0.6pt, dashed, draw opacity=0.5] (a) at (210:1){\small{$\cdot$}};
          \node[red, line width=0.6pt, dashed, draw opacity=0.5] (a) at (225:1){\small{$\cdot$}};
            \draw[blue,line width=.6pt] (300:1.3) -- (300:1.6);
             \draw[blue,line width=.6pt] (330:1) -- (330:1.3);
            \draw[blue,line width=.6pt] (360:1.3) -- (360:1.6);
            \draw[blue,line width=.6pt] (30:1) -- (30:1.3);
             \draw[blue,line width=.6pt] (60:1.6) -- (60:1.3);
          \node[blue, line width=0.6pt, dashed, draw opacity=0.5] (a) at (75:1)
            {\small{$\cdot$}};
            \node[blue, line width=0.6pt, dashed, draw opacity=0.5] (a) at (60:1){\small{$\cdot$}};
          \node[blue, line width=0.6pt, dashed, draw opacity=0.5] (a) at (45:1){\small{$\cdot$}};
         \node[blue, line width=0.6pt, dashed, draw opacity=0.5] (a) at (45:1.45)
            {\small{$\cdot$}};
            \node[blue, line width=0.6pt, dashed, draw opacity=0.5] (a) at (30:1.45){\small{$\cdot$}};
          \node[blue, line width=0.6pt, dashed, draw opacity=0.5] (a) at (15:1.45){\small{$\cdot$}};
          \node[blue, line width=0.6pt, dashed, draw opacity=0.5] (a) at (60:1.85)
            {$b_1$};
            \node[blue, line width=0.6pt, dashed, draw opacity=0.5] (a) at (300:1.9)
            {$b_n$};
            \node[blue, line width=0.6pt, dashed, draw opacity=0.5] (a) at (0:2.1)
            {$b_{n-1}$};
          \node[blue, line width=0.6pt, dashed, draw opacity=0.5] (a) at (330:0.65)
            {$d_n$};
              \node[red, line width=0.6pt, dashed, draw opacity=0.5] (a) at (120:1.4)
            {$a_1$};
            \node[red, line width=0.6pt, dashed, draw opacity=0.5] (a) at (180:1.4)
            {$a_2$};
            \node[red, line width=0.6pt, dashed, draw opacity=0.5] (a) at (240:1.3)
            {$a_m$};
          \node[red, line width=0.6pt, dashed, draw opacity=0.5] (a) at (140:0.25)
            {$c_1$};
            \node[line width=0.6pt, dashed, draw opacity=0.5] (a) at (-90:1.7) {$y$};
            \node[line width=0.6pt, dashed, draw opacity=0.5] (a) at (-80:0.3)
            {$v$};
            \node[line width=0.6pt, dashed, draw opacity=0.5] (a) at (80:0.3) {$u$};
            \node[line width=0.6pt, dashed, draw opacity=0.5] (a) at (90:1.75) {$x$};
        \end{tikzpicture}
    \end{aligned}
    \otimes \begin{aligned}
        \begin{tikzpicture}
             \draw[line width=.6pt,black] (0,0.5)--(0,1.5);
             \draw[line width=.6pt,black] (0,-0.5)--(0,-1.5);
             \draw[red,line width=.6pt] (120:0.8) -- (120:1.1);
             \draw[red,line width=.6pt] (150:0.8) -- (150:0.5);
            \draw[red,line width=.6pt] (180:0.8) -- (180:1.1);
            \draw[red,line width=.6pt] (210:0.8) -- (210:0.5);
          \draw[red,line width=.6pt] (240:0.8) -- (240:1.1);
             \draw[red] (0,0.8) arc[start angle=90, end angle=270, radius=0.8];
             \draw[blue] (0,1.3) arc[start angle=90, end angle=-90, radius=1.3];
            \node[red, line width=0.6pt, dashed, draw opacity=0.5] (a) at (225:0.6)
            {\small{$\cdot$}};
            \node[red, line width=0.6pt, dashed, draw opacity=0.5] (a) at (240:0.6){\small{$\cdot$}};
          \node[red, line width=0.6pt, dashed, draw opacity=0.5] (a) at (255:0.6){\small{$\cdot$}};
             \node[red, line width=0.6pt, dashed, draw opacity=0.5] (a) at (195:1)
            {\small{$\cdot$}};
            \node[red, line width=0.6pt, dashed, draw opacity=0.5] (a) at (210:1){\small{$\cdot$}};
          \node[red, line width=0.6pt, dashed, draw opacity=0.5] (a) at (225:1){\small{$\cdot$}};
            \draw[blue,line width=.6pt] (300:1.3) -- (300:1.6);
             \draw[blue,line width=.6pt] (330:1) -- (330:1.3);
            \draw[blue,line width=.6pt] (360:1.3) -- (360:1.6);
            \draw[blue,line width=.6pt] (30:1) -- (30:1.3);
             \draw[blue,line width=.6pt] (60:1.6) -- (60:1.3);
          \node[blue, line width=0.6pt, dashed, draw opacity=0.5] (a) at (75:1)
            {\small{$\cdot$}};
            \node[blue, line width=0.6pt, dashed, draw opacity=0.5] (a) at (60:1){\small{$\cdot$}};
          \node[blue, line width=0.6pt, dashed, draw opacity=0.5] (a) at (45:1){\small{$\cdot$}};
         \node[blue, line width=0.6pt, dashed, draw opacity=0.5] (a) at (45:1.45)
            {\small{$\cdot$}};
            \node[blue, line width=0.6pt, dashed, draw opacity=0.5] (a) at (30:1.45){\small{$\cdot$}};
          \node[blue, line width=0.6pt, dashed, draw opacity=0.5] (a) at (15:1.45){\small{$\cdot$}};
          \node[blue, line width=0.6pt, dashed, draw opacity=0.5] (a) at (60:1.85)
            {$b_1'$};
            \node[blue, line width=0.6pt, dashed, draw opacity=0.5] (a) at (300:1.9)
            {$b_n'$};
            \node[blue, line width=0.6pt, dashed, draw opacity=0.5] (a) at (0:2.1)
            {$b_{n-1}'$};
          \node[blue, line width=0.6pt, dashed, draw opacity=0.5] (a) at (330:0.65)
            {$d_n'$};
              \node[red, line width=0.6pt, dashed, draw opacity=0.5] (a) at (120:1.4)
            {$a_1'$};
            \node[red, line width=0.6pt, dashed, draw opacity=0.5] (a) at (180:1.4)
            {$a_2'$};
            \node[red, line width=0.6pt, dashed, draw opacity=0.5] (a) at (240:1.3)
            {$a_m'$};
          \node[red, line width=0.6pt, dashed, draw opacity=0.5] (a) at (140:0.25)
            {$c_1'$};
        \node[line width=0.6pt, dashed, draw opacity=0.5] (a) at (-90:1.7) {$y'$};
        \node[line width=0.6pt, dashed, draw opacity=0.5] (a) at (-70:0.3)
        {$v'$};
        \node[line width=0.6pt, dashed, draw opacity=0.5] (a) at (70:0.4) {$u'$};
        \node[line width=0.6pt, dashed, draw opacity=0.5] (a) at (90:1.75) {$x'$};
        \end{tikzpicture}
    \end{aligned}
    \right) \nonumber \\
    & = \delta_{u,x'}\delta_{v,y'}\prod_{i=1}^m\delta_{c_i,a'_i}\prod_{j=1}^n\delta_{d_j,b_j'} \; 
    \begin{aligned}
        \begin{tikzpicture}
             \draw[line width=.6pt,black] (0,0.5)--(0,2.5);
             \draw[line width=.6pt,black] (0,-0.5)--(0,-2.5);
             \draw[red] (0,0.8) arc[start angle=90, end angle=270, radius=0.8];
             \draw[blue] (0,1.3) arc[start angle=90, end angle=-90, radius=1.3];
            \draw[red] (0,1.6) arc[start angle=90, end angle=270, radius=1.6];
           \draw[blue] (0,2.1) arc[start angle=90, end angle=-90, radius=2.1];
           \draw[red,line width=.6pt] (120:1.6) -- (120:1.9);
           \draw[red,line width=.6pt] (120:0.5) -- (120:0.8);
           \draw[red,line width=.6pt] (150:0.8) -- (150:1.6);
           \draw[red,line width=.6pt] (180:1.6) -- (180:1.9);
           \draw[red,line width=.6pt] (150:0.8) -- (150:1.6);
           \draw[red,line width=.6pt] (210:0.8) -- (210:1.6);
           \draw[red,line width=.6pt] (240:1.6) -- (240:1.9);
           \draw[red,line width=.6pt] (180:0.5) -- (180:0.8);
           \node[red, line width=0.6pt, dashed, draw opacity=0.5] (a) at (210:1.9){\small{$\cdot$}};
           \node[red, line width=0.6pt, dashed, draw opacity=0.5] (a) at (200:1.9){\small{$\cdot$}};
           \node[red, line width=0.6pt, dashed, draw opacity=0.5] (a) at (220:1.9){\small{$\cdot$}};
           \node[red, line width=0.6pt, dashed, draw opacity=0.5] (a) at (240:1.2){\small{$\cdot$}};
           \node[red, line width=0.6pt, dashed, draw opacity=0.5] (a) at (250:1.2){\small{$\cdot$}};
           \node[red, line width=0.6pt, dashed, draw opacity=0.5] (a) at (230:1.2){\small{$\cdot$}};
           \node[red, line width=0.6pt, dashed, draw opacity=0.5] (a) at (240:0.6){\small{$\cdot$}};
           \node[red, line width=0.6pt, dashed, draw opacity=0.5] (a) at (255:0.6){\small{$\cdot$}};
           \node[red, line width=0.6pt, dashed, draw opacity=0.5] (a) at (225:0.6){\small{$\cdot$}};
           \node[red, line width=0.6pt, dashed, draw opacity=0.5] (a) at (120:2.2)
            {$a_1$};
            \node[red, line width=0.6pt, dashed, draw opacity=0.5] (a) at (180:2.2)
            {$a_2$};
            \node[red, line width=0.6pt, dashed, draw opacity=0.5] (a) at (240:2.2)
            {$a_m$};
            \node[red, line width=0.6pt, dashed, draw opacity=0.5] (a) at (135:1.2)
            {$c_1$};
            \node[red, line width=0.6pt, dashed, draw opacity=0.5] (a) at (195:1.2)
            {$c_2$};
            \node[red, line width=0.6pt, dashed, draw opacity=0.5] (a) at (120:0.3)
            {$c_1'$};
            \draw[blue,line width=.6pt] (60:2.1) -- (60:2.4);
            \draw[blue,line width=.6pt] (30:1.3) -- (30:2.1);
            \draw[blue,line width=.6pt] (0:2.1) -- (0:2.4);
            \draw[blue,line width=.6pt] (0:1) -- (0:1.3);
            \draw[blue,line width=.6pt] (330:1.3) -- (330:2.1);
            \draw[blue,line width=.6pt] (300:2.1) -- (300:2.4);
            \draw[blue,line width=.6pt] (300:1) -- (300:1.3);
            \node[blue, line width=0.6pt, dashed, draw opacity=0.5] (a) at (60:2.7)
            {$b_1$};
            \node[blue, line width=0.6pt, dashed, draw opacity=0.5] (a) at (0:2.9)
            {$b_{n-1}$};
            \node[blue, line width=0.6pt, dashed, draw opacity=0.5] (a) at (300:2.7)
            {$b_{n}$};
            \node[blue, line width=0.6pt, dashed, draw opacity=0.5] (a) at (30:2.4){\small{$\cdot$}};
            \node[blue, line width=0.6pt, dashed, draw opacity=0.5] (a) at (20:2.4){\small{$\cdot$}};
            \node[blue, line width=0.6pt, dashed, draw opacity=0.5] (a) at (40:2.4){\small{$\cdot$}};
            \node[blue, line width=0.6pt, dashed, draw opacity=0.5] (a) at (45:1){\small{$\cdot$}};
            \node[blue, line width=0.6pt, dashed, draw opacity=0.5] (a) at (60:1){\small{$\cdot$}};
            \node[blue, line width=0.6pt, dashed, draw opacity=0.5] (a) at (30:1){\small{$\cdot$}};
            \node[blue, line width=0.6pt, dashed, draw opacity=0.5] (a) at (0:1.7){\small{$\cdot$}};
            \node[blue, line width=0.6pt, dashed, draw opacity=0.5] (a) at (10:1.7){\small{$\cdot$}};
            \node[blue, line width=0.6pt, dashed, draw opacity=0.5] (a) at (350:1.7){\small{$\cdot$}};
            \node[blue, line width=0.6pt, dashed, draw opacity=0.5] (a) at (40:1.7)
            {$d_1$};
            \node[blue, line width=0.6pt, dashed, draw opacity=0.5] (a) at (320:1.7)
            {$d_n$};
            \node[blue, line width=0.6pt, dashed, draw opacity=0.5] (a) at (300:0.7)
            {$d_n'$};
            \node[line width=0.6pt, dashed, draw opacity=0.5] (a) at (-95:2.4) {$y$};
            \node[line width=0.6pt, dashed, draw opacity=0.5] (a) at (-96:1.4)
            {$v$};
            \node[line width=0.6pt, dashed, draw opacity=0.5] (a) at (96:1.4) {$u$};
            \node[line width=0.6pt, dashed, draw opacity=0.5] (a) at (-70:0.3)
        {$v'$};
        \node[line width=0.6pt, dashed, draw opacity=0.5] (a) at (70:0.4) {$u'$};
            \node[line width=0.6pt, dashed, draw opacity=0.5] (a) at (95:2.4) {$x$};
        \end{tikzpicture}
    \end{aligned}\;. \label{eq:tube-space-prod}
\end{align}
Note that we have omitted the morphism labels and object labels around the arcs for simplicity. For the gluing process to be meaningful, the number of internal edges in the first diagram must match the number of external edges in the second diagram. Hence, we set $m_0=m_1=m$ and $n_0=n_1=n$ to ensure compatibility. It is clear that this operation is associative by the diagrams. One can easily show that the unit is of the form 
\begin{equation}
    1=\sum_{a_1,\cdots,b_n,c,d}\begin{aligned}
        \begin{tikzpicture}
             \draw[line width=.6pt,black] (0,0.5)--(0,1.5);
             \draw[line width=.6pt,black] (0,-0.5)--(0,-1.5);
             \draw[red,line width=.6pt] (120:0.8) -- (120:1.1);
             \draw[red,line width=.6pt] (150:0.8) -- (150:0.5);
            \draw[red,line width=.6pt] (180:0.8) -- (180:1.1);
            \draw[red,line width=.6pt] (210:0.8) -- (210:0.5);
          \draw[red,line width=.6pt] (240:0.8) -- (240:1.1);
             \draw[red,dotted] (0,0.8) arc[start angle=90, end angle=270, radius=0.8];
             \draw[blue,dotted] (0,1.3) arc[start angle=90, end angle=-90, radius=1.3];
            \node[red, line width=0.6pt, dashed, draw opacity=0.5] (a) at (225:0.6)
            {\small{$\cdot$}};
            \node[red, line width=0.6pt, dashed, draw opacity=0.5] (a) at (240:0.6){\small{$\cdot$}};
          \node[red, line width=0.6pt, dashed, draw opacity=0.5] (a) at (255:0.6){\small{$\cdot$}};
             \node[red, line width=0.6pt, dashed, draw opacity=0.5] (a) at (195:1)
            {\small{$\cdot$}};
            \node[red, line width=0.6pt, dashed, draw opacity=0.5] (a) at (210:1){\small{$\cdot$}};
          \node[red, line width=0.6pt, dashed, draw opacity=0.5] (a) at (225:1){\small{$\cdot$}};
            \draw[blue,line width=.6pt] (300:1.3) -- (300:1.6);
             \draw[blue,line width=.6pt] (330:1) -- (330:1.3);
            \draw[blue,line width=.6pt] (360:1.3) -- (360:1.6);
            \draw[blue,line width=.6pt] (30:1) -- (30:1.3);
             \draw[blue,line width=.6pt] (60:1.6) -- (60:1.3);
          \node[blue, line width=0.6pt, dashed, draw opacity=0.5] (a) at (75:1)
            {\small{$\cdot$}};
            \node[blue, line width=0.6pt, dashed, draw opacity=0.5] (a) at (60:1){\small{$\cdot$}};
          \node[blue, line width=0.6pt, dashed, draw opacity=0.5] (a) at (45:1){\small{$\cdot$}};
         \node[blue, line width=0.6pt, dashed, draw opacity=0.5] (a) at (45:1.45)
            {\small{$\cdot$}};
            \node[blue, line width=0.6pt, dashed, draw opacity=0.5] (a) at (30:1.45){\small{$\cdot$}};
          \node[blue, line width=0.6pt, dashed, draw opacity=0.5] (a) at (15:1.45){\small{$\cdot$}};
          \node[blue, line width=0.6pt, dashed, draw opacity=0.5] (a) at (60:1.85)
            {$b_1$};
            \node[blue, line width=0.6pt, dashed, draw opacity=0.5] (a) at (300:1.9)
            {$b_n$};
            \node[blue, line width=0.6pt, dashed, draw opacity=0.5] (a) at (0:2.1)
            {$b_{n-1}$};
          \node[blue, line width=0.6pt, dashed, draw opacity=0.5] (a) at (330:0.65)
            {$b_n$};
              \node[red, line width=0.6pt, dashed, draw opacity=0.5] (a) at (120:1.4)
            {$a_1$};
            \node[red, line width=0.6pt, dashed, draw opacity=0.5] (a) at (180:1.4)
            {$a_2$};
            \node[red, line width=0.6pt, dashed, draw opacity=0.5] (a) at (240:1.3)
            {$a_m$};
          \node[red, line width=0.6pt, dashed, draw opacity=0.5] (a) at (140:0.25)
            {$a_1$};
            \node[black, line width=0.6pt, dashed, draw opacity=0.5] (a) at (-90:1.75)
            {$d$};
            \node[black, line width=0.6pt, dashed, draw opacity=0.5] (a) at (90:1.7)
            {$c$};
        \end{tikzpicture}
    \end{aligned}.
\end{equation}

Similarly, for $\mathbf{T}^{m,s;n,t}$, both a right $\mathbf{T}^{m,m;n,n}$-action and a left $\mathbf{T}^{s,s;t,t}$-action can be established by merging the corresponding basis diagrams and implementing the topological local moves. For example, the left action of $\mathbf{T}^{s,s;t,t}$ on 
$\mathbf{T}^{m,s;n,t}$ is represented as 
\begin{align}
    &\begin{aligned}
        \begin{tikzpicture}
             \draw[line width=.6pt,black] (0,0.5)--(0,1.5);
             \draw[line width=.6pt,black] (0,-0.5)--(0,-1.5);
             \draw[red,line width=.6pt] (120:0.8) -- (120:1.1);
             \draw[red,line width=.6pt] (150:0.8) -- (150:0.5);
            \draw[red,line width=.6pt] (180:0.8) -- (180:1.1);
            \draw[red,line width=.6pt] (210:0.8) -- (210:0.5);
          \draw[red,line width=.6pt] (240:0.8) -- (240:1.1);
             \draw[red] (0,0.8) arc[start angle=90, end angle=270, radius=0.8];
             \draw[blue] (0,1.3) arc[start angle=90, end angle=-90, radius=1.3];
            \node[red, line width=0.6pt, dashed, draw opacity=0.5] (a) at (225:0.6)
            {\small{$\cdot$}};
            \node[red, line width=0.6pt, dashed, draw opacity=0.5] (a) at (240:0.6){\small{$\cdot$}};
          \node[red, line width=0.6pt, dashed, draw opacity=0.5] (a) at (255:0.6){\small{$\cdot$}};
             \node[red, line width=0.6pt, dashed, draw opacity=0.5] (a) at (195:1)
            {\small{$\cdot$}};
            \node[red, line width=0.6pt, dashed, draw opacity=0.5] (a) at (210:1){\small{$\cdot$}};
          \node[red, line width=0.6pt, dashed, draw opacity=0.5] (a) at (225:1){\small{$\cdot$}};
            \draw[blue,line width=.6pt] (300:1.3) -- (300:1.6);
             \draw[blue,line width=.6pt] (330:1) -- (330:1.3);
            \draw[blue,line width=.6pt] (360:1.3) -- (360:1.6);
            \draw[blue,line width=.6pt] (30:1) -- (30:1.3);
             \draw[blue,line width=.6pt] (60:1.6) -- (60:1.3);
          \node[blue, line width=0.6pt, dashed, draw opacity=0.5] (a) at (75:1)
            {\small{$\cdot$}};
            \node[blue, line width=0.6pt, dashed, draw opacity=0.5] (a) at (60:1){\small{$\cdot$}};
          \node[blue, line width=0.6pt, dashed, draw opacity=0.5] (a) at (45:1){\small{$\cdot$}};
         \node[blue, line width=0.6pt, dashed, draw opacity=0.5] (a) at (45:1.45)
            {\small{$\cdot$}};
            \node[blue, line width=0.6pt, dashed, draw opacity=0.5] (a) at (30:1.45){\small{$\cdot$}};
          \node[blue, line width=0.6pt, dashed, draw opacity=0.5] (a) at (15:1.45){\small{$\cdot$}};
          \node[blue, line width=0.6pt, dashed, draw opacity=0.5] (a) at (60:1.85)
            {$b_1$};
            \node[blue, line width=0.6pt, dashed, draw opacity=0.5] (a) at (300:1.9)
            {$b_t$};
            \node[blue, line width=0.6pt, dashed, draw opacity=0.5] (a) at (0:2.1)
            {$b_{t-1}$};
          \node[blue, line width=0.6pt, dashed, draw opacity=0.5] (a) at (330:0.65)
            {$d_t$};
              \node[red, line width=0.6pt, dashed, draw opacity=0.5] (a) at (120:1.4)
            {$a_1$};
            \node[red, line width=0.6pt, dashed, draw opacity=0.5] (a) at (180:1.4)
            {$a_2$};
            \node[red, line width=0.6pt, dashed, draw opacity=0.5] (a) at (240:1.3)
            {$a_s$};
          \node[red, line width=0.6pt, dashed, draw opacity=0.5] (a) at (140:0.25)
            {$c_1$};
            \node[line width=0.6pt, dashed, draw opacity=0.5] (a) at (-90:1.7) {$y$};
            \node[line width=0.6pt, dashed, draw opacity=0.5] (a) at (-80:0.3)
            {$v$};
            \node[line width=0.6pt, dashed, draw opacity=0.5] (a) at (80:0.3) {$u$};
            \node[line width=0.6pt, dashed, draw opacity=0.5] (a) at (90:1.75) {$x$};
        \end{tikzpicture}
    \end{aligned}
    \triangleright \begin{aligned}
        \begin{tikzpicture}
             \draw[line width=.6pt,black] (0,0.5)--(0,1.5);
             \draw[line width=.6pt,black] (0,-0.5)--(0,-1.5);
             \draw[red,line width=.6pt] (120:0.8) -- (120:1.1);
             \draw[red,line width=.6pt] (150:0.8) -- (150:0.5);
            \draw[red,line width=.6pt] (180:0.8) -- (180:1.1);
            \draw[red,line width=.6pt] (210:0.8) -- (210:0.5);
          \draw[red,line width=.6pt] (240:0.8) -- (240:1.1);
             \draw[red] (0,0.8) arc[start angle=90, end angle=270, radius=0.8];
             \draw[blue] (0,1.3) arc[start angle=90, end angle=-90, radius=1.3];
            \node[red, line width=0.6pt, dashed, draw opacity=0.5] (a) at (225:0.6)
            {\small{$\cdot$}};
            \node[red, line width=0.6pt, dashed, draw opacity=0.5] (a) at (240:0.6){\small{$\cdot$}};
          \node[red, line width=0.6pt, dashed, draw opacity=0.5] (a) at (255:0.6){\small{$\cdot$}};
             \node[red, line width=0.6pt, dashed, draw opacity=0.5] (a) at (195:1)
            {\small{$\cdot$}};
            \node[red, line width=0.6pt, dashed, draw opacity=0.5] (a) at (210:1){\small{$\cdot$}};
          \node[red, line width=0.6pt, dashed, draw opacity=0.5] (a) at (225:1){\small{$\cdot$}};
            \draw[blue,line width=.6pt] (300:1.3) -- (300:1.6);
             \draw[blue,line width=.6pt] (330:1) -- (330:1.3);
            \draw[blue,line width=.6pt] (360:1.3) -- (360:1.6);
            \draw[blue,line width=.6pt] (30:1) -- (30:1.3);
             \draw[blue,line width=.6pt] (60:1.6) -- (60:1.3);
          \node[blue, line width=0.6pt, dashed, draw opacity=0.5] (a) at (75:1)
            {\small{$\cdot$}};
            \node[blue, line width=0.6pt, dashed, draw opacity=0.5] (a) at (60:1){\small{$\cdot$}};
          \node[blue, line width=0.6pt, dashed, draw opacity=0.5] (a) at (45:1){\small{$\cdot$}};
         \node[blue, line width=0.6pt, dashed, draw opacity=0.5] (a) at (45:1.45)
            {\small{$\cdot$}};
            \node[blue, line width=0.6pt, dashed, draw opacity=0.5] (a) at (30:1.45){\small{$\cdot$}};
          \node[blue, line width=0.6pt, dashed, draw opacity=0.5] (a) at (15:1.45){\small{$\cdot$}};
          \node[blue, line width=0.6pt, dashed, draw opacity=0.5] (a) at (60:1.85)
            {$b_1'$};
            \node[blue, line width=0.6pt, dashed, draw opacity=0.5] (a) at (300:1.9)
            {$b_t'$};
            \node[blue, line width=0.6pt, dashed, draw opacity=0.5] (a) at (0:2.1)
            {$b_{t-1}'$};
          \node[blue, line width=0.6pt, dashed, draw opacity=0.5] (a) at (330:0.65)
            {$d_n'$};
              \node[red, line width=0.6pt, dashed, draw opacity=0.5] (a) at (120:1.4)
            {$a_1'$};
            \node[red, line width=0.6pt, dashed, draw opacity=0.5] (a) at (180:1.4)
            {$a_2'$};
            \node[red, line width=0.6pt, dashed, draw opacity=0.5] (a) at (240:1.35)
            {$a_s'$};
          \node[red, line width=0.6pt, dashed, draw opacity=0.5] (a) at (140:0.25)
            {$c_1'$};
        \node[line width=0.6pt, dashed, draw opacity=0.5] (a) at (-90:1.7) {$y'$};
        \node[line width=0.6pt, dashed, draw opacity=0.5] (a) at (-70:0.3)
        {$v'$};
        \node[line width=0.6pt, dashed, draw opacity=0.5] (a) at (70:0.4) {$u'$};
        \node[line width=0.6pt, dashed, draw opacity=0.5] (a) at (90:1.75) {$x'$};
        \end{tikzpicture}
    \end{aligned}
     \nonumber \\
    & = \delta_{u,x'}\delta_{v,y'}\prod_{i=1}^s\delta_{c_i,a'_i}\prod_{j=1}^t\delta_{d_j,b_j'} \; 
    \begin{aligned}
        \begin{tikzpicture}
             \draw[line width=.6pt,black] (0,0.5)--(0,2.5);
             \draw[line width=.6pt,black] (0,-0.5)--(0,-2.5);
             \draw[red] (0,0.8) arc[start angle=90, end angle=270, radius=0.8];
             \draw[blue] (0,1.3) arc[start angle=90, end angle=-90, radius=1.3];
            \draw[red] (0,1.6) arc[start angle=90, end angle=270, radius=1.6];
           \draw[blue] (0,2.1) arc[start angle=90, end angle=-90, radius=2.1];
           \draw[red,line width=.6pt] (120:1.6) -- (120:1.9);
           \draw[red,line width=.6pt] (120:0.5) -- (120:0.8);
           \draw[red,line width=.6pt] (150:0.8) -- (150:1.6);
           \draw[red,line width=.6pt] (180:1.6) -- (180:1.9);
           \draw[red,line width=.6pt] (150:0.8) -- (150:1.6);
           \draw[red,line width=.6pt] (210:0.8) -- (210:1.6);
           \draw[red,line width=.6pt] (240:1.6) -- (240:1.9);
           \draw[red,line width=.6pt] (180:0.5) -- (180:0.8);
           \node[red, line width=0.6pt, dashed, draw opacity=0.5] (a) at (210:1.9){\small{$\cdot$}};
           \node[red, line width=0.6pt, dashed, draw opacity=0.5] (a) at (200:1.9){\small{$\cdot$}};
           \node[red, line width=0.6pt, dashed, draw opacity=0.5] (a) at (220:1.9){\small{$\cdot$}};
           \node[red, line width=0.6pt, dashed, draw opacity=0.5] (a) at (240:1.2){\small{$\cdot$}};
           \node[red, line width=0.6pt, dashed, draw opacity=0.5] (a) at (250:1.2){\small{$\cdot$}};
           \node[red, line width=0.6pt, dashed, draw opacity=0.5] (a) at (230:1.2){\small{$\cdot$}};
           \node[red, line width=0.6pt, dashed, draw opacity=0.5] (a) at (240:0.6){\small{$\cdot$}};
           \node[red, line width=0.6pt, dashed, draw opacity=0.5] (a) at (255:0.6){\small{$\cdot$}};
           \node[red, line width=0.6pt, dashed, draw opacity=0.5] (a) at (225:0.6){\small{$\cdot$}};
           \node[red, line width=0.6pt, dashed, draw opacity=0.5] (a) at (120:2.2)
            {$a_1$};
            \node[red, line width=0.6pt, dashed, draw opacity=0.5] (a) at (180:2.2)
            {$a_2$};
            \node[red, line width=0.6pt, dashed, draw opacity=0.5] (a) at (240:2.2)
            {$a_s$};
            \node[red, line width=0.6pt, dashed, draw opacity=0.5] (a) at (135:1.2)
            {$c_1$};
            \node[red, line width=0.6pt, dashed, draw opacity=0.5] (a) at (195:1.2)
            {$c_2$};
            \node[red, line width=0.6pt, dashed, draw opacity=0.5] (a) at (120:0.3)
            {$c_1'$};
            \draw[blue,line width=.6pt] (60:2.1) -- (60:2.4);
            \draw[blue,line width=.6pt] (30:1.3) -- (30:2.1);
            \draw[blue,line width=.6pt] (0:2.1) -- (0:2.4);
            \draw[blue,line width=.6pt] (0:1) -- (0:1.3);
            \draw[blue,line width=.6pt] (330:1.3) -- (330:2.1);
            \draw[blue,line width=.6pt] (300:2.1) -- (300:2.4);
            \draw[blue,line width=.6pt] (300:1) -- (300:1.3);
            \node[blue, line width=0.6pt, dashed, draw opacity=0.5] (a) at (60:2.7)
            {$b_1$};
            \node[blue, line width=0.6pt, dashed, draw opacity=0.5] (a) at (0:2.9)
            {$b_{t-1}$};
            \node[blue, line width=0.6pt, dashed, draw opacity=0.5] (a) at (300:2.7)
            {$b_{t}$};
            \node[blue, line width=0.6pt, dashed, draw opacity=0.5] (a) at (30:2.4){\small{$\cdot$}};
            \node[blue, line width=0.6pt, dashed, draw opacity=0.5] (a) at (20:2.4){\small{$\cdot$}};
            \node[blue, line width=0.6pt, dashed, draw opacity=0.5] (a) at (40:2.4){\small{$\cdot$}};
            \node[blue, line width=0.6pt, dashed, draw opacity=0.5] (a) at (45:1){\small{$\cdot$}};
            \node[blue, line width=0.6pt, dashed, draw opacity=0.5] (a) at (60:1){\small{$\cdot$}};
            \node[blue, line width=0.6pt, dashed, draw opacity=0.5] (a) at (30:1){\small{$\cdot$}};
            \node[blue, line width=0.6pt, dashed, draw opacity=0.5] (a) at (0:1.7){\small{$\cdot$}};
            \node[blue, line width=0.6pt, dashed, draw opacity=0.5] (a) at (10:1.7){\small{$\cdot$}};
            \node[blue, line width=0.6pt, dashed, draw opacity=0.5] (a) at (350:1.7){\small{$\cdot$}};
            \node[blue, line width=0.6pt, dashed, draw opacity=0.5] (a) at (40:1.7)
            {$d_1$};
            \node[blue, line width=0.6pt, dashed, draw opacity=0.5] (a) at (320:1.7)
            {$d_t$};
            \node[blue, line width=0.6pt, dashed, draw opacity=0.5] (a) at (300:0.7)
            {$d_n'$};
            \node[line width=0.6pt, dashed, draw opacity=0.5] (a) at (-95:2.4) {$y$};
            \node[line width=0.6pt, dashed, draw opacity=0.5] (a) at (-96:1.4)
            {$v$};
            \node[line width=0.6pt, dashed, draw opacity=0.5] (a) at (96:1.4) {$u$};
            \node[line width=0.6pt, dashed, draw opacity=0.5] (a) at (-70:0.3)
        {$v'$};
        \node[line width=0.6pt, dashed, draw opacity=0.5] (a) at (70:0.4) {$u'$};
            \node[line width=0.6pt, dashed, draw opacity=0.5] (a) at (95:2.4) {$x$};
        \end{tikzpicture}
    \end{aligned}\;.
\end{align}
The right $\mathbf{T}^{m,m;n,n}$-action on $\mathbf{T}^{m,s;n,t}$ is represented by stacking diagrams of $\mathbf{T}^{m,m;n,n}$ from inside of diagrams of $\mathbf{T}^{m,s;n,t}$. Obviously these two actions are independent and hence endow the tube space $\mathbf{T}^{m,s;n,t}$ a bimodule structure.

To summarize, we have the following result:

\begin{theorem}
The tube space $\mathbf{T}^{m,m;n,n}$ are algebras for all  $m,n\in \mathbb{N}$. 
The tube space $\mathbf{T}^{m,s;n,t}$ forms a right-$\mathbf{T}^{m,m;n,n}$ and left-$\mathbf{T}^{s,s;t,t}$ bimodule. 
These structures form a Morita context in the sense that
\begin{equation} \label{eq:Morita-equiv}
   \mathbf{T}^{m,s;n,t}\otimes_{\mathbf{T}^{m,m;n,n}} \mathbf{T}^{s,m;t,n} \cong \mathbf{T}^{s,s;t,t},\quad \mathbf{T}^{s,m;t,n} \otimes_{\mathbf{T}^{s,s;t,t}} \mathbf{T}^{m,s;n,t}\cong \mathbf{T}^{m,m;n,n}.
\end{equation}
Thus $\mathbf{T}^{m,m;n,n}$'s  are Morita equivalent for all  $m,n\in \mathbb{N}$.
\end{theorem}

\begin{proof}
    The first assertion has been shown, so it suffices to construct the isomorphisms. Define $\alpha:\mathbf{T}^{m,s;n,t}\otimes_{\mathbf{T}^{m,m;n,n}} \mathbf{T}^{s,m;t,n} \to \mathbf{T}^{s,s;t,t}$ by stacking the diagrams as in the definition of $\mu$ in Eq.~\eqref{eq:tube-space-prod}. This is a $\mathbf{T}^{s,s;t,t}|\mathbf{T}^{s,s;t,t}$-bimodule morphism. Define $\beta:\mathbf{T}^{s,s;t,t}\to\mathbf{T}^{m,s;n,t}\otimes_{\mathbf{T}^{m,m;n,n}} \mathbf{T}^{s,m;t,n}$ as follows: 
    \begin{align*}
        \beta&\left(\begin{aligned}
        \begin{tikzpicture}
             \draw[line width=.6pt,black] (0,0.5)--(0,1.5);
             \draw[line width=.6pt,black] (0,-0.5)--(0,-1.5);
             \draw[red,line width=.6pt] (120:0.8) -- (120:1.1);
             \draw[red,line width=.6pt] (150:0.8) -- (150:0.5);
            \draw[red,line width=.6pt] (180:0.8) -- (180:1.1);
            \draw[red,line width=.6pt] (210:0.8) -- (210:0.5);
          \draw[red,line width=.6pt] (240:0.8) -- (240:1.1);
             \draw[red] (0,0.8) arc[start angle=90, end angle=270, radius=0.8];
             \draw[blue] (0,1.3) arc[start angle=90, end angle=-90, radius=1.3];
            \node[red, line width=0.6pt, dashed, draw opacity=0.5] (a) at (225:0.6)
            {\small{$\cdot$}};
            \node[red, line width=0.6pt, dashed, draw opacity=0.5] (a) at (240:0.6){\small{$\cdot$}};
          \node[red, line width=0.6pt, dashed, draw opacity=0.5] (a) at (255:0.6){\small{$\cdot$}};
             \node[red, line width=0.6pt, dashed, draw opacity=0.5] (a) at (195:1)
            {\small{$\cdot$}};
            \node[red, line width=0.6pt, dashed, draw opacity=0.5] (a) at (210:1){\small{$\cdot$}};
          \node[red, line width=0.6pt, dashed, draw opacity=0.5] (a) at (225:1){\small{$\cdot$}};
            \draw[blue,line width=.6pt] (300:1.3) -- (300:1.6);
             \draw[blue,line width=.6pt] (330:1) -- (330:1.3);
            \draw[blue,line width=.6pt] (360:1.3) -- (360:1.6);
            \draw[blue,line width=.6pt] (30:1) -- (30:1.3);
             \draw[blue,line width=.6pt] (60:1.6) -- (60:1.3);
          \node[blue, line width=0.6pt, dashed, draw opacity=0.5] (a) at (75:1)
            {\small{$\cdot$}};
            \node[blue, line width=0.6pt, dashed, draw opacity=0.5] (a) at (60:1){\small{$\cdot$}};
          \node[blue, line width=0.6pt, dashed, draw opacity=0.5] (a) at (45:1){\small{$\cdot$}};
         \node[blue, line width=0.6pt, dashed, draw opacity=0.5] (a) at (45:1.45)
            {\small{$\cdot$}};
            \node[blue, line width=0.6pt, dashed, draw opacity=0.5] (a) at (30:1.45){\small{$\cdot$}};
          \node[blue, line width=0.6pt, dashed, draw opacity=0.5] (a) at (15:1.45){\small{$\cdot$}};
          \node[blue, line width=0.6pt, dashed, draw opacity=0.5] (a) at (60:1.85)
            {$b_1$};
            \node[blue, line width=0.6pt, dashed, draw opacity=0.5] (a) at (300:1.9)
            {$b_t$};
            \node[blue, line width=0.6pt, dashed, draw opacity=0.5] (a) at (0:2.1)
            {$b_{t-1}$};
          \node[blue, line width=0.6pt, dashed, draw opacity=0.5] (a) at (330:0.65)
            {$d_t$};
              \node[red, line width=0.6pt, dashed, draw opacity=0.5] (a) at (120:1.4)
            {$a_1$};
            \node[red, line width=0.6pt, dashed, draw opacity=0.5] (a) at (180:1.4)
            {$a_2$};
            \node[red, line width=0.6pt, dashed, draw opacity=0.5] (a) at (240:1.3)
            {$a_s$};
          \node[red, line width=0.6pt, dashed, draw opacity=0.5] (a) at (140:0.25)
            {$c_1$};
        \end{tikzpicture}
    \end{aligned}\right)  = \begin{aligned}
        \begin{tikzpicture}
             \draw[line width=.6pt,black] (0,0.5)--(0,1.5);
             \draw[line width=.6pt,black] (0,-0.5)--(0,-1.5);
             \draw[red,line width=.6pt] (120:0.8) -- (120:1.1);
             \draw[red,line width=.6pt] (150:0.8) -- (150:0.5);
            \draw[red,line width=.6pt] (180:0.8) -- (180:1.1);
            \draw[red,line width=.6pt] (210:0.8) -- (210:0.5);
          \draw[red,line width=.6pt] (240:0.8) -- (240:1.1);
             \draw[red,dotted] (0,0.8) arc[start angle=90, end angle=270, radius=0.8];
             \draw[blue] (0,1.3) arc[start angle=90, end angle=-90, radius=1.3];
            \node[red, line width=0.6pt, dashed, draw opacity=0.5] (a) at (225:0.6)
            {\small{$\cdot$}};
            \node[red, line width=0.6pt, dashed, draw opacity=0.5] (a) at (240:0.6){\small{$\cdot$}};
          \node[red, line width=0.6pt, dashed, draw opacity=0.5] (a) at (255:0.6){\small{$\cdot$}};
             \node[red, line width=0.6pt, dashed, draw opacity=0.5] (a) at (195:1)
            {\small{$\cdot$}};
            \node[red, line width=0.6pt, dashed, draw opacity=0.5] (a) at (210:1){\small{$\cdot$}};
          \node[red, line width=0.6pt, dashed, draw opacity=0.5] (a) at (225:1){\small{$\cdot$}};
            \draw[blue,line width=.6pt] (300:1.3) -- (300:1.6);
             \draw[blue,line width=.6pt] (330:1) -- (330:1.3);
            \draw[blue,line width=.6pt] (360:1.3) -- (360:1.6);
            \draw[blue,line width=.6pt] (30:1) -- (30:1.3);
             \draw[blue,line width=.6pt] (60:1.6) -- (60:1.3);
          \node[blue, line width=0.6pt, dashed, draw opacity=0.5] (a) at (75:1)
            {\small{$\cdot$}};
            \node[blue, line width=0.6pt, dashed, draw opacity=0.5] (a) at (60:1){\small{$\cdot$}};
          \node[blue, line width=0.6pt, dashed, draw opacity=0.5] (a) at (45:1){\small{$\cdot$}};
         \node[blue, line width=0.6pt, dashed, draw opacity=0.5] (a) at (45:1.45)
            {\small{$\cdot$}};
            \node[blue, line width=0.6pt, dashed, draw opacity=0.5] (a) at (30:1.45){\small{$\cdot$}};
          \node[blue, line width=0.6pt, dashed, draw opacity=0.5] (a) at (15:1.45){\small{$\cdot$}};
          \node[blue, line width=0.6pt, dashed, draw opacity=0.5] (a) at (60:1.85)
            {$b_1$};
            \node[blue, line width=0.6pt, dashed, draw opacity=0.5] (a) at (300:1.9)
            {$b_t$};
            \node[blue, line width=0.6pt, dashed, draw opacity=0.5] (a) at (0:2.1)
            {$b_{t-1}$};
          \node[blue, line width=0.6pt, dashed, draw opacity=0.5] (a) at (330:0.65)
            {$\one$};
              \node[red, line width=0.6pt, dashed, draw opacity=0.5] (a) at (120:1.4)
            {$a_1$};
            \node[red, line width=0.6pt, dashed, draw opacity=0.5] (a) at (180:1.4)
            {$a_2$};
            \node[red, line width=0.6pt, dashed, draw opacity=0.5] (a) at (240:1.3)
            {$a_s$};
          \node[red, line width=0.6pt, dashed, draw opacity=0.5] (a) at (140:0.25)
            {$\one$};
        \end{tikzpicture}
    \end{aligned}
    \otimes \begin{aligned}
        \begin{tikzpicture}
             \draw[line width=.6pt,black] (0,0.5)--(0,1.5);
             \draw[line width=.6pt,black] (0,-0.5)--(0,-1.5);
             \draw[red,line width=.6pt] (120:0.8) -- (120:1.1);
             \draw[red,line width=.6pt] (150:0.8) -- (150:0.5);
            \draw[red,line width=.6pt] (180:0.8) -- (180:1.1);
            \draw[red,line width=.6pt] (210:0.8) -- (210:0.5);
          \draw[red,line width=.6pt] (240:0.8) -- (240:1.1);
             \draw[red] (0,0.8) arc[start angle=90, end angle=270, radius=0.8];
             \draw[blue,dotted] (0,1.3) arc[start angle=90, end angle=-90, radius=1.3];
            \node[red, line width=0.6pt, dashed, draw opacity=0.5] (a) at (225:0.6)
            {\small{$\cdot$}};
            \node[red, line width=0.6pt, dashed, draw opacity=0.5] (a) at (240:0.6){\small{$\cdot$}};
          \node[red, line width=0.6pt, dashed, draw opacity=0.5] (a) at (255:0.6){\small{$\cdot$}};
             \node[red, line width=0.6pt, dashed, draw opacity=0.5] (a) at (195:1)
            {\small{$\cdot$}};
            \node[red, line width=0.6pt, dashed, draw opacity=0.5] (a) at (210:1){\small{$\cdot$}};
          \node[red, line width=0.6pt, dashed, draw opacity=0.5] (a) at (225:1){\small{$\cdot$}};
            \draw[blue,line width=.6pt] (300:1.3) -- (300:1.6);
             \draw[blue,line width=.6pt] (330:1) -- (330:1.3);
            \draw[blue,line width=.6pt] (360:1.3) -- (360:1.6);
            \draw[blue,line width=.6pt] (30:1) -- (30:1.3);
             \draw[blue,line width=.6pt] (60:1.6) -- (60:1.3);
          \node[blue, line width=0.6pt, dashed, draw opacity=0.5] (a) at (75:1)
            {\small{$\cdot$}};
            \node[blue, line width=0.6pt, dashed, draw opacity=0.5] (a) at (60:1){\small{$\cdot$}};
          \node[blue, line width=0.6pt, dashed, draw opacity=0.5] (a) at (45:1){\small{$\cdot$}};
         \node[blue, line width=0.6pt, dashed, draw opacity=0.5] (a) at (45:1.45)
            {\small{$\cdot$}};
            \node[blue, line width=0.6pt, dashed, draw opacity=0.5] (a) at (30:1.45){\small{$\cdot$}};
          \node[blue, line width=0.6pt, dashed, draw opacity=0.5] (a) at (15:1.45){\small{$\cdot$}};
          \node[blue, line width=0.6pt, dashed, draw opacity=0.5] (a) at (60:1.85)
            {$\one$};
            \node[blue, line width=0.6pt, dashed, draw opacity=0.5] (a) at (300:1.9)
            {$\one$};
            \node[blue, line width=0.6pt, dashed, draw opacity=0.5] (a) at (0:2.1)
            {$\one$};
          \node[blue, line width=0.6pt, dashed, draw opacity=0.5] (a) at (330:0.65)
            {$d_t$};
              \node[red, line width=0.6pt, dashed, draw opacity=0.5] (a) at (120:1.4)
            {$\one$};
            \node[red, line width=0.6pt, dashed, draw opacity=0.5] (a) at (180:1.4)
            {$\one$};
            \node[red, line width=0.6pt, dashed, draw opacity=0.5] (a) at (240:1.3)
            {$\one$};
          \node[red, line width=0.6pt, dashed, draw opacity=0.5] (a) at (140:0.25)
            {$c_1$};
        \end{tikzpicture}
    \end{aligned}. 
    \end{align*}
    Clearly $\beta$ is also a $\mathbf{T}^{s,s;t,t}|\mathbf{T}^{s,s;t,t}$-bimodule morphism. It is immediate that $\alpha\comp\beta=\id$ and $\beta\comp\alpha=\id$ hold, thus the first isomorphism in Eq.~\eqref{eq:Morita-equiv} is proven. The second isomorphism can be demonstrated similarly. 
\end{proof}

\section{Gapped boundary theory}
\label{sec:bdtheory}

Gapped boundary conditions for the Levin-Wen string-net model can be established through various equivalent approaches \cite{Kitaev2012boundary,hu2018boundary,jia2023boundary}. Here, we will outline the application of the Kitaev-Kong construction to the multifusion string-net model. In this context, the gapped boundaries are characterized by module categories over the input UMFC $\ED=\Rep(W)$ for some weak Hopf algebra $W$, which means that the bulk gauge weak Hopf symmetry is $W$.

It is crucial to highlight that, in contrast to the original Levin-Wen string-net model, even though the bulk phase is a UMTC, the ensuing boundary phase deviates from being a UFC. Instead, it assumes the structure of a UMFC, thereby introducing a notable distinction from the characteristics observed in the original model.

\begin{definition}
A gapped boundary phase is called stable if the boundary topological phase is a UFC. This inspires us to give the following two definitions:
\begin{enumerate}
    \item For a UMFC $\ED$, a $\ED$-module category $\EM$ is called stable if the functor category $\Fun_{\ED}(\EM,\EM)$ is a UFC.
    \item For a weak Hopf algebra $W$, a $W$-comodule algebra $\mathfrak{A}$ is called stable if the functor category $\Fun_{\Rep(W)}({_{\mathfrak{A}}}\Mod,{_{\mathfrak{A}}}\Mod)$ is a UFC.
\end{enumerate}

\end{definition}

The lattice Hamiltonian of the gapped boundary $_{\ED}\EM$ can be formulated as follows. It is important to note that the boundary is oriented such that the bulk resides on the left-hand side of the boundary. For a boundary vertex $v_b$, it is connected to one bulk edge $e$ and two boundary edges $e_1$ and $e_0$: $e_1$ ends at $v_b$, while $e_0$ originates from $v_b$.
The bulk edge $e$ is labeled with simple objects $j$ in $\ED$, but boundary edges $e_1$ and $e_0$ are labeled with simple objects $x_1$ and $x_0$ in $\EM$.
We define $Q_{v_b}$ as the projection operator such that, when applied to a boundary string-net configuration, $Q_{v_b}=1$ if $\Hom_{\EM}(j\otimes x_1,x_0)\neq 0$, and $Q_{v_b}=0$ otherwise.
The boundary face operator $B^k_{f_b}$ with $k\in \Irr (\ED)$ is defined as inserting a loop labeled by $k$ in the boundary face $f_b$, and $ B_{f_b}=\sum_k\frac{d_k}{\sum_l d_l^2} B^k_{f_b}$.
The boundary Hamiltonian is of the following form
\begin{equation}
    H[\partial \Sigma, {_{\ED}}\EM]=-\sum_{v_b} Q_{v_b}-\sum_{f_b}B_{f_b}.
\end{equation}
Notice that despite the similarity in expression to the bulk vertex and face operators, the boundary vertex and face operators are inherently distinct. When computing the matrix form of $B_{f_b}$, we require boundary topological local moves, which rely on the $\ED$-module structure of $\EM$. And in the definition of $Q_{v_b}$ we need the fusion of bulk label and boundary label, which is also based on the $\ED$-module structure of $\EM$.
An intriguing observation is that any boundary model can be transformed into a UMFC string-net model. However, it is worth noting that the UMFC resulting from this transformation is larger than the original one, as will be discussed in Sec.~\ref{sec:DefectSN}.

\subsection{Boundary gauge symmetry}

As previously discussed, the bulk gauge symmetry of the multifusion string-net model $\mathsf{SN}_{\ED}$ is described by a weak Hopf algebra $W$ such that $\Rep(W)\simeq\ED$. An important question that arises is how to model the boundary gauge symmetry corresponding to the left $\ED$-module category $\EM$.
Here, we propose that the boundary gauge symmetry can be equivalently characterized by:
\begin{enumerate}
\item $W$-comodule algebra $\mathfrak{A}$. By definition, a left $W$-comodule algebra is a left $W$-comodule with left coaction $\beta_{\mathfrak{A}}:\mathfrak{A}\to W\otimes \mathfrak{A}$ such that $\beta_{\mathfrak{A}}(xy)=\beta_{\mathfrak{A}}(x)\beta_{\mathfrak{A}}(y)$ and $\beta_{\mathfrak{A}}(1_{\mathfrak{A}})=1_W\otimes 1_{\mathfrak{A}}$. A right $W$-comodule algebra can be defined similarly.
\item $W$-module algebra $\mathfrak{M}$. By definition, a left $W$-module algebra is a left $W$-module with left action $h\triangleright x$ such that $h\triangleright (xy)=\sum_{(h)}(h^{(1)}\triangleright x)(h^{(2)}\triangleright y)$ and $h\triangleright 1_\mathfrak{M}=\varepsilon(h)1_\mathfrak{M}$. A right $W$-module algebra can be defined similarly.
\end{enumerate}
The equivalence between these two approaches is carefully discussed in our previous work \cite{jia2023boundary}.
Hereinafter, we will take the $W$-comodule algebra as a gapped boundary gauge symmetry.
For a given boundary gauge symmetry $\mathfrak{A}$, the module category $\EM={_{\mathfrak{A}}}\Mod$ is a module category over $\ED=\Rep(W)$.
The boundary excitations are characterized by $\ED$-module  functor category $\Fun_{\ED}(\EM,\EM)$; via Eilenberg-Watts theorem, the boundary excitations can be equivalently characterized by the category $_{\mathfrak{A}}\Mod_{\mathfrak{A}}^W$ of $W$-equivariant $\mathfrak{A}|\mathfrak{A}$-bimodules.

\subsection{Boundary tube algebra}

The boundary charge symmetry is given by the boundary tube algebra.
For the boundary, as a generalization of the algebra constructed in Ref.~\cite{Kitaev2012boundary}, for a left module category $\EM$ over a UMFC $\ED$,  we have the following tube algebra:

\begin{equation}
 \begin{aligned}
    \tilde{\mathbf{L}}({_{\ED}\EM})
 \end{aligned}  =\operatorname{span}\left\{  
   \begin{aligned}
        \begin{tikzpicture}
             \draw[line width=.6pt,black] (0,0.5)--(0,1.3);
             \draw[line width=.6pt,black] (0,-0.5)--(0,-1.3);
             \draw[red] (0,0.8) arc[start angle=90, end angle=270, radius=0.8];
            \node[ line width=0.6pt, dashed, draw opacity=0.5] (a) at (-1,0){$a$};
            \node[ line width=0.6pt, dashed, draw opacity=0.5] (a) at (-0.2,-1.3){$d$};
            \node[ line width=0.6pt, dashed, draw opacity=0.5] (a) at (0.2,-0.8){$\nu$};
            \node[ line width=0.6pt, dashed, draw opacity=0.5] (a) at (0,-0.3){$e$};
            \node[ line width=0.6pt, dashed, draw opacity=0.5] (a) at (0,0.3){$f$};
            \node[ line width=0.6pt, dashed, draw opacity=0.5] (a) at (-0.2,1.3){$g$};
            \node[ line width=0.6pt, dashed, draw opacity=0.5] (a) at (0.2,0.8){$\zeta$};
        \end{tikzpicture}
    \end{aligned}
   :\quad  a\in \Irr(\ED),f,g,e,d\in \Irr(\EM), \nu,\zeta\in \Hom_{\EM} 
    \right\}.
    \label{eq:LeftBdTubebasis}
\end{equation}
Building upon the discourse presented in Sec.~\ref{sec:tube}, we can endow $\tilde{\mathbf{L}}({_{\ED}\EM})$ with a $C^*$ weak Hopf algebra structure.
Analogously, for a right $\ED$-module $\EM_{\ED}$, we establish the $C^*$ weak Hopf tube algebra:
\begin{equation}
 \begin{aligned}
    \tilde{\mathbf{R}}(\EM_{\ED})
 \end{aligned}  =\operatorname{span}\left\{  
   \begin{aligned}
        \begin{tikzpicture}
             \draw[line width=.6pt,black] (0,0.5)--(0,1.3);
             \draw[line width=.6pt,black] (0,-0.5)--(0,-1.3);
            \draw[blue] (0,0.9) arc[start angle=90, end angle=-90, radius=0.9];
            \node[ line width=0.6pt, dashed, draw opacity=0.5] (a) at (0.7,0){$b$};
            \node[ line width=0.6pt, dashed, draw opacity=0.5] (a) at (-0.2,-1.3){$d$};
            \node[ line width=0.6pt, dashed, draw opacity=0.5] (a) at (0.2,-1.2){$\nu$};
            \node[ line width=0.6pt, dashed, draw opacity=0.5] (a) at (0,-0.3){$e$};
            \node[ line width=0.6pt, dashed, draw opacity=0.5] (a) at (0,0.3){$f$};
            \node[ line width=0.6pt, dashed, draw opacity=0.5] (a) at (-0.2,1.3){$g$};
            \node[ line width=0.6pt, dashed, draw opacity=0.5] (a) at (0.2,1.2){$\zeta$};
        \end{tikzpicture}
    \end{aligned}
   :\quad  b\in \Irr(\ED),f,g,e,d\in \Irr(\EM), \nu,\zeta\in \Hom_{\EM} 
    \right\}.
    \label{eq:RightBdTubebasis}
\end{equation}

\begin{remark}
   We will also use the notations $\mathbf{Tube}({_{\ED}}\EM)$ or $\mathbf{Tube}(\EM{_{\ED}})$ for left and right boundary tube algebra whenever there is no ambiguity.
\end{remark}

\begin{proposition}
   Both the boundary tube algebras  $ \tilde{\mathbf{L}}({_{\ED}\EM})$ and  $ \tilde{\mathbf{R}}(\EM_{\ED})$ are $C^*$ weak Hopf algebras. More precisely, the $C^*$ weak Hopf algebra structure of $ \tilde{\mathbf{L}}({_{\ED}\EM})$ is given by 
   \begin{gather}
       1=\sum_{a,b} \;\begin{aligned}
        \begin{tikzpicture}
             \draw[line width=.6pt,black] (0,0.5)--(0,1.3);
             \draw[line width=.6pt,black] (0,-0.5)--(0,-1.3);
             \draw[red, dotted] (0,0.8) arc[start angle=90, end angle=270, radius=0.8];
            \node[ line width=0.6pt, dashed, draw opacity=0.5] (a) at (-0.2,-1){$a$};
            \node[ line width=0.6pt, dashed, draw opacity=0.5] (a) at (-0.2,1){$b$};
        \end{tikzpicture}
    \end{aligned}\;, \quad \quad 
    \mu\left( \begin{aligned}
        \begin{tikzpicture}
             \draw[line width=.6pt,black] (0,0.5)--(0,1.3);
             \draw[line width=.6pt,black] (0,-0.5)--(0,-1.3);
             \draw[red] (0,0.8) arc[start angle=90, end angle=270, radius=0.8];
            \node[ line width=0.6pt, dashed, draw opacity=0.5] (a) at (-1,0){$a$};
            \node[ line width=0.6pt, dashed, draw opacity=0.5] (a) at (-0.2,-1){$d$};
            \node[ line width=0.6pt, dashed, draw opacity=0.5] (a) at (0.2,-0.8){$\nu$};
            \node[ line width=0.6pt, dashed, draw opacity=0.5] (a) at (0,-0.3){$e$};
            \node[ line width=0.6pt, dashed, draw opacity=0.5] (a) at (0,0.3){$f$};
            \node[ line width=0.6pt, dashed, draw opacity=0.5] (a) at (-0.2,1){$g$};
            \node[ line width=0.6pt, dashed, draw opacity=0.5] (a) at (0.2,0.8){$\zeta$};
        \end{tikzpicture}
    \end{aligned}\otimes \begin{aligned}
        \begin{tikzpicture}
             \draw[line width=.6pt,black] (0,0.5)--(0,1.3);
             \draw[line width=.6pt,black] (0,-0.5)--(0,-1.3);
             \draw[red] (0,0.8) arc[start angle=90, end angle=270, radius=0.8];
            \node[ line width=0.6pt, dashed, draw opacity=0.5] (a) at (-1,0){$a'$};
            \node[ line width=0.6pt, dashed, draw opacity=0.5] (a) at (-0.2,-1){$d'$};
            \node[ line width=0.6pt, dashed, draw opacity=0.5] (a) at (0.2,-0.8){$\nu'$};
            \node[ line width=0.6pt, dashed, draw opacity=0.5] (a) at (0,-0.3){$e'$};
            \node[ line width=0.6pt, dashed, draw opacity=0.5] (a) at (0,0.3){$f'$};
            \node[ line width=0.6pt, dashed, draw opacity=0.5] (a) at (-0.2,1){$g'$};
            \node[ line width=0.6pt, dashed, draw opacity=0.5] (a) at (0.2,0.8){$\zeta'$};
        \end{tikzpicture}
    \end{aligned}\right)=\delta_{f,g'}\delta_{e,d'}\;\begin{aligned}
        \begin{tikzpicture}
             \draw[line width=.6pt,black] (0,0.5)--(0,1.9);
             \draw[line width=.6pt,black] (0,-0.5)--(0,-1.9);
             \draw[red] (0,0.8) arc[start angle=90, end angle=270, radius=0.8];
            \node[ line width=0.6pt, dashed, draw opacity=0.5] (a) at (-0.2,1.05){$f$};
             \node[ line width=0.6pt, dashed, draw opacity=0.5] (a) at (-0.2,-1.05){$e$};
            \node[ line width=0.6pt, dashed, draw opacity=0.5] (a) at (-1,0){$a'$};
            \node[ line width=0.6pt, dashed, draw opacity=0.5] (a) at (0.3,-0.8){$\nu'$};
            \node[ line width=0.6pt, dashed, draw opacity=0.5] (a) at (0,-0.3){$e'$};
            \node[ line width=0.6pt, dashed, draw opacity=0.5] (a) at (0,0.3){$f'$};
            \node[ line width=0.6pt, dashed, draw opacity=0.5] (a) at (0.3,0.8){$\zeta'$};
            \draw[red] (0,1.4) arc[start angle=90, end angle=270, radius=1.4];
            \node[ line width=0.6pt, dashed, draw opacity=0.5] (a) at (0.3,1.4){$\zeta$};
          \node[ line width=0.6pt, dashed, draw opacity=0.5] (a) at (0.25,-1.5){$\nu$};
            \node[ line width=0.6pt, dashed, draw opacity=0.5] (a) at (-1.5,0.5){$a$};
         \node[ line width=0.6pt, dashed, draw opacity=0.5] (a) at (-0.3,1.7){$g$};
          \node[ line width=0.6pt, dashed, draw opacity=0.5] (a) at (-0.2,-1.7){$d$};
        \end{tikzpicture}
    \end{aligned}\;, \\
    \varepsilon\left(\begin{aligned}\begin{tikzpicture}
             \draw[line width=.6pt,black] (0,0.5)--(0,1.3);
             \draw[line width=.6pt,black] (0,-0.5)--(0,-1.3);
             \draw[red] (0,0.8) arc[start angle=90, end angle=270, radius=0.8];
            \node[ line width=0.6pt, dashed, draw opacity=0.5] (a) at (-1,0){$a$};
            \node[ line width=0.6pt, dashed, draw opacity=0.5] (a) at (-0.2,-1){$d$};
            \node[ line width=0.6pt, dashed, draw opacity=0.5] (a) at (0.2,-0.8){$\nu$};
            \node[ line width=0.6pt, dashed, draw opacity=0.5] (a) at (0,-0.3){$e$};
            \node[ line width=0.6pt, dashed, draw opacity=0.5] (a) at (0,0.3){$f$};
            \node[ line width=0.6pt, dashed, draw opacity=0.5] (a) at (-0.2,1){$g$};
            \node[ line width=0.6pt, dashed, draw opacity=0.5] (a) at (0.2,0.8){$\zeta$};
        \end{tikzpicture}
    \end{aligned}\right)
    =\frac{\delta_{e,f}\delta_{d,g}}{d_g}
    \begin{aligned}\begin{tikzpicture}
             \draw[line width=.6pt,black] (0,-1.3)--(0,1.3);
             \draw[red] (0,0.8) arc[start angle=90, end angle=270, radius=0.8];
             \draw[black] (0,1.3) arc[start angle=90, end angle=-90, radius=1.3];
            \node[ line width=0.6pt, dashed, draw opacity=0.5] (a) at (1.1,0){$g$};
            \node[ line width=0.6pt, dashed, draw opacity=0.5] (a) at (-0.6,0){$a$};
            \node[ line width=0.6pt, dashed, draw opacity=0.5] (a) at (0.2,-0.8){$\nu$};
            \node[ line width=0.6pt, dashed, draw opacity=0.5] (a) at (0.2,0){$f$};
            \node[ line width=0.6pt, dashed, draw opacity=0.5] (a) at (0.2,0.8){$\zeta$};
        \end{tikzpicture}
    \end{aligned}=\delta_{e,f}\delta_{d,g}\delta_{\nu,\zeta}\sqrt{\frac{d_a d_f}{d_g}}\;, \\ 
    \Delta\left(\begin{aligned}\begin{tikzpicture}
             \draw[line width=.6pt,black] (0,0.5)--(0,1.3);
             \draw[line width=.6pt,black] (0,-0.5)--(0,-1.3);
             \draw[red] (0,0.8) arc[start angle=90, end angle=270, radius=0.8];
            \node[ line width=0.6pt, dashed, draw opacity=0.5] (a) at (-1,0){$a$};
            \node[ line width=0.6pt, dashed, draw opacity=0.5] (a) at (-0.2,-1){$d$};
            \node[ line width=0.6pt, dashed, draw opacity=0.5] (a) at (0.2,-0.8){$\nu$};
            \node[ line width=0.6pt, dashed, draw opacity=0.5] (a) at (0,-0.3){$e$};
            \node[ line width=0.6pt, dashed, draw opacity=0.5] (a) at (0,0.3){$f$};
            \node[ line width=0.6pt, dashed, draw opacity=0.5] (a) at (-0.2,1){$g$};
            \node[ line width=0.6pt, dashed, draw opacity=0.5] (a) at (0.2,0.8){$\zeta$};
        \end{tikzpicture}
    \end{aligned}\right)
    =\sum_{i,j,\sigma} \sqrt{ \frac{d_j}{d_ad_i}}
    \begin{aligned}\begin{tikzpicture}
             \draw[line width=.6pt,black] (0,0.5)--(0,1.3);
             \draw[line width=.6pt,black] (0,-0.5)--(0,-1.3);
             \draw[red] (0,0.8) arc[start angle=90, end angle=270, radius=0.8];
            \node[ line width=0.6pt, dashed, draw opacity=0.5] (a) at (-1,0){$a$};
            \node[ line width=0.6pt, dashed, draw opacity=0.5] (a) at (-0.2,-1){$j$};
            \node[ line width=0.6pt, dashed, draw opacity=0.5] (a) at (0.2,-0.8){$\sigma$};
            \node[ line width=0.6pt, dashed, draw opacity=0.5] (a) at (0,-0.3){$i$};
            \node[ line width=0.6pt, dashed, draw opacity=0.5] (a) at (0,0.3){$f$};
            \node[ line width=0.6pt, dashed, draw opacity=0.5] (a) at (-0.2,1){$g$};
            \node[ line width=0.6pt, dashed, draw opacity=0.5] (a) at (0.2,0.8){$\zeta$};
        \end{tikzpicture}
    \end{aligned} 
    \otimes 
    \begin{aligned}\begin{tikzpicture}
             \draw[line width=.6pt,black] (0,0.5)--(0,1.3);
             \draw[line width=.6pt,black] (0,-0.5)--(0,-1.3);
             \draw[red] (0,0.8) arc[start angle=90, end angle=270, radius=0.8];
            \node[ line width=0.6pt, dashed, draw opacity=0.5] (a) at (-1,0){$a$};
            \node[ line width=0.6pt, dashed, draw opacity=0.5] (a) at (-0.2,-1){$d$};
            \node[ line width=0.6pt, dashed, draw opacity=0.5] (a) at (0.2,-0.8){$\nu$};
            \node[ line width=0.6pt, dashed, draw opacity=0.5] (a) at (0,-0.3){$e$};
            \node[ line width=0.6pt, dashed, draw opacity=0.5] (a) at (0,0.3){$i$};
            \node[ line width=0.6pt, dashed, draw opacity=0.5] (a) at (-0.2,1){$j$};
            \node[ line width=0.6pt, dashed, draw opacity=0.5] (a) at (0.2,0.8){$\sigma$};
        \end{tikzpicture}
\end{aligned}\;, \\
S\left( 
          \begin{aligned}\begin{tikzpicture}
             \draw[line width=.6pt,black] (0,0.5)--(0,1.3);
             \draw[line width=.6pt,black] (0,-0.5)--(0,-1.3);
             \draw[red] (0,0.8) arc[start angle=90, end angle=270, radius=0.8];
            \node[ line width=0.6pt, dashed, draw opacity=0.5] (a) at (-1,0){$a$};
            \node[ line width=0.6pt, dashed, draw opacity=0.5] (a) at (-0.2,-1){$d$};
            \node[ line width=0.6pt, dashed, draw opacity=0.5] (a) at (0.2,-0.8){$\nu$};
            \node[ line width=0.6pt, dashed, draw opacity=0.5] (a) at (0,-0.3){$e$};
            \node[ line width=0.6pt, dashed, draw opacity=0.5] (a) at (0,0.3){$f$};
            \node[ line width=0.6pt, dashed, draw opacity=0.5] (a) at (-0.2,1){$g$};
            \node[ line width=0.6pt, dashed, draw opacity=0.5] (a) at (0.2,0.8){$\zeta$};
        \end{tikzpicture}
    \end{aligned}
          \right) =\frac{d_f}{d_g}\;\begin{aligned}\begin{tikzpicture}
             \draw[line width=.6pt,black] (0,0.5)--(0,1.3);
             \draw[line width=.6pt,black] (0,-0.5)--(0,-1.3);
             \draw[red] (0,0.8) arc[start angle=90, end angle=270, radius=0.8];
            \node[ line width=0.6pt, dashed, draw opacity=0.5] (a) at (-1,0){$\bar{a}$};
            \node[ line width=0.6pt, dashed, draw opacity=0.5] (a) at (-0.2,-1){$f$};
            \node[ line width=0.6pt, dashed, draw opacity=0.5] (a) at (0.2,-0.8){$\zeta$};
            \node[ line width=0.6pt, dashed, draw opacity=0.5] (a) at (0,-0.3){$g$};
            \node[ line width=0.6pt, dashed, draw opacity=0.5] (a) at (0,0.3){$d$};
            \node[ line width=0.6pt, dashed, draw opacity=0.5] (a) at (-0.2,1){$e$};
            \node[ line width=0.6pt, dashed, draw opacity=0.5] (a) at (0.2,0.8){$\nu$};
        \end{tikzpicture}
    \end{aligned}\;, \quad  \quad 
    \left( \begin{aligned}\begin{tikzpicture}
             \draw[line width=.6pt,black] (0,0.5)--(0,1.3);
             \draw[line width=.6pt,black] (0,-0.5)--(0,-1.3);
             \draw[red] (0,0.8) arc[start angle=90, end angle=270, radius=0.8];
            \node[ line width=0.6pt, dashed, draw opacity=0.5] (a) at (-1,0){$a$};
            \node[ line width=0.6pt, dashed, draw opacity=0.5] (a) at (-0.2,-1){$d$};
            \node[ line width=0.6pt, dashed, draw opacity=0.5] (a) at (0.2,-0.8){$\nu$};
            \node[ line width=0.6pt, dashed, draw opacity=0.5] (a) at (0,-0.3){$e$};
            \node[ line width=0.6pt, dashed, draw opacity=0.5] (a) at (0,0.3){$f$};
            \node[ line width=0.6pt, dashed, draw opacity=0.5] (a) at (-0.2,1){$g$};
            \node[ line width=0.6pt, dashed, draw opacity=0.5] (a) at (0.2,0.8){$\zeta$};
        \end{tikzpicture}
    \end{aligned} \right)^*= \frac{d_e}{d_d} \; \begin{aligned}\begin{tikzpicture}
             \draw[line width=.6pt,black] (0,0.5)--(0,1.3);
             \draw[line width=.6pt,black] (0,-0.5)--(0,-1.3);
             \draw[red] (0,0.8) arc[start angle=90, end angle=270, radius=0.8];
            \node[ line width=0.6pt, dashed, draw opacity=0.5] (a) at (-1,0){$\bar{a}$};
            \node[ line width=0.6pt, dashed, draw opacity=0.5] (a) at (-0.2,-1){$e$};
            \node[ line width=0.6pt, dashed, draw opacity=0.5] (a) at (0.2,-0.8){$\nu$};
            \node[ line width=0.6pt, dashed, draw opacity=0.5] (a) at (0,-0.3){$d$};
            \node[ line width=0.6pt, dashed, draw opacity=0.5] (a) at (0,0.3){$g$};
            \node[ line width=0.6pt, dashed, draw opacity=0.5] (a) at (-0.2,1){$f$};
            \node[ line width=0.6pt, dashed, draw opacity=0.5] (a) at (0.2,0.8){$\zeta$};
        \end{tikzpicture}
    \end{aligned}\;.
   \end{gather}
   The $C^*$ weak Hopf algebra structure of $ \tilde{\mathbf{R}}(\EM_{\ED})$ is given similarly. 
\end{proposition}

\begin{proof}
    The proof is the same as that for $\mathbf{Tube}({_\ED}\ED_\ED)$. 
\end{proof}

\begin{remark}
    For the smooth boundary $\EM=\ED$ (regard $\ED$ as a left $\ED$-module category in a natural way), the boundary tube algebra $ \tilde{\mathbf{L}}({_{\ED}\ED})$ is isomorphic to ${\mathbf{L}}({_{\ED}\ED})$ which is a subalgebra of the bulk tube algebra $\mathbf{Tube}({_{\ED}}\ED_{\ED})$. Similarly,  $ \tilde{\mathbf{R}}(\ED_{\ED})\cong {\mathbf{R}}(\ED_{\ED})\hookrightarrow \mathbf{Tube}({_{\ED}}\ED_{\ED})$.
    We can introduce a pairing between basis diagrams of $ \tilde{\mathbf{L}}({_{\ED}\ED})$ and  $ \tilde{\mathbf{R}}(\ED_{\ED})$ as follows:
\begin{align}
\left\langle 
   \begin{aligned}
        \begin{tikzpicture}
             \draw[line width=.6pt,black] (0,0.5)--(0,1.3);
             \draw[line width=.6pt,black] (0,-0.5)--(0,-1.3);
             \draw[red] (0,0.8) arc[start angle=90, end angle=270, radius=0.8];
            \node[ line width=0.6pt, dashed, draw opacity=0.5] (a) at (-1,0){$a$};
            \node[ line width=0.6pt, dashed, draw opacity=0.5] (a) at (-0.2,-1.3){$d$};
            \node[ line width=0.6pt, dashed, draw opacity=0.5] (a) at (0.2,-0.8){$\nu$};
            \node[ line width=0.6pt, dashed, draw opacity=0.5] (a) at (0,-0.3){$e$};
            \node[ line width=0.6pt, dashed, draw opacity=0.5] (a) at (0,0.3){$f$};
            \node[ line width=0.6pt, dashed, draw opacity=0.5] (a) at (-0.2,1.3){$g$};
            \node[ line width=0.6pt, dashed, draw opacity=0.5] (a) at (0.2,0.8){$\zeta$};
        \end{tikzpicture}
    \end{aligned},
   \begin{aligned}
        \begin{tikzpicture}
             \draw[line width=.6pt,black] (0,0.5)--(0,1.3);
             \draw[line width=.6pt,black] (0,-0.5)--(0,-1.3);
            \draw[blue] (0,0.9) arc[start angle=90, end angle=-90, radius=0.9];
            \node[ line width=0.6pt, dashed, draw opacity=0.5] (a) at (0.7,0){$\bar{b}$};
            \node[ line width=0.6pt, dashed, draw opacity=0.5] (a) at (-0.2,-1.3){$\bar{d'}$};
            \node[ line width=0.6pt, dashed, draw opacity=0.5] (a) at (0.2,-1.2){$\nu'$};
            \node[ line width=0.6pt, dashed, draw opacity=0.5] (a) at (0,-0.3){$\bar{e'}$};
            \node[ line width=0.6pt, dashed, draw opacity=0.5] (a) at (0,0.2){$\bar{f'}$};
            \node[ line width=0.6pt, dashed, draw opacity=0.5] (a) at (-0.2,1.3){$\bar{g'}$};
            \node[ line width=0.6pt, dashed, draw opacity=0.5] (a) at (0.2,1.2){$\zeta'$};
        \end{tikzpicture}
    \end{aligned}
\right \rangle  
& = \delta_{e,e'}\delta_{d,d'}\delta_{f,f'}\delta_{g,g'}\delta_{\nu,\nu'}\delta_{\zeta,\zeta'}\delta_{a,b}\nonumber\\
& =
\frac{\delta_{e,e'}\delta_{d,d'}\delta_{f,f'}\delta_{g,g'} }{\sqrt{d_fd_gd_dd_e}}
\begin{aligned}
        \begin{tikzpicture}
             \draw[red] (0,0.8) arc[start angle=90, end angle=270, radius=0.8];
            \node[ line width=0.6pt, dashed, draw opacity=0.5] (a) at (-1,0){$a$};
            \node[ line width=0.6pt, dashed, draw opacity=0.5] (a) at (0.25,-1.3){$d$};
            \node[ line width=0.6pt, dashed, draw opacity=0.5] (a) at (-0.3,-0.8){$\nu$};
            \node[ line width=0.6pt, dashed, draw opacity=0.5] (a) at (0.25,-0.3){$e$};
            \node[ line width=0.6pt, dashed, draw opacity=0.5] (a) at (0.25,0.3){$f$};
            \node[ line width=0.6pt, dashed, draw opacity=0.5] (a) at (0.23,1.3){$g$};
            \node[ line width=0.6pt, dashed, draw opacity=0.5] (a) at (-0.3,0.8){$\zeta$};
            \draw[black] (0.5,0.8) arc[start angle=-360, end angle=0, radius=0.25];
            \draw[black] (0.5,-0.8) arc[start angle=-360, end angle=0, radius=0.25];
            \draw[blue] (0.5,0.8) arc[start angle=90, end angle=-90, radius=0.8];
            \node[ line width=0.6pt, dashed, draw opacity=0.5] (a) at (1.45,0){$\bar{b}$};
            \node[ line width=0.6pt, dashed, draw opacity=0.5] (a) at (0.8,-1){$\nu'$};
            \node[ line width=0.6pt, dashed, draw opacity=0.5] (a) at (0.8,1.2){$\zeta'$};
        \end{tikzpicture}
    \end{aligned}\;.
\end{align}
Thus, we can perceive the bases of $ \tilde{\mathbf{L}}({_{\ED}\ED})$ and $ \tilde{\mathbf{R}}(\ED_{\ED})$ as dual to each other. 
This pairing is crucial for us to regard the bulk tube algebra as a quantum double, the details will be studied in our future work.
\end{remark}

\begin{remark}
 We can also introduce more bulk edges like that in Sec.~\ref{sec:TubeMorita} to construct the boundary tube space $\tilde{\mathbf{L}}^{m_0,m_1}({_{\ED}\EM})$.
Different algebras $\tilde{\mathbf{L}}^{m,m}({_{\ED}\EM})$
 and  $\tilde{\mathbf{L}}^{n,n}({_{\ED}\EM})$ are Morita equivalent for $m,n\in \mathbb{N}$ with Morita context given by bimodules $\tilde{\mathbf{L}}^{m,n}({_{\ED}\EM})$ and $\tilde{\mathbf{L}}^{n,m}({_{\ED}\EM})$.   
\end{remark}

\begin{remark}
By leveraging the boundary-bulk duality, the boundary tube algebra can be viewed as the bulk gauge symmetry. Utilizing the boundary tube algebra as input data for the weak Hopf quantum double model yields a lattice model with identical topological excitations as those derived from the multifusion string-net bulk.
\end{remark}

\subsection{Anyon condensation}

The gapped boundary phase can be regarded as a phase obtained from the bulk phase via anyon condensation \cite{Bais2002,bais2003hopf,Bais2009,Kong2014,eliens2010anyon,burnell2018anyon}. 
There are at least three equivalent ways to understand bulk-to-boundary anyon condensation.

    \vspace{1em} 
    \emph{Shrinking approach.} ---
    For the multifusion string-net $\mathsf{SN}_{\ED}$, the bulk excitations can be regarded as the defects over the domain wall $_{\ED}\ED_{\ED}$, \emph{viz.}, we regarded $\ED$ as a $\ED|\ED$-bimodule category.
    The corresponding UMTC is given by $\Fun_{\ED|\ED}(\ED,\ED)= \mathcal{Z}(\ED)$.
    The topological phase of $_{\ED}\EM$ boundary is the left $\ED$-module functor category $\Fun_{\ED}(\EM,\EM)$. The anyon condensation can be regarded as a shrink of the region between the bulk $_{\ED}\ED_{\ED}$ domain wall and $_{\ED}\EM$ boundary.
    This shrinking process is achieved by the Deligne tensor product between module categories $_{\ED}\ED_{\ED}$ and $_{\ED}\EM$: $\ED\boxtimes_{\ED} \EM\simeq \EM$.
    A bulk excitation is a $\ED|\ED$-bimodule functor $F:\ED\to \ED$; after anyon condensation, it becomes the functor $\EM\simeq \ED\boxtimes_{\ED}\EM\overset{F\boxtimes \id_{\EM}}{\to} \ED\boxtimes_{\ED}\EM\simeq \EM$, which is the same as anyon condensation of Levin-Wen string-net \cite{Kitaev2012boundary}.

    \vspace{1em}
    \emph{Lagrangian algebra approach.} ---
    Unlike the shrinking approach above and the tube algebra approach below, it is more convenient to distinguish the stable and unstable boundaries.
    For the stable boundary,  the bulk-to-boundary anyon condensation is described by the one-step anyon condensation theory \cite{Kong2014,Cong2017,jia2022electricmagnetic,jia2023anyon}, and there are two equivalent descriptions. The bulk phase is given by the UMTC $\mathcal{Z}(\ED)$. The gapped boundary is determined by a Lagrangian algebra $\Acal\in \mathcal{Z}(\ED)$. The boundary phase is given by the category of $\Acal$-modules  in $\mathcal{Z}(\ED)$\,\footnote{For readers not familiar with tensor category theory, it is important to note that the concept of an $\Acal$-module here is distinct from the module over an algebra as typically defined in $\Vect$ (the category of vector spaces).
    An $\Acal$-module in $\mathcal{Z}(\ED)$ is an object $M\in \mathcal{Z}(\ED)$ equipped with a morphism $\mu_{M}:\Acal\otimes M \to M$ that satisfies the conditions: $\mu_M\comp (\mu_{\Acal}\otimes \id_M)=\mu_M\comp ( \id_{\Acal}\otimes\mu_M)\comp a_{\Acal,\Acal,M}$ and $\mu_M\comp ( \eta_{\Acal}\otimes\id_M)=l_M$, where $\eta:\one \to \Acal$ and $\mu_{\Acal}: \Acal\otimes \Acal \to \Acal$ are structure morphisms of $\Acal$, $a_{\Acal,\Acal,M}$ is the associativity isomorphisms, and $l_M$ is the left unit isomorphisms in $\mathcal{Z}(\ED)$.}, which is usually denoted as $\mathcal{Z}(\ED)_{\Acal}$.
    The Lagrangian algebra $\Acal$ also has a Frobenius algebra structure, from which we can construct a pre-quotient category $\mathcal{Z}(\ED)/\Acal$, then via taking Karoubi envelope, we obtain the boundary phase $\EuScript{Q}(\mathcal{Z}(\ED);\Acal)$. It can be proved \cite{MUGER2004galois} that  $\EuScript{Q}(\mathcal{Z}(\ED);\Acal)\simeq \mathcal{Z}(\ED)_{\Acal}$, thus two descriptions are equivalent.

    For the general case, including scenarios that are not necessarily stable, the bulk-to-boundary anyon condensation is described by a more intricate two-step condensation theory. In this part, we will provide an outline of the anyon condensation theory for multifusion string-nets. A more comprehensive discussion will be presented in \cite{jia2023anyon}. It is worth recalling that the bulks we consider are always stable, i.e., their topological excitations are characterized by a unitary braided fusion category. The general multifusion phase will also be addressed in \cite{jia2023anyon}.
    Before we proceed to discuss the anyon condensation theory for multifusion topological phase, let us first introduce some necessary definitions.

\begin{definition}
A condensable algebra $\Acal$ in a braided fusion category $\EP$ is an object $\Acal\in \EP$ equipped with two structure morphisms:  multiplication $\mu_{\Acal}:\Acal \otimes \Acal\to  \Acal$ and unit map $\eta_{\Acal}: \mathbf{1}\to \Acal$, which satisfy
\begin{enumerate}
\item $(\Acal,\mu_{\Acal},\eta_{\Acal})$ is an algebra.

\item $\Acal$ is commutative, i.e., $\mu_{\Acal} \comp c_{\Acal,\Acal}=\mu_{\Acal}$ where $c_{\Acal,\Acal}$ is the braiding map in $\EP$.

\item $\Acal$ is separable, i.e., $\mu_{A}: \Acal \otimes \Acal \to \Acal$ splits as a morphism of $\Acal|\Acal$-bimodules. Namely, there is an $\Acal|\Acal$-bimodule map $e_{\Acal}: \Acal \rightarrow \Acal \otimes \Acal$ such that $\mu_{\Acal} \comp e_{\Acal}=\mathrm{id}_{\Acal}$. 

\item $\Acal$ is connected, viz., $\operatorname{dim} \operatorname{Hom}_{\EP}(\mathbf{1}, \Acal)=1$, the vacuum charge $\mathbf{1}$ only appears once in the decomposition of $\Acal=\mathbf{1}\oplus a\oplus b\oplus \cdots$.
\end{enumerate}
If an algebra satisfies conditions 1 and 2, it is referred to as an \'{e}tale algebra \cite{davydov2013witt}. If a condensable algebra $\Acal$ additionally satisfies $(\operatorname{dim} \Acal)^2=\operatorname{dim} \EP:=\sum_{a\in \mathrm{Irr}(\EP)}(d_a)^2$, then $\Acal$ is termed a Lagrangian algebra.
\end{definition}

A left $\Acal$-module in $\EP$ is an object $M$ equipped with a morphism $\mu_M:\Acal\otimes M\to M$ such that, $\mu_{M}\comp (\mu_{\Acal}\otimes \id_M)=\mu_M\comp (\id_{\Acal}\otimes \mu_M)\comp a_{\Acal,\Acal,M}$ (where $a_{\Acal,\Acal,M}$ is the associator in $\EP$) and $\mu_{M}\comp (\eta_{\Acal}\otimes \id_M)=l_M$ (where $l_M$ is the left unit isomorphism in $\EP$). 
The right $\Acal$-module can be defined similarly.
For a commutative algebra, a left $\Acal$-module $M$ is called local if and only if $\mu_M=\mu_M\comp c_{M,\Acal}\comp c_{\Acal,M}$.

It can be proved that for Lagrangian algebra $\Acal\in \EP$, the category $\EP_{\Acal}^{\rm loc}$ of local $\Acal$-modules in $\EP$ is equivalent (as braided fusion category) to the category $\Hilb$ of finite-dimensional Hilbert spaces (corresponding to the trivial topological phase). And the category $\EP_{\Acal}$ of $\Acal$-modules is a UFC.

To characterize the gapped boundary of multifusion string-net, we need extra data, an algebra $\Bcal$ in $\EP$ equipped with an algebra homomorphism $\zeta_{\Acal,\Bcal}:\Acal \to \Bcal$. We say that $\Bcal$ is an algebra over $\Acal$ if and only if
$\mu_{\Bcal}\comp (\zeta_{\Acal,\Bcal}\otimes \id_\Bcal)= \mu_{\Bcal}\comp (\id_{\Bcal}\otimes \zeta_{\Acal,\Bcal})\comp c^{-1}_{\Bcal,\Acal}$.

With the above preparation, we are now ready to give our main result:

\begin{theorem}
    For a multifusion string-net with bulk phase $\EP=\mathcal{Z}(\ED)$, the bulk-to-boundary anyon condensation is characterized by a pair of algebras $(\Acal,\Bcal)$ where $\Acal$ is a Lagrangian algebra and $\Bcal$ is a separable algebra over $\Acal$. The boundary phase is characterized by the category $\EP_{\Bcal|\Bcal}$ of $\Bcal|\Bcal$-bimodules, which is a UMFC.
    The bulk-to-boundary condensation is given by the monoidal functor 
    \begin{equation}
        \bullet \otimes \Bcal: \EP \to \EP_{\Bcal|\Bcal}.
    \end{equation}
\end{theorem}

\begin{proof}
  Notice that the condensed phase is characterized by local $\Acal$-modules in $\EP$. Since $\Acal$ is a Lagrangian algebra, we have $\EP_{\Acal}^{\rm loc}\simeq \Hilb$ as unitary braided fusion categories. This implies that the condensed phase is a trivial topological phase. The boundary phase $\EP_{\Bcal|\Bcal}$ is a UMFC.
\end{proof}

    
    \vspace{1em}
    \emph{Tube algebra approach.} ---
    Using the tube algebra description of the topological excitations that we have established before, we have another description of the bulk-to-boundary anyon condensation.
    Since the bulk phase is given by the representation category of the bulk tube algebra $\Rep(\mathbf{Tube}({_{\ED}}\ED_{\ED}))$, and the boundary phase is given by the representation category of the boundary tube algebra $\Rep(\mathbf{Tube}({_{\ED}}\EM))$, the anyon condensation is a monoidal functor $\mathbf{Cond}:\Rep(\mathbf{Tube}({_{\ED}}\ED_{\ED})) \to \Rep(\mathbf{Tube}({_{\ED}}\EM))$.

    Consider a disk region $\Gamma$ in the bulk; if there are excitations inside this region, we examine the tube region $\operatorname{Tube}_{\Gamma}$ surrounding this disk and construct the corresponding tube algebra. Anyon condensation can be viewed as enlarging the disk such that it can be treated as a boundary disk $\Gamma'$.
    For this boundary disk, we will have the corresponding boundary tube region $\operatorname{Tube}_{\Gamma'}$, and through certain boundary topological local moves, we can derive the corresponding boundary tube algebra. In this manner, the bulk excitation transitions into a boundary excitation. See Fig.~\ref{fig:tubeBulkBoundary} for an illustration.

\begin{figure}[t]
		\centering
		\includegraphics[width=9cm]{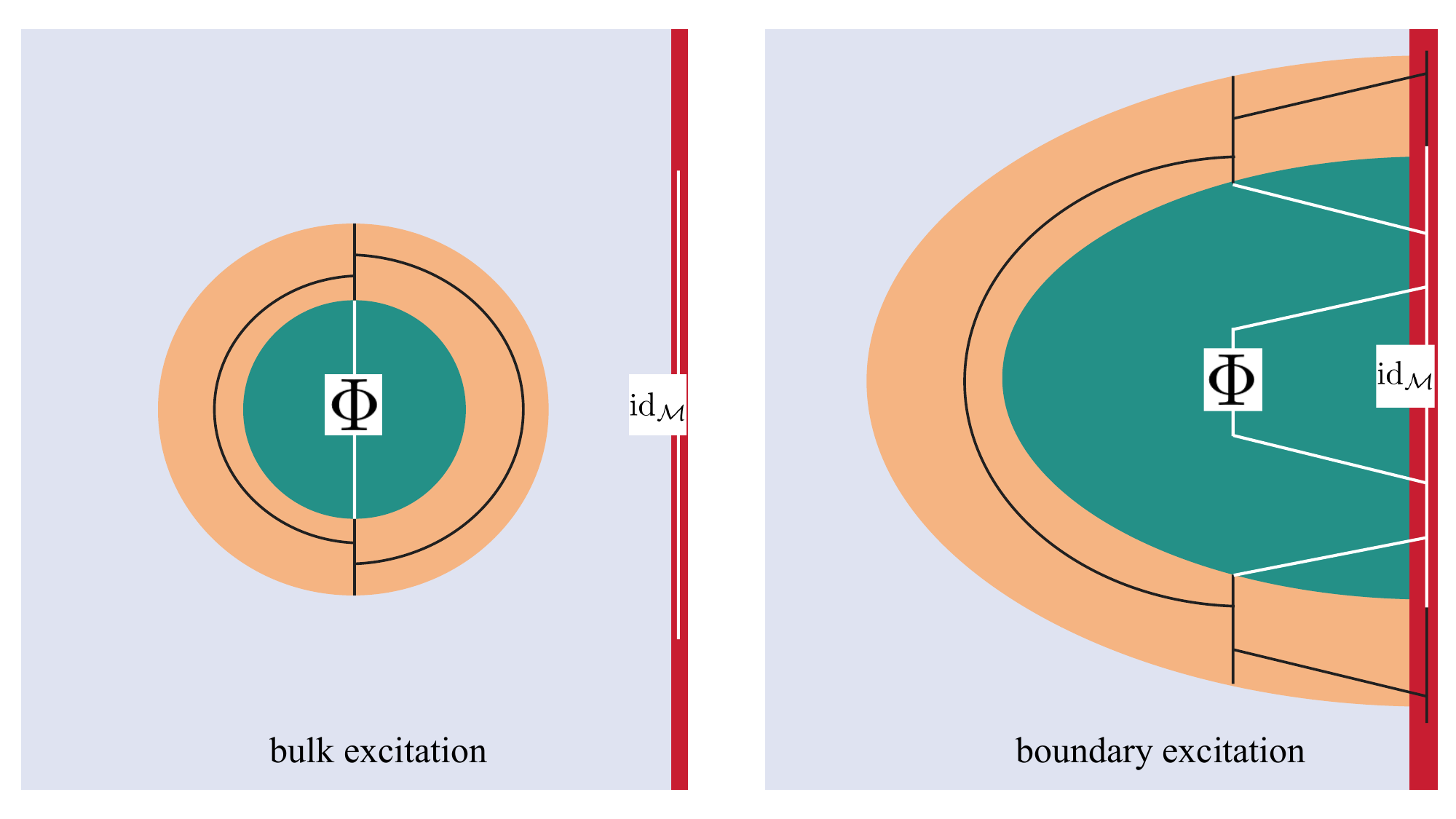}
		\caption{Illustration of bulk-to-boundary anyon condensation based on tube algebra. \label{fig:tubeBulkBoundary}}
\end{figure}

More precisely, let us consider a bulk excitation $M\in \Rep(\mathbf{Tube}({_{\ED}}\ED_{\ED}))$, we denote its basis diagrammatically as
\begin{equation}
        \begin{aligned}
        \begin{tikzpicture}
             \draw[line width=.6pt,black] (0,0)--(0,0.6);
             \draw [draw=black,fill=gray, fill opacity=0.2] (-0.3,0) rectangle (0.3,-0.6);
             \draw[line width=.6pt,black] (0,-0.6)--(0,-1.2);
              \node[ line width=0.6pt, dashed, draw opacity=0.5] (a) at (0,-.3){$\Phi$};
             \node[ line width=0.6pt, dashed, draw opacity=0.5] (a) at (0.2,-.9){$x$};
               \node[ line width=0.6pt, dashed, draw opacity=0.5] (a) at (0.2,.3){$y$};
        \end{tikzpicture}
    \end{aligned}.
\end{equation}
For the gapped boundary characterized by a $\ED$-module $\EM$, the boundary vacuum is the identity functor $\id_{\EM}\in \Fun_{\ED}(\EM,\EM)$. Using a construction similar to that given in subsection~\ref{subsec:FuntorExcitation}, we obtain a $\mathbf{Tube}({_{\ED}}\EM)$-module $V_{\operatorname{id}_{\EM}}$, which represents the boundary vacuum charge. The basis consists of identity morphisms $\operatorname{id}_m$, where $m\in \EM$.
Via parallel move, we obtain 
\begin{equation}
     \begin{aligned}
        \begin{tikzpicture}
             \draw[line width=.6pt,black] (0,0)--(0,0.6);
             \draw [draw=black,fill=gray, fill opacity=0.2] (-0.3,0) rectangle (0.3,-0.6);
             \draw[line width=.6pt,black] (0,-0.6)--(0,-1.2);
              \node[ line width=0.6pt, dashed, draw opacity=0.5] (a) at (0,-.3){$\Phi$};
             \node[ line width=0.6pt, dashed, draw opacity=0.5] (a) at (0.2,-.9){$x$};
               \node[ line width=0.6pt, dashed, draw opacity=0.5] (a) at (0.2,.3){$y$};
        \end{tikzpicture}
    \end{aligned}\quad 
      \begin{aligned}
        \begin{tikzpicture}
             \draw[line width=.6pt,black] (0,0.6)--(0,-1.2);
             \node[ line width=0.6pt, dashed, draw opacity=0.5] (a) at (0.2,-.9){$m$};
        \end{tikzpicture}
    \end{aligned}=
    \sum_{s,t,\mu,\nu}  \frac{1}{Y_{s}^{xm}Y_t^{ym}}    \begin{aligned}
        \begin{tikzpicture}
             \draw[line width=.6pt,black] (0,0)--(0,0.6);
             \draw [draw=black,fill=gray, fill opacity=0.2] (-0.3,0) rectangle (0.3,-0.6);
             \draw[line width=.6pt,black] (0,-0.6)--(0,-1.2);
              \node[ line width=0.6pt, dashed, draw opacity=0.5] (a) at (0,-.3){$\Phi$};
             \node[ line width=0.6pt, dashed, draw opacity=0.5] (a) at (0.2,-.9){$x$};
               \node[ line width=0.6pt, dashed, draw opacity=0.5] (a) at (0.2,.3){$y$};
           \draw[line width=.6pt,black] (0.8,0.6)--(0.8,-1.2);
           \draw[line width=.6pt,black] (0.4,-1.5)--(0.8,-1.2);
          \draw[line width=.6pt,black] (0.4,-1.5)--(0,-1.2);
                   \node[ line width=0.6pt, dashed, draw opacity=0.5] (a) at (1.0,-.9){$m$};
   \draw[line width=.6pt,black] (0.4,0.9)--(0,0.6);
     \draw[line width=.6pt,black] (0.4,0.9)--(0.8,0.6);
     \draw[line width=.6pt,black] (0.4,0.9)--(.4,1.2);
          \draw[line width=.6pt,black] (0.4,-1.5)--(.4,-1.8);
                      \node[ line width=0.6pt, dashed, draw opacity=0.5] (a) at (.6,1.1){$t$};
                             \node[ line width=0.6pt, dashed, draw opacity=0.5] (a) at (.6,-1.7){$s$};
           \draw[line width=.6pt,black] (0.4,1.2)--(0,1.5);           \draw[line width=.6pt,black] (0.4,1.2)--(0.8,1.5);               \draw[line width=.6pt,black] (0.4,-1.8)--(0,-2.1);           \draw[line width=.6pt,black] (0.4,-1.8)--(0.8,-2.1);   
        \node[ line width=0.6pt, dashed, draw opacity=0.5] (a) at (-0.2,1.5){$y$};
                \node[ line width=0.6pt, dashed, draw opacity=0.5] (a) at (1,1.5){$m$};
                      \node[ line width=0.6pt, dashed, draw opacity=0.5] (a) at (-0.2,-2.1){$x$};
                \node[ line width=0.6pt, dashed, draw opacity=0.5] (a) at (1,-2.1){$m$};
        \end{tikzpicture}
    \end{aligned}.\label{eq:BulkPhi}
\end{equation}
Now, consider the boundary tube algebra $\mathbf{Tube}^{2;2}({_{\ED}}\EM)$ with the basis given by (see Fig.~\ref{fig:tubeBulkBoundary})
\begin{equation}
 \left\{  \begin{aligned}
        \begin{tikzpicture}
             \draw[line width=.6pt,black] (0,0.5)--(0,1.3);
             \draw[line width=.6pt,black] (0,-0.5)--(0,-1.3);
             \draw[red] (0,0.8) arc[start angle=90, end angle=270, radius=0.8];
             \draw[line width=.6pt,black] (0,-0.5)--(0,-1.3);
             \draw[line width=.6pt,black] (120:0.8)-- ++(120:0.3);
             \draw[line width=.6pt,black] (150:0.8)--++(150:-0.3);
               \draw[line width=.6pt,black] (210:0.8)-- ++(210:-0.3);
             \draw[line width=.6pt,black] (240:0.8)--++(240:0.3);
        \end{tikzpicture}
    \end{aligned}: \text{edge}\in \ED,\EM; \text{vertex}\in \Hom_{\ED},\Hom_{\EM} \right\}
\label{eq:BdTubebasis22}
\end{equation}
where we have omitted the edge and vertex labels.
There is a natural action of $\mathbf{Tube}^{2;2}({_{\ED}}\EM)$ on the basis elements given in Eq.~\eqref{eq:BulkPhi}, which means that we can map
bulk excitation $M\in \Rep(\mathbf{Tube}({_{\ED}}\ED_{\ED}))$ to boundary excitation  $F_1(M)\in \Rep(\mathbf{Tube}^{2;2}({_{\ED}}\EM))$ with $F_1$ denotes the corresponding functor.
Notice that  $\mathbf{Tube}^{2;2}({_{\ED}}\EM)$ is Morita equivalent to $\mathbf{Tube}^{0;0}({_{\ED}}\EM)=\mathbf{Tube}({_{\ED}}\EM)$; we denote the equivalence functor as $F_2:\Rep(\mathbf{Tube}^{2;2}({_{\ED}}\EM)) \to \Rep(\mathbf{Tube}({_{\ED}}\EM))$
The anyon condensation is thus described by $\mathbf{Cond}=F_2\circ F_1$.

\begin{remark}
Notice that Eq.~\eqref{eq:VerticalFusion} is the vertical fusion of two bulk excitations, where we need to use the comultiplication of the tube algebra. The vertical fusion corresponds to the composition $X\circ Y$ of functors $X,Y\in \Fun_{\ED|\ED}(\ED,\ED)$.
The anyon condensation we discussed above is indeed the horizontal fusion; it can also be generalized to the bulk topological excitations. Notice that in this case, the fusion corresponds to $X\boxtimes_{\ED}Y$ for $X,Y\in\Fun_{\ED|\ED}(\ED,\ED)$.
\end{remark}

\section{Gapped domain wall theory}
\label{sec:walltheory}

A gapped domain wall that separates two topological phases (UMTCs), denoted as $\EP_1$ and $\EP_2$, can be converted into a gapped boundary of the stacked topological phase $\EP_1\boxtimes \bar{\EP}_2$. Here, $\bar{\EP}_2$ represents the phase obtained by applying the time-reversal transformation to $\EP_2$. Mathematically, this transformation results in a UMTC with the same objects as $\EP_2$ but featuring anti-braiding structures.
This means that we can apply the result of gapped boundary in Sec.~\ref{sec:bdtheory} to the gapped domain wall.
In this section, rather than adopting the aforementioned approach, we will delve into the domain wall tube algebra and domain wall excitations within the framework of the domain wall picture.
This viewpoint is valuable for advancing our understanding of gapped domain walls and gapped boundaries. It is noteworthy that gapped boundaries can be conceived as gapped domain walls that delineate a topological phase from a trivial or vacuum phase denoted as $\mathsf{Hilb}$.

For vertices $v_w$ on the domain wall, we define the vertex operator $Q_{v_w}$ in a similar way as that for the boundary (notice that the domain wall vertex is either for left $\EC$-module structure over $\EM$ or for right $\ED$-module structure over $\EM$).
The face operator is also constructed by inserting loops inside the boundary faces.
The domain wall Hamiltonian is thus of the form
\begin{equation}
    H[\Sigma_1\cap \Sigma_2, {_{\EC}\EM_{\ED}}]=-\sum_{v_w}Q_{v_w}-\sum_{f_{w,L}} B_{f_{w,L}}-\sum_{f_{w,R}} B_{f_{w,R}}.
\end{equation}
If we set $\EC=\ED$ and $\EM=\ED$, we obtain the trivial ${_{\ED}}\ED_{\ED}$ domain wall. It is observed that both the vertex operator and the two-sided face operators become identical to those for the $\ED$ bulk.

\begin{remark}
Consider two domain walls $_{\EC}\EM_{\ED}$ that separating multifusion string-nets with input UMFCs $\EC$ and $\ED$, and $_{\ED}\EN_{\EE}$ that separating multifusion string-nets with input UMFCs $\ED$ and $\EE$, then two domain walls can fuse to obtain a new domain wall that separating multifusion string-nets with input UMFCs $\EC$ and $\EE$; this fusion process is mathematically characterized by the tensor product of module categories: $\EM\boxtimes_{\ED}\EN$.

A gapped domain wall $\EM$ is called invertible or transparent if it can be annihilated by its inverse domain wall $\EM^{\rm op}$.
Mathematically, for the $\EC|\ED$-bimodule category $\EM$, its opposite category $\EM^{\rm op}$ is a $\ED|\EC$-bimodule category; $\EM$ is called invertible if and only if there are equivalences $\EM\boxtimes_{\ED}\EM^{\rm op}\simeq \EC$ and $\EM^{\rm op}\boxtimes_{\EC}\EM\simeq \ED$ as bimodule categories \cite{etingof2010fusion}.
The invertible domain walls are crucial because that the topological excitations can tunnel through the domain wall from one side to another. 
\end{remark}

\subsection{Domain wall gauge symmetry}

Consider two $2d$ bulks of multifusion string-net with their respective UMFCs denoted as $\ED_1$ and $\ED_2$. The associated weak Hopf gauge symmetries are represented by $W_1$ and $W_2$ for  $\ED_1=\Rep(W_1)$ and $\ED_2=\Rep(W_2)$ respectively. The gauge symmetry of the $1d$ domain wall is characterized by the $W_1|W_2$-bicomodule algebra $\mathfrak{J}$, or equivalently by $W_1|W_2$-bimodule algebra $\mathfrak{K}$.

Assume the domain wall gauge symmetry to be $\mathfrak{J}$, then the module category $_{\mathfrak{J}}\Mod$ over $\mathfrak{J}$ is a $\ED_1|\ED_2$-bimodule category.
Thus the irreducible modules over $\mathfrak{J}$ are labels of the domain wall string.

For two domain walls with gauge symmetry $\mathfrak{J}_1$ and $\mathfrak{J}_2$ respectively, the point defect between two domain walls are characterized by $\mathfrak{J}_1|\mathfrak{J}_2$-bimodules.
If $\mathfrak{J}_1=\mathfrak{J}_2=\mathfrak{J}$, these bimodules describe the topological excitations of the domain wall with gauge symmetry $\mathfrak{J}$.

\subsection{Domain wall tube algebra}

For a gapped domain wall $_{\EC}\EM_{\ED}$, we have the tube algebra $\mathbf{Tube}({_{\EC}}\EM_{\ED})$ spanned by the following basis:
\begin{equation}
 \left\{  
   \begin{aligned}
        \begin{tikzpicture}
             \draw[line width=.6pt,black] (0,0.5)--(0,1.5);
             \draw[line width=.6pt,black] (0,-0.5)--(0,-1.5);
             \draw[red] (0,0.8) arc[start angle=90, end angle=270, radius=0.8];
             \draw[blue] (0,1.3) arc[start angle=90, end angle=-90, radius=1.3];
            \node[ line width=0.6pt, dashed, draw opacity=0.5] (a) at (0,1.7){$h$};
             \node[ line width=0.6pt, dashed, draw opacity=0.5] (a) at (0,-1.7){$c$};
            \node[ line width=0.6pt, dashed, draw opacity=0.5] (a) at (-1,0){$a$};
            \node[ line width=0.6pt, dashed, draw opacity=0.5] (a) at (1.5,0){$b$};
            \node[ line width=0.6pt, dashed, draw opacity=0.5] (a) at (-0.2,-1){$d$};
            \node[ line width=0.6pt, dashed, draw opacity=0.5] (a) at (-0.4,-1.3){$\mu$};
            \node[ line width=0.6pt, dashed, draw opacity=0.5] (a) at (0.2,-0.8){$\nu$};
            \node[ line width=0.6pt, dashed, draw opacity=0.5] (a) at (0,-0.3){$e$};
            \node[ line width=0.6pt, dashed, draw opacity=0.5] (a) at (0,0.3){$f$};
            \node[ line width=0.6pt, dashed, draw opacity=0.5] (a) at (-0.2,1){$g$};
            \node[ line width=0.6pt, dashed, draw opacity=0.5] (a) at (-0.4,1.3){$\gamma$};
            \node[ line width=0.6pt, dashed, draw opacity=0.5] (a) at (0.2,0.8){$\zeta$};
        \end{tikzpicture}
    \end{aligned}
   : 
   \begin{aligned}
     & a\in \Irr(\EC),b\in \Irr(\ED),f,g,h,c,d,e\in \Irr(\EM), \\
     & \mu,\nu,\gamma,\zeta\in \Hom_{\EM}  
   \end{aligned}
    \right\}.
    \label{eq:tubebasis-M}
\end{equation}

\begin{remark}
When considering the  defects between two domain walls $_{\EC}\EM_{\ED}$ and $_{\EC}\EN_{\ED}$, we can construct a tube algebra $\mathbf{Tube}[_{\EC}\EM_{\ED};{_{\EC}\EN_{\ED}}]$ by setting the labels $f,g,h\in \Irr(\EM)$, $\zeta,\gamma\in \Hom_{\EM}$ and $c,d,e\in  \Irr(\EN)$, $\mu,\nu\in \Hom_{\EN}$. 
It is clear that this forms an algebra, but there is no longer a coalgebra structure.
The defects between domain walls are thus characterized by the modules over $\mathbf{Tube}[_{\EC}\EM_{\ED};{_{\EC}\EN_{\ED}}]$.
\end{remark}

\begin{proposition}
    The domain wall tube algebra $\mathbf{Tube}({_{\EC}}\EM_{\ED})$ is a $C^*$ weak Hopf algebra with the following structure morphisms: 
    \begin{gather}
        1=\sum_{a,b} \;\begin{aligned}
        \begin{tikzpicture}
             \draw[line width=.6pt,black] (0,0.5)--(0,1.5);
             \draw[line width=.6pt,black] (0,-0.5)--(0,-1.5);
             \draw[red, dotted] (0,0.8) arc[start angle=90, end angle=270, radius=0.8];
             \draw[blue,dotted] (0,1.3) arc[start angle=90, end angle=-90, radius=1.3];
            \node[ line width=0.6pt, dashed, draw opacity=0.5] (a) at (-0.2,-1){$a$};
            \node[ line width=0.6pt, dashed, draw opacity=0.5] (a) at (-0.2,1){$b$};
        \end{tikzpicture}
    \end{aligned}\;, \\
    \mu\left( \begin{aligned}
        \begin{tikzpicture}
             \draw[line width=.6pt,black] (0,0.5)--(0,1.5);
             \draw[line width=.6pt,black] (0,-0.5)--(0,-1.5);
             \draw[red] (0,0.8) arc[start angle=90, end angle=270, radius=0.8];
             \draw[blue] (0,1.3) arc[start angle=90, end angle=-90, radius=1.3];
            \node[ line width=0.6pt, dashed, draw opacity=0.5] (a) at (0,1.7){$h$};
             \node[ line width=0.6pt, dashed, draw opacity=0.5] (a) at (0,-1.7){$c$};
            \node[ line width=0.6pt, dashed, draw opacity=0.5] (a) at (-1,0){$a$};
            \node[ line width=0.6pt, dashed, draw opacity=0.5] (a) at (1.5,0){$b$};
            \node[ line width=0.6pt, dashed, draw opacity=0.5] (a) at (-0.2,-1){$d$};
            \node[ line width=0.6pt, dashed, draw opacity=0.5] (a) at (-0.4,-1.3){$\mu$};
            \node[ line width=0.6pt, dashed, draw opacity=0.5] (a) at (0.2,-0.8){$\nu$};
            \node[ line width=0.6pt, dashed, draw opacity=0.5] (a) at (0,-0.3){$e$};
            \node[ line width=0.6pt, dashed, draw opacity=0.5] (a) at (0,0.3){$f$};
            \node[ line width=0.6pt, dashed, draw opacity=0.5] (a) at (-0.2,1){$g$};
            \node[ line width=0.6pt, dashed, draw opacity=0.5] (a) at (-0.4,1.3){$\gamma$};
            \node[ line width=0.6pt, dashed, draw opacity=0.5] (a) at (0.2,0.8){$\zeta$};
        \end{tikzpicture}
    \end{aligned}\otimes \begin{aligned}
        \begin{tikzpicture}
             \draw[line width=.6pt,black] (0,0.5)--(0,1.5);
             \draw[line width=.6pt,black] (0,-0.5)--(0,-1.5);
             \draw[red] (0,0.8) arc[start angle=90, end angle=270, radius=0.8];
             \draw[blue] (0,1.3) arc[start angle=90, end angle=-90, radius=1.3];
            \node[ line width=0.6pt, dashed, draw opacity=0.5] (a) at (0,1.7){$h'$};
             \node[ line width=0.6pt, dashed, draw opacity=0.5] (a) at (0,-1.7){$c'$};
            \node[ line width=0.6pt, dashed, draw opacity=0.5] (a) at (-1,0){$a'$};
            \node[ line width=0.6pt, dashed, draw opacity=0.5] (a) at (1.5,0){$b'$};
            \node[ line width=0.6pt, dashed, draw opacity=0.5] (a) at (-0.2,-1){$d'$};
            \node[ line width=0.6pt, dashed, draw opacity=0.5] (a) at (-0.4,-1.4){$\mu'$};
            \node[ line width=0.6pt, dashed, draw opacity=0.5] (a) at (0.2,-0.8){$\nu'$};
            \node[ line width=0.6pt, dashed, draw opacity=0.5] (a) at (0,-0.3){$e'$};
            \node[ line width=0.6pt, dashed, draw opacity=0.5] (a) at (0,0.3){$f'$};
            \node[ line width=0.6pt, dashed, draw opacity=0.5] (a) at (-0.2,1){$g'$};
            \node[ line width=0.6pt, dashed, draw opacity=0.5] (a) at (-0.4,1.3){$\gamma'$};
            \node[ line width=0.6pt, dashed, draw opacity=0.5] (a) at (0.2,0.8){$\zeta'$};
        \end{tikzpicture}
    \end{aligned}\right)=\delta_{f,h'}\delta_{e,c'}\begin{aligned}
        \begin{tikzpicture}
             \draw[line width=.6pt,black] (0,0.5)--(0,2.5);
             \draw[line width=.6pt,black] (0,-0.5)--(0,-2.5);
             \draw[red] (0,0.8) arc[start angle=90, end angle=270, radius=0.8];
             \draw[blue] (0,1.3) arc[start angle=90, end angle=-90, radius=1.3];
            \node[ line width=0.6pt, dashed, draw opacity=0.5] (a) at (0.2,1.5){$f$};
             \node[ line width=0.6pt, dashed, draw opacity=0.5] (a) at (0.2,-1.5){$e$};
            \node[ line width=0.6pt, dashed, draw opacity=0.5] (a) at (-1,0){$a'$};
            \node[ line width=0.6pt, dashed, draw opacity=0.5] (a) at (1.5,0){$b'$};
            \node[ line width=0.6pt, dashed, draw opacity=0.5] (a) at (-0.2,-1){$d'$};
            \node[ line width=0.6pt, dashed, draw opacity=0.5] (a) at (-0.2,-1.3){$\mu'$};
            \node[ line width=0.6pt, dashed, draw opacity=0.5] (a) at (0.2,-0.8){$\nu'$};
            \node[ line width=0.6pt, dashed, draw opacity=0.5] (a) at (0,-0.3){$e'$};
            \node[ line width=0.6pt, dashed, draw opacity=0.5] (a) at (0,0.3){$f'$};
            \node[ line width=0.6pt, dashed, draw opacity=0.5] (a) at (-0.2,1){$g'$};
            \node[ line width=0.6pt, dashed, draw opacity=0.5] (a) at (-0.4,1.3){$\gamma'$};
            \node[ line width=0.6pt, dashed, draw opacity=0.5] (a) at (0.2,0.8){$\zeta'$};
            \draw[red] (0,1.6) arc[start angle=90, end angle=270, radius=1.6];
            \node[ line width=0.6pt, dashed, draw opacity=0.5] (a) at (-0.3,1.8){$\zeta$};
          \node[ line width=0.6pt, dashed, draw opacity=0.5] (a) at (-0.3,-1.8){$\nu$};
           \draw[blue] (0,2.1) arc[start angle=90, end angle=-90, radius=2.1];
            \node[ line width=0.6pt, dashed, draw opacity=0.5] (a) at (-1.3,0.4){$a$};
                \node[ line width=0.6pt, dashed, draw opacity=0.5] (a) at (1.8,0.4){$b$};
         \node[ line width=0.6pt, dashed, draw opacity=0.5] (a) at (0.2,1.8){$g$};
          \node[ line width=0.6pt, dashed, draw opacity=0.5] (a) at (0.2,-1.8){$d$};
        \node[ line width=0.6pt, dashed, draw opacity=0.5] (a) at (0.2,2.3){$\gamma$};
        \node[ line width=0.6pt, dashed, draw opacity=0.5] (a) at (0.2,-2.3){$\mu$};
       \node[ line width=0.6pt, dashed, draw opacity=0.5] (a) at (-0.3,2.5){$h$};
       \node[ line width=0.6pt, dashed, draw opacity=0.5] (a) at (-0.3,-2.5){$c$};
        \end{tikzpicture}
    \end{aligned}\;, \\
    \varepsilon\left(\begin{aligned}\begin{tikzpicture}
             \draw[line width=.6pt,black] (0,0.5)--(0,1.5);
             \draw[line width=.6pt,black] (0,-0.5)--(0,-1.5);
             \draw[red] (0,0.8) arc[start angle=90, end angle=270, radius=0.8];
             \draw[blue] (0,1.3) arc[start angle=90, end angle=-90, radius=1.3];
            \node[ line width=0.6pt, dashed, draw opacity=0.5] (a) at (0,1.7){$h$};
             \node[ line width=0.6pt, dashed, draw opacity=0.5] (a) at (0,-1.7){$c$};
            \node[ line width=0.6pt, dashed, draw opacity=0.5] (a) at (-1,0){$a$};
            \node[ line width=0.6pt, dashed, draw opacity=0.5] (a) at (1.5,0){$b$};
            \node[ line width=0.6pt, dashed, draw opacity=0.5] (a) at (-0.2,-1){$d$};
            \node[ line width=0.6pt, dashed, draw opacity=0.5] (a) at (-0.4,-1.3){$\mu$};
            \node[ line width=0.6pt, dashed, draw opacity=0.5] (a) at (0.2,-0.8){$\nu$};
            \node[ line width=0.6pt, dashed, draw opacity=0.5] (a) at (0,-0.3){$e$};
            \node[ line width=0.6pt, dashed, draw opacity=0.5] (a) at (0,0.3){$f$};
            \node[ line width=0.6pt, dashed, draw opacity=0.5] (a) at (-0.2,1){$g$};
            \node[ line width=0.6pt, dashed, draw opacity=0.5] (a) at (-0.4,1.3){$\gamma$};
            \node[ line width=0.6pt, dashed, draw opacity=0.5] (a) at (0.2,0.8){$\zeta$};
        \end{tikzpicture}
    \end{aligned}\right)
    =\frac{\delta_{e,f}\delta_{c,h}}{d_h}
    \begin{aligned}\begin{tikzpicture}
             \draw[line width=.6pt,black] (0,-1.8)--(0,1.8);
             \draw[red] (0,0.8) arc[start angle=90, end angle=270, radius=0.8];
             \draw[blue] (0,1.3) arc[start angle=90, end angle=-90, radius=1.3];
             \draw[black] (0,1.8) arc[start angle=90, end angle=-90, radius=1.8];
            \node[ line width=0.6pt, dashed, draw opacity=0.5] (a) at (1.6,0){$h$};
            \node[ line width=0.6pt, dashed, draw opacity=0.5] (a) at (-0.6,0){$a$};
            \node[ line width=0.6pt, dashed, draw opacity=0.5] (a) at (1.1,0){$b$};
            \node[ line width=0.6pt, dashed, draw opacity=0.5] (a) at (-0.2,-1){$d$};
            \node[ line width=0.6pt, dashed, draw opacity=0.5] (a) at (-0.4,-1.3){$\mu$};
            \node[ line width=0.6pt, dashed, draw opacity=0.5] (a) at (0.2,-0.8){$\nu$};
            \node[ line width=0.6pt, dashed, draw opacity=0.5] (a) at (0.3,0){$f$};
            \node[ line width=0.6pt, dashed, draw opacity=0.5] (a) at (-0.2,1){$g$};
            \node[ line width=0.6pt, dashed, draw opacity=0.5] (a) at (-0.4,1.3){$\gamma$};
            \node[ line width=0.6pt, dashed, draw opacity=0.5] (a) at (0.2,0.8){$\zeta$};
        \end{tikzpicture}
    \end{aligned}=\delta_{e,f}\delta_{c,h}\delta_{d,g}\delta_{\nu,\zeta}\delta_{\mu,\gamma} \sqrt{\frac{d_a d_f d_b}{d_h}}\;, \\
    \Delta\left(\begin{aligned}\begin{tikzpicture}
             \draw[line width=.6pt,black] (0,0.5)--(0,1.5);
             \draw[line width=.6pt,black] (0,-0.5)--(0,-1.5);
             \draw[red] (0,0.8) arc[start angle=90, end angle=270, radius=0.8];
             \draw[blue] (0,1.3) arc[start angle=90, end angle=-90, radius=1.3];
            \node[ line width=0.6pt, dashed, draw opacity=0.5] (a) at (0,1.7){$h$};
             \node[ line width=0.6pt, dashed, draw opacity=0.5] (a) at (0,-1.7){$c$};
            \node[ line width=0.6pt, dashed, draw opacity=0.5] (a) at (-1,0){$a$};
            \node[ line width=0.6pt, dashed, draw opacity=0.5] (a) at (1.5,0){$b$};
            \node[ line width=0.6pt, dashed, draw opacity=0.5] (a) at (-0.2,-1){$d$};
            \node[ line width=0.6pt, dashed, draw opacity=0.5] (a) at (-0.4,-1.3){$\mu$};
            \node[ line width=0.6pt, dashed, draw opacity=0.5] (a) at (0.2,-0.8){$\nu$};
            \node[ line width=0.6pt, dashed, draw opacity=0.5] (a) at (0,-0.3){$e$};
            \node[ line width=0.6pt, dashed, draw opacity=0.5] (a) at (0,0.3){$f$};
            \node[ line width=0.6pt, dashed, draw opacity=0.5] (a) at (-0.2,1){$g$};
            \node[ line width=0.6pt, dashed, draw opacity=0.5] (a) at (-0.4,1.3){$\gamma$};
            \node[ line width=0.6pt, dashed, draw opacity=0.5] (a) at (0.2,0.8){$\zeta$};
        \end{tikzpicture}
    \end{aligned}\right)
    =\sum_{i,j,k,\rho,\sigma} \sqrt{ \frac{d_j}{d_ad_i}}\sqrt{ \frac{d_k}{d_jd_b}}
    \begin{aligned}\begin{tikzpicture}
             \draw[line width=.6pt,black] (0,0.5)--(0,1.5);
             \draw[line width=.6pt,black] (0,-0.5)--(0,-1.5);
             \draw[red] (0,0.8) arc[start angle=90, end angle=270, radius=0.8];
             \draw[blue] (0,1.3) arc[start angle=90, end angle=-90, radius=1.3];
            \node[ line width=0.6pt, dashed, draw opacity=0.5] (a) at (0,1.7){$h$};
             \node[ line width=0.6pt, dashed, draw opacity=0.5] (a) at (0,-1.7){$k$};
            \node[ line width=0.6pt, dashed, draw opacity=0.5] (a) at (-1,0){$a$};
            \node[ line width=0.6pt, dashed, draw opacity=0.5] (a) at (1.5,0){$b$};
            \node[ line width=0.6pt, dashed, draw opacity=0.5] (a) at (-0.2,-1){$j$};
            \node[ line width=0.6pt, dashed, draw opacity=0.5] (a) at (-0.3,-1.4){$\rho$};
            \node[ line width=0.6pt, dashed, draw opacity=0.5] (a) at (0.2,-0.8){$\sigma$};
            \node[ line width=0.6pt, dashed, draw opacity=0.5] (a) at (0,-0.3){$i$};
            \node[ line width=0.6pt, dashed, draw opacity=0.5] (a) at (0,0.3){$f$};
            \node[ line width=0.6pt, dashed, draw opacity=0.5] (a) at (-0.2,1){$g$};
            \node[ line width=0.6pt, dashed, draw opacity=0.5] (a) at (-0.4,1.3){$\gamma$};
            \node[ line width=0.6pt, dashed, draw opacity=0.5] (a) at (0.2,0.8){$\zeta$};
        \end{tikzpicture}
    \end{aligned} 
    \otimes 
    \begin{aligned}\begin{tikzpicture}
             \draw[line width=.6pt,black] (0,0.5)--(0,1.5);
             \draw[line width=.6pt,black] (0,-0.5)--(0,-1.5);
             \draw[red] (0,0.8) arc[start angle=90, end angle=270, radius=0.8];
             \draw[blue] (0,1.3) arc[start angle=90, end angle=-90, radius=1.3];
            \node[ line width=0.6pt, dashed, draw opacity=0.5] (a) at (0,1.7){$k$};
             \node[ line width=0.6pt, dashed, draw opacity=0.5] (a) at (0,-1.7){$c$};
            \node[ line width=0.6pt, dashed, draw opacity=0.5] (a) at (-1,0){$a$};
            \node[ line width=0.6pt, dashed, draw opacity=0.5] (a) at (1.5,0){$b$};
            \node[ line width=0.6pt, dashed, draw opacity=0.5] (a) at (-0.2,-1){$d$};
            \node[ line width=0.6pt, dashed, draw opacity=0.5] (a) at (-0.4,-1.3){$\mu$};
            \node[ line width=0.6pt, dashed, draw opacity=0.5] (a) at (0.2,-0.8){$\nu$};
            \node[ line width=0.6pt, dashed, draw opacity=0.5] (a) at (0,-0.3){$e$};
            \node[ line width=0.6pt, dashed, draw opacity=0.5] (a) at (0,0.3){$i$};
            \node[ line width=0.6pt, dashed, draw opacity=0.5] (a) at (-0.2,1){$j$};
            \node[ line width=0.6pt, dashed, draw opacity=0.5] (a) at (-0.4,1.3){$\rho$};
            \node[ line width=0.6pt, dashed, draw opacity=0.5] (a) at (0.2,0.8){$\sigma$};
        \end{tikzpicture}
\end{aligned}\;, \\
S\left( 
          \begin{aligned}\begin{tikzpicture}
             \draw[line width=.6pt,black] (0,0.5)--(0,1.5);
             \draw[line width=.6pt,black] (0,-0.5)--(0,-1.5);
             \draw[red] (0,0.8) arc[start angle=90, end angle=270, radius=0.8];
             \draw[blue] (0,1.3) arc[start angle=90, end angle=-90, radius=1.3];
            \node[ line width=0.6pt, dashed, draw opacity=0.5] (a) at (0,1.7){$h$};
             \node[ line width=0.6pt, dashed, draw opacity=0.5] (a) at (0,-1.7){$c$};
            \node[ line width=0.6pt, dashed, draw opacity=0.5] (a) at (-1,0){$a$};
            \node[ line width=0.6pt, dashed, draw opacity=0.5] (a) at (1.5,0){$b$};
            \node[ line width=0.6pt, dashed, draw opacity=0.5] (a) at (-0.2,-1){$d$};
            \node[ line width=0.6pt, dashed, draw opacity=0.5] (a) at (-0.4,-1.3){$\mu$};
            \node[ line width=0.6pt, dashed, draw opacity=0.5] (a) at (0.2,-0.8){$\nu$};
            \node[ line width=0.6pt, dashed, draw opacity=0.5] (a) at (0,-0.3){$e$};
            \node[ line width=0.6pt, dashed, draw opacity=0.5] (a) at (0,0.3){$f$};
            \node[ line width=0.6pt, dashed, draw opacity=0.5] (a) at (-0.2,1){$g$};
            \node[ line width=0.6pt, dashed, draw opacity=0.5] (a) at (-0.4,1.3){$\gamma$};
            \node[ line width=0.6pt, dashed, draw opacity=0.5] (a) at (0.2,0.8){$\zeta$};
        \end{tikzpicture}
    \end{aligned}
          \right) =\frac{d_f}{d_h}\;\begin{aligned}\begin{tikzpicture}
             \draw[line width=.6pt,black] (0,0.5)--(0,1.5);
             \draw[line width=.6pt,black] (0,-0.5)--(0,-1.5);
             \draw[red] (0,1.3) arc[start angle=90, end angle=270, radius=1.3];
             \draw[blue] (0,0.8) arc[start angle=90, end angle=-90, radius=0.8];
            \node[ line width=0.6pt, dashed, draw opacity=0.5] (a) at (0,1.7){$e$};
             \node[ line width=0.6pt, dashed, draw opacity=0.5] (a) at (0,-1.7){$f$};
            \node[ line width=0.6pt, dashed, draw opacity=0.5] (a) at (-1,0){$\bar{a}$};
            \node[ line width=0.6pt, dashed, draw opacity=0.5] (a) at (1,0){$\bar{b}$};
            \node[ line width=0.6pt, dashed, draw opacity=0.5] (a) at (-0.2,-1){$g$};
            \node[ line width=0.6pt, dashed, draw opacity=0.5] (a) at (0.2,-1.3){$\zeta$};
            \node[ line width=0.6pt, dashed, draw opacity=0.5] (a) at (-0.2,-0.7){$\gamma$};
            \node[ line width=0.6pt, dashed, draw opacity=0.5] (a) at (0,-0.3){$h$};
            \node[ line width=0.6pt, dashed, draw opacity=0.5] (a) at (0,0.3){$c$};
            \node[ line width=0.6pt, dashed, draw opacity=0.5] (a) at (-0.2,1){$d$};
            \node[ line width=0.6pt, dashed, draw opacity=0.5] (a) at (-0.4,1.4){$\nu$};
            \node[ line width=0.6pt, dashed, draw opacity=0.5] (a) at (0.3,1.0){$\mu$};
        \end{tikzpicture}
    \end{aligned}\;, \quad \left( \begin{aligned}\begin{tikzpicture}
             \draw[line width=.6pt,black] (0,0.5)--(0,1.5);
             \draw[line width=.6pt,black] (0,-0.5)--(0,-1.5);
             \draw[red] (0,0.8) arc[start angle=90, end angle=270, radius=0.8];
             \draw[blue] (0,1.3) arc[start angle=90, end angle=-90, radius=1.3];
            \node[ line width=0.6pt, dashed, draw opacity=0.5] (a) at (0,1.7){$h$};
             \node[ line width=0.6pt, dashed, draw opacity=0.5] (a) at (0,-1.7){$c$};
            \node[ line width=0.6pt, dashed, draw opacity=0.5] (a) at (-1,0){$a$};
            \node[ line width=0.6pt, dashed, draw opacity=0.5] (a) at (1.5,0){$b$};
            \node[ line width=0.6pt, dashed, draw opacity=0.5] (a) at (-0.2,-1){$d$};
            \node[ line width=0.6pt, dashed, draw opacity=0.5] (a) at (-0.4,-1.3){$\mu$};
            \node[ line width=0.6pt, dashed, draw opacity=0.5] (a) at (0.2,-0.8){$\nu$};
            \node[ line width=0.6pt, dashed, draw opacity=0.5] (a) at (0,-0.3){$e$};
            \node[ line width=0.6pt, dashed, draw opacity=0.5] (a) at (0,0.3){$f$};
            \node[ line width=0.6pt, dashed, draw opacity=0.5] (a) at (-0.2,1){$g$};
            \node[ line width=0.6pt, dashed, draw opacity=0.5] (a) at (-0.4,1.3){$\gamma$};
            \node[ line width=0.6pt, dashed, draw opacity=0.5] (a) at (0.2,0.8){$\zeta$};
        \end{tikzpicture}
    \end{aligned} \right)^*= \frac{d_e}{d_c} \;
\begin{aligned}\begin{tikzpicture}
             \draw[line width=.6pt,black] (0,0.5)--(0,1.5);
             \draw[line width=.6pt,black] (0,-0.5)--(0,-1.5);
             \draw[red] (0,1.3) arc[start angle=90, end angle=270, radius=1.3];
             \draw[blue] (0,0.8) arc[start angle=90, end angle=-90, radius=0.8];
            \node[ line width=0.6pt, dashed, draw opacity=0.5] (a) at (0,1.7){$f$};
             \node[ line width=0.6pt, dashed, draw opacity=0.5] (a) at (0,-1.7){$e$};
            \node[ line width=0.6pt, dashed, draw opacity=0.5] (a) at (-1,0){$\bar{a}$};
            \node[ line width=0.6pt, dashed, draw opacity=0.5] (a) at (1,0){$\bar{b}$};
            \node[ line width=0.6pt, dashed, draw opacity=0.5] (a) at (-0.2,-1){$d$};
            \node[ line width=0.6pt, dashed, draw opacity=0.5] (a) at (0.2,-1.3){$\nu$};
            \node[ line width=0.6pt, dashed, draw opacity=0.5] (a) at (-0.2,-0.7){$\mu$};
            \node[ line width=0.6pt, dashed, draw opacity=0.5] (a) at (0,-0.3){$c$};
            \node[ line width=0.6pt, dashed, draw opacity=0.5] (a) at (0,0.3){$h$};
            \node[ line width=0.6pt, dashed, draw opacity=0.5] (a) at (-0.2,1){$g$};
            \node[ line width=0.6pt, dashed, draw opacity=0.5] (a) at (-0.4,1.4){$\zeta$};
            \node[ line width=0.6pt, dashed, draw opacity=0.5] (a) at (0.3,1.0){$\gamma$};
        \end{tikzpicture}
    \end{aligned}\;. 
    \end{gather}
\end{proposition}

\begin{proof}
    The proof is the same as that for $\mathbf{Tube}({_\ED}\ED_\ED)$. 
\end{proof}

There are two canonical subalgebras of $\mathbf{Tube}({_{\EC}}\EM_{\ED})$ which can be constructed in a similar way as that in Eqs.~\eqref{eq:tubebasis-L} and \eqref{eq:tubebasis-R}.
We will denote them as $\mathbf{L}({_{\EC}}\EM)$ and $\mathbf{R}(\EM_{\ED})$.

\begin{proposition}
    The domain wall tube algebra $\mathbf{Tube}({_{\EC}}\EM_{\ED})$  is isomorphic to the crossed product $\mathbf{R}(\EM_{\ED}) \Join \mathbf{L}({_{\EC}}\EM)$.
\end{proposition}

\begin{proof}
    The isomorphism $\Phi:\mathbf{R}(\EM_{\ED}) \Join \mathbf{L}({_{\EC}}\EM)\to \mathbf{Tube}({_{\EC}}\EM_{\ED})$ is given by 
    \begin{equation}
        \Phi:\begin{aligned}
        \begin{tikzpicture}
             \draw[line width=.6pt,black] (0,0.5)--(0,1.5);
             \draw[line width=.6pt,black] (0,-0.5)--(0,-1.5);
             \draw[red,dotted] (0,0.8) arc[start angle=90, end angle=270, radius=0.8];
             \draw[blue] (0,1.3) arc[start angle=90, end angle=-90, radius=1.3];
            \node[ line width=0.6pt, dashed, draw opacity=0.5] (a) at (0,1.7){$h$};
             \node[ line width=0.6pt, dashed, draw opacity=0.5] (a) at (0,-1.7){$c$};
            \node[ line width=0.6pt, dashed, draw opacity=0.5] (a) at (-1,0){$\one$};
            \node[ line width=0.6pt, dashed, draw opacity=0.5] (a) at (1.5,0){$b$};
            \node[ line width=0.6pt, dashed, draw opacity=0.5] (a) at (0,-0.3){$d$};
            \node[ line width=0.6pt, dashed, draw opacity=0.5] (a) at (-0.4,-1.3){$\mu$};
            \node[ line width=0.6pt, dashed, draw opacity=0.5] (a) at (0,0.3){$g$};
            \node[ line width=0.6pt, dashed, draw opacity=0.5] (a) at (-0.4,1.3){$\gamma$};
        \end{tikzpicture}
    \end{aligned}
    \Join 
    \begin{aligned}
        \begin{tikzpicture}
             \draw[line width=.6pt,black] (0,0.5)--(0,1.5);
             \draw[line width=.6pt,black] (0,-0.5)--(0,-1.5);
             \draw[red] (0,0.8) arc[start angle=90, end angle=270, radius=0.8];
             \draw[blue,dotted] (0,1.3) arc[start angle=90, end angle=-90, radius=1.3];
            \node[ line width=0.6pt, dashed, draw opacity=0.5] (a) at (-1,0){$a$};
            \node[ line width=0.6pt, dashed, draw opacity=0.5] (a) at (1.5,0){$\one$};
            \node[ line width=0.6pt, dashed, draw opacity=0.5] (a) at (-0.2,-1.3){$d'$};
            \node[ line width=0.6pt, dashed, draw opacity=0.5] (a) at (0.2,-0.8){$\nu$};
            \node[ line width=0.6pt, dashed, draw opacity=0.5] (a) at (0,-0.3){$e$};
            \node[ line width=0.6pt, dashed, draw opacity=0.5] (a) at (0,0.3){$f$};
            \node[ line width=0.6pt, dashed, draw opacity=0.5] (a) at (-0.2,1.3){$g'$};
            \node[ line width=0.6pt, dashed, draw opacity=0.5] (a) at (0.2,0.8){$\zeta$};
        \end{tikzpicture}
    \end{aligned}
    \mapsto  \delta_{g,g'}\delta_{d,d'}\begin{aligned}\begin{tikzpicture}
             \draw[line width=.6pt,black] (0,0.5)--(0,1.5);
             \draw[line width=.6pt,black] (0,-0.5)--(0,-1.5);
             \draw[red] (0,0.8) arc[start angle=90, end angle=270, radius=0.8];
             \draw[blue] (0,1.3) arc[start angle=90, end angle=-90, radius=1.3];
            \node[ line width=0.6pt, dashed, draw opacity=0.5] (a) at (0,1.7){$h$};
             \node[ line width=0.6pt, dashed, draw opacity=0.5] (a) at (0,-1.7){$c$};
            \node[ line width=0.6pt, dashed, draw opacity=0.5] (a) at (-1,0){$a$};
            \node[ line width=0.6pt, dashed, draw opacity=0.5] (a) at (1.5,0){$b$};
            \node[ line width=0.6pt, dashed, draw opacity=0.5] (a) at (-0.2,-1){$d$};
            \node[ line width=0.6pt, dashed, draw opacity=0.5] (a) at (-0.4,-1.3){$\mu$};
            \node[ line width=0.6pt, dashed, draw opacity=0.5] (a) at (0.2,-0.8){$\nu$};
            \node[ line width=0.6pt, dashed, draw opacity=0.5] (a) at (0,-0.3){$e$};
            \node[ line width=0.6pt, dashed, draw opacity=0.5] (a) at (0,0.3){$f$};
            \node[ line width=0.6pt, dashed, draw opacity=0.5] (a) at (-0.2,1){$g$};
            \node[ line width=0.6pt, dashed, draw opacity=0.5] (a) at (-0.4,1.3){$\gamma$};
            \node[ line width=0.6pt, dashed, draw opacity=0.5] (a) at (0.2,0.8){$\zeta$};
        \end{tikzpicture}
    \end{aligned}\;.
    \end{equation}
    The proof is the same as that for $\mathbf{Tube}({_\ED}\ED_\ED)$. 
\end{proof}

\section{Defective string-net as a restricted multifusion string-net}
\label{sec:DefectSN}

In this section, we will elucidate a crucial application of the multifusion string-net model: a defective topological phase.
The main result is as follows:
Every defective string-net model is a restricted multifusion string-net model.
This implies that the multifusion string-net provides a general framework for investigating various defects of the $2d$ topological phases.

\subsection{Defective string-net from UMFC}
\label{subsec:DefectSN}

Consider a $1d$ surface with defects, including the $1d$ domain wall or boundary and $0d$ defects between different domain walls and boundaries. 
Notice that all $1d$ defects must have an orientation.
For an input UMFC $\ED=\oplus_{i,j}\ED_{i,j}$, we can label $2d$ bulks with diagonal components $\ED_{i,i}$, $1d$ defects with non-diagonal component $\ED_{i,j}$, the domain wall $0d$ defects with $\Fun_{\ED_{i,i}|\ED_{j,j}}(\ED_{i,j},\ED_{i,j})$, and the boundary $0d$ defects $\Fun_{\ED_{i,i}}(\ED_{i,j},\ED_{i,j})$.
Since $\ED_{i,j}$ is a $\ED_{i,i}|\ED_{j,j}$-bimodule category, and 
\begin{equation}
    \ED_{i,j}\boxtimes_{\ED_{j,j}}\ED_{j,k}\simeq \ED_{i,k},
\end{equation}
using the fact that $\Fun_{\EC|\ED}(\EM,\EN)=\EM^{\rm op}\boxtimes_{\EC}\EN$, we see that $0d$ defects can also be labeled by objects in the UMFC.
For example, we have a well-defined labeling, e.g., as follows:
\begin{equation}
   \begin{aligned}
        \begin{tikzpicture}
             \draw[dotted,fill=red] (0,1) -- (0,0) -- (-0.866,-0.5) -- (0,1);
         \draw[dotted,fill=green] (0,1) -- (0,0) -- (0.866,-0.5) -- (0,1);
           \draw[dotted,fill=cyan]  (0,0) -- (0.866,-0.5) -- (-0.866,-0.5) -- (0,0);
             \draw[line width=1pt,black,-latex] (0,0)--(0,1);
             \draw[line width=1pt,black,-latex] (-0.866,-0.5) -- (0,0);
            \draw[line width=1pt,black,-latex] (0.866,-0.5) -- (0,0);  
            \node[ line width=0.6pt, dashed, draw opacity=0.5] (a) at (-.4,0.3){\small{$\ED_{i,i}$}};
            \node[ line width=0.6pt, dashed, draw opacity=0.5] (a) at (0.4,0.3){\small{$\ED_{j,j}$}};
           \node[ line width=0.6pt, dashed, draw opacity=0.5] (a) at (0,-0.5){\small{$\ED_{k,k}$}};
           \node[ line width=0.6pt, dashed, draw opacity=0.5] (a) at (0,1.3){\small{$\ED_{i,j}$}};
            \node[ line width=0.6pt, dashed, draw opacity=0.5] (a) at (-1,-0.7){\small{$\ED_{i,k}$}};
           \node[ line width=0.6pt, dashed, draw opacity=0.5] (a) at (1,-0.7){\small{$\ED_{k,j}$}}; 
        \end{tikzpicture}
    \end{aligned} 
\end{equation}
the $0d$ defect is labeled by the functor category $\Fun_{\ED_{i,i}|\ED_{j,j}}(\ED_{i,j},\ED_{i,j})$.
Similar to \eqref{eq:DefectFace}, we can introduce face operators and vertex operators in the vicinity of the defects.
We see that different bulks, domain walls and boundaries are labeled with different components of the input UMFC $\ED$; this kind of model will be called a \emph{restricted multifusion string-net}.
The above discussion implies that: \emph{a restricted multifusion string-net gives  a defective string-net}.
This can be summarized as follows:

\begin{proposition}
For a string-net model on a disk with boundary and boundary defects, the input data can be determined by a multifusion category $\ED=\oplus_{i,j\in I}\ED_{i,j}$ whose diagonal components $\ED_{i,i}$ are UFCs that satisfy $\mathcal{Z}(\ED_{i,i})\simeq \mathcal{Z}(\ED_{j,j})$ for all $i,j\in I$. And we have:
\begin{enumerate}
    \item The bulk input data is $\ED_{0,0}$ and the bulk topological phase is given by the Drinfeld center $\mathcal{Z}(\ED_{0,0})$.
    \item The input data of type-$i$ boundary is $\ED_{0,i}$ and the corresponding boundary excitation is characterized by $\ED_{i,i}$.
    \item The boundary defect between type-$i$ and type-$j$ boundaries (the $i|j$-boundary defects\,\footnote{Since there is a given orientation of the boundary, the $i|j$-defect is different from the $j|i$-defect.}) are described by $\ED_{i,j}$.
\end{enumerate}
\end{proposition}

Notice that the space spanned by restricted string-net configurations is a subspace of the corresponding multifusion string-net, where there is no restriction of the string labels for different bulks and domain walls or boundaries, meaning that they all are chosen arbitrarily from $\Irr(\ED)$.

For the reverse direction, we need to consider how to represent a given defective string-net as a restricted multifusion string-net.
This will be illustrated in the next subsection.

\subsection{Example: defective Levin-Wen string-net as multifusion string-net}
\label{sec:LWbd}

To understand why all defective string-net models can be regarded as restricted multifusion string-nets, let us consider the Levin-Wen string-net model as an example. We will demonstrate that \emph{any defective Levin-Wen string-net can be naturally interpreted as a multifusion string-net}, as previously briefly outlined in Refs.~\cite{Kitaev2012boundary,kong2012universal}.
In the context of a defective Levin-Wen string-net, we will focus on boundaries and boundary defects, as illustrated in Fig.~\ref{fig:StringNetBd}. This formulation proves to be sufficiently general, as a domain wall can be seamlessly transformed into a boundary, and domain wall defects can be equivalently transformed into boundary defects.
We will take the disk as an example; the generalization to a more general case is straightforward.

Recall that, in the Kitaev-Kong construction, for a string-net model with the bulk UFC $\EC$, the input data of the gapped boundary is an indecomposable $\EC$-module category $\EM$ \cite{Kitaev2012boundary}.
In the construction of a string-net with a boundary, careful consideration is required for handling both the boundary data in the module category $\EM$ and the bulk data in $\EC$. The labels for boundary edges are represented by simple objects denoted as $a \in \Irr(\EM)$, where no duality structure is inherently present.
To elucidate the meaning of $\bar{a}$, it is necessary to consider $\bar{a}$ as an object in $\EM^{\rm op}$, where $\EM^{\rm op}$ is the category that shares the same objects as $\EM$ but with reversed arrows. This prompts the question of how to define $b \otimes \bar{a}$ and $\bar{a} \otimes b$.
It turns out that $b \otimes \bar{a}$ can be regarded as an object in $\EC$ and $\bar{a}\otimes a$ can be regarded as an object in $\EC^{\vee}_{\EM}:=\Fun_{\EC}(\EM,\EM)$ \cite{Kitaev2012boundary}.
The boundary defects (excitations) are characterized by the $\EC$-module functors between $\EM$ and $\EM$.
For string-net evaluations involving boundaries to be well-defined, we are tasked with embedding both $\EC$ and $\EM$ into a multifusion category.

It is convenient to denote $\EM_0=\EC$ and $\EM_1=\EM$, which are two left $\EC$-module categories.
We can embed $\EC$ and $\EM$ into a multifusion category denoted as $\EC^{\mathrm{def}}_{\EM}:=\EC_{0,0}\oplus \EC_{0,1}\oplus \EC_{1,0}\oplus \EC_{1,1}$.
This expanded category can be understood as a direct sum of four components: 
\begin{enumerate}
    \item UFC $\EC_{0,0}:=\EC$. Notice that we can regard $\EC$ as a category of $\EC$-module functors since we have the equivalence $\EC \simeq \Fun_{\EC}(\EC,\EC)=\Fun_{\EC}(\EM_0,\EM_0)$.
    \item Left $\EC$-module category $\EC_{0,1}:=\EM$.
    \item $\EC_{1,0}:=\EM^{\rm op}$ (the opposite category of $\EM$, which has the same objects as $\EM$ but with reversed morphisms).
    \item $\EC_{1,1}:=\EC_{\EM}^{\vee}$, which encapsulates the category of $\EC$-module functors mapping $\EM$ onto itself and inherently serves as a UFC.
\end{enumerate}

The duality and fusion operation are closed for $\EC$ and $\EC_{\EM}^{\vee}$. The dual object of  $m\in\EM$ in $\EC_\EM^{\rm def}$ can be regarded as $\bar{m}\in \EM^{\rm op}$ which is the same object as $m$ but all arrows are reversed.
Fusion between objects in $\EM$ and $\EM^{\rm op}$ are defined as
\begin{equation}
m\otimes \bar{n}=\bigoplus_{i\in \operatorname{Irr}(\EC)}N_{in}^{m}i\in \EC.
\end{equation}
Fusion between $i\in \EC$ and $m\in\EM$ (resp.~$\bar{n}\in\EM^{\rm op}$) is given by the left (resp.~right) $\EC$-module category structure: $i\otimes m\in \EM$ (resp.~$\bar{n}\otimes i\in \EM^{\rm op}$).
We also need the fusion between $m\in\EM$ and $\tau\in\EC^{\vee}_{\EM}$: Since $\tau$ is a $\EC$-module functor from $\EM$ to itself, we have $m\otimes \tau =\tau(m)\in \EM$. The last one is the fusion of $\bar{m}\in \EM^{\rm op}$ and $n\in \EM$. We can define $\bar{m}\otimes n$ as a $\EC$-module functor by taking each $x\in \EM$ to $(x\otimes \bar{m})\otimes n$, thus $\bar{m}\otimes n\in \EC_{\EM}^{\vee}$.
In summary, we have a multifusion category
\begin{equation}\EC^{\mathrm{def}}_{\EM}=\left(
\begin{array}{cc}
\EC_{00}& \EC_{01}\\
\EC_{10} &\EC_{11}
\end{array}\right)
\end{equation}
with fusion defined as 
\begin{equation}\left(
\begin{array}{cc}
x_{00}& x_{01}\\
x_{10} &x_{11}
\end{array}\right)\otimes \left(
\begin{array}{cc}
y_{00}& y_{01}\\
y_{10} &y_{11}
\end{array}\right)=\left(
\begin{array}{cc}
(x_{00}\otimes y_{00})\oplus (x_{01}\otimes y_{10}) & (x_{00}\otimes y_{01})\oplus (x_{01}\otimes y_{11})\\
(x_{10}\otimes y_{00})\oplus (x_{11}\otimes y_{10}) &(x_{10}\otimes y_{01})\oplus (x_{11}\otimes y_{11}) 
\end{array}\right). \label{eq:mfusion}
\end{equation}
The tensor unit is $\one_{\EC}\oplus \one_{\EC_{\EM}^{\vee}}=\left(\begin{array}{cc} \one_{\EC} & \\ &\one_{\EC_{\EM}^{\vee}}\end{array}\right)$.
Now, we have all of the data we need: $\{N_{ab}^c, \delta_{abc}, d_a, {N'}_{am}^n, \delta'_{amn},d_m\}$, which determine the model.
Thus, to summarize, we see that a defective Levin-Wen string-net can be regarded as a multifusion string-net, for which the bulk is labeled by a diagonal component $\EC_{0,0}$, the boundary is labeled by the category $\EC_{0,1}$ (or $\EC_{1,0}$) which is in the same row (or column) as $\EC_{0,0}$. The boundary defects are characterized by another diagonal component $\EC_{1,1}$. The construction can be straightforwardly generalized to the case of $\Sigma=\Sigma_{g,l}$ with $l$-different boundaries.

For a string-net model with bulk determined by a UFC $\EC$, there are two types of defects: gapped boundaries (codimension~1) and boundary defects (codimension~2).
The gapped boundaries are classified by equivalence classes of $\EC$-module categories $\EM_0,\cdots,\EM_n$, where $\EM_0=\EC$ is the smooth boundary and $\EM_n=\Vect$ is the rough boundary. Note that there exist smooth and rough boundaries for any string-net model.

\begin{figure}[t]
		\centering
		\includegraphics[width=12cm]{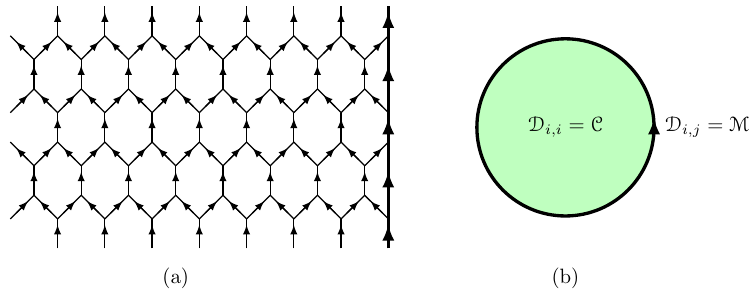}
		\caption{(a) A trivalent lattice in the vicinity of the boundary (bold arrows). (b) A disk with the bulk labeled by $\ED_{i,i}=\EC$ and the boundary labeled by $\ED_{i,j}=\EM$. \label{fig:StringNetBd}} 
\end{figure} 

The above observation can be generalized to the general case and we have the following result:

\begin{proposition}
    Any defective Levin-Wen string-net model can be regarded as a restricted multifusion string-net. 
    For the Levin-Wen string-net model with bulk input UFC $\EC$, the gapped boundaries are classified by equivalence classes of $\EC$-module categories $\EM_0,\cdots,\EM_n$, where $\EM_0=\EC$ is the smooth boundary and $\EM_n=\Vect$ is the rough boundary. 
    The defective Levin-Wen string-net can be regarded as a multifusion string-net with input UMFC as $\EC_{\EM_0,\cdots,\EM_n}^{\rm def}$, where $\EC_{0,0}=\EC \simeq \mathsf{Fun}_{\EC}(\mathsf{Vect},\mathsf{Vect})$, $\EC_{i,j}=\mathsf{Fun}_{\EC}(\EM_i,\EM_j)$. 
\end{proposition}

\begin{proof}
    It is proved in Ref.~\cite{etingof2005fusion} that, for a given UFC $\EC$ and an indecomposable $\EC$-module category $\EM$, the category $\EC^{\mathrm{def}}_{\EM}$ is a unitary multifusion category.
\end{proof}

\begin{table}[t]
\centering \small 
\begin{tabular} {|l|c|c|c|} 
\hline
   & UMFC string-net & Kitaev-Kong string-net  & weak Hopf QD  \\ \hline
 Bulk  &  $\ED_{i,i}$ & $\EC=\ED_{i,i}=\mathsf{Rep}(W)$ & WHA $W$ \\ \hline
 Bulk phase  & 
 $\mathcal{Z}(\ED_{i,i})$ & $\mathcal{Z}(\EC)\simeq \mathsf{Fun}_{\EC|\EC}(\EC,\EC) $ & $\mathsf{Rep}(D(W))$  \\ \hline
 Boundary & $\ED_{i,j}$ & $\EM\simeq \ED_{i,j}\simeq   \mathsf{Mod}_{\mathfrak{A}}$ & $W$-comodule algebra $\mathfrak{A}$ \\\hline
 Boundary phase & $\ED_{j,j}$ & $\mathsf{Fun}_{\EC}(\EM,\EM)$ & ${^W_{\mathfrak{A}}}\mathsf{Mod}_{\mathfrak{A}}$ \\\hline
Boundary defect & $\ED_{j,k}$ & $\mathsf{Fun}_{\EC}(\EM,\EN)$  & 
${^W_{\mathfrak{B}}}\mathsf{Mod}_{\mathfrak{A}}$ \\\hline
\end{tabular}
\caption{The dictionary between $W$-comodule algebra description of weak Hopf quantum double boundary and string-net boundary.\label{tab:bdTop2}}
\end{table}

\begin{remark}
The above discussion also applies to the defective multifusion string-net model.
\end{remark}

Conversely, given a UMFC, it is also possible to construct a defective Levin-Wen string-net as we have pointed out before.
Suppose that we have a UMFC $\ED=\oplus_{i,j\in I}\ED_{i,j}$ with $|\Irr(\ED_{i,j})|=n_{i,j}$. Each diagonal component has a distinguished object $\one_i\in \ED_{i,i}$, and we have $\one_i\otimes X_{i,j}\cong X_{i,j}$ and $X_{i,j}\otimes \one_j\cong X_{i,j}$ for all $X_{i,j}\in \ED_{i,j}$. The tensor unit of $\ED$ is the direct sum of these objects $\one=\oplus_{i\in I}\one_i$.
The bulk's input data can be chosen as an arbitrary diagonal component $\ED_{i,i}$.
The boundary input data now can be chosen as some component $\ED_{i,j}$ which is a left $\ED_{i,i}$-module category. Since the dual object of $X_{i,j}\in \ED_{i,j}$ is $Y_{j,i}\in\ED_{j,i}$, the complete data for the string-net model with a single boundary (i.e., on a disk) is the following sub-UMFC of $\ED$
\begin{equation}\ED[i;j]=\left(
\begin{array}{cc}
\ED_{i,i}& \ED_{i,j}\\
\ED_{j,i} &\ED_{j,j}
\end{array}\right)
\end{equation}
with fusion defined as 
\begin{align}
&\left(
\begin{array}{cc}
X_{i,i}& X_{i,j}\\
X_{j,i} &X_{j,j}
\end{array}\right)\otimes \left(
\begin{array}{cc}
Y_{i,i}& Y_{i,j}\\
Y_{j,i} &Y_{j,j}
\end{array}\right)\nonumber \\
= &\left(
\begin{array}{cc}
(X_{i,i}\otimes Y_{i,i})\oplus (X_{i,j}\otimes Y_{j,i}) & (X_{i,i}\otimes Y_{i,j})\oplus (X_{i,j}\otimes Y_{j,j})\\
(X_{j,i}\otimes Y_{i,i})\oplus (X_{j,j}\otimes Y_{j,i}) &(X_{j,i}\otimes Y_{i,j})\oplus (X_{j,j}\otimes Y_{j,j}) 
\end{array}\right). \label{eq:mfusionD}
\end{align}
Notice that for a fixed bulk $\ED_{i,i}$, the UMFC gives $|I|$ different boundaries which are characterized by $\ED[i,j]$ with $j\in I$. 
There always exists a special case that $\ED[i;i]=\ED_{i,i}$, called smooth boundary, where both the bulk and boundary are characterized by the UFC $\ED_{i,i}$.

\begin{proposition}
	For a string-net with gapped boundary characterized by $\ED[i;j]$, where the bulk label is taken from $\ED_{i,i}$, and the boundary label is taken from $\ED_{i,j}$:
\begin{enumerate}
		\item The bulk phase is given by the Drinfeld center $\mathcal{Z}(\ED_{i,i})$ of $\ED_{i,i}$.
		\item The boundary phase is given by the $\ED_{j,j}$.
\end{enumerate}
\end{proposition}

\begin{proof}
This is evident from our previous discussion.
\end{proof}


\section{Conclusion and discussion}
\label{sec:conclusion}

In this work, we introduced the generalized multifusion string-net model and extensively discussed its macroscopic and microscopic properties.
Despite the progress that has been made, there remain several unresolved issues that warrant further investigation:

(i) A lattice realization of the correspondence between the multifusion string-net model and weak Hopf lattice gauge theory.
While for the Hopf case, the lattice realization of the equivalence is established in Refs.~\cite{Buerschaper2009mapping,buerschaper2013electric,jia2023boundary}, the problem for weak Hopf gauge symmetry largely remains open.

(ii) Applications in symmetry-enriched topological (SET) phases.
It has been noted in Ref.~\cite{chang2015enriching} that the multifusion string-net model provides a suitable framework for studying the SET phase.
A systematic and thorough investigation in this direction is still necessary.

(iii) The entanglement properties of the multifusion string-net are also of crucial importance. We know that the topological entanglement entropy is related to the quantum dimension of the topological excitations, and the entanglement entropy is sensitive to the defects of the model. Examining the entanglement features of the multifusion string-net model can provide us with deeper insights into its topological properties.

These problems will be the focus of investigation in our future endeavors.

\subsection*{Acknowledgements}
Z.J. sincerely thanks Liang Kong for sharing his valuable insights on the boundary theory of the Levin-Wen string-net model. Z.J. also acknowledges Quan Chen for bringing several references on DHR-modules to his attention. 
S.T. would like to express gratitude to Zhengwei Liu, Shuang Ming, Yilong Wang, Jinsong Wu and Sebastien Palcoux for their many helpful conversations.  
All authors are grateful for the referee's valuable suggestions. 
Z.J. and D.K. are supported by the National
Research Foundation in Singapore and A*STAR under
its CQT Bridging Grant.
S.T. is supported by postdoc fund from BIMSA.




\bibliographystyle{apsrev4-1-title}
\bibliography{mybib}

\end{document}